\documentclass[12pt,british]{book}
\usepackage{mathpazo}
\usepackage[scaled=0.95]{helvet}
\usepackage{courier}
\usepackage[T1]{fontenc}
\usepackage[latin9]{inputenc}
\usepackage{color}
\usepackage{babel}
\usepackage{pdfpages}
\usepackage{amscd}
\usepackage{amsmath}
\usepackage{amsthm}
\usepackage{amssymb}
\usepackage{cancel}
\usepackage{graphicx}
\usepackage[unicode=true,pdfusetitle,
 bookmarks=true,bookmarksnumbered=false,bookmarksopen=false,
 breaklinks=true,pdfborder={0 0 0},pdfborderstyle={},backref=false,colorlinks=true]
 {hyperref}
\usepackage{breakurl}

\makeatletter

\numberwithin{equation}{section}
\newcommand{\lyxrightaddress}[1]{
	\par {\raggedleft \begin{tabular}{l}\ignorespaces
	#1
	\end{tabular}
	\vspace{1.4em}
	\par}
}
\theoremstyle{plain}
\newtheorem{thm}{\protect\theoremname}[section]
\theoremstyle{definition}
\newtheorem{defn}[thm]{\protect\definitionname}
\theoremstyle{definition}
\newtheorem{example}[thm]{\protect\examplename}
\theoremstyle{plain}
\newtheorem{assumption}[thm]{\protect\assumptionname}
\theoremstyle{remark}
\newtheorem{rem}[thm]{\protect\remarkname}
\theoremstyle{plain}
\newtheorem{lem}[thm]{\protect\lemmaname}
\theoremstyle{plain}
\newtheorem{prop}[thm]{\protect\propositionname}
\theoremstyle{plain}
\newtheorem{cor}[thm]{\protect\corollaryname}
\theoremstyle{plain}
\newtheorem{ax}[thm]{\protect\axiomname}
\theoremstyle{definition}
\newtheorem{condition}[thm]{\protect\conditionname}
\theoremstyle{definition}
\newtheorem*{defn*}{\protect\definitionname}
\theoremstyle{plain}
\newtheorem*{conjecture*}{\protect\conjecturename}

\@ifundefined{date}{}{\date{}}
\usepackage{cite}
\makeatother
\usepackage[matrix,frame,arrow]{xy}
\usepackage{amsmath}
\newcommand{\bra}[1]{\left\langle{#1}\right\vert}
\newcommand{\ket}[1]{\left\vert{#1}\right\rangle}
\newcommand{\qw}[1][-1]{\ar @{-} [0,#1]}

\newcommand{\gate}[1]{*{\xy *+<.6em>{#1};p\save+LU;+RU **\dir{-}\restore\save+RU;+RD **\dir{-}\restore\save+RD;+LD **\dir{-}\restore\POS+LD;+LU **\dir{-}\endxy} \qw}

\newcommand{\measureD}[1]{*{\xy*+=+<.5em>{\vphantom{\rule{0em}{.1em}#1}}*\cir{r_l};p\save*!R{#1} \restore\save+UC;+UC-<.5em,0em>*!R{\hphantom{#1}}+L **\dir{-} \restore\save+DC;+DC-<.5em,0em>*!R{\hphantom{#1}}+L **\dir{-} \restore\POS+UC-<.5em,0em>*!R{\hphantom{#1}}+L;+DC-<.5em,0em>*!R{\hphantom{#1}}+L **\dir{-} \endxy} \qw}

\newcommand{\multimeasureD}[2]{*+<1em,.9em>{\hphantom{#2}}\save[0,0].[#1,0];p\save !C *{#2},p+LU+<0em,0em>;+RU+<-.8em,0em> **\dir{-}\restore\save +LD;+LU **\dir{-}\restore\save +LD;+RD-<.8em,0em> **\dir{-} \restore\save +RD+<0em,.8em>;+RU-<0em,.8em> **\dir{-} \restore \POS !UR*!UR{\cir<.9em>{r_d}};!DR*!DR{\cir<.9em>{d_l}}\restore \qw}

\newcommand{\multigate}[2]{*+<1em,.9em>{\hphantom{#2}} \qw \POS[0,0].[#1,0];p !C *{#2},p \save+LU;+RU **\dir{-}\restore\save+RU;+RD **\dir{-}\restore\save+RD;+LD **\dir{-}\restore\save+LD;+LU **\dir{-}\restore}
\newcommand{\ghost}[1]{*+<1em,.9em>{\hphantom{#1}} \qw}
\newcommand{\Qcircuit}[1][0em]{\xymatrix @*=<#1>} 

\newcommand{\pureghost}[1]{*+<1em,.9em>{\hphantom{#1}}}
\newcommand{\multiprepareC}[2]{*+<1em,.9em>{\hphantom{#2}}\save[0,0].[#1,0];p\save !C
  *{#2},p+RU+<0em,0em>;+LU+<+.8em,0em> **\dir{-}\restore\save +RD;+RU **\dir{-}\restore\save
  +RD;+LD+<.8em,0em> **\dir{-} \restore\save +LD+<0em,.8em>;+LU-<0em,.8em> **\dir{-} \restore \POS
  !UL*!UL{\cir<.9em>{u_r}};!DL*!DL{\cir<.9em>{l_u}}\restore}
\newcommand{\prepareC}[1]{*{\xy*+=+<.5em>{\vphantom{#1\rule{0em}{.1em}}}*\cir{l^r};p\save*!L{#1} \restore\save+UC;+UC+<.5em,0em>*!L{\hphantom{#1}}+R **\dir{-} \restore\save+DC;+DC+<.5em,0em>*!L{\hphantom{#1}}+R **\dir{-} \restore\POS+UC+<.5em,0em>*!L{\hphantom{#1}}+R;+DC+<.5em,0em>*!L{\hphantom{#1}}+R **\dir{-} \endxy}}
\newcommand{\poloFantasmaCn}[1]{{{}^{#1}_{\phantom{#1}}}}

\makeatletter
\newcommand{\rA}{\mathrm{A}}
\newcommand{\rB}{\mathrm{B}}
\newcommand{\rC}{\mathrm{C}}
\newcommand{\rD}{\mathrm{D}}
\newcommand{\rE}{\mathrm{E}}

\newcommand{\rI}{\mathrm{I}}
\newcommand{\rM}{\mathrm{M}}
\newcommand{\rS}{\mathrm{S}}
\newcommand{\rX}{\mathrm{X}}

\newcommand{\cA}{\mathcal{A}}
\newcommand{\cB}{\mathcal{B}}
\newcommand{\cC}{\mathcal{C}}
\newcommand{\cD}{\mathcal{D}}
\newcommand{\cI}{\mathcal{I}}

\newcommand{\cT}{\mathcal{T}}
\newcommand{\cU}{\mathcal{U}}
\newcommand{\cV}{\mathcal{V}}

\@ifundefined{showcaptionsetup}{}{%
 \PassOptionsToPackage{caption=false}{subfig}}
\usepackage{subfig}
\makeatother

\providecommand{\assumptionname}{Assumption}
\providecommand{\axiomname}{Axiom}
\providecommand{\conditionname}{Condition}
\providecommand{\conjecturename}{Conjecture}
\providecommand{\corollaryname}{Corollary}
\providecommand{\definitionname}{Definition}
\providecommand{\examplename}{Example}
\providecommand{\lemmaname}{Lemma}
\providecommand{\propositionname}{Proposition}
\providecommand{\remarkname}{Remark}
\providecommand{\theoremname}{Theorem}

\begin{document}
\includepdf{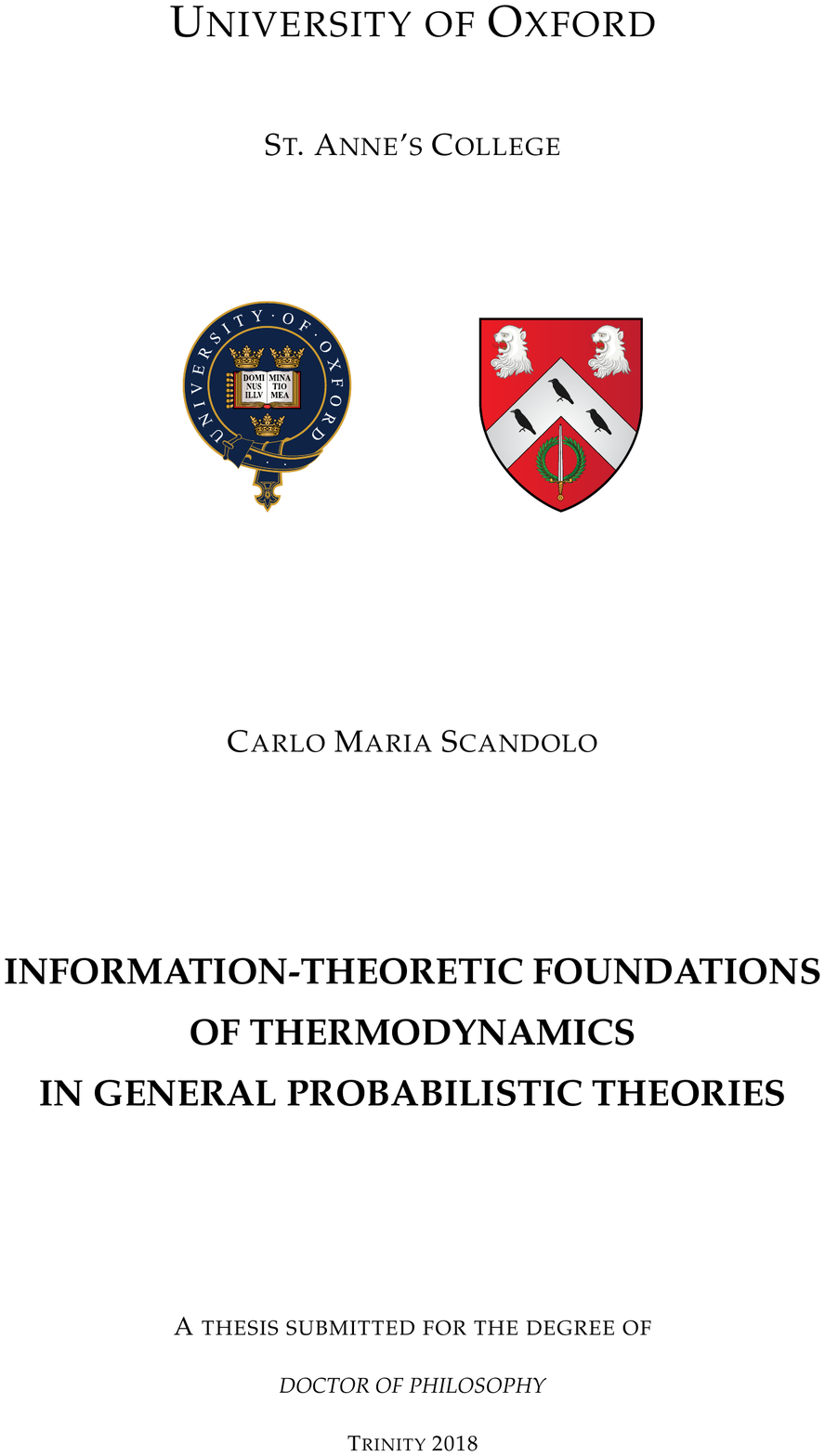}

\lyxrightaddress{To my parents, Antonio and Silvana}

\tableofcontents{}

\chapter*{\markboth{ABSTRACT}{ABSTRACT}Abstract}

In this thesis we study the informational underpinnings of thermodynamics
and statistical mechanics. To this purpose, we use an abstract framework\textemdash general
probabilistic theories\textemdash , capable of describing arbitrary
physical theories, which allows one to abstract the informational
content of a theory from the concrete details of its formalism. In
this framework, we extend the treatment of microcanonical thermodynamics,
namely the thermodynamics of systems with a well-defined energy, beyond
the known cases of classical and quantum theory. We formulate two
requirements a theory should satisfy to have a well-defined microcanonical
thermodynamics. We adopt the recent approach of resource theories,
where one studies the transitions between states that can be accomplished
with a restricted set of physical operations. We formulate three different
resource theories, differing in the choice of the restricted set of
physical operations.

To bridge the gap between the objective dynamics of particles and
the subjective world of probabilities, one of the core issues in the
foundations of statistical mechanics, we propose four information-theoretic
axioms. They are satisfied by quantum theory and more exotic alternatives,
including a suitable extension of classical theory where classical
systems interact with each other creating entangled states. The axioms
identify a class of theories where every mixed state can be modelled
as the reduced state of a pure entangled state. In these theories
it is possible to introduce well-behaved notions of majorisation,
entropy, and Gibbs states, allowing for an information-theoretic derivation
of Landauer's principle. The three resource theories define the same
notion of resource if and only if, on top of the four axioms, the
dynamics of the underlying theory satisfy a condition called ``unrestricted
reversibility''. Under this condition we derive a duality between
microcanonical thermodynamics and pure bipartite entanglement.

\chapter{Introduction}

Thermodynamics is a powerful phenomenological paradigm encompassing
several scientific disciplines, from physics to chemistry, up to biology
and engineering. Its principles form a framework every experimental
observation must adhere to. Its two most important laws express energy
conservation (the first law), and the existence of irreversible processes,
or, loosely speaking, of an arrow of time (the second law). Whilst
it is fairly easy and reasonable to accept a principle such as energy
conservation, the second law and its consequent arrow of time have
caused a great bewilderment among scientists and philosophers, since
it was difficult to find an explanation for the origin of irreversibility.

It was also necessary to find a place in the structure of physics
for the new concepts introduced by thermodynamics: work, heat, temperature,
etc. Were they fundamental or could they be derived from other concepts?
Several theories of thermodynamics were proposed. The first leaned
towards the view that heat and temperature are primitive notions,
but later, the classic works by Maxwell \cite{Maxwell-1,Maxwell-2},
Boltzmann \cite{Boltzmann}, and Gibbs \cite{Gibbs} undertook a reduction
of the laws of thermodynamics to the laws of the underlying dynamics
of particles and fields. This reduction led to the establishment of
statistical mechanics as the standard paradigm for the foundations
for thermodynamics \cite{Callen,Huang,Kardar1,Kardar2}. This has
worked even for quantum systems, where quantum statistical mechanics
was able to predict new, genuinely quantum, phenomena, such as Bose-Einstein
condensation, later observed in a laboratory \cite{BEC}, and related
to other important phenomena: superfluidity and superconductivity
\cite{Superfluidity}. For this reason, in the following we will use
the terms ``statistical mechanics'' and ``thermodynamics'' nearly
as synonyms.

However, the statistical paradigm in turn led to novel questions,
the central one now being how to reconcile the use of statistical
notions (mixed states), associated with the incomplete knowledge of
an agent, with the picture of Nature provided by classical and quantum
mechanics, where the fundamental dynamics are deterministic (and reversible).
Different proposals have been made for classical statistical mechanics,
the best known of which are ergodic theory \cite{Birkhoff-ergodic,vonNeumann-ergodic,Ergodic1,Ergodic2,Fasano-Marmi},
and Jaynes' maximum entropy approach \cite{Jaynes1}, then extended
to quantum statistical mechanics \cite{Jaynes2}. Deffner and Zurek
refer to these attempts as to a 
\begin{quotation}
```half-way' house, populated by fictitious but useful concepts such
as \emph{ensembles}'' \cite{Zurek}.
\end{quotation}
Quantum theory, instead, offers a \emph{radically new} opportunity.
As originally noted by Schrödinger \cite{Schrodinger}, a system and
its environment can be jointly in a pure state, whilst the system
is individually in a mixed state. Here the mixed state does \emph{not}
represent an ensemble of identical systems, but rather the state of
a \emph{single} quantum system. Based on this idea, Popescu, Short,
and Winter \cite{Popescu-Short-Winter}, and Goldstein, Lebowitz,
Tumulka, and Zanghì \cite{Canonical-typicality} proposed that entanglement
could be the starting point for a new, genuinely quantum foundation
of statistical mechanics. The idea was that, when the environment
is large enough, the system is approximately in the equilibrium state
for the typical joint pure states of the system and the environment.
This idea has been explored in a variety of settings \cite{Bocchieri,Seth-Lloyd,Lubkin,Gemmer-Otte-Mahler,Mahler-book,Concentration-measure,BrandaoQIP2015,Thermo-recent},
whose common inspiration is the idea that quantum entanglement can
provide a new foundation for statistical mechanics, and ultimately,
thermodynamics. Furthermore, even irreversibility can be explained
because some degrees of freedom are traced out.

The success of the statistical mechanical paradigm is deeply tied
to the fact that the physical systems under investigation are composed
of an enormously large number of particles, which guarantees the applicability
of statistical methods. However, more recently, the scope of thermodynamics
in the quantum regime has been extended from quantum gases to microscopic
systems far from the thermodynamic limit, so crucial for the development
of nanotechnology \cite{delRio,Xuereb,Anders-thermo}. In this new
regime one studies the thermodynamic transformations and the fluctuations
of quantum systems with very few particles, a scenario often called
the \emph{single-shot} regime. Clearly one cannot use the standard
tools of statistical mechanics, and a new way to address this new
regime is to adopt a resource-theoretic approach \cite{Quantum-resource-1,Quantum-resource-2,Gour-review},
where one starts from a subset of quantum operations that are ``free''
or ``easy to implement'', and characterises the transitions that
can be accomplished by these free operations. This establishes a preorder
on quantum states based on their value as thermodynamic resources,
from the most valuable to the least valuable (the so-called ``free
states''), which correspond to equilibrium states. This emergence
of thermalisation through the repeated applications of some transformations
is the reason why resource theories have been so successful in the
study of quantum thermodynamics \cite{delRio,Xuereb,Anders-thermo}.
In this approach thermodynamic potentials often emerge as functions
assigning a value to resources compatibly with the resource preorder.
Many recent results in quantum thermodynamics have been obtained in
this way \cite{Horodecki-Oppenheim,Nicole,Athermality1,Horodecki-Oppenheim-2,Athermality2,2ndlaws,Gibbs-preserving-maps,quantum2ndlaw,Lostaglio-Jennings-Rudolph,Lostaglio-coherence,Nicole-beyond,Nicole-non-commuting,Non-commuting-Bristol,David-non-commuting,3rdlaw,3rd-law-Wilming,Sparaciari,Scharlau,Boes,Multiresource}.

The approach to quantum thermodynamics based on resource theories
uses a lot of concepts and techniques from (quantum) information theory.
This should not surprise, as from the early development of statistical
mechanics it became clear that thermodynamic concepts are intimately
tied to information-theoretic ones, as shown by the paradigmatic examples
of Maxwell's demon \cite{Maxwelldemon2}, and the closely related
Szilard engine \cite{Szilard}. In these examples, the knowledge possessed
by a (microscopic) observer was used to set up a physical process
violating the second law of thermodynamics. This was clearly paradoxical,
and called for an explanation in order to reaffirm the validity of
the famed second law. Since the paradoxes were all based on the information
possessed by an observer, if a solution was to be found, it would
involve some information-theoretic concepts. The correct solution
came many years later, when Landauer found out something very surprising:
the act of erasing and overwriting the memory of a computing device
has a physical effect \cite{Landauer}. More precisely, if we erase
an unknown bit at temperature $T$, there is an associated heat dissipation
of $kT\ln2$, where $k$ is Boltzmann constant. This opened the way
to a solution of the paradoxes, which was found by Bennett \cite{Bennett-Maxwell,Bennett2003}.
He understood that if an observer is to act on a physical system based
on the result of their previous observation, they have to \emph{store}
their observation somewhere. Since infinite storage does not exist,
at some point they will have to erase their memory, spending energy
to do it, according to Landauer's result. In conclusion, microscopic
observers such as Maxwell's demon \emph{cannot} be used for cyclic
work extraction at no additional cost, because their usage would involve
some energy dissipation at some stage of the protocol.

In this thesis, in accordance with the recent trend in theoretical
physics of grounding physics on information theory, we study the information-theoretic
foundations of thermodynamics and statistical mechanics. However,
we do it in a different way from before: instead of analysing them
in quantum theory, we do it in arbitrary physical theories. This is
a novel area of research, and this DPhil thesis is the first doctoral
thesis on this topic.

If quantum theory is the ultimate theory of Nature, at least in the
microscopic domain, why do we bother to study the foundations of thermodynamics
in general physical theories? There are manifold answers. One is that
one of the major trends in theoretical physics has always been to
generalise known results and broaden their scope, therefore it is
natural to study thermodynamics in its full generality. This may seem
a truly ambitious and hard feat. However this is precisely in the
spirit of thermodynamics: being a very general phenomenological paradigm,
it should be theory-independent in its essence, therefore it should
be possible to define and study it in abstract terms. Another answer
is that some features of quantum theory, and by extension, of quantum
thermodynamics, are best understood when looked at them ``from the
outside''. Contrasting quantum behaviour with the behaviour one observes
in general theories can provide a new insights into why the world
is quantum, this time from a thermodynamic angle. Another reason is
that, working in an abstract way, we can capture the information-theoretic
essence of thermodynamics, without ``being distracted'' by the concrete
details of the formalism of a specific physical theory. 

Clearly, the first thing we need is a theory-independent way to address
physical theories, an abstract framework that allows us to describe
all their common traits, without plunging deep into the details of
their specific formalisms. Fortunately, such a framework exists, and
it is that of \emph{general probabilistic theories} \cite{Hardy-informational-1,Barrett,Barnum-1,Janotta-Hinrichsen,chiribella2016quantum,Barnum2016},
which identifies the two main ingredients of any physical theory to
be its \emph{compositional structure} (how to build experiments) and
its \emph{probabilistic structure} (how to assign probabilities to
experimental observations). This is even more obvious in the variant
of the formalism known as \emph{operational probabilistic theories}
\cite{Chiribella-purification,Chiribella-informational,hardy2011,Hardy-informational-2,Chiribella14,QuantumFromPrinciples,hardy2013,chiribella2017quantum},
arisen from the marriage of the graphical language of category theory
\cite{Abramsky2004,Coecke-Kindergarten,Coecke-Picturalism,Selinger,Coecke2016,Coecke2017picturing}
with probability theory. As opposed to the original framework for
general probabilistic theories, based on the convex geometry of states
\cite{Barrett,Barnum-2,Wilce-formalism,Janotta-Hinrichsen,Barnum2016},
the focus of operational probabilistic theories is on physical processes
and their composition. This is the approach adopted in this thesis:
since thermodynamics is a theory concerned with processes and transitions,
it is natural to resort to a formalism where processes play centre
stage.

Historically, general probabilistic theories were introduced as a
framework from which to derive quantum theory by imposing suitable
information-theoretic principles, leading to various quantum reconstructions
\cite{Hardy-informational-1,Chiribella-informational,QuantumFromPrinciples,Hardy-informational-2,Brukner,Masanes-physical-derivation,Wilce-4.5,Masanes+all,Barnum-interference,hardy2013,Dakic2016,Muller2016}.
However, their scope is broader than just this: besides helping us
gain an operational understanding of quantum theory and why Nature
is quantum, general probabilistic theories are important also for
studying extensions and restrictions of quantum structures \cite{Scandolo14}.
Indeed, several proposals for a theory of quantum gravity have called
for a modification of the quantum laws to a more general form (see
e.g.\ \cite{modHeisenberg}). On the other side, sometimes one considers
sub-theories of quantum theory, arising for instance from an experimental
limitation on the states or operations one can implement in a laboratory.
In this case, what is the resulting theory like? To give an example,
Bartlett, Rudolph, and Spekkens studied Gaussian quantum theory, and
found out that it admits a semiclassical explanation as an epistemically
restricted theory \cite{Spekkens-Gaussian,Spekkens-epistricted}.

In our quest for the foundations of thermodynamics, we will be inspired
by the recent results obtained in quantum thermodynamics for microscopic
systems. Therefore we adopt the resource-theoretic approach, which
was shown to be not at all specific to quantum theory, but rather
applicable to a broad range of theories and situations \cite{Resource-theories,Resource-monoid,Resource-knowledge,Resource-currencies,EastThesis,East-article}.

We will find out that not all physical theories, which in principle
may be extremely counter-intuitive, are suitable to support a sensible
thermodynamics. Therefore we have to introduce some axioms in order
to restrict ourselves to thermodynamically relevant theories. Specifically,
inspired by the results about typicality \cite{Canonical-typicality,Popescu-Short-Winter},
we start from entanglement \cite{Chiribella-purification,Chiribella-Scandolo-entanglement},
turning it into an axiomatic foundation for statistical mechanics.
We explore the hypothesis that the physical systems admitting a well-behaved
statistical mechanics are exactly those where, at least in principle,
mixed states can be modelled as the local states of larger systems,
globally in a pure state \cite{TowardsThermo}. This modelling is
possible in quantum theory, where it provides the stepping stone for
the derivation of the microcanonical and canonical states in \cite{Canonical-typicality,Popescu-Short-Winter,Zurek}.
But the foundational role of entanglement is not limited to quantum
theory. We show that even classical statistical mechanics, where entanglement
is absent, can find a new foundation if classical theory is regarded
as part of a larger physical theory where classical mixed states can
be obtained as marginals of pure states of non-classical composite
systems \cite{TowardsThermo}. Remarkably, the mere fact that classical
systems \emph{could} be entangled with some other physical systems
determines some of their properties, and opens the way to the use
of typicality arguments like in the quantum case. The same approach
is applicable to several extensions of quantum theory, including quantum
theory with superselection rules \cite{Preskill-superselection,Fermionic1,Fermionic2,TowardsThermo,Purity},
and variants of quantum theory with real amplitudes \cite{Stuckelberg,Araki-real,Wootters-real,Hardy-real}.
In this framework, we demand the validity of four information-theoretic
axioms, informally stated as follows:
\begin{description}
\item [{Causality\ \cite{Chiribella-purification}}] No signal can be
sent from the future to the past.
\item [{Purity\ Preservation\ \cite{Scandolo14}}] The composition of
two pure transformations is a pure transformation.
\item [{Pure\ Sharpness\ \cite{QPL15}}] Every system has at least one
pure sharp observable.
\item [{Purification\ \cite{Chiribella-purification}}] Every state can
be modelled as the marginal of a pure state. Such a modelling is unique
up to local reversible transformations.
\end{description}
We call the theories satisfying these axioms \emph{sharp theories
with purification}, a notable example being quantum theory itself.
We show that the validity of above axioms implies that these theories
have some nearly quantum behaviour (e.g.\ the existence of entanglement),
yet they need not be quantum \cite{TowardsThermo,Purity,HOI}. Their
key feature is that they admit a level of description in which all
processes are pure and reversible, and all measurements are sharp.
As such, we believe them to play a really fundamental role in physics,
and quantum theory is an example of this. These axioms enforcing purity
at the fundamental level are also interesting from a thermodynamic
point of view. For example, Causality can be related to the ability
to discard systems, and therefore to restrict ourselves to a smaller
subsystem of a larger system. On the other side, Purification, being
the foundation for all extension and dilation theorems \cite{Chiribella-purification,Chiribella14},
can be thought of as the ability for a thermodynamic observer to enlarge
their system in order to always have an isolated system.

We study the simplest situation: microcanonical thermodynamics, describing
a system with fixed energy, first in arbitrary physical theories,
and then in sharp theories with purification. We formulate two requirements
a theory should satisfy to have a well-defined microcanonical thermodynamics,
and we show that the axioms of sharp theories with purification guarantee
that they are satisfied. The following step is to introduce a resource-theoretic
treatment of thermodynamics in this regime. Clearly, it is natural
to choose the microcanonical state as free, but what about the choice
of free operations? We have essentially three possibilities \cite{Purity}:
\begin{description}
\item [{random\ reversible\ channels}] arising from reversible dynamics
with randomly fluctuating parameters;
\item [{noisy\ operations}] generated by preparing ancillas in the microcanonical
state, turning on a reversible dynamic, and discarding the ancillas;
\item [{unital\ channels}] defined as the processes that preserve the
microcanonical state.
\end{description}
In sharp theories with purification the three sets of operations satisfy
some remarkable inclusion relations like in quantum theory, with random
reversible channels included in the set of noisy operations, and noisy
operations included in the set of unital channels.

We show that the preorder induced by unital channel is completely
characterised by a suitable majorisation criterion \cite{Purity}.
As a consequence, the functions that assign a value to states compatibly
with the preorder, which are measures of mixedness, bear a close resemblance
to entropies; in more mathematical terms they are Schur-concave functions
\cite{Olkin}. In this setting we show that it is possible to put
forward a definition of Shannon-von Neumann entropy with similar properties
to its quantum counterpart, which allows us to prove an operational
version of Landauer's principle \cite{TowardsThermo}.

If we want majorisation to completely characterise the preorder of
all the three resource theories, the physical theory must satisfy
an additional axiom, called ``unrestricted reversibility'', which
comes in three equivalent flavours in sharp theories with purification
\cite{Purity}:
\begin{description}
\item [{Permutability\ \cite{Hardy-informational-2}}] Every permutation
of every maximal set of perfectly distinguishable pure states can
be implemented by a reversible transformation.
\item [{Strong\ Symmetry\ \cite{Muller-self-duality,Barnum-interference}}] For
every two maximal sets of perfectly distinguishable pure states, there
exists a reversible transformation converting the states in one set
into the states in the other.
\item [{Reversible\ controllability\ \cite{Control-reversible}}] For
every pair of systems $\mathrm{A}$ and $\mathrm{B}$, it is possible
to reversibly implement any control-reversible transformation.
\end{description}
When unrestricted reversibility holds, the three resource theories
identify the same notion of resource\textemdash purity\textemdash and
in this case we can prove a duality between purity and pure-state
entanglement \cite{Chiribella-Scandolo-entanglement}.

\paragraph{Published work}

The core of this work is taken from \cite{TowardsThermo,Purity},
with some minor parts from \cite{HOI,Objectivity}. Specifically,
most of the material presented in chapter~\ref{chap:Sharp-theories-with}
comes from \cite{TowardsThermo}, whereas chapter~\ref{chap:Operational-thermodynamics}
contains material from both \cite{TowardsThermo} (the part about
entropies, mixedness monotones, and Landauer's principle) and \cite{Purity}
(the rest).

\paragraph{Structure}

The thesis is structured as follows: in chapter~\ref{chap:General-probabilistic-theories}
we introduce the basic framework of general probabilistic theories,
presented mainly in the operational-probabilistic variant. In the
same chapter we introduce the first axiom, Causality, stating that
information propagates from the past to the future, and we analyse
its consequences. Causality will remain a standing assumption throughout
the rest of the thesis. In chapter~\ref{chap:Resource-theories}
we present the main tool we use to study thermodynamics in general
physical theories, namely resource theories. Sharp theories with purification
are introduced in chapter~\ref{chap:Sharp-theories-with}, where
their general properties are studied in detail. The key thermodynamic
results are exposed in chapter~\ref{chap:Operational-thermodynamics}:
we start from microcanonical thermodynamics, examined in great detail
both in arbitrary physical theories and in sharp theories with purification.
Then we introduce thermal states, by which we obtain an information-theoretic
derivation of Landauer's principle. Finally conclusions are drawn
in chapter~\ref{chap:Conclusions}, with an outlook on further directions
of research.

In this thesis we assume that the reader is already familiar with
the basic framework and terminology of quantum mechanics and elementary
quantum information theory, in particular mixed states and quantum
channels. Good references in this respect are \cite{Nielsen-Chuang,Preskill,Wilde,Watrous}.
The other, more advanced, concepts will be thoroughly explained when
they are introduced.

\section{List of publications and preprints}

The work presented in this thesis contains material from the following
publications and preprints:
\begin{enumerate}
\item G.\ Chiribella, C.\ M.\ Scandolo, \emph{Microcanonical thermodynamics
in general physical theories}, New J.\ Phys.\ \textbf{19} (12),
123043 (2017) \cite{Purity}.
\item G.\ Chiribella, C.\ M.\ Scandolo, \emph{Entanglement as an axiomatic
foundation for statistical mechanics}, arXiv:1608.04459 {[}quant-ph{]}
(2016) \cite{TowardsThermo}.
\item H.\ Barnum, C.\ M.\ Lee, C.\ M.\ Scandolo, J.\ H.\ Selby, \emph{Ruling
out Higher-Order Interference from Purity Principles}, Entropy \textbf{19}
(6), 253 (2017) \cite{HOI}.
\item C.\ M.\ Scandolo, R.\ Salazar, J.\ K.\ Korbicz, P.\ Horodecki,
\emph{Is it possible to be objective in every physical theory?} arXiv:1805.12126
{[}quant-ph{]} (2018) \cite{Objectivity}.
\end{enumerate}

\chapter{General probabilistic theories\label{chap:General-probabilistic-theories}}

In this chapter we present the framework our investigation is conducted
in. Known as ``general probabilistic theories'' (GPTs) \cite{Hardy-informational-1,Barrett,Barnum-1,Janotta-Hinrichsen,chiribella2016quantum,Barnum2016},
it is general enough to accommodate essentially every physical theory,
admitting probabilistic processes. The idea behind it is that a theory
is defined by what an agent can do in a laboratory, and by the observations
they collect, and the predictions they make.

GPTs come in two flavours: one based on convex geometry \cite{Barrett,Barnum-2,Wilce-formalism,Janotta-Hinrichsen},
and the other, more general, based on the compositional structure
of physical theories \cite{Chiribella-purification,Chiribella-informational,hardy2011,Hardy-informational-2,Chiribella14,QuantumFromPrinciples,hardy2013,chiribella2017quantum}.
The theories described in the latter approach are often called \emph{operational
probabilistic theories} (OPTs). These two approaches are almost equivalent,
but OPTs are able to describe also non-convex theories, arising e.g.\ from
the lack of Causality.

For convex theories, one can translate concepts of one approach into
the other, but the scope of the two approaches remains slightly different.
The convex approach is more low-level: one starts from the state space
of single systems, and builds composites from it. Its weak point is
that one must specify all the details and constructions, and this
can become cumbersome when one studies the composition of systems
\cite{Graydon-QPL,BarnumGraydonWilceCCEJA,Galley}. However, this
approach is often inescapable if one is to deal with a concrete model.
On the other hand, OPTs use the high-level language of circuits borrowed
from category theory \cite{Abramsky2004,Coecke-Kindergarten,Coecke-Picturalism,Selinger,Coecke2016,Coecke2017picturing},
and take composition as a primitive, rather than derive it from the
structure of the state space. Its strong point is that it can be used
to derive results about a theory without specifying its concrete details
too much.

In this chapter and the rest of this thesis we will mainly adopt the
OPT variant: its focus on processes and their composition matches
beautifully with the scope of thermodynamics, which is all about processes
and transformations between states. Therefore, in the following, the
term ``GPT'' will be used as a synonym of ``OPT'', or more precisely,
of a general probabilistic theory treated in the OPT approach. Here
we present the principles underpinning the OPT framework, where the
ideas of process and composition play a central role. Even states
are viewed as processes, specifically as preparation processes. The
analysis of the operational structure of a physical theory is done
by introducing a diagrammatic language, which will be used throughout
this thesis. Then, in section~\ref{sec:The-probabilistic-structure},
we insert the probabilistic ingredient: every theory must be able
to provide probabilities of experimental outcomes.

In section~\ref{sec:Causality} we introduce the axiom of Causality
\cite{Chiribella-purification}, which is often implicitly assumed
in a lot of GPT literature. Causality will be a background assumption
throughout this thesis. The choice of Causality as an axiom for a
physical theory is motivated by several of its consequences, e.g.\ the
lack of time loops or the no-signalling principle between different
physical systems. Finally, this axiom is also appealing from a thermodynamic
perspective, for it guarantees the ability to discard systems, and
therefore to restrict ourselves to a subsystem of a larger thermodynamic
system.

\section{Events, tests, and the operational structure}

As the name suggests, an operational probabilistic theory is made
of two parts: the \emph{operational} one and the \emph{probabilistic}
one. The operational part is the more fundamental: it describes how
to build experiments in a laboratory by composing and connecting the
associated devices. As such, it is the essential ingredient of every
experimentally testable physical theory. The probabilistic part is
built on top of that, and it deals with the predictive power of the
theory, namely the ability to predict the likelihood of the various
experimental observations.

In this section we introduce a formalism able to describe the operational
structure of every physical theory \cite{Chiribella-purification,hardy2011,Hardy-informational-2,hardy2013,QuantumFromPrinciples},
which has a nice graphical representation in terms of diagrams and
circuits, taken from the graphical languages for symmetric monoidal
categories \cite{Coecke-Kindergarten,Coecke-Picturalism,Selinger,Coecke2016,Coecke2017picturing}.

\subsection{Systems and tests}

In an operational theory, there are two primitive notions: \emph{systems}
and \emph{tests}. We can have an intuition about their meaning by
thinking of a concrete experimental situation. A \emph{test} represents
the application of a physical device (beam-splitter, polarimeter,
Stern-Gerlach magnet, etc.). Every device has an input and an output,
which will be called \emph{input} and \emph{output system} respectively.
In this way, somehow systems play the role of labels attached to the
input and output ports of a device.

We denote systems by capital letters in Roman character: $\mathrm{A}$,
$\mathrm{B}$, etc. There is also a particular system, the \emph{trivial
system}, which simply means ``nothing'', or the degrees of freedom
the theory does not deal with. We will denote it by letter $\mathrm{I}$.
A device with the trivial system as input is simply a device with
\emph{no} input, and a device with the trivial system as output is
simply a device with \emph{no} output.

The application of a physical device can yield various outcomes. Each
of them corresponds to a particular event that occurred in the laboratory,
which can be identified by the experimenter by ``reading'' the device
pointer. Therefore, we can give the following characterisation of
tests.
\begin{defn}
A \emph{test} with input system $\mathrm{A}$ and output system $\mathrm{B}$
is a collection of \emph{events} $\left\{ \mathcal{C}_{i}\right\} _{i\in X}$
that can occur in an experiment, labelled by the outcome $i$ in some
set $X$. $X$ is called \emph{outcome set.}
\end{defn}

We will often say that $\left\{ \mathcal{C}_{i}\right\} _{i\in X}$
is a test from system $\mathrm{A}$ to system $\mathrm{B}$; if $\mathrm{A}$
and $\mathrm{B}$ coincide, we say that $\left\{ \mathcal{C}_{i}\right\} _{i\in X}$
is a test on system $\mathrm{A}$.

To clarify the role of outcome $i$ better, we can regard it as what
the experimenter actually sees when they perform their experiment
(a sequence of digits, a spot in a photographic plate, the device
pointer, etc.). The outcome set $X$ is the set containing all the
possible outcomes for a given test. In the following we will assume
that all outcome sets are finite. This will simplify the later mathematical
treatment, and will match a finite-dimensionality assumption we will
make in section~\ref{sec:The-probabilistic-structure}.

We can represent a test graphically as a box with incoming and outgoing
wires representing the input and output systems respectively.\[
\begin{aligned}\Qcircuit @C=1em @R=.7em @!R {    & \qw \poloFantasmaCn{\rA} &  \gate{\left\{\cC_{i}\right\}_{i \in X}} & \qw \poloFantasmaCn{\rB} &\qw}\end{aligned}
\]When there is no ambiguity, we will omit the outcome set $X$. If
we want to express that the specific event $\mathcal{C}_{i}$ has
occurred, we will write\[
\begin{aligned}\Qcircuit @C=1em @R=.7em @!R {    & \qw \poloFantasmaCn{\rA} &  \gate{\cC_{i}} & \qw \poloFantasmaCn{\rB} &\qw}\end{aligned}~,
\]without braces.

Whenever the trivial system $\mathrm{I}$ is involved, we omit the
corresponding wire and letter. Specifically, when we have no physical
input for our device\textemdash which means the trivial system as
input\textemdash we have a \emph{preparation-test} (a collection of
\emph{preparation-events}), which we represent as\[
\begin{aligned}\Qcircuit @C=1em @R=.7em @!R {    &  \prepareC{\left\{\rho_{i}\right\}} & \qw \poloFantasmaCn{\rA} &\qw}\end{aligned}~:= ~ \begin{aligned}\Qcircuit @C=1em @R=.7em @!R {    & \qw \poloFantasmaCn{\rI} &  \gate{\left\{\rho_{i}\right\}} & \qw \poloFantasmaCn{\rA} &\qw}\end{aligned}~,
\]namely with a rounded box on its left side. Intuitively, preparation-tests
prepare a system in a particular ``random state'', although we will
clarify this statement later. Similarly, when we have no physical
output for our device\textemdash i.e.\ the trivial system as output\textemdash we
have an \emph{observation-test} (a collection of \emph{observation-events}),
which we represent as\[
\begin{aligned}\Qcircuit @C=1em @R=.7em @!R {   & \qw \poloFantasmaCn{\rA} &\measureD{\left\{a_{i}\right\}}}\end{aligned}~:= ~ \begin{aligned}\Qcircuit @C=1em @R=.7em @!R {    & \qw \poloFantasmaCn{\rA} &  \gate{\left\{a_{i}\right\}} & \qw \poloFantasmaCn{\rI} &\qw}\end{aligned}~,
\]namely with a rounded box on its right side. Intuitively, observation-tests
destroy a system while acquiring some information from it, so they
are related to demolition measurements. Finally, if we have a test
$\left\{ p_{i}\right\} _{i\in X}$ from the trivial system to itself,
we omit both the wires and the box.\[
\left\{p_{i}\right\}~:=~  \begin{aligned}\Qcircuit @C=1em @R=.7em @!R {    & \qw \poloFantasmaCn{\rI} &  \gate{\left\{p_{i}\right\}} & \qw \poloFantasmaCn{\rI} &\qw}\end{aligned}
\]
\begin{defn}
We say that a test is \emph{deterministic} if its outcome set has
one element.
\end{defn}

If a test is deterministic, we omit the braces and simply write $\mathcal{C}$
instead of $\left\{ \mathcal{C}\right\} $. In a non-deterministic
test, we cannot predict which particular outcome we will obtain. On
the contrary, the outcome of a deterministic test is completely determined.
Since we are not able to predict the outcome of non-deterministic
tests, we set up a probabilistic structure that enables us to define
probabilities for the various outcomes. We will address this issue
in section~\ref{sec:The-probabilistic-structure}, but first some
other notions are needed.

\subsection{Sequential and parallel composition}

Since we are implementing a graphical language which has a direct
link to experimental apparatuses, the next step is to describe how
to connect devices. Devices can be connected sequentially or in parallel.
Let us start from sequential composition. Intuitively, two devices
can be connected sequentially, i.e.\ one after another, if the output
system of the former is the input system of the latter.
\begin{defn}
If $\left\{ \mathcal{C}_{i}\right\} _{i\in X}$ is a test from $\mathrm{A}$
to $\mathrm{B}$ with outcome set $X$, and $\left\{ \mathcal{D}_{j}\right\} _{j\in Y}$
is a test from $\mathrm{B}$ to $\mathrm{C}$ with outcome set $Y$,
we can consider the \emph{sequential composition} $\left\{ \mathcal{D}_{j}\circ\mathcal{C}_{i}\right\} _{\left(i,j\right)\in X\times Y}$
, which is a test from $\mathrm{A}$ to $\mathrm{C}$ and has outcome
set $X\times Y$.
\end{defn}

The graphical representation is quite intuitive: suppose we want to
compose the event $\mathcal{D}_{j}$ after the event $\mathcal{C}_{i}$;
we simply write\[
\begin{aligned}\Qcircuit @C=1em @R=.7em @!R {    & \qw \poloFantasmaCn{\rA} &  \gate{\cD_{j} \circ \cC_{i}} & \qw \poloFantasmaCn{\rC} &\qw}\end{aligned}~:=~ \begin{aligned}\Qcircuit @C=1em @R=.7em @!R {    & \qw \poloFantasmaCn{\rA} &  \gate{\cC_{i}} & \qw \poloFantasmaCn{\rB} &  \gate{\cD_{j}} & \qw \poloFantasmaCn{\rC} &\qw}\end{aligned}~.
\]From this notation, and from its operational meaning, we immediately
get that sequential composition is associative.

Sequential composition yields a natural ordering on tests. Indeed,
some tests are performed first and other later. In graphical language
this ordering goes from left to right: every box follows all the others
on its left. However, we must not confuse this ordering with ``temporal''
or ``causal'' ordering. We will come back to this point in section~\ref{sec:Causality}.

Now let us see an example of sequential composition of tests.
\begin{example}
Consider the diagram\[
\begin{aligned}\Qcircuit @C=1em @R=.7em @!R {  \prepareC{\left\{\rho_{i}\right\}}  & \qw \poloFantasmaCn{\rA} &  \gate{\left\{\cC_{j}\right\}} & \qw \poloFantasmaCn{\rB} &\measureD{\left\{b_{k}\right\}}}\end{aligned}~.
\]It gives instructions on how to build the experiment: first, we initialise
system $\mathrm{A}$ with the preparation-test $\left\{ \rho_{i}\right\} $,
then we perform the test $\left\{ \mathcal{C}_{j}\right\} $ from
$\mathrm{A}$ to $\mathrm{B}$ and finally we acquire some information
from $\mathrm{B}$ by destroying it with the observation-test $\left\{ b_{k}\right\} $.

If we wish to express which events actually occurred, we write\begin{equation}\label{eq:example joint}
\begin{aligned}\Qcircuit @C=1em @R=.7em @!R {  \prepareC{\rho_{i}}  & \qw \poloFantasmaCn{\rA} &  \gate{\cC_{j}} & \qw \poloFantasmaCn{\rB} &\measureD{b_{k}}}\end{aligned}~.
\end{equation}This means that the preparation-event $\rho_{i}$, the event $\mathcal{C}_{j}$,
and the observation-event $b_{k}$ occurred. 
\end{example}

We will often make use of the following short-hand notations, inspired
by quantum theory, to mean some common diagrams occurring in our analysis.
\begin{enumerate}
\item \[ \left(a_j\middle|\rho_i\right)~:=\!\!\!\!\begin{aligned}\Qcircuit @C=1em @R=.7em @!R { & \prepareC{\rho_i} & \qw \poloFantasmaCn{\rA} &\measureD{a_j}}\end{aligned}~; \]
\item \[ \left(b_k\middle|\cC_j\middle|\rho_i\right)~:=\!\!\!\!\begin{aligned}\Qcircuit @C=1em @R=.7em @!R { & \prepareC{\rho_i} & \qw \poloFantasmaCn{\rA} &\gate{\cC_j} &\qw \poloFantasmaCn{\rB} &\measureD{b_k}}\end{aligned}~; \]
\item \[ \left|\rho_j\right)\left(a_i\right|~:=~\begin{aligned}\Qcircuit @C=1em @R=.7em @!R { & \qw \poloFantasmaCn{\rA} &\measureD{a_i}&\prepareC{\rho_j}&\qw \poloFantasmaCn{\rB} &\qw}\end{aligned}~. \]
\end{enumerate}
Let us now define the identity test.
\begin{defn}
The \emph{identity test} for system $\mathrm{A}$ is a deterministic
test $\mathcal{I}_{\mathrm{A}}$ on $\mathrm{A}$ such that $\mathcal{C}_{i}\circ\mathcal{I}_{\mathrm{A}}=\mathcal{C}_{i}$
for every event $\mathcal{C}_{i}$ from $\mathrm{A}$ to $\mathrm{B}$,
and $\mathcal{I}_{\mathrm{A}}\circ\mathcal{D}_{i}=\mathcal{D}_{i}$
for every event $\mathcal{D}_{i}$ from $\mathrm{B}$ to $\mathrm{A}$.
\end{defn}

Graphically, we have\[
\begin{aligned}\Qcircuit @C=1em @R=.7em @!R {    & \qw \poloFantasmaCn{\rA} &  \gate{\cI} & \qw \poloFantasmaCn{\rA} &  \gate{\cC_{i}} & \qw \poloFantasmaCn{\rB} &\qw}\end{aligned}~=~\begin{aligned}\Qcircuit @C=1em @R=.7em @!R {    & \qw \poloFantasmaCn{\rA} &  \gate{\cC_{i}} & \qw \poloFantasmaCn{\rB} &\qw}\end{aligned}
\]for every $\mathcal{C}_{i}$, and\[
\begin{aligned}\Qcircuit @C=1em @R=.7em @!R {    & \qw \poloFantasmaCn{\rB} &  \gate{\cD_{i}} & \qw \poloFantasmaCn{\rA} &  \gate{\cI} & \qw \poloFantasmaCn{\rA} &\qw}\end{aligned}~=~\begin{aligned}\Qcircuit @C=1em @R=.7em @!R {    & \qw \poloFantasmaCn{\rB} &  \gate{\cD_{i}} & \qw \poloFantasmaCn{\rA} &\qw}\end{aligned}
\]for every $\mathcal{D}_{i}$. According to this definition, it is
clear that for every system $\mathrm{A}$ the identity test $\mathcal{I}_{\mathrm{A}}$
is unique.

Applying the identity test is just like doing nothing. For this reason
we will often omit the box for the identity test, and write just a
plain wire.

We sometimes want to ``identify'' similar systems, namely systems
that behave exactly in the same way from an operational point of view,
yet they are distinct. In quantum mechanics, for example, we can consider
the polarisation of a photon and the spin of an electron. Although
they are completely different physical systems, they are described
by the same Hilbert space.\footnote{Or by isomorphic Hilbert spaces, to be precise.}
\begin{defn}
We say that system $\mathrm{A}$ and system $\mathrm{A}'$ are \emph{operationally
equivalent} (and we write $\mathrm{A}\approx\mathrm{A}'$) if there
is a deterministic test $\mathcal{U}_{1}$ from $\mathrm{A}$ to $\mathrm{A}'$
and a deterministic test $\mathcal{U}_{2}$ from $\mathrm{A}'$ to
$\mathrm{A}$, such that\[
\begin{aligned}\Qcircuit @C=1em @R=.7em @!R {    & \qw \poloFantasmaCn{\rA} &  \gate{\cU_{1}} & \qw \poloFantasmaCn{\rA'} &  \gate{\cU_{2}} & \qw \poloFantasmaCn{\rA} &\qw}\end{aligned}~=~\begin{aligned}\Qcircuit @C=1em @R=.7em @!R {    & \qw \poloFantasmaCn{\rA} &  \gate{\cI} & \qw \poloFantasmaCn{\rA} &\qw}\end{aligned}~,
\]where $\mathcal{I}_{\mathrm{A}}$ is the identity test on $\mathrm{A}$,
and\[
\begin{aligned}\Qcircuit @C=1em @R=.7em @!R {    & \qw \poloFantasmaCn{\rA'} &  \gate{\cU_{2}} & \qw \poloFantasmaCn{\rA} &  \gate{\cU_{1}} & \qw \poloFantasmaCn{\rA'} &\qw}\end{aligned}~=~\begin{aligned}\Qcircuit @C=1em @R=.7em @!R {    & \qw \poloFantasmaCn{\rA'} &  \gate{\cI} & \qw \poloFantasmaCn{\rA'} &\qw}\end{aligned}~,
\]where $\mathcal{I}_{\mathrm{A}'}$ is the identity test on $\mathrm{A}'$.
\end{defn}

If $\mathrm{A}\approx\mathrm{A}'$, we can transform tests on system
$\mathrm{A}$ into tests on system $\mathrm{A}'$ by taking the sequential
composition with the intertwining tests $\mathcal{U}_{1}$ and $\mathcal{U}_{2}$.
Indeed, if $\mathcal{C}_{i}$ is an event on system $\mathrm{A}$,
the corresponding event $\mathcal{C}_{i}'$ on system $\mathrm{A}'$
is\[
\begin{aligned}\Qcircuit @C=1em @R=.7em @!R {    & \qw \poloFantasmaCn{\rA'} &  \gate{\cC'_{i}} & \qw \poloFantasmaCn{\rA'} &\qw}\end{aligned}~:=~\begin{aligned}\Qcircuit @C=1em @R=.7em @!R {    & \qw \poloFantasmaCn{\rA'} &  \gate{\cU_{2}} & \qw \poloFantasmaCn{\rA} &  \gate{\cC_{i}} & \qw \poloFantasmaCn{\rA}&  \gate{\cU_{1}} & \qw \poloFantasmaCn{\rA'} &\qw}\end{aligned}~.
\]

Now we move to the other type of composition: parallel composition.
If we have two systems $\mathrm{A}$ and $\mathrm{B}$, we can consider
them together, forming the composite system $\mathrm{AB}$.
\begin{defn}
If $\mathrm{A}$ and $\mathrm{B}$ are two systems, the corresponding
\emph{composite system} is $\mathrm{\mathrm{AB}}$. System composition
has the following properties.
\begin{enumerate}
\item $\mathrm{AI}=\mathrm{IA}=\mathrm{A}$ for every system $\mathrm{A}$,
where $\mathrm{I}$ is the trivial system;
\item \label{enu:commutativity composition}$\mathrm{AB}\approx\mathrm{BA}$
for all systems $\mathrm{A}$ and $\mathrm{B}$;
\item $\mathrm{A}\left(\mathrm{BC}\right)=\left(\mathrm{AB}\right)\mathrm{C}$
for all systems $\mathrm{A}$, $\mathrm{B}$, $\mathrm{C}$.
\end{enumerate}
\end{defn}

These properties have a fairly intuitive meaning.
\begin{enumerate}
\item When we combine a system with ``nothing'', we still have the original
system.
\item The composition of systems does not depend on the order we compose
them.
\item This particular form of associativity allows us to write simply $\mathrm{ABC}$,
without parentheses. Again, the order of composition is irrelevant.
\end{enumerate}
We represent composite systems diagrammatically as a collection of
wires one under another. We will typically omit the wire for the trivial
system.

We can represent an event $\mathcal{C}_{i}$ from system $\mathrm{AB}$
to system $\mathrm{CD}$ as a box with multiple wires, one for each
system.\[
\begin{aligned}\Qcircuit @C=1em @R=.7em @!R {    & \qw \poloFantasmaCn{\mathrm{AB}} &  \gate{\cC_{i}} & \qw \poloFantasmaCn{\mathrm{CD}} &\qw}\end{aligned}~ = ~ \begin{aligned}\Qcircuit @C=1em @R=.7em @!R {     & \qw \poloFantasmaCn{\rA}  & \multigate{1}{\cC_{i}} & \qw \poloFantasmaCn{\rC} &\qw  \\    & \qw \poloFantasmaCn{\rB}  &\ghost{\cC_{i}} & \qw \poloFantasmaCn{\rD} &  \qw }\end{aligned}
\]By property~\ref{enu:commutativity composition}, it is completely
irrelevant to write $\mathrm{A}$ or $\mathrm{B}$ on the upper input
wire, and the same holds for every wire. For composite systems we
depict preparation-events as\[
\begin{aligned}\Qcircuit @C=1em @R=.7em @!R { & \multiprepareC{1}{\rho_{i}}    & \qw \poloFantasmaCn{\rA} &  \qw   \\  & \pureghost{\rho_{i}}    & \qw \poloFantasmaCn{\rB}  &  \qw }\end{aligned}~,
\]and observation-events as\[
\begin{aligned}\Qcircuit @C=1em @R=.7em @!R {     & \qw \poloFantasmaCn{\rA} &  \multimeasureD{1}{a_{i}}  \\      & \qw \poloFantasmaCn{\rB}  &  \ghost{a_{i}} }\end{aligned}~.
\]

Now we can define the parallel composition of tests.
\begin{defn}
Let $\left\{ \mathcal{C}_{i}\right\} _{i\in X}$ be a test from $\mathrm{A}$
to $\mathrm{B}$, and let $\left\{ \mathcal{D}_{j}\right\} _{j\in Y}$
be a test from $\mathrm{C}$ to $\mathrm{D}$. The \emph{parallel
composition} $\left\{ \mathcal{C}_{i}\otimes\mathcal{D}_{j}\right\} _{\left(i,j\right)\in X\times Y}$
(or tensor product) is a test from $\mathrm{AC}$ to $\mathrm{BD}$
with outcome set $X\times Y$, and it is represented diagrammatically
as\[
\begin{aligned}\Qcircuit @C=1em @R=.7em @!R {     & \qw \poloFantasmaCn{\rA}  & \multigate{1}{\cC_{i}\otimes \cD_{j}} & \qw \poloFantasmaCn{\rB} &\qw  \\    & \qw \poloFantasmaCn{\rC}  &\ghost{\cC_{i}\otimes \cD_{j}} & \qw \poloFantasmaCn{\rD} &  \qw }\end{aligned} ~:=~ \begin{aligned}\Qcircuit @C=1em @R=.7em @!R {     & \qw \poloFantasmaCn{\rA}  & \gate{\cC_{i}} & \qw \poloFantasmaCn{\rB} &\qw  \\    & \qw \poloFantasmaCn{\rC}  &\gate{\cD_{j}} & \qw \poloFantasmaCn{\rD} &  \qw }\end{aligned}~.
\]
\end{defn}

Again, from its operational meaning, it is immediate that parallel
composition is associative. Note that this is captured by the graphical
notation we are using.

We can combine parallel and sequential composition: suppose $\mathcal{A}_{i}$
is an event from $\mathrm{A}$ to $\mathrm{B}$, $\mathcal{B}_{j}$
is an event from $\mathrm{B}$ to $\mathrm{C}$; $\mathcal{D}_{k}$
is an event from $\mathrm{D}$ to $\mathrm{E}$ and $\mathcal{E}_{l}$
is an event from $\mathrm{E}$ to $\mathrm{F}$. Then we have\[
\begin{aligned}\Qcircuit @C=1em @R=.7em @!R {     & \qw \poloFantasmaCn{\rA}  & \multigate{1}{\left(\cB_{j} \circ \cA_{i}\right) \otimes \left(\mathcal{E}_{l} \circ \cD_{k}\right)} & \qw \poloFantasmaCn{\rC} &\qw  \\    & \qw \poloFantasmaCn{\rD}  &\ghost{\left(\cB_{j} \circ \cA_{i}\right) \otimes \left(\mathcal{E}_{l} \circ \cD_{k}\right)} & \qw \poloFantasmaCn{\mathrm{F}} &  \qw }\end{aligned}~ =~ \begin{aligned}\Qcircuit @C=1em @R=.7em @!R {     & \qw \poloFantasmaCn{\rA}  & \gate{\cB_{j} \circ \cA_{i}} & \qw \poloFantasmaCn{\rC} &\qw  \\    & \qw \poloFantasmaCn{\rD}  &\gate{\mathcal{E}_{l} \circ \cD_{k}} & \qw \poloFantasmaCn{\mathrm{F}} &  \qw }\end{aligned}=
\]
\[
= \begin{aligned}\Qcircuit @C=1em @R=.7em @!R {& \qw \poloFantasmaCn{\rA}  & \gate{\cA_{i}} & \qw \poloFantasmaCn{\rB} & \gate{\cB_{j}}& \qw \poloFantasmaCn{\rC} &\qw \\    & \qw \poloFantasmaCn{\rD}  & \gate{\cD_{k}} & \qw \poloFantasmaCn{\rE}  & \gate{\mathcal{E}_{l}}& \qw \poloFantasmaCn{\mathrm{F}} &\qw}\end{aligned}~=~\begin{aligned}\Qcircuit @C=1em @R=.7em @!R {     & \qw \poloFantasmaCn{\rA}  & \multigate{1}{\left(\cB_{j} \otimes \mathcal{E}_{l}\right) \circ \left(\cA_{i} \otimes \cD_{k}\right)} & \qw \poloFantasmaCn{\rC} &\qw  \\    & \qw \poloFantasmaCn{\rD}  &\ghost{\left(\cB_{j} \otimes \mathcal{E}_{l}\right) \circ \left(\cA_{i} \otimes \cD_{k}\right)} & \qw \poloFantasmaCn{\mathrm{F}} &  \qw }\end{aligned}~.
\]

Let us analyse the properties of the deterministic test that intertwines
system $\mathrm{AB}$ and $\mathrm{BA}$. In practice, it swaps system
$\mathrm{A}$ and system $\mathrm{B}$, so we call it $\mathtt{SWAP}$.
Clearly swapping the systems twice yields the original system, thus
$\mathtt{SWAP}^{-1}=\mathtt{SWAP}$. Moreover it swaps the events
in a parallel composition.\[
\begin{aligned}\Qcircuit @C=1em @R=.7em @!R {     & \qw \poloFantasmaCn{\rA}  & \gate{\cC_{i}} & \qw \poloFantasmaCn{\rB} &\multigate{1}{\mathtt{SWAP}} & \qw \poloFantasmaCn{\rD} &\qw  \\    & \qw \poloFantasmaCn{\rC}  &\gate{\cD_{j}} & \qw \poloFantasmaCn{\rD} &  \ghost{\mathtt{SWAP}}& \qw \poloFantasmaCn{\rB} &\qw}\end{aligned}~=~\begin{aligned}\Qcircuit @C=1em @R=.7em @!R { & \qw \poloFantasmaCn{\rA} &\multigate{1}{\mathtt{SWAP}} & \qw \poloFantasmaCn{\rC}      & \gate{\cD_{j}}  &\qw   \poloFantasmaCn{\rD} &\qw \\ & \qw \poloFantasmaCn{\rC} &  \ghost{\mathtt{SWAP}}& \qw \poloFantasmaCn{\rA}     &\gate{\cC_{i}}  &\qw \poloFantasmaCn{\rC} &\qw }\end{aligned}
\]

Note that we can compose preparation-tests only in parallel; the same
holds for observation-tests. We will often write sequential composition
as a product: if $\mathcal{C}_{i}$ is an event from $\mathrm{A}$
to $\mathrm{B}$ and $\mathcal{D}_{j}$ is an event from $\mathrm{B}$
to $\mathrm{C}$, we will write $\mathcal{D}_{j}\circ\mathcal{C}_{i}$
simply as $\mathcal{D}_{j}\mathcal{C}_{i}$.

Now we can define operational theories.
\begin{defn}
An \emph{operational theory} is given by a collection of systems,
closed under composition, and a collection of tests, closed under
sequential and parallel composition.
\end{defn}

It is easy to see, from the properties presented above, that an operational
theory is described by a strict symmetric monoidal category \cite{Chiribella-purification,Categories-practising,Chiribella14,QuantumFromPrinciples,Coecke2017picturing}.

In the following we will assume that tests $\left\{ \mathcal{A}_{i}\right\} $
from system $\mathrm{A}$ to system $\mathrm{B}$ are performed through
a deterministic interaction between the systems and the measurement
apparatus $\mathrm{X}$, which is read by the observer with an observation-test.
\begin{assumption}[Physicalisation of readout \cite{Chiribella14}]
Every test $\left\{ \mathcal{A}_{i}\right\} _{i\in X}$ can be realised
as follows:\[
\begin{aligned}\Qcircuit @C=1em @R=.7em @!R { & \qw \poloFantasmaCn{\rA} &\gate{\left\{\cA_i \right\}} & \qw \poloFantasmaCn{\rB}      &\qw }\end{aligned}~=~\begin{aligned}\Qcircuit @C=1em @R=.7em @!R { & \qw \poloFantasmaCn{\rA} &\multigate{1}{\cA} & \qw \poloFantasmaCn{\rB}       &\qw \\ &  &  \pureghost{\cA}& \qw \poloFantasmaCn{\rX}      &\measureD{\left\{e_i\right\}}}\end{aligned}~,
\]where $\mathcal{A}$ is a deterministic test from $\mathrm{A}$ to
$\mathrm{BX}$, and $\left\{ e_{i}\right\} _{i\in X}$ is an observation-test
on $\mathrm{X}$.
\end{assumption}

\section{The probabilistic structure\label{sec:The-probabilistic-structure}}

Now we can add the probabilistic ingredient to our theory: basically,
we want to assign a number in the interval $\left[0,1\right]$ to
every event from the trivial system to itself.
\begin{defn}
An \emph{operational-probabilistic theory} (OPT) is an operational
theory where, for every test $\left\{ p_{i}\right\} _{i\in X}$ on
the trivial system $\mathrm{I}$, one has $p_{i}\in\left[0,1\right]$
and $\sum_{i\in X}p_{i}=1$.

Moreover, the sequential and parallel compositions of two events on
the trivial system are given by the product of probabilities: $p_{i}\circ p_{j}=p_{i}\otimes p_{j}=p_{i}p_{j}$.
\end{defn}

This definition states that every event from $\mathrm{I}$ to itself
can be interpreted as a probability. Consequently, we can associate
a probability with every diagram with no external wires.
\begin{example}
Let us consider eq.~\eqref{eq:example joint} again. It is a diagram
without external wires; indeed the sequential composition of the three
events is an event from the trivial system $\mathrm{I}$ to itself
(no input and no output). So we have $p_{ijk}:=\left(b_{k}\middle|\mathcal{C}_{j}\middle|\rho_{i}\right)$,
that is the \emph{joint probability} of having the preparation-event
$\rho_{i}$, the event $\mathcal{C}_{j}$, and the observation-event
$b_{k}$.
\end{example}

Henceforth we will focus only on OPTs, namely on operational theories
with a probabilistic structure.

Sometimes it happens that we obtain the same physical configuration
with different experimental procedures. For instance, in quantum theory
consider the mixed state $\rho=\frac{1}{2}\mathbf{1}$ of a qubit.
This state can be prepared either by having no information on the
state of the system, or by taking the partial trace of one of the
Bell states. The issue is now how to distinguish different experimental
preparations, or find out when they are equivalent.

Let us consider, for instance, preparation-events. If we compose a
preparation-event with an observation-event, we get $p_{ij}=\left(a_{j}\middle|\rho_{i}\right)$,
the joint probability of having the preparation-event $\rho_{i}$
and the observation-event $a_{j}$.

If we have a preparation-event $\rho_{i}$ on $\mathrm{A}$, we can
associate a real-valued function $\widehat{\rho}_{i}$ with it. This
function acts on observation-events $a_{j}$ on $\mathrm{A}$ and
yields the joint probability $p_{ij}$.
\[
\widehat{\rho}_{i}:a_{j}\mapsto\left(a_{j}\middle|\rho_{i}\right)=p_{ij}
\]
Similarly, if we have an observation-event $a_{j}$ on $\mathrm{A}$,
we can associate a real-valued function $\widehat{a}_{j}$ with it.
This function acts on preparation-events $\rho_{i}$ on $\mathrm{A}$
and yields the joint probability $p_{ij}$.
\[
\widehat{a}_{j}:\rho_{i}\mapsto\left(a_{j}\middle|\rho_{i}\right)=p_{ij}
\]

From a probabilistic point of view, we cannot distinguish two preparations
of the system if they yield the same probabilities for \emph{all}
observation-tests, even if the preparations were obtained operatively
in completely different ways. If we consider an experimenter, they
can distinguish two unknown preparations of the system by examining
the statistics they get from performing, in principle, all possible
measurements on the system. If they find any difference in the statistics,
then they conclude that the preparations were different. A very similar
argument holds for observation-events.

In this vein, we can introduce an equivalence relation between preparation-events
(and similarly between observation-events). If $\rho_{i}$ and $\sigma_{j}$
are two preparation-events on system $\mathrm{A}$, we say that they
are \emph{tomographically equivalent} (or \emph{tomographically indistinguishable}),
written as $\rho_{i}\sim\sigma_{j}$, if $\widehat{\rho}_{i}=\widehat{\sigma}_{j}$,
namely if for every observation-event $a_{k}$ on $\mathrm{A}$ we
have $\left(a_{k}\middle|\rho_{i}\right)=\left(a_{k}\middle|\sigma_{j}\right)$.
Similarly, if $a_{i}$ and $b_{j}$ are two observation-events on
$\mathrm{A}$, we say that they are \emph{tomographically equivalent}
(or \emph{tomographically indistinguishable}), written as $a_{i}\sim b_{j}$,
if $\widehat{a}_{i}=\widehat{b}_{j}$, namely if for every preparation-event
$\rho_{k}$ on $\mathrm{A}$ we have $\left(a_{i}\middle|\rho_{k}\right)=\left(b_{j}\middle|\rho_{k}\right)$.
\begin{defn}
\label{def:states}Equivalence classes of tomographically indistinguishable
preparation-events are called \emph{states}. The set of states of
system $\mathrm{A}$ is denoted as $\mathsf{St}\left(\mathrm{A}\right)$.

Equivalence classes of tomographically indistinguishable observation-events
are called \emph{effects}. The set of effects of system $\mathrm{A}$
is denoted as $\mathsf{Eff}\left(\mathrm{A}\right)$.
\end{defn}

Therefore, two states $\rho_{1}$ and $\rho_{2}$ of system $\mathrm{A}$
are equal if and only if $\left(a\middle|\rho_{1}\right)=\left(a\middle|\rho_{2}\right)$
for every effect $a\in\mathsf{Eff}\left(\mathrm{A}\right)$. Similarly,
two effects $a_{1}$ and $a_{2}$ of system $\mathrm{A}$ are equal
if and only if $\left(a_{1}\middle|\rho\right)=\left(a_{2}\middle|\rho\right)$
for every state $\rho\in\mathsf{St}\left(\mathrm{A}\right)$. The
process of reconstructing a state (or an effect) from the statistics
of measurements is called \emph{tomography}.

We can assume that equivalence classes were taken from the very beginning,
so from now on we will say that a preparation-test is made of states,
and that an observation-test is made of effects.
\begin{example}
The trivial system has a unique deterministic state and a unique deterministic
effect: it is the number 1. All the other effects and states are elements
of $\left[0,1\right]$.
\end{example}

Let us see what states and effects are in quantum mechanics.
\begin{example}
\label{ex:states effects in quantum mechanics}In quantum mechanics
we associate a Hilbert space $\mathcal{H}_{\mathrm{A}}$ with every
system $\mathrm{A}$. Deterministic states are density operators,
which means trace-class positive operators with trace equal to 1.
A non-deterministic preparation-test is sometimes called \emph{quantum
information source}: it is a collection of trace-class positive operators
$\rho_{i}$, with $\mathrm{tr}\:\rho_{i}\leq1$. This is essentially
a random preparation: a state $\rho_{i}$ is prepared with a probability
given by $\mathrm{tr}\:\rho_{i}$. Therefore in quantum mechanics
$\mathsf{St}\left(\mathrm{A}\right)$ is the set of trace-class positive
operators with trace less than or equal to one.

An effect is, instead, represented by a positive operator $P$, with
$P\leq\mathbf{1}$, where $\mathbf{1}$ is the identity operator.
Observation-tests are then POVMs. The pairing between states and effect
is given by the trace: $\left(P\middle|\rho\right)=\mathrm{tr}\:P\rho$.
In quantum mechanics there is only one deterministic effect: the identity
$\mathbf{1}$. This is not a coincidence, but it follows from Causality
(see section~\ref{sec:Causality}).
\end{example}

According to definition~\ref{def:states}, states and effects are
in fact real-valued functions. As a consequence we can take linear
combinations of them with real coefficients; in other words they span
real vector spaces. Let $\mathsf{St}_{\mathbb{R}}\left(\mathrm{A}\right)$
be the vector space spanned by states, and let $\mathsf{Eff}_{\mathbb{R}}\left(\mathrm{A}\right)$
be the vector space spanned by effects. These vector spaces can be
finite- or infinite-dimensional. In our presentation, to avoid mathematical
subtleties and simplify the treatment, we will assume that they are
finite-dimensional. Clearly, $\mathsf{Eff}_{\mathbb{R}}\left(\mathrm{A}\right)$
is the dual vector space of $\mathsf{St}_{\mathbb{R}}\left(\mathrm{A}\right)$
and $\mathsf{St}_{\mathbb{R}}\left(\mathrm{A}\right)$ is the dual
vector space of $\mathsf{Eff}_{\mathbb{R}}\left(\mathrm{A}\right)$.
For finite-dimensional vector spaces, we have $\dim\mathsf{St}_{\mathbb{R}}\left(\mathrm{A}\right)=\dim\mathsf{Eff}_{\mathbb{R}}\left(\mathrm{A}\right)$.

We can extend the notion of tomography to the full vector space $\mathsf{St}_{\mathbb{R}}\left(\mathrm{A}\right)$:
$\xi,\eta\in\mathsf{St}_{\mathbb{R}}\left(\mathrm{A}\right)$ are
equal if and only if, for every $a\in\mathsf{Eff}\left(\mathrm{A}\right)$
(or for every $X\in\mathsf{Eff}_{\mathbb{R}}\left(\mathrm{A}\right)$)
we have $\left(a\middle|\xi\right)=\left(a\middle|\eta\right)$ (or
$\left(X\middle|\xi\right)=\left(X\middle|\eta\right)$). A similar
fact holds for the vector space $\mathsf{Eff}_{\mathbb{R}}\left(\mathrm{A}\right)$.
The fact that these vector spaces have finite dimension means that
a finite set of effects (resp.\ states) is sufficient to do tomography
on states (resp.\ effects).
\begin{rem}
\label{rem:local tomography}Consider the states of a bipartite system
$\mathrm{AB}$. Clearly to do tomography one takes the effects of
$\mathrm{AB}$. Some of them will be of the product form, i.e.\ $a\otimes b$,
where $a\in\mathsf{Eff}\left(\mathrm{A}\right)$ and $b\in\mathsf{Eff}\left(\mathrm{B}\right)$,
some others will not. If considering only product effects is enough
to do tomography on all bipartite states, we say that the theory satisfies
Local Tomography \cite{Araki-real,Bergia-local-tomography,Hardy-informational-1,Chiribella-purification,Masanes-physical-derivation,Hardy-informational-2,Masanes+all}.
Quantum theory on complex Hilbert spaces satisfies Local Tomography,
but quantum theory on real Hilbert spaces does \emph{not} \cite{Stuckelberg,Araki-real,Wootters-real,Hardy-real}.
If a theory satisfies Local Tomography, we have $\mathsf{St}_{\mathbb{R}}\left(\mathrm{AB}\right)=\mathsf{St}_{\mathbb{R}}\left(\mathrm{A}\right)\otimes\mathsf{St}_{\mathbb{R}}\left(\mathrm{B}\right)$
\cite{Chiribella-purification}. In this case, product states are
enough to characterise all states of a composite system, because every
state of the composite system can be written as a linear combination
of product states. In other words, there are no ``genuinely new''
states arising when systems are composed. In the following we will
\emph{not} assume Local Tomography as an axiom, in fact we will provide
concrete examples of theories that violate it.
\end{rem}

We can also take linear combinations with non-negative coefficients,
they are called \emph{conical combinations}. Using conical combinations,
states (resp.\ effects) span a convex cone, the cone of states $\mathsf{St}_{+}\left(\mathrm{A}\right)$
(resp.\ the cone of effects $\mathsf{Eff}_{+}\left(\mathrm{A}\right)$).
These two cones $K$ are \emph{proper} by construction: it means that
\begin{enumerate}
\item $0\in K$;
\item let $\xi\neq0$; if $\xi\in\mathsf{St}_{+}\left(\mathrm{A}\right)$,
then $-\xi\notin\mathsf{St}_{+}\left(\mathrm{A}\right)$;
\item for every vector $\xi\in V$, there exist $\xi_{+},\xi_{-}\in K$
such that $\xi=\xi_{+}-\xi_{-}$.
\end{enumerate}
Here $K$ denotes the cones $\mathsf{St}_{+}\left(\mathrm{A}\right)$
and $\mathsf{Eff}_{+}\left(\mathrm{A}\right)$, and $V$ the corresponding
vector spaces $\mathsf{St}_{\mathbb{R}}\left(\mathrm{A}\right)$ and
$\mathsf{Eff}_{\mathbb{R}}\left(\mathrm{A}\right)$. Once we define
the cone of states, we can consider the \emph{dual cone} $\mathsf{St}_{+}^{*}\left(\mathrm{A}\right)$,
defined as the set of linear functionals $X$ such that $\left(X\middle|\xi\right)\geq0$
for all $\xi\in\mathsf{St}_{+}\left(\mathrm{A}\right)$. Clearly,
the elements of the cone of effects are in the dual cone, because
a conical combination of effects yields a non-negative number when
applied to a state, so $\mathsf{Eff}_{+}\left(\mathrm{A}\right)\subseteq\mathsf{St}_{+}^{*}\left(\mathrm{A}\right)$.
\begin{defn}[No-restriction hypothesis \cite{Chiribella-purification}]
We say that a theory is \emph{non-restricted}, or that it satisfies
the no-restriction hypothesis, if $\mathsf{Eff}_{+}\left(\mathrm{A}\right)=\mathsf{St}_{+}^{*}\left(\mathrm{A}\right)$
for every system.
\end{defn}

While this may look just a statement of mathematical interest, it
has some important physical implications. Consider the subset of $\mathsf{St}_{+}^{*}\left(\mathrm{A}\right)$
made of linear functionals $f$ such that $\left(f\middle|\rho\right)\in\left[0,1\right]$
for all states $\rho$. In a non-restricted theory, these elements
$f$ are also valid effects. In other words, the no-restriction hypothesis
states that every mathematically allowed effect is also a physical
effect. Clearly the no-restriction hypothesis concerns more the mathematical
structure of the theory than its operational one. Indeed, it is the
duty of the physical theory to specify what objects are to be considered
physical effects, even if they are admissible in principle, based
on their mathematical properties. For this reason, the no-restriction
hypothesis has been questioned various times on the basis of its lack
of operational motivation \cite{Chiribella-purification,Janotta-Lal,No-restriction2}.
Moreover, recently it has been shown that theories with almost quantum
correlations \cite{Almost-quantum} violate it \cite{Almost-no-restriction}.
\begin{example}
The trivial system of every theory has remarkable properties. We know
that states are in $\left[0,1\right]$, therefore $\mathsf{St}_{\mathbb{R}}\left(\mathrm{I}\right)=\mathbb{R}$.
The cone $\mathsf{St}_{+}\left(\mathrm{I}\right)$ is the set of non-negative
numbers. Similarly, $\mathsf{Eff}_{\mathbb{R}}\left(\mathrm{I}\right)=\mathbb{R}$,
and $\mathsf{Eff}_{+}\left(\mathrm{I}\right)=\mathbb{R}_{\geq0}$.
\end{example}

\begin{example}
Let us see what $\mathsf{St}_{\mathbb{R}}\left(\mathrm{A}\right)$
and $\mathsf{Eff}_{\mathbb{R}}\left(\mathrm{A}\right)$ are in finite-dimensional
quantum theory, namely when the Hilbert space is finite-dimensional
($\mathcal{H}\approx\mathbb{C}^{n}$, for $n\geq2$). $\mathsf{St}_{\mathbb{R}}\left(\mathrm{A}\right)$
is the vector space of hermitian matrices of order $n$. It is a real
vector space with dimension $n^{2}$. $\mathsf{Eff}_{\mathbb{R}}\left(\mathrm{A}\right)$
is again the vector space of hermitian matrices of order $n$.

Instead, the cones $\mathsf{St}_{+}\left(\mathrm{A}\right)$ and $\mathsf{Eff}_{+}\left(\mathrm{A}\right)$
are both the convex cone of positive semidefinite matrices. In quantum
theory, the no-restriction hypothesis is valid: $\mathsf{Eff}_{+}\left(\mathrm{A}\right)=\mathsf{St}_{+}^{*}\left(\mathrm{A}\right)$.
\end{example}

Now it is time to move our attention to the equivalence classes of
tomographically indistinguishable events for general tests.

First of all, note that every event $\mathcal{C}_{i}$ from $\mathrm{A}$
to $\mathrm{B}$ induces a linear operator $\widehat{\mathcal{C}}_{i}$
from $\mathsf{St}_{\mathbb{R}}\left(\mathrm{A}\right)$ to $\mathsf{St}_{\mathbb{R}}\left(\mathrm{B}\right)$.
We define $\widehat{\mathcal{C}}_{i}$ via its action on the spanning
set of states $\mathsf{St}\left(\mathrm{A}\right)$, as follows: 
\begin{equation}
\widehat{\mathcal{C}}_{i}:\rho_{\mathrm{A}}\mapsto\mathcal{C}_{i}\rho_{\mathrm{A}},\label{eq:linear extension}
\end{equation}
for every $\rho_{\mathrm{A}}\in\mathsf{St}\left(\mathrm{A}\right)$.
Note that $\mathcal{C}_{i}\rho_{\mathrm{A}}$ is a state of $\mathrm{B}$.
We want to check whether the linear extension of~\eqref{eq:linear extension}
is well defined. If $\xi\in\mathsf{St}_{\mathbb{R}}\left(\mathrm{A}\right)$,
we can express it as a linear combination of states, $\xi=\sum_{j}\lambda_{j}\rho_{j}$,
where $\lambda_{j}\in\mathbb{R}$ for every $j$. The obvious linear
extension of~\eqref{eq:linear extension} is $\widehat{\mathcal{C}}_{i}\xi:=\sum_{j}\lambda_{j}\widehat{\mathcal{C}}_{i}\rho_{j}$.
The problem is that, in general, $\xi$ does \emph{not} have a unique
expression as a linear combination of states. Suppose that $\xi=\sum_{j}\lambda_{j}\rho_{j}$
and $\xi=\sum_{j}\mu_{j}\sigma_{j}$, where $\mu_{j}\in\mathbb{R}$
for every $j$. Our extension $\widehat{\mathcal{C}}_{i}$ is well-defined
if and only if $\sum_{j}\lambda_{j}\widehat{\mathcal{C}}_{i}\rho_{j}=\sum_{j}\mu_{j}\widehat{\mathcal{C}}_{i}\sigma_{j}$
whenever $\sum_{j}\lambda_{j}\rho_{j}=\sum_{j}\mu_{j}\sigma_{j}$.
Using the linearity of summations, this problem is equivalent to checking
if $\sum_{j}\lambda_{j}\widehat{\mathcal{C}}_{i}\rho_{j}=0$ whenever
$\sum_{j}\lambda_{j}\rho_{j}=0$.

We have $\sum_{j}\lambda_{j}\rho_{j}=0$ if and only if $\sum_{j}\lambda_{j}\left(a\middle|\rho_{j}\right)=0$
for every effect $a\in\mathsf{Eff}\left(\mathrm{A}\right)$. Let $b$
be an arbitrary effect on $\mathrm{B}$. Then $b\widehat{\mathcal{C}}_{i}$
is an effect on $\mathrm{A}$, therefore $\sum_{j}\lambda_{j}\left(b\middle|\widehat{\mathcal{C}}_{i}\middle|\rho_{j}\right)=0$.
Since $b$ is arbitrary, this implies that $\sum_{j}\lambda_{j}\widehat{\mathcal{C}}_{i}\rho_{j}=0$.
This proves that the linear extension is well-defined.

Likewise, for every system $\mathrm{S}$, the event $\mathcal{C}_{i}\otimes\mathcal{I}_{\mathrm{S}}$
from $\mathrm{AS}$ to $\mathrm{BS}$ will induce a linear operator
from $\mathsf{St}_{\mathbb{R}}\left(\mathrm{AS}\right)$ to $\mathsf{St}_{\mathbb{R}}\left(\mathrm{BS}\right)$.
Two events $\mathcal{C}_{i}$ and $\mathcal{C}_{i}'$ from $\mathrm{A}$
to $\mathrm{B}$ are \emph{tomographically indistinguishable} if,
for every system $\mathrm{S}$, the linear operators associated with
$\mathcal{C}_{i}\otimes\mathcal{I}_{\mathrm{S}}$ and $\mathcal{C}_{i}'\otimes\mathcal{I}_{\mathrm{S}}$
are the same.

In other words, for every system $\mathrm{S}$, and every state $\rho\in\mathsf{St}\left(\mathrm{AS}\right)$,
\[ \begin{aligned}\Qcircuit @C=1em @R=.7em @!R { & \multiprepareC{1}{\rho} & \qw \poloFantasmaCn{\rA} & \gate{\cC_i} & \qw \poloFantasmaCn{\rB}&\qw \\ & \pureghost{\rho} & \qw \poloFantasmaCn{\rS} & \qw &\qw&\qw}\end{aligned} ~= \!\!\!\! \begin{aligned}\Qcircuit @C=1em @R=.7em @!R { & \multiprepareC{1}{\rho} & \qw \poloFantasmaCn{\rA} & \gate{\cC'_i} & \qw \poloFantasmaCn{\rB} &\qw \\ & \pureghost{\rho} & \qw \poloFantasmaCn{\rS} & \qw &\qw &\qw}\end{aligned}~.\]Recalling
the definition of equal states, $\mathcal{C}_{i}$ and $\mathcal{C}_{i}'$
are tomographically indistinguishable if and only if, for every system
$\mathrm{S}$, every state $\rho\in\mathsf{St}\left(\mathrm{AS}\right)$,
and every effect $E\in\mathsf{St}\left(\mathrm{BS}\right)$ one has\[ \begin{aligned}\Qcircuit @C=1em @R=.7em @!R { & \multiprepareC{1}{\rho} & \qw \poloFantasmaCn{\rA} & \gate{\cC_i} & \qw \poloFantasmaCn{\rB}&\multimeasureD{1}{E} \\ & \pureghost{\rho} & \qw \poloFantasmaCn{\rS} &\qw &\qw &\ghost{E}}\end{aligned} ~= \!\!\!\! \begin{aligned}\Qcircuit @C=1em @R=.7em @!R { & \multiprepareC{1}{\rho} & \qw \poloFantasmaCn{\rA} & \gate{\cC'_i} & \qw \poloFantasmaCn{\rB} &\multimeasureD{1}{E} \\ & \pureghost{\rho} & \qw \poloFantasmaCn{\rS} & \qw &\qw &\ghost{E}}\end{aligned}~.\]

Again, we take the quotient set of events modulo the indistinguishability
relation.
\begin{defn}
Equivalence classes of indistinguishable events from $\mathrm{A}$
to $\mathrm{B}$ are called \emph{transformations} from $\mathrm{A}$
to $\mathrm{B}$.
\end{defn}

The set of transformations from $\mathrm{A}$ to $\mathrm{B}$ is
denoted by $\mathsf{Transf}\left(\mathrm{A},\mathrm{B}\right)$. The
set of transformations from $\mathrm{A}$ to itself is denoted simply
by $\mathsf{Transf}\left(\mathrm{A}\right)$.
\begin{rem}
One may wonder why we have given such a definition of tomographically
indistinguishable events, involving an ancillary system $\mathrm{S}$.
The most obvious way of defining tomographic indistinguishability
would have been to say that $\mathcal{C}_{i}$ and $\mathcal{C}_{i}'$
are indistinguishable if $\mathcal{C}_{i}\rho=\mathcal{C}_{i}'\rho$
for every $\rho\in\mathsf{St}\left(\mathrm{A}\right)$. Actually,
this is not enough for OPTs. Indeed, Wootters provided a counterexample
concerning quantum mechanics with real Hilbert space \cite{Wootters-real}.
It can be shown that there exist two events that are locally indistinguishable,
but if we add an ancillary system, they produce distinct output states.
The condition $\mathcal{C}_{i}\rho=\mathcal{C}_{i}'\rho$ for every
$\rho\in\mathsf{St}\left(\mathrm{A}\right)$ is sufficient to identify
indistinguishable events if the theory satisfies Local Tomography
(see \cite{Chiribella-purification} for further details).
\end{rem}

Again, we will assume that equivalence classes have been taken from
the very beginning, so we will consider tests as collections of transformations.
Transformations span a real vector space, denoted by $\mathsf{Transf}_{\mathbb{R}}\left(\mathrm{A},\mathrm{B}\right)$.
\begin{defn}
A deterministic transformation $\mathcal{C}\in\mathsf{Transf}\left(\mathrm{A},\mathrm{B}\right)$
is called \emph{channel}.
\end{defn}

We will denote the set of channels from $\mathrm{A}$ to $\mathrm{B}$
by $\mathsf{DetTransf}\left(\mathrm{A},\mathrm{B}\right)$ (or by
$\mathsf{DetTransf}\left(\mathrm{A}\right)$ if $\mathrm{B}=\mathrm{A}$).

Among all possible channels, reversible ones are particularly important.
\begin{defn}
A channel $\mathcal{U}\in\mathsf{Transf}\left(\mathrm{A},\mathrm{B}\right)$
is said \emph{reversible }if it is invertible, namely if there is
another channel $\mathcal{U}^{-1}\in\mathsf{Transf}\left(\mathrm{B},\mathrm{A}\right)$,
called the \emph{inverse}, such that $\mathcal{U}^{-1}\mathcal{U}=\mathcal{I}_{\mathrm{A}}$
and $\mathcal{U}\mathcal{U}^{-1}=\mathcal{I}_{\mathrm{B}}$.
\end{defn}

Clearly, reversible channels on $\mathrm{A}$ form a group, denoted
$G_{\mathrm{A}}$. Now, we can rephrase the definition of operationally
equivalent systems: two systems $\mathrm{A}$ and $\mathrm{A}'$ are
operationally equivalent if there exists a reversible channel from
$\mathrm{A}$ to $\mathrm{A}'$.

Before moving on, let us see what transformations, channels, and reversible
channels are in quantum theory.
\begin{example}
A test in quantum theory from $\mathcal{H}_{\mathrm{A}}$ to $\mathcal{H}_{\mathrm{B}}$
is a collection of completely positive, trace non-increasing linear
maps $\left\{ \mathcal{C}_{k}\right\} $, called quantum operations
\cite{Nielsen-Chuang,Preskill}, such that $\sum_{k}\mathcal{C}_{k}$
is a trace-preserving map. Each quantum operation maps linear operators
on $\mathcal{H}_{\mathrm{A}}$ into linear operators on $\mathcal{H}_{\mathrm{B}}$.
A test is a quantum instrument, namely a collection of quantum operations
\cite{Nielsen-Chuang}.

A channel is a completely positive trace-preserving map from linear
operators on $\mathcal{H}_{\mathrm{A}}$ to linear operators on $\mathcal{H}_{\mathrm{B}}$.

Finally, reversible channels are unitary channels. They act on $\mathrm{A}$
as $\mathcal{U}\left(\rho\right)=U\rho U^{\dagger}$, where $U$ is
a unitary operator. It follows that two systems are operationally
equivalent if and only if their Hilbert spaces have the same dimension,
otherwise it is not possible to define unitary operators from one
space to the other.
\end{example}

\subsection{Purity and coarse-graining}

Even in an abstract probabilistic theory, it makes sense to define
pure and mixed states, or, more generally, pure and non-pure transformations.
The idea behind it is \emph{coarse-graining}. Let us clarify this
idea with the example of the roll of a die \cite{Chiribella-educational}.
In this random experiment, there are some atomic events, which cannot
be decomposed further: they are the numbers from 1 to 6. So, an atomic
event is, for example, ``the outcome of the roll is 2''. However,
we can consider the event ``the outcome of the roll is odd''. This
event is the union of the atomic events relative to 1, 3, 5. We have
just done a coarse-graining: we joined together some outcomes, neglecting
some information. Indeed, if we know only that the outcome was ``odd'',
we cannot retrieve any information about which number actually came
out. In this vein, we give the following definition.
\begin{defn}
A test $\left\{ \mathcal{B}_{j}\right\} _{j\in Y}$ is a \emph{coarse-graining}
of the test $\left\{ \mathcal{A}_{i}\right\} _{i\in X}$ if there
is a partition $\left\{ X_{j}\right\} _{j\in Y}$ of $X$ such that
$\mathcal{B}_{j}=\sum_{i\in X_{j}}\mathcal{A}_{i}$. In this case,
we say that $\left\{ \mathcal{A}_{i}\right\} _{i\in X}$ is a \emph{refinement}
of $\left\{ \mathcal{B}_{j}\right\} _{j\in Y}$.
\end{defn}

As we can see, this definition gives a precise characterisation of
what we mean by ``joining together outcomes''. A test that refines
another extracts more information than the other. It is clear that
if $\left\{ \mathcal{B}_{j}\right\} _{i\in Y}$ is a coarse-graining
of the test $\left\{ \mathcal{A}_{i}\right\} _{i\in X}$, it has fewer
outcomes. This concept is easily explained in fig.~\ref{fig:coarse-graining}.
\begin{figure}
\begin{centering}
\includegraphics[scale=0.8]{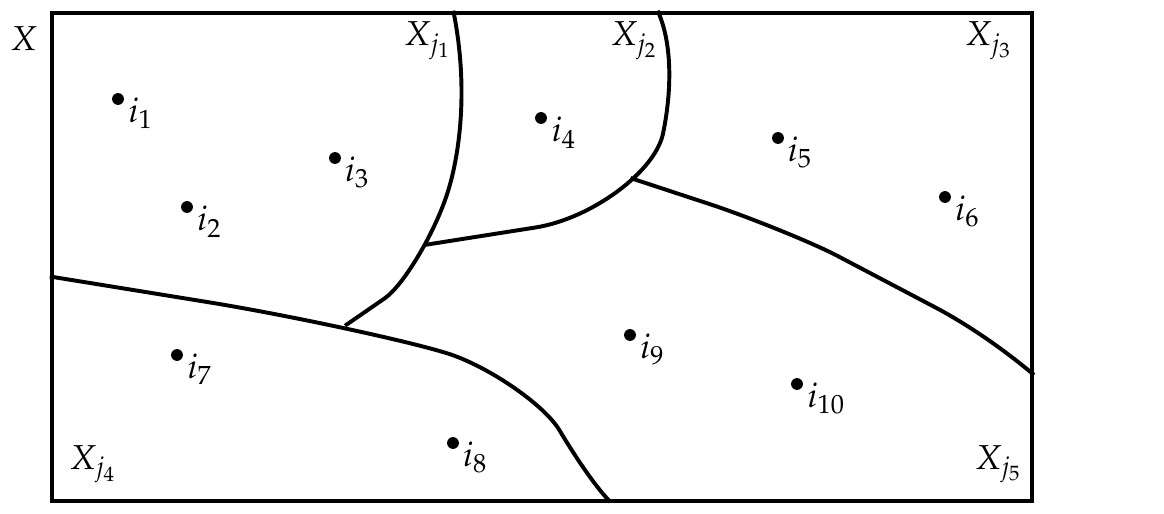}
\par\end{centering}
\caption{\label{fig:coarse-graining}The outcome set $X$ of the test $\left\{ \mathcal{A}_{i}\right\} _{i\in X}$
has 10 outcomes. To perform a coarse-graining of it, we lump together
some of its outcomes, relabelling them as a new outcome. For example,
the outcomes $i_{1}$, $i_{2}$, and $i_{3}$ are relabelled as $j_{1}$.
This gives rise to a partition $\left\{ X_{j}\right\} _{j\in Y}$
of $X$. We associate a new transformation with each set in the partition,
such that it is the sum of the transformations associated with the
outcomes contained in that set. Thus $\mathcal{B}_{j_{1}}=\mathcal{A}_{i_{1}}+\mathcal{A}_{i_{2}}+\mathcal{A}_{i_{3}}$.
The new test $\left\{ \mathcal{B}_{j}\right\} _{j\in Y}$ has 5 outcomes.}
\end{figure}

By performing a coarse-graining, we can associate a deterministic
transformation with every test. Indeed, let us take a test $\left\{ \mathcal{C}_{i}\right\} _{i\in X}$
from $\mathrm{A}$ to $\mathrm{B}$ and let us sum over \emph{all}
the outcomes $i\in X$. Then we obtain the channel $\mathcal{C}=\sum_{i\in X}\mathcal{C}_{i}$
from $\mathrm{A}$ to $\mathrm{B}$, which is called the channel associated
with the test $\left\{ \mathcal{C}_{i}\right\} _{i\in X}$. Similarly,
we can obtain a deterministic state by summing all the states in a
preparation-test; and we can get a deterministic effect by summing
all the effects in an observation-test.

We can consider also refinements of single transformations.
\begin{defn}
Let $\mathcal{C}$ be a transformation from system $\mathrm{A}$ to
system $\mathrm{B}$. Consider a test $\left\{ \mathcal{D}_{i}\right\} _{i\in X}$
from system $\mathrm{A}$ to system $\mathrm{B}$ and a subset $X_{0}\subseteq X$
such that $\mathcal{C}=\sum_{i\in X_{0}}\mathcal{D}_{i}$. Each transformation
$\mathcal{D}_{i}$, for $i\in X_{0}$ is a \emph{refinement} of $\mathcal{C}$.
\end{defn}

We can always obtain a refinement of a transformation $\mathcal{T}$
by taking a subset of a test, made of $\left\{ p_{i}\mathcal{T}\right\} _{i\in X_{0}}$,
with the property that $p_{i}\in\left(0,1\right]$ for every $i\in X_{0}$,
and $\sum_{i\in X_{0}}p_{i}=1$. Some transformations cannot be refined
further, and they admit only trivial refinements of this form above.
\begin{defn}
A transformation $\mathcal{T}$ is \emph{pure} if it has only trivial
refinements.
\end{defn}

In other words, it is not possible to extract further information
from a pure transformation.

Clearly, this definition applies also to states, which are particular
transformations from the trivial system $\mathrm{I}$ to a system
$\mathrm{A}$. We will denote the set of pure states of system $\mathrm{A}$
by $\mathsf{PurSt}\left(\mathrm{A}\right)$. The non-pure states are
called \emph{mixed}. In this way, a pure state represents maximal
knowledge about the preparation of a system, whereas a mixed state
expresses some lack of information. Similarly, we will denote the
set of pure effects of system $\mathrm{A}$ by $\mathsf{PurEff}\left(\mathrm{A}\right)$.

Let us see some examples in quantum theory.
\begin{example}
If we diagonalise a density operator $\rho=\sum_{j}p_{j}\ket{\psi_{j}}\bra{\psi_{j}}$,
each term $p_{j}\ket{\psi_{j}}\bra{\psi_{j}}$ is a refinement\footnote{With a little abuse of terminology we also say that $\ket{\psi_{j}}\bra{\psi_{j}}$
is a refinement of $\rho$.} of $\rho$. More generally, a refinement of $\rho$ is a state $\sigma$
such that $\sigma\leq\rho$. This means that the support\footnote{Recall the support of a matrix is the orthogonal complement of its
kernel.} of $\sigma$ is contained in the support of $\rho$ (see \cite[appendix A.1]{Scandolo-thesis}
for a proof). A pure state is a density operator $\lambda\ket{\psi}\bra{\psi}$,
with $\lambda\in\left(0,1\right]$, viz.\ proportional to a rank-one
projector.

In quantum mechanics, we can associate Kraus operators $\left\{ M_{k}\right\} $
with every quantum operation $\mathcal{C}$, such that $\mathcal{C}\left(\rho\right)=\sum_{k}M_{k}\rho M_{k}^{\dagger}$,
for every state $\rho$ \cite{Nielsen-Chuang,Preskill}. A quantum
operation is pure if and only it has only one Kraus operator.
\end{example}

Reversible channels are \emph{not} pure in general, unless some axioms
are imposed on the theory (see section~\ref{sec:The-axioms-and}).
\begin{example}
Consider finite-dimensional classical theory. Here states are vectors
of non-negative numbers, whose entries sum to a number less than or
equal to 1. In symbols, states are vectors $\mathbf{p}\in\mathbb{R}^{d}$,
such that $p_{i}\geq0$ for all $i=1,\ldots,d$, and $\sum_{i=1}^{d}p_{i}\leq1$.
The cone of states is spanned by the $d$ pure states $\delta_{i}$,
where $\delta_{i}$ denotes the vector with all zero entries, except
the $i$th, which is 1. The identity channel can be written as $\mathcal{I}=\sum_{i=1}^{d}\left|\delta_{i}\right)\left(\delta_{i}^{\dagger}\right|$,
where the $\delta_{i}^{\dagger}$'s are the effects such that $\left(\delta_{i}^{\dagger}\middle|\delta_{j}\right)=\delta_{ij}$,
where $\delta_{ij}$ is Kronecker delta. Indeed, every element of
the cone of states can be written as $\xi=\sum_{i=1}^{d}\lambda_{i}\delta_{i}$,
where $\lambda_{i}\geq0$, and
\[
\mathcal{I}\xi=\sum_{i=1}^{d}\sum_{j=1}^{d}\lambda_{j}\left|\delta_{i}\right)\left(\delta_{i}^{\dagger}\middle|\delta_{j}\right)=\sum_{i=1}^{d}\lambda_{i}\delta_{i}=\xi.
\]
Since classical theory satisfies Local Tomography, this is enough
to conclude that $\sum_{i=1}^{d}\left|\delta_{i}\right)\left(\delta_{i}^{\dagger}\right|$
is the identity channel. This means that the identity channel is the
coarse-graining of the test $\left\{ \left|\delta_{i}\right)\left(\delta_{i}^{\dagger}\right|\right\} _{i=1}^{d}$.
Therefore the identity channel in classical theory, albeit reversible,
is \emph{not} pure.
\end{example}

However, although not necessarily pure themselves, reversible channels
send pure states into pure states \cite{Chiribella-Scandolo-entanglement}.
They do not alter the ``purity'' of a state: they also map mixed
states into mixed states.
\begin{lem}
Let $\mathcal{U}$ be a reversible channel from $\mathrm{A}$ to $\mathrm{B}$.
Then $\psi\in\mathsf{St}\left(\mathrm{A}\right)$ is pure if and only
if $\mathcal{U}\psi\in\mathsf{St}\left(\mathrm{B}\right)$ is pure.
\end{lem}

\begin{proof}
Necessity. Let us write $\mathcal{U}\psi$ as a coarse-graining of
other states.
\begin{equation}
\mathcal{U}\psi=\sum_{i}\rho_{i}\label{eq:pure reversible}
\end{equation}
Let us show that each refinement $\rho_{i}$ of $\mathcal{U}\psi$
is trivial, that is $\rho_{i}=p_{i}\mathcal{U}\psi$, for some $p_{i}\in\left(0,1\right]$,
with $\sum_{i}p_{i}=1$. By applying $\mathcal{U}^{-1}$ to both sides
of eq.~\eqref{eq:pure reversible}, we have $\psi=\sum_{i}\mathcal{U}^{-1}\rho_{i}$.
Since $\psi$ is pure, each refinement $\mathcal{U}^{-1}\rho_{i}$
is trivial, namely 
\begin{equation}
\mathcal{U}^{-1}\rho_{i}=p_{i}\psi,\label{eq:pure reversible 2}
\end{equation}
for some $p_{i}\in\left(0,1\right]$, with $\sum_{i}p_{i}=1$. By
applying $\mathcal{U}$ to both sides of eq.~\eqref{eq:pure reversible 2},
we have $\rho_{i}=p_{i}\mathcal{U}\psi$. Since every refinement $\mathcal{U}\psi$
is trivial, $\mathcal{U}\psi$ is pure.

Sufficiency follows from necessity, by applying the reversible channel
$\mathcal{U}^{-1}$ to $\mathcal{U}\psi$, which is pure by hypothesis.
\end{proof}
A similar statement holds also for effects: $b\in\mathsf{Eff}\left(\mathrm{B}\right)$
is pure if and only if $b\mathcal{U}\in\mathsf{Eff}\left(\mathrm{A}\right)$
is pure.

\subsection{Norms for states, effects, and transformations}

We can define a norm in the vector space of states $\mathsf{St}_{\mathbb{R}}\left(\mathrm{A}\right)$.
It is defined as follows \cite{Chiribella-purification}.
\begin{defn}
Let $\xi\in\mathsf{St}_{\mathbb{R}}\left(\mathrm{A}\right)$. The
\emph{operational norm} of $\xi$ is
\[
\left\Vert \xi\right\Vert :=\sup_{a\in\mathsf{Eff}\left(\mathrm{A}\right)}\left(a\middle|\xi\right)-\inf_{a\in\mathsf{Eff}\left(\mathrm{A}\right)}\left(a\middle|\xi\right).
\]
\end{defn}

Let us show that it is indeed a norm. First of all, note that $\left\Vert \xi\right\Vert \geq0$
because $\sup_{a\in\mathsf{Eff}\left(\mathrm{A}\right)}\left(a\middle|\xi\right)\geq\inf_{a\in\mathsf{Eff}\left(\mathrm{A}\right)}\left(a\middle|\xi\right)$.
Then, let us show that $\left\Vert \xi\right\Vert =0$ only if $\xi=0$.
If $\left\Vert \xi\right\Vert =0$, then $\sup_{a\in\mathsf{Eff}\left(\mathrm{A}\right)}\left(a\middle|\xi\right)=\inf_{a\in\mathsf{Eff}\left(\mathrm{A}\right)}\left(a\middle|\xi\right)$.
Now, we have that $\sup_{a\in\mathsf{Eff}\left(\mathrm{A}\right)}\left(a\middle|\xi\right)\geq0$,
and $\inf_{a\in\mathsf{Eff}\left(\mathrm{A}\right)}\left(a\middle|\xi\right)\leq0$.
Indeed, write $\xi$ as $\xi=\xi_{+}-\xi_{-}$, where $\xi_{+},\xi_{-}\in\mathsf{St}_{+}\left(\mathrm{A}\right)$.
Then $\left(a\middle|\xi\right)=\left(a\middle|\xi_{+}\right)-\left(a\middle|\xi_{-}\right)$,
where $\left(a\middle|\xi_{+}\right)$ and $\left(a\middle|\xi_{-}\right)$
are both non-negative. As to the supremum, note that $\left(a\middle|\xi\right)\geq\left(a\middle|-\xi_{-}\right)$,
therefore
\[
\sup_{a\in\mathsf{Eff}\left(\mathrm{A}\right)}\left(a\middle|\xi\right)\geq\sup_{a\in\mathsf{Eff}\left(\mathrm{A}\right)}\left(a\middle|-\xi_{-}\right)=0.
\]
As to the infimum, note that $\left(a\middle|\xi\right)\leq\left(a\middle|\xi_{+}\right)$,
whence
\[
\inf_{a\in\mathsf{Eff}\left(\mathrm{A}\right)}\left(a\middle|\xi\right)\leq\inf_{a\in\mathsf{Eff}\left(\mathrm{A}\right)}\left(a\middle|\xi_{+}\right)=0.
\]
Then the only possibility of having $\sup_{a\in\mathsf{Eff}\left(\mathrm{A}\right)}\left(a\middle|\xi\right)=\inf_{a\in\mathsf{Eff}\left(\mathrm{A}\right)}\left(a\middle|\xi\right)$
is when
\begin{equation}
\sup_{a\in\mathsf{Eff}\left(\mathrm{A}\right)}\left(a\middle|\xi\right)=\inf_{a\in\mathsf{Eff}\left(\mathrm{A}\right)}\left(a\middle|\xi\right)=0.\label{eq:sup=00003Dinf}
\end{equation}
Since
\[
\sup_{a\in\mathsf{Eff}\left(\mathrm{A}\right)}\left(a\middle|\xi\right)\geq\left(a\middle|\xi\right)\geq\inf_{a\in\mathsf{Eff}\left(\mathrm{A}\right)}\left(a\middle|\xi\right),
\]
eq.~\eqref{eq:sup=00003Dinf} implies $\left(a\middle|\xi\right)=0$
for every effect $a$, so $\xi=0$.

Let us prove that $\left\Vert \lambda\xi\right\Vert =\left|\lambda\right|\left\Vert \xi\right\Vert $
for every $\lambda\in\mathbb{R}$, and every $\xi$. Let $\lambda\geq0$.
Then $\sup_{a\in\mathsf{Eff}\left(\mathrm{A}\right)}\left(a\middle|\lambda\xi\right)=\lambda\sup_{a\in\mathsf{Eff}\left(\mathrm{A}\right)}\left(a\middle|\xi\right)$,
and similarly $\inf_{a\in\mathsf{Eff}\left(\mathrm{A}\right)}\left(a\middle|\lambda\xi\right)=\lambda\inf_{a\in\mathsf{Eff}\left(\mathrm{A}\right)}\left(a\middle|\xi\right).$
Therefore
\[
\left\Vert \lambda\xi\right\Vert =\lambda\sup_{a\in\mathsf{Eff}\left(\mathrm{A}\right)}\left(a\middle|\xi\right)-\lambda\inf_{a\in\mathsf{Eff}\left(\mathrm{A}\right)}\left(a\middle|\xi\right)=\lambda\left\Vert \xi\right\Vert .
\]
Now, let $\lambda<0$. In this case $\sup_{a\in\mathsf{Eff}\left(\mathrm{A}\right)}\left(a\middle|\lambda\xi\right)=\lambda\inf_{a\in\mathsf{Eff}\left(\mathrm{A}\right)}\left(a\middle|\xi\right)$,
and similarly $\inf_{a\in\mathsf{Eff}\left(\mathrm{A}\right)}\left(a\middle|\lambda\xi\right)=\lambda\sup_{a\in\mathsf{Eff}\left(\mathrm{A}\right)}\left(a\middle|\xi\right)$.
Hence
\[
\left\Vert \lambda\xi\right\Vert =\lambda\inf_{a\in\mathsf{Eff}\left(\mathrm{A}\right)}\left(a\middle|\xi\right)-\lambda\sup_{a\in\mathsf{Eff}\left(\mathrm{A}\right)}\left(a\middle|\xi\right)=\left(-\lambda\right)\left\Vert \xi\right\Vert .
\]
In conclusion $\left\Vert \lambda\xi\right\Vert =\left|\lambda\right|\left\Vert \xi\right\Vert $.

Finally, we need to prove the triangle inequality: $\left\Vert \xi+\eta\right\Vert \leq\left\Vert \xi\right\Vert +\left\Vert \eta\right\Vert $,
for every $\xi$ and $\eta$. We have
\begin{equation}
\left\Vert \xi+\eta\right\Vert =\sup_{a\in\mathsf{Eff}\left(\mathrm{A}\right)}\left(a\middle|\xi+\eta\right)-\inf_{a\in\mathsf{Eff}\left(\mathrm{A}\right)}\left(a\middle|\xi+\eta\right).\label{eq:norm sum}
\end{equation}
Now
\begin{equation}
\sup_{a\in\mathsf{Eff}\left(\mathrm{A}\right)}\left(a\middle|\xi+\eta\right)\leq\sup_{a\in\mathsf{Eff}\left(\mathrm{A}\right)}\left(a\middle|\xi\right)+\sup_{a\in\mathsf{Eff}\left(\mathrm{A}\right)}\left(a\middle|\eta\right)\label{eq:sup sum}
\end{equation}
and
\begin{equation}
\inf_{a\in\mathsf{Eff}\left(\mathrm{A}\right)}\left(a\middle|\xi+\eta\right)\geq\inf_{a\in\mathsf{Eff}\left(\mathrm{A}\right)}\left(a\middle|\xi\right)+\inf_{a\in\mathsf{Eff}\left(\mathrm{A}\right)}\left(a\middle|\eta\right),\label{eq:inf sum}
\end{equation}
as it is not hard to show. Thus, putting eqs.~\eqref{eq:norm sum},
\eqref{eq:sup sum}, and \eqref{eq:inf sum} together, $\left\Vert \xi+\eta\right\Vert \leq\left\Vert \xi\right\Vert +\left\Vert \eta\right\Vert $.
This shows that the operational norm is indeed a norm. In subsection~\ref{subsec:Extending-the-diagonalisation}
we will show that this norm is in fact the 1-norm (or the trace norm
in quantum theory). For the trivial system, where $\mathsf{St}_{\mathbb{R}}\left(\mathrm{I}\right)=\mathbb{R}$,
we have that $\left\Vert \xi\right\Vert =\left|\xi\right|$. Indeed,
either $\xi\in\mathsf{St}_{+}\left(\mathrm{I}\right)$ or $-\xi\in\mathsf{St}_{+}\left(\mathrm{I}\right)$.
Therefore 
\[
\left\Vert \xi\right\Vert =\sup_{k\in\left[0,1\right]}k\left|\xi\right|=\left|\xi\right|,
\]
where we have used the fact that the effects of the trivial system
are elements in $\left[0,1\right]$.

Clearly, for a state, $\left\Vert \rho\right\Vert \leq1$, because
effects yield probabilities when applied to a state. Therefore the
set of states is bounded.
\begin{defn}
A state $\rho\in\mathsf{St}\left(\mathrm{A}\right)$ is \emph{normalised}
if $\left\Vert \rho\right\Vert =1$.
\end{defn}

We will denote the set of normalised states of system $\mathrm{A}$
by $\mathsf{St}_{1}\left(\mathrm{A}\right)$, and the set of normalised
\emph{pure} states by $\mathsf{PurSt}_{1}\left(\mathrm{A}\right)$.

This norm is non-decreasing under the action of transformations \cite[lemma 1]{Chiribella-purification}.
\begin{prop}
\label{prop:norm non-increasing}For every vector $\xi\in\mathsf{St}_{\mathbb{R}}\left(\mathrm{A}\right)$,
$\left\Vert \mathcal{T}\xi\right\Vert _{\mathrm{B}}\leq\left\Vert \xi\right\Vert _{\mathrm{A}}$,
where $\mathcal{T}$ is a transformation from $\mathrm{A}$ to $\mathrm{B}$.
If $\mathcal{T}$ is a reversible channel, then one has the equality.
\end{prop}

\begin{proof}
By definition
\[
\left\Vert \mathcal{T}\xi\right\Vert _{\mathrm{B}}=\sup_{b\in\mathsf{Eff}\left(\mathrm{B}\right)}\left(b\middle|\mathcal{T}\middle|\xi\right)-\inf_{b\in\mathsf{Eff}\left(\mathrm{B}\right)}\left(b\middle|\mathcal{T}\middle|\xi\right).
\]
Now, $b\mathcal{T}$ is an effect of $\mathrm{A}$, so $\sup_{b\in\mathsf{Eff}\left(\mathrm{B}\right)}\left(b\middle|\mathcal{T}\middle|\xi\right)\leq\sup_{a\in\mathsf{Eff}\left(\mathrm{A}\right)}\left(a\middle|\xi\right)$,
and $\inf_{b\in\mathsf{Eff}\left(\mathrm{B}\right)}\left(b\middle|\mathcal{T}\middle|\xi\right)\geq\inf_{a\in\mathsf{Eff}\left(\mathrm{A}\right)}\left(a\middle|\xi\right)$.
We conclude that
\[
\left\Vert \mathcal{T}\xi\right\Vert _{\mathrm{B}}\leq\sup_{a\in\mathsf{Eff}\left(\mathrm{A}\right)}\left(a\middle|\xi\right)-\inf_{a\in\mathsf{Eff}\left(\mathrm{A}\right)}\left(a\middle|\xi\right)=\left\Vert \xi\right\Vert _{\mathrm{A}}.
\]
If $\mathcal{T}$ is a reversible channel, then one has 
\[
\left\Vert \xi\right\Vert _{\mathrm{A}}=\left\Vert \mathcal{T}^{-1}\mathcal{T}\xi\right\Vert _{\mathrm{A}}\leq\left\Vert \mathcal{T}\xi\right\Vert _{\mathrm{B}},
\]
which means $\left\Vert \mathcal{T}\xi\right\Vert _{\mathrm{B}}=\left\Vert \xi\right\Vert _{\mathrm{A}}$.
\end{proof}
In particular, if we have $\xi=\rho-\sigma$, with $\rho$ and $\sigma$
two states, for every transformation
\[
\left\Vert \mathcal{T}\rho-\mathcal{T}\sigma\right\Vert \leq\left\Vert \rho-\sigma\right\Vert ,
\]
thus recovering the same statement as in quantum theory.

The norm of states is the starting point to give the definition of
a norm for the elements of the vector space of transformations.
\begin{defn}
Let $\Delta\in\mathsf{Transf}_{\mathbb{R}}\left(\mathrm{A},\mathrm{B}\right)$.
The\emph{ norm} of $\Delta$ is
\[
\left\Vert \Delta\right\Vert :=\sup_{\mathrm{S}}\sup_{\rho\in\mathsf{St}\left(\mathrm{AS}\right)}\left\Vert \left(\Delta\otimes\mathcal{I}_{\mathrm{S}}\right)\rho\right\Vert .
\]
\end{defn}

Let us show that this is indeed a norm. First of all, note that, for
every system $\mathrm{S}$, and every state $\rho\in\mathsf{St}\left(\mathrm{AS}\right)$,
\[
\sup_{\mathrm{S}}\sup_{\rho\in\mathsf{St}\left(\mathrm{AS}\right)}\left\Vert \left(\Delta\otimes\mathcal{I}_{\mathrm{S}}\right)\rho\right\Vert \geq\left\Vert \left(\Delta\otimes\mathcal{I}_{\mathrm{S}}\right)\rho\right\Vert \geq0,
\]
whence $\left\Vert \Delta\right\Vert \geq0$. Moreover, the same inequality
shows that if $\left\Vert \Delta\right\Vert =0$, then $\left\Vert \left(\Delta\otimes\mathcal{I}_{\mathrm{S}}\right)\rho\right\Vert =0$,
which in turn implies $\left(\Delta\otimes\mathcal{I}_{\mathrm{S}}\right)\rho=0$
for every system $\mathrm{S}$, and every state $\rho\in\mathsf{St}\left(\mathrm{AS}\right)$.
We conclude that $\Delta=0$.

Now, let us show that $\left\Vert \lambda\Delta\right\Vert =\left|\lambda\right|\left\Vert \Delta\right\Vert $,
for every $\lambda\in\mathbb{R}$. Now,
\[
\left\Vert \lambda\Delta\right\Vert =\sup_{\mathrm{S}}\sup_{\rho\in\mathsf{St}\left(\mathrm{AS}\right)}\left\Vert \left(\lambda\Delta\otimes\mathcal{I}_{\mathrm{S}}\right)\rho\right\Vert =\sup_{\mathrm{S}}\sup_{\rho\in\mathsf{St}\left(\mathrm{AS}\right)}\left|\lambda\right|\left\Vert \left(\Delta\otimes\mathcal{I}_{\mathrm{S}}\right)\rho\right\Vert =
\]
\[
=\left|\lambda\right|\sup_{\mathrm{S}}\sup_{\rho\in\mathsf{St}\left(\mathrm{AS}\right)}\left\Vert \left(\Delta\otimes\mathcal{I}_{\mathrm{S}}\right)\rho\right\Vert =\left|\lambda\right|\left\Vert \Delta\right\Vert .
\]
Finally, we prove the triangle inequality: $\left\Vert \Delta+\Xi\right\Vert \leq\left\Vert \Delta\right\Vert +\left\Vert \Xi\right\Vert $,
for every $\Delta,\Xi\in\mathsf{Transf}_{\mathbb{R}}\left(\mathrm{A},\mathrm{B}\right)$.
\[
\left\Vert \Delta+\Xi\right\Vert =\sup_{\mathrm{S}}\sup_{\rho\in\mathsf{St}\left(\mathrm{AS}\right)}\left\Vert \left[\left(\Delta+\Xi\right)\otimes\mathcal{I}_{\mathrm{S}}\right]\rho\right\Vert =\sup_{\mathrm{S}}\sup_{\rho\in\mathsf{St}\left(\mathrm{AS}\right)}\left\Vert \left(\Delta\otimes\mathcal{I}_{\mathrm{S}}\right)\rho+\left(\Xi\otimes\mathcal{I}_{\mathrm{S}}\right)\rho\right\Vert .
\]
Now, for every system $\mathrm{S}$, and every state $\rho\in\mathsf{St}\left(\mathrm{AS}\right)$,
we have
\[
\left\Vert \left(\Delta\otimes\mathcal{I}_{\mathrm{S}}\right)\rho+\left(\Xi\otimes\mathcal{I}_{\mathrm{S}}\right)\rho\right\Vert \leq\left\Vert \left(\Delta\otimes\mathcal{I}_{\mathrm{S}}\right)\rho\right\Vert +\left\Vert \left(\Xi\otimes\mathcal{I}_{\mathrm{S}}\right)\rho\right\Vert \leq
\]
\[
\leq\sup_{\mathrm{S}}\sup_{\rho\in\mathsf{St}\left(\mathrm{AS}\right)}\left\Vert \left(\Delta\otimes\mathcal{I}_{\mathrm{S}}\right)\rho\right\Vert +\sup_{\mathrm{S}}\sup_{\rho\in\mathsf{St}\left(\mathrm{AS}\right)}\left\Vert \left(\Xi\otimes\mathcal{I}_{\mathrm{S}}\right)\rho\right\Vert .
\]
In conclusion, $\left\Vert \Delta+\Xi\right\Vert \leq\left\Vert \Delta\right\Vert +\left\Vert \Xi\right\Vert $,
thus proving that we are dealing with an actual norm. Note that for
a generic transformation $\mathcal{T}\in\mathsf{Transf}\left(\mathrm{A},\mathrm{B}\right)$
we have
\[
\left\Vert \mathcal{T}\right\Vert =\sup_{\mathrm{S}}\sup_{\rho\in\mathsf{St}\left(\mathrm{AS}\right)}\left\Vert \left(\mathcal{T}\otimes\mathcal{I}_{\mathrm{S}}\right)\rho\right\Vert \leq1,
\]
because $\left(\mathcal{T}\otimes\mathcal{I}_{\mathrm{S}}\right)\rho$
is a physical state. Therefore the set of transformations is bounded.

When $\Delta\in\mathsf{Eff}_{\mathbb{R}}\left(\mathrm{A}\right)$
the norm of $\Delta$ takes a simpler form \cite[lemma 8]{Chiribella-purification}.
\begin{prop}
Let $X\in\mathsf{Eff}_{\mathbb{R}}\left(\mathrm{A}\right)$. Then
\[
\left\Vert X\right\Vert =\sup_{\rho\in\mathsf{St}\left(\mathrm{A}\right)}\left|\left(X\middle|\rho\right)\right|.
\]
\end{prop}

Note that this norm is the operator norm of the linear functional
$X$ on the set of states. Clearly, the set of effects $\mathsf{Eff}\left(\mathrm{A}\right)$
of a generic system $\mathrm{A}$ is bounded too, for $\left\Vert a\right\Vert \leq1$.
The effects $a$ such that $\left\Vert a\right\Vert =1$ are called
\emph{normalised}, and their set for system $\mathrm{A}$ will be
denoted by $\mathsf{Eff}_{1}\left(\mathrm{A}\right)$. Similarly,
the set of normalised pure effects will be denoted by $\mathsf{PurEff}_{1}\left(\mathrm{A}\right)$.

\subsection{Setting up a topology}

The definition of tomography for states, effects, and transformations
naturally yields a topology. The idea is that a sequence of states,
effects, transformations converges to a limit if the states, effects,
transformations in the sequence become tomographically indistinguishable
from the limit. Let us clarify this idea separately for states, effects,
and transformations.

We say that a sequence of states $\left\{ \rho_{n}\right\} $ of system
$\mathrm{A}$ converges to $\rho\in\mathsf{St}_{\mathbb{R}}\left(\mathrm{A}\right)$
if, for every effect $a\in\mathsf{Eff}\left(\mathrm{A}\right)$ we
have 
\begin{equation}
\lim_{n\rightarrow+\infty}\left(a\middle|\rho_{n}\right)=\left(a\middle|\rho\right).\label{eq:topology states}
\end{equation}
With this topology, the issue of the convergence of a sequence of
states is turned into the convergence of a sequence of real numbers.
Notice we wrote $\rho\in\mathsf{St}_{\mathbb{R}}\left(\mathrm{A}\right)$,
because in general the limit may not be a state, but it is just an
element of the vector space of states. The limit is a state if the
set of states is topologically closed. Physically this means that
every vector that can be arbitrarily well approximated by a sequence
of states must be a state. This is fairly natural to assume, and we
will do it.

We can extend the topology defined by eq.~\eqref{eq:topology states}
to all vectors of $\mathsf{St}_{\mathbb{R}}\left(\mathrm{A}\right)$,
for every system $\mathrm{A}$: $\left\{ \xi_{n}\right\} $ converges
to $\xi$ if
\[
\lim_{n\rightarrow+\infty}\left(a\middle|\xi_{n}\right)=\left(a\middle|\xi\right)
\]
for every effect $a\in\mathsf{Eff}\left(\mathrm{A}\right)$.

Dually, we define a topology on the set of effects: a sequence $\left\{ a_{n}\right\} $
of system $\mathrm{A}$ converges to $a\in\mathsf{Eff}_{\mathbb{R}}\left(\mathrm{A}\right)$
if, for every state $\rho\in\mathsf{St}\left(\mathrm{A}\right)$,
we have
\[
\lim_{n\rightarrow+\infty}\left(a_{n}\middle|\rho\right)=\left(a\middle|\rho\right).
\]
Again, if for every convergent sequence of effects, the limit $a$
is an effect too, the set of effects is closed for every system. We
can extend this topology to all $\mathsf{Eff}_{\mathbb{R}}\left(\mathrm{A}\right)$:
$\left\{ X_{n}\right\} $ converges to $X$ if, for every state $\rho$
\[
\lim_{n\rightarrow+\infty}\left(X_{n}\middle|\rho\right)=\left(X\middle|\rho\right).
\]
Finally, let us look at transformations. Recall that transformations
are defined by their action on half of bipartite states. Therefore
it is natural to say that a sequence $\left\{ \mathcal{T}_{n}\right\} $
of transformations from $\mathrm{A}$ to $\mathrm{B}$ converges to
$\mathcal{T}\in\mathsf{Transf}_{\mathbb{R}}\left(\mathrm{A},\mathrm{B}\right)$
if, for every system $\mathrm{S}$, one has 
\begin{equation}
\lim_{n\rightarrow+\infty}\left(\mathcal{T}_{n}\otimes\mathcal{I}_{\mathrm{S}}\right)\rho_{\mathrm{AS}}=\left(\mathcal{T}\otimes\mathcal{I}_{\mathrm{S}}\right)\rho_{\mathrm{AS}},\label{eq:convergence transformations}
\end{equation}
for every state $\rho\in\mathsf{St}\left(\mathrm{AS}\right)$. In
this way, we turn the convergence of a sequence of transformations
into the convergence of a sequence of states, defined above. Again,
if $\mathsf{Transf}\left(\mathrm{A},\mathrm{B}\right)$ is closed,
then every limit of every convergent sequence of transformations will
be a transformation.

Recalling eq.~\eqref{eq:topology states}, the condition of eq.~\eqref{eq:convergence transformations}
becomes\[ \lim_{n\rightarrow+\infty}\!\!\!\!\begin{aligned}\Qcircuit @C=1em @R=.7em @!R { & \multiprepareC{1}{\rho} & \qw \poloFantasmaCn{\rA} & \gate{\cT_n} & \qw \poloFantasmaCn{\rB} &\multimeasureD{1}{E}\\ & \pureghost{\rho} & \qw \poloFantasmaCn{\rS} &\qw &\qw & \ghost{E}}\end{aligned} ~= \!\!\!\! \begin{aligned}\Qcircuit @C=1em @R=.7em @!R { & \multiprepareC{1}{\rho} & \qw \poloFantasmaCn{\rA} & \gate{\cT} & \qw \poloFantasmaCn{\rB} &\multimeasureD{1}{E}\\ & \pureghost{\rho} & \qw \poloFantasmaCn{\rS} &\qw &\qw & \ghost{E}}\end{aligned}~, \]for
every effect $E\in\mathsf{Eff}\left(\mathrm{BS}\right)$. Note that
this equation covers also the previous cases for states and effects:
it is enough to take $\mathrm{A}$ (resp.\ $\mathrm{B}$) to be the
trivial system. Therefore we can give the following definition \cite{QuantumFromPrinciples}.
\begin{defn}[Operational topology]
A sequence of transformations $\left\{ \mathcal{T}_{n}\right\} $
from system $\mathrm{A}$ to system $\mathrm{B}$ converges to $\mathcal{T}$
if\[ \lim_{n\rightarrow+\infty}\!\!\!\!\begin{aligned}\Qcircuit @C=1em @R=.7em @!R { & \multiprepareC{1}{\rho} & \qw \poloFantasmaCn{\rA} & \gate{\cT_n} & \qw \poloFantasmaCn{\rB} &\multimeasureD{1}{E}\\ & \pureghost{\rho} & \qw \poloFantasmaCn{\rS} &\qw &\qw & \ghost{E}}\end{aligned} ~= \!\!\!\! \begin{aligned}\Qcircuit @C=1em @R=.7em @!R { & \multiprepareC{1}{\rho} & \qw \poloFantasmaCn{\rA} & \gate{\cT} & \qw \poloFantasmaCn{\rB} &\multimeasureD{1}{E}\\ & \pureghost{\rho} & \qw \poloFantasmaCn{\rS} &\qw &\qw & \ghost{E}}\end{aligned}~, \]for
every system $\mathrm{S},$ every state $\rho\in\mathsf{St}\left(\mathrm{AS}\right)$,
and every effect $E\in\mathsf{Eff}\left(\mathrm{BS}\right)$.
\end{defn}

Henceforth we will assume that all sets of states, effects, and transformations
are topologically closed. The idea behind this is that a transformation
(including a state or an effect) that can be arbitrarily well approximated
by physical transformations, must be a physical transformation too.
For the trivial system this is translated as follows.
\begin{lem}
The set $\mathsf{St}\left(\mathrm{I}\right)$ is either $\left\{ 0,1\right\} $
or the whole interval $\left[0,1\right]$.
\end{lem}

\begin{proof}
If only 0 and 1 are allowed probabilities, there is nothing to prove.
If instead the theory admits a probability $p\in\left(0,1\right)$,
we have a test $\left\{ p,1-p\right\} $ on the trivial system, with
which, repeating it several times and doing a suitable coarse-graining,
we can approximate every element of $\left[0,1\right]$ arbitrarily
well (see \cite{Chiribella-purification} for more details). Now,
the closure in the operational topology for the trivial system coincides
with the closure in the usual topology. Indeed $x_{n}\rightarrow x$
operationally if and only if $kx_{n}\rightarrow kx$ in the usual
topology, for $k\in\left[0,1\right]$. Clearly, this is true if and
only if $x_{n}\rightarrow x$ in the usual topology. Since the set
of states is closed (in the usual topology), we have $\mathsf{St}\left(\mathrm{I}\right)=\left[0,1\right]$.
\end{proof}
In the former case, when $\mathsf{St}\left(\mathrm{I}\right)=\left\{ 0,1\right\} $,
we say that the theory is \emph{deterministic}.

\subsubsection{Equivalent topologies for states, effects, and transformations.}

We have seen that in $\mathsf{St}_{\mathbb{R}}\left(\mathrm{A}\right)$,
$\mathsf{Eff}_{\mathbb{R}}\left(\mathrm{A}\right)$, and $\mathsf{Transf}_{\mathbb{R}}\left(\mathrm{A},\mathrm{B}\right)$
we can introduce a norm, which induces another topology: we say that
a sequence $\left\{ \xi_{n}\right\} $ converges to $\xi$ in norm
if $\lim_{n\rightarrow+\infty}\left\Vert \xi_{n}-\xi\right\Vert =0$.
The topology of the norm is \emph{stronger} than the operational topology
defined above. This means that if a sequence converges in norm, it
also converges in the operational topology. Let us prove this statement
separately for states and transformations (which covers effects too).
\begin{prop}
\label{prop:norm -> operational}Let $\left\{ \xi_{n}\right\} $ be
a sequence of elements of $\mathsf{St}_{\mathbb{R}}\left(\mathrm{A}\right)$.
If $\left\{ \xi_{n}\right\} $ converges to $\xi$ in norm, it converges
to $\xi$ also operationally.
\end{prop}

\begin{proof}
Suppose $\xi_{n}$ converges to $\xi$ in norm, then $\lim_{n\rightarrow+\infty}\left\Vert \xi_{n}-\xi\right\Vert =0$.
Now to prove that $\xi_{n}$ converges to $\xi$ operationally, it
is enough to prove that $\lim_{n\rightarrow+\infty}\left|\left(a\middle|\xi_{n}-\xi\right)\right|=0$
for every $a\in\mathsf{Eff}\left(\mathrm{A}\right)$. Define $\eta_{n}:=\xi_{n}-\xi$.
Now, let us evaluate $\left(a\middle|\eta_{n}\right)$. Clearly 
\begin{equation}
\left(a\middle|\eta_{n}\right)\leq\sup_{a\in\mathsf{Eff}\left(\mathrm{A}\right)}\left(a\middle|\eta_{n}\right).\label{eq:first inequality equivalence}
\end{equation}
Recall that for every vector $x\in\mathsf{St}_{\mathbb{R}}\left(\mathrm{A}\right)$,
$\inf_{a\in\mathsf{Eff}\left(\mathrm{A}\right)}\left(a\middle|x\right)\leq0$.
Therefore we can add the term $-\inf_{a\in\mathsf{Eff}\left(\mathrm{A}\right)}\left(a\middle|\eta_{n}\right)$
to the right-hand side of eq.~\eqref{eq:first inequality equivalence}:
\[
\left(a\middle|\eta_{n}\right)\leq\sup_{a\in\mathsf{Eff}\left(\mathrm{A}\right)}\left(a\middle|\eta_{n}\right)-\inf_{a\in\mathsf{Eff}\left(\mathrm{A}\right)}\left(a\middle|\eta_{n}\right)=\left\Vert \eta_{n}\right\Vert ,
\]
where we have recognised the definition of the norm in $\mathsf{St}_{\mathbb{R}}\left(\mathrm{A}\right)$.
We are done if we show that $\left|\left(a\middle|\eta_{n}\right)\right|\leq\left\Vert \eta_{n}\right\Vert $.
Note that
\[
\left|\left(a\middle|\eta_{n}\right)\right|=\begin{cases}
\left(a\middle|\eta_{n}\right) & \textrm{if }\left(a\middle|\eta_{n}\right)\geq0\\
\left(a\middle|-\eta_{n}\right) & \textrm{if }\left(a\middle|\eta_{n}\right)<0
\end{cases}.
\]
Repeating the same argument, we get
\[
\left(a\middle|-\eta_{n}\right)\leq\left\Vert -\eta_{n}\right\Vert =\left\Vert \eta_{n}\right\Vert ,
\]
thus showing that $\left|\left(a\middle|\eta_{n}\right)\right|\leq\left\Vert \eta_{n}\right\Vert $
for every $a\in\mathsf{Eff}\left(\mathrm{A}\right)$. Since $\left\Vert \eta_{n}\right\Vert \rightarrow0$,
we have the thesis.
\end{proof}
Let us move to the corresponding statement for transformations.
\begin{prop}
\label{prop:norm -> operational transf}Let $\left\{ \Delta_{n}\right\} $
be a sequence of elements of $\mathsf{Transf}_{\mathbb{R}}\left(\mathrm{A},\mathrm{B}\right)$.
Then if $\left\{ \Delta_{n}\right\} $ converges to $\Delta$ in norm,
it converges to $\Delta$ also operationally.
\end{prop}

\begin{proof}
Suppose $\Delta_{n}$ converges to $\Delta$ in norm, then $\lim_{n\rightarrow+\infty}\left\Vert \Xi_{n}\right\Vert =0$,
where $\Xi_{n}:=\Delta_{n}-\Delta$. Now $\Delta_{n}$ converges to
$\Delta$ operationally, if for every system $\mathrm{S}$, and every
state $\rho\in\mathsf{St}\left(\mathrm{AS}\right)$, $\left(\Delta_{n}\otimes\mathcal{I}_{\mathrm{S}}\right)\rho$
converges to $\left(\Delta\otimes\mathcal{I}_{\mathrm{S}}\right)\rho$
operationally, or in other words, if and only if $\lim_{n\rightarrow+\infty}\left(\Xi_{n}\otimes\mathcal{I}_{\mathrm{S}}\right)\rho=0$
operationally. Recalling the proof of proposition~\ref{prop:norm -> operational},
for every effect $E\in\mathsf{Eff}\left(\mathrm{BS}\right)$,
\[
\left|\left(E\middle|\Xi_{n}\otimes\mathcal{I}_{\mathrm{S}}\middle|\rho\right)\right|\leq\left\Vert \left(\Xi_{n}\otimes\mathcal{I}_{\mathrm{S}}\right)\rho\right\Vert .
\]
Now
\[
\left\Vert \left(\Xi_{n}\otimes\mathcal{I}_{\mathrm{S}}\right)\rho\right\Vert \leq\sup_{\mathrm{S}}\sup_{\rho\in\mathsf{St}\left(\mathrm{AS}\right)}\left\Vert \left(\Xi_{n}\otimes\mathcal{I}_{\mathrm{S}}\right)\rho\right\Vert =\left\Vert \Xi_{n}\right\Vert ,
\]
so $\left|\left(E\middle|\Xi_{n}\otimes\mathcal{I}_{\mathrm{S}}\middle|\rho\right)\right|\leq\left\Vert \Xi_{n}\right\Vert $.
Therefore if $\left\Vert \Xi_{n}\right\Vert \rightarrow0$, it implies
that $\left|\left(E\middle|\Xi_{n}\otimes\mathcal{I}_{\mathrm{S}}\middle|\rho\right)\right|\rightarrow0$,
viz.\ $\lim_{n\rightarrow+\infty}\left(\Xi_{n}\otimes\mathcal{I}_{\mathrm{S}}\right)\rho=0$
operationally. This concludes the proof.
\end{proof}
We have assumed that the set of states, effects, and transformations
are closed in the operational topology. Now we will make a stronger
assumption.
\begin{assumption}
The sets $\mathsf{St}\left(\mathrm{A}\right)$, $\mathsf{Eff}\left(\mathrm{A}\right)$,
and $\mathsf{Transf}\left(\mathrm{A},\mathrm{B}\right)$ are \emph{closed}
both in the operational and the norm topology, for all systems $\mathrm{A}$
and $\mathrm{B}$.
\end{assumption}

This assumption has far-reaching consequences. For example, the sets
of states, effects, and transformations, being bounded and closed
in the topology of the norm, are \emph{compact} in this topology,
as we are in finite dimension. This fact can be extended also to the
operational topology. Indeed, an important consequence of this assumption
is that the operational topology and the topology of the norm are
\emph{equivalent}: a sequence $\left\{ \xi_{n}\right\} $ converges
to $\xi$ operationally if and only if it does it in norm. Let us
prove the statement separately for states and transformations.
\begin{prop}
\label{prop:equivalence topologies}Let $\left\{ \xi_{n}\right\} $
be a sequence of elements of $\mathsf{St}_{\mathbb{R}}\left(\mathrm{A}\right)$.
Then $\xi_{n}$ converges to $\xi$ operationally if and only if it
converges to $\xi$ in norm.
\end{prop}

\begin{proof}
We proved sufficiency in proposition~\ref{prop:norm -> operational}.
Let us prove necessity. Consider a sequence $\xi_{n}$ that converges
to $\xi$ operationally. Then, for every $a\in\mathsf{Eff}\left(\mathrm{A}\right)$,
we have $\lim_{n\rightarrow+\infty}\left(a\middle|\eta_{n}\right)=0$,
where $\eta_{n}:=\xi_{n}-\xi$. Now, since the set of effects is compact
(in the topology of the norm), and we are in finite dimension, the
supremum $\sup_{a\in\mathsf{Eff}\left(\mathrm{A}\right)}\left(a\middle|\eta_{n}\right)$
of the linear function $\eta_{n}$ is in fact a maximum, and achieved
on the effect $a^{*}$. Similarly, the infimum $\inf_{a\in\mathsf{Eff}\left(\mathrm{A}\right)}\left(a\middle|\eta_{n}\right)$
is in fact a minimum, and achieved on the effect $a_{*}$. Then, by
hypothesis, $\lim_{n\rightarrow+\infty}\left(a^{*}\middle|\eta_{n}\right)=0$
and $\lim_{n\rightarrow+\infty}\left(a_{*}\middle|\eta_{n}\right)=0$.
Therefore
\[
\lim_{n\rightarrow+\infty}\left\Vert \eta_{n}\right\Vert =\lim_{n\rightarrow+\infty}\left(\left(a^{*}\middle|\eta_{n}\right)-\left(a_{*}\middle|\eta_{n}\right)\right)=0.
\]
This proves that $\xi_{n}$ that converges to $\xi$ in norm too.
\end{proof}
This means that when we consider the conversion of states in the limit\textemdash e.g.\ in
data compression or in the asymptotic conversion of states\textemdash we
can choose either topology, according to its convenience in the problem
we want to address.

Let us move to the similar statement for the topologies in the vector
space of transformations, which covers the case of effects as well.
\begin{prop}
Let $\left\{ \Delta_{n}\right\} $ be a sequence of elements of $\mathsf{Transf}_{\mathbb{R}}\left(\mathrm{A},\mathrm{B}\right)$.
Then $\Delta_{n}$ converges to $\Delta$ operationally if and only
if it converges to $\Delta$ in norm.
\end{prop}

\begin{proof}
We proved sufficiency in proposition~\ref{prop:norm -> operational transf},
let us show necessity. We have that $\Delta_{n}$ converges to $\Delta$
operationally if and only if, for every system $\mathrm{S}$ and every
state $\rho\in\mathsf{St}\left(\mathrm{AS}\right)$, we have that
$\left(\Delta_{n}\otimes\mathcal{I}_{\mathrm{S}}\right)\rho$ converges
to $\left(\Delta\otimes\mathcal{I}_{\mathrm{S}}\right)\rho$ operationally.
By proposition~\ref{prop:equivalence topologies}, this means that
\[
\lim_{n\rightarrow+\infty}\left\Vert \left(\Xi_{n}\otimes\mathcal{I}_{\mathrm{S}}\right)\rho\right\Vert =0,
\]
where $\Xi_{n}:=\Delta_{n}-\Delta$. Then in the definition of $\left\Vert \Xi_{n}\right\Vert $,
by compactness, the supremum is achieved on system $\mathrm{S}^{*}$
and on the state $\rho^{*}\in\mathsf{St}\left(\mathrm{A}\mathrm{S}^{*}\right)$.
Therefore
\[
\lim_{n\rightarrow+\infty}\left\Vert \Xi_{n}\right\Vert =\lim_{n\rightarrow+\infty}\left\Vert \left(\Xi_{n}\otimes\mathcal{I}_{\mathrm{S}^{*}}\right)\rho^{*}\right\Vert =0,
\]
because, by hypothesis, $\lim_{n\rightarrow+\infty}\left\Vert \left(\Xi_{n}\otimes\mathcal{I}_{\mathrm{S}}\right)\rho\right\Vert =0$
for every system $\mathrm{S}$ and every $\rho\in\mathsf{St}\left(\mathrm{AS}\right)$.
This concludes the proof.
\end{proof}
One concludes that the two topologies are completely equivalent. Specifically,
the statements about the compactness of the set of states, effects,
and transformations are valid both in the operational and the norm
topology.

We can use proposition~\ref{prop:equivalence topologies} to prove
a property of limits of sequences of transformations we will use in
the following.
\begin{lem}
\label{lem:lemma product}Let $\left\{ \mathcal{A}_{n}\right\} $
be a sequence of transformations from $\mathrm{A}$ to $\mathrm{B}$
converging to $\mathcal{A}$, and let $\left\{ \mathcal{B}_{n}\right\} $
be a sequence of transformations from $\mathrm{A}$ to $\mathrm{B}$
converging to $\mathcal{B}$. Then $\left\{ \mathcal{B}_{n}\mathcal{A}_{n}\right\} $
converges to $\mathcal{BA}$.
\end{lem}

\begin{proof}
We have that $\mathcal{A}_{n}$ converges to $\mathcal{A}$ if, for
every $\mathrm{S}$, $\left\{ \left(\mathcal{A}_{n}\otimes\mathcal{I}_{\mathrm{S}}\right)\rho_{\mathrm{AS}}\right\} $
converges to $\left(\mathcal{A}\otimes\mathcal{I}_{\mathrm{S}}\right)\rho_{\mathrm{AS}}$.
Similarly, $\left\{ \left(\mathcal{B}_{n}\otimes\mathcal{I}_{\mathrm{S}}\right)\rho_{\mathrm{AS}}\right\} $
converges to $\left(\mathcal{B}\otimes\mathcal{I}_{\mathrm{S}}\right)\rho_{\mathrm{AS}}$.
By proposition~\ref{prop:equivalence topologies}, this happens if
and only if $\left\Vert \left(\mathcal{A}_{n}\otimes\mathcal{I}_{\mathrm{S}}\right)\rho_{\mathrm{AS}}-\left(\mathcal{A}\otimes\mathcal{I}_{\mathrm{S}}\right)\rho_{\mathrm{AS}}\right\Vert \rightarrow0$,
and $\left\Vert \left(\mathcal{B}_{n}\otimes\mathcal{I}_{\mathrm{S}}\right)\rho_{\mathrm{AS}}-\left(\mathcal{B}\otimes\mathcal{I}_{\mathrm{S}}\right)\rho_{\mathrm{AS}}\right\Vert \rightarrow0$.
Now we are ready to prove that $\left\{ \mathcal{B}_{n}\mathcal{A}_{n}\right\} $
converges to $\mathcal{BA}$, viz.\ that $\left\Vert \left(\mathcal{B}_{n}\mathcal{A}_{n}\otimes\mathcal{I}_{\mathrm{S}}\right)\rho_{\mathrm{AS}}-\left(\mathcal{BA}\otimes\mathcal{I}_{\mathrm{S}}\right)\rho_{\mathrm{AS}}\right\Vert \rightarrow0$.
\[
\left\Vert \left(\mathcal{B}_{n}\mathcal{A}_{n}\otimes\mathcal{I}_{\mathrm{S}}\right)\rho_{\mathrm{AS}}-\left(\mathcal{BA}\otimes\mathcal{I}_{\mathrm{S}}\right)\rho_{\mathrm{AS}}\right\Vert \leq\left\Vert \mathcal{B}_{n}\left(\mathcal{A}_{n}\otimes\mathcal{I}_{\mathrm{S}}\right)\rho_{\mathrm{AS}}-\mathcal{B}_{n}\left(\mathcal{A}\otimes\mathcal{I}_{\mathrm{S}}\right)\rho_{\mathrm{AS}}\right\Vert +
\]
\[
+\left\Vert \mathcal{B}_{n}\left(\mathcal{A}\otimes\mathcal{I}_{\mathrm{S}}\right)\rho_{\mathrm{AS}}-\mathcal{B}\left(\mathcal{A}\otimes\mathcal{I}_{\mathrm{S}}\right)\rho_{\mathrm{AS}}\right\Vert 
\]
Now, $\left\Vert \mathcal{B}_{n}\left(\mathcal{A}\otimes\mathcal{I}_{\mathrm{S}}\right)\rho_{\mathrm{AS}}-\mathcal{B}\left(\mathcal{A}\otimes\mathcal{I}_{\mathrm{S}}\right)\rho_{\mathrm{AS}}\right\Vert \rightarrow0$
because $\mathcal{B}_{n}\rightarrow\mathcal{B}$. Let us apply the
monotonicity of the norm:
\[
\left\Vert \mathcal{B}_{n}\left(\mathcal{A}_{n}\otimes\mathcal{I}_{\mathrm{S}}\right)\rho_{\mathrm{AS}}-\mathcal{B}_{n}\left(\mathcal{A}\otimes\mathcal{I}_{\mathrm{S}}\right)\rho_{\mathrm{AS}}\right\Vert \leq\left\Vert \left(\mathcal{A}_{n}\otimes\mathcal{I}_{\mathrm{S}}\right)\rho_{\mathrm{AS}}-\left(\mathcal{A}\otimes\mathcal{I}_{\mathrm{S}}\right)\rho_{\mathrm{AS}}\right\Vert .
\]
Then $\left\Vert \mathcal{B}_{n}\left(\mathcal{A}_{n}\otimes\mathcal{I}_{\mathrm{S}}\right)\rho_{\mathrm{AS}}-\mathcal{B}_{n}\left(\mathcal{A}\otimes\mathcal{I}_{\mathrm{S}}\right)\rho_{\mathrm{AS}}\right\Vert \rightarrow0$,
because $\mathcal{A}_{n}\rightarrow\mathcal{A}$, by which we conclude
that
\[
\left\Vert \left(\mathcal{B}_{n}\mathcal{A}_{n}\otimes\mathcal{I}_{\mathrm{S}}\right)\rho_{\mathrm{AS}}-\left(\mathcal{BA}\otimes\mathcal{I}_{\mathrm{S}}\right)\rho_{\mathrm{AS}}\right\Vert \rightarrow0.
\]
\end{proof}

\section{Causality and its consequences\label{sec:Causality}}

In this section we will examine how the direction in which ``information
flows'' in a theory constrains the structure of the theory itself.
In causal theories information propagates in the same order as the
input-output order given by sequential composition. More poetically,
in causal theories information propagates from the past to the future,
and one is able to choose later experiments depending on the outcomes
of present observations. Causality is a standard setting in most physical
descriptions of Nature, and in the GPT literature it is often assumed
implicitly, without even mentioning it directly. Causality will be
the first axiom we impose in this thesis. Besides being reasonable
from a physical point of view, Causality will bring a lot of interesting
and important consequences to a theory, making it easier to tackle.

Let us begin with the formal definition of causal theory.
\begin{defn}[Causality \cite{Chiribella-purification}]
\label{def:causality}A theory is \emph{causal} if for every preparation-test
$\left\{ \rho_{i}\right\} _{i\in X}$ and every observation-test $\left\{ a_{j}\right\} _{j\in Y}$
on any system $\mathrm{A}$, the probability $p_{i}:=\sum_{j\in Y}\left(a_{j}\middle|\rho_{i}\right)$
is \emph{independent} of the observation-test $\left\{ a_{j}\right\} _{j\in Y}$.
\end{defn}

In other words, if $\left\{ a_{j}\right\} _{j\in Y}$ and $\left\{ b_{k}\right\} _{k\in Z}$
are two observation-tests, we have
\begin{equation}
\sum_{j\in Y}\left(a_{j}\middle|\rho_{i}\right)=\sum_{k\in Z}\left(b_{k}\middle|\rho_{i}\right).\label{eq:causality}
\end{equation}
Loosely speaking, the preparation of the system does not depend on
the choice of subsequent (or ``future'') measurements, a sort of
no-signalling condition from the future. In this way, the direction
in which information flows, as witnessed by marginal probabilities
in definition~\ref{def:causality}, coincides with the ordering given
by sequential composition. In general, this is not obvious, as the
following example shows.
\begin{example}
Consider a theory where the states of a system are the quantum operations
on that system. Specifically, deterministic states are quantum channels.
Then we can consider the channels of this higher theory to be quantum
``supermaps'', which map quantum operations into quantum operations
\cite{Circuit-architecture,Supermaps,Hierarchy-combs,Perinotti1,Perinotti2,Gour-super}.

Let us consider a preparation of a state $\mathcal{C}_{i}$ followed
by a measurement $\mathcal{A}_{j}$, which we represent in the higher
theory as\[
\begin{aligned}\Qcircuit @C=1em @R=.7em @!R {  &\prepareC{\cC_{i}}  & \qw \poloFantasmaCn{\rA} &  \measureD{\cA_{j}}}\end{aligned}~.
\]Note that the measurement follows the preparation in the composition
order. But if we recall that $\mathcal{C}_{i}$ is a quantum operation,
namely a box with an input and an output wire, in quantum theory such
a diagram will look like\[
\begin{aligned}\Qcircuit @C=1em @R=.7em @!R {  &\multiprepareC{1}{\rho_{j}}  & \qw \poloFantasmaCn{\rA} &  \gate{\cC_{i}} & \qw \poloFantasmaCn{\rA} &\multimeasureD{1}{a_{j}}\\ &\pureghost{\rho_{j}}&\qw \poloFantasmaCn{\rS} &\qw &\qw &\ghost{a_j}}\end{aligned}~,
\]for some system $\mathrm{S}$. Note that the effect $\mathcal{A}_{j}$
is split into two parts: one is before the quantum operation and the
other is after, otherwise we could not have a diagram with no external
wires. Therefore, in the theory in which states are quantum operations,
the preparation of a state is influenced by a subsequent measurement,
so information does not propagate in the same direction as sequential
composition.
\end{example}

Definition~\ref{def:causality} can be recast in an equivalent way
\cite[lemma 4]{Chiribella-purification}, which is often more practical
to work with. 
\begin{prop}
A theory is causal if and only if for every system $\mathrm{A}$ there
is a unique deterministic effect $u_{\mathrm{A}}$.
\end{prop}

\begin{proof}
Necessity. Suppose, by contradiction, that there are two deterministic
effects $u$ and $u'$ for system $\mathrm{A}$. Deterministic effects
are particular examples of observation-tests. Eq.~\eqref{eq:causality}
then states that $\left(u\middle|\rho_{i}\right)=\left(u'\middle|\rho_{i}\right)$
for every $\rho_{i}\in\mathsf{St}\left(\mathrm{A}\right)$. This means
that $u=u'$.

Sufficiency. Suppose there is a unique deterministic effect $u_{\mathrm{A}}$
for system $\mathrm{A}$, and consider the observation-test $\left\{ a_{j}\right\} _{j\in Y}$.
By doing a coarse-graining over the effects, we obtain the deterministic
effect $u'=\sum_{j\in Y}a_{j}$. Since the deterministic effect is
unique, it must be $u'=u$. Hence, for every state $\rho_{i}$, we
have
\[
\sum_{j\in Y}\left(a_{j}\middle|\rho_{i}\right)=\left(u\middle|\rho_{i}\right),
\]
and the right-hand side does not depend on the choice of the observation-test.
This means that the theory is causal.
\end{proof}
\begin{example}
We saw in example~\ref{ex:states effects in quantum mechanics} that
in quantum mechanics there is only one deterministic effect, the identity
operator. Hence quantum mechanics is a causal theory.
\end{example}

We noticed that if we perform a coarse-graining over the effects in
an observation-test, we have a deterministic effect. By the uniqueness
of the deterministic effect, we have that if $\left\{ a_{i}\right\} _{i\in X}$
is an observation-test on system $\mathrm{A}$, then $\sum_{i\in X}a_{i}=u$,
where $u$ is the deterministic effect of $\mathrm{A}$. This is a
necessary condition for $\left\{ a_{i}\right\} _{i\in X}$ to be an
observation-test. Specifically, this means that if $a$ is an effect,
$u-a$ is an effect too.

Let us see a straightforward corollary of uniqueness of the deterministic
effect.
\begin{cor}
In a causal theory, if $u_{\mathrm{A}}$ and $u_{\mathrm{B}}$ are
the deterministic effects of systems $\mathrm{A}$ and $\mathrm{B}$
respectively, then the deterministic effect for system $\mathrm{AB}$
is $u_{\mathrm{A}}\otimes u_{\mathrm{B}}$.
\end{cor}

\begin{proof}
The parallel composition of two single-outcome tests is clearly a
single-outcome test, hence the effect $u_{\mathrm{A}}\otimes u_{\mathrm{B}}$
is deterministic and acts on $\mathrm{AB}$. By the uniqueness of
the deterministic effect, we conclude that $u_{\mathrm{AB}}=u_{\mathrm{A}}\otimes u_{\mathrm{B}}$.
\end{proof}
In a causal theory, we can define marginal states. Suppose we have
a bipartite state of system $\mathrm{AB}$, and we are interested
in the state of subsystem $\mathrm{A}$. We want to throw away all
the information concerning system $\mathrm{B}$. This operation resembles
marginalisation in probability theory, whence the name. In quantum
mechanics, this operation is simply given by taking the partial trace
over $\mathrm{B}$.
\begin{defn}
The \emph{marginal state} (\emph{marginal} for short)\emph{ }$\rho_{\mathrm{A}}$
on system $\mathrm{A}$ of a bipartite state $\sigma_{\mathrm{AB}}$
is obtained by applying the deterministic effect to $\mathrm{B}$:\[
\begin{aligned}\Qcircuit @C=1em @R=.7em @!R { & \prepareC{\rho}    & \qw \poloFantasmaCn{\rA} &  \qw   }\end{aligned}~=\!\!\!\!\begin{aligned}\Qcircuit @C=1em @R=.7em @!R { & \multiprepareC{1}{\sigma}    & \qw \poloFantasmaCn{\rA} &  \qw   \\  & \pureghost{\sigma}    & \qw \poloFantasmaCn{\rB}  &   \measureD{u} }\end{aligned}~.
\]
\end{defn}

For this reason, we will sometimes use the notation $\mathrm{tr}_{\mathrm{A}}$
for the unique deterministic effect on $\mathrm{A}$. Therefore, $\rho_{\mathrm{A}}=\mathrm{tr}_{\mathrm{B}}\sigma_{\mathrm{AB}}$.

In a causal theory, we have also useful characterisations of channels
and tests \cite[lemma 5]{Chiribella-purification}.
\begin{prop}
\label{prop:characterization channel}Let $\mathcal{C}\in\mathsf{Transf}\left(\mathrm{A},\mathrm{B}\right)$.
$\mathcal{C}$ is a channel if and only if $u_{\mathrm{B}}\mathcal{C}=u_{\mathrm{A}}$.
\end{prop}

\begin{proof}
Necessity is straightforward. Since a channel is a deterministic transformation,
then $u_{\mathrm{B}}\mathcal{C}$ is a deterministic effect on system
$\mathrm{A}$. By the uniqueness of the deterministic effect, $u_{\mathrm{B}}\mathcal{C}=u_{\mathrm{A}}$.

Sufficiency. Suppose we have a test $\left\{ \mathcal{C}_{i}\right\} _{i\in X}$
from system $\mathrm{A}$ to system $\mathrm{B}$ such that $\mathcal{C}:=\mathcal{C}_{i_{0}}$
satisfies $u_{\mathrm{B}}\mathcal{C}=u_{\mathrm{A}}$. We want to
prove that $\left\{ \mathcal{C}_{i}\right\} _{i\in X}$ is a deterministic
test. Let us consider the channel $\mathcal{C}'$ associated with
the test $\left\{ \mathcal{C}_{i}\right\} _{i\in X}$, namely $\mathcal{C}'=\sum_{i\in X}\mathcal{C}_{i}$.
Since $\mathcal{C}'$ is a channel, we have $u_{\mathrm{B}}\mathcal{C}'=u_{\mathrm{A}}$.
Recalling the expression of $\mathcal{C}'$, we have
\[
u_{\mathrm{A}}=u_{\mathrm{B}}\mathcal{C}'=u_{\mathrm{B}}\mathcal{C}_{i_{0}}+u_{\mathrm{B}}\sum_{i\neq i_{0}}\mathcal{C}_{i}=u_{\mathrm{A}}+u_{\mathrm{B}}\sum_{i\neq i_{0}}\mathcal{C}_{i},
\]
because $u_{\mathrm{B}}\mathcal{C}_{i_{0}}=u_{\mathrm{A}}$. This
means $u_{\mathrm{B}}\sum_{i\neq i_{0}}\mathcal{C}_{i}=0$, namely
$\sum_{i\neq i_{0}}\mathcal{C}_{i}=0$. Therefore $\mathcal{C}=\mathcal{C}'$,
whence the test was in fact deterministic. Thus $\mathcal{C}$ is
a channel.
\end{proof}
Note that in quantum theory this is precisely the statement that a
quantum operation is a quantum channel if and only if it is trace-preserving.
This is even more obvious if we write $u_{\mathrm{A}}$ as $\mathrm{tr}_{\mathrm{A}}$. 

Specifically, if $\mathrm{A}$ is the trivial system, we have that
a state $\rho_{\mathrm{B}}$ is deterministic if and only if $\mathrm{tr}\:\rho=1$.
Moreover, for every test $\left\{ \mathcal{C}_{i}\right\} _{i\in X}$
from $\mathrm{A}$ to $\mathrm{B}$, we can consider the associated
channel $\sum_{i\in X}\mathcal{C}_{i}$. Therefore we have
\begin{equation}
\sum_{i\in X}u_{\mathrm{B}}\mathcal{C}_{i}=u_{\mathrm{A}}.\label{eq:necessity test}
\end{equation}
This is a necessary condition. In quantum theory this is the statement
that the quantum channel associated with a quantum instrument is trace-preserving.

Suppose we have two parties sharing a bipartite state. In a causal
theory it is impossible for a party to send a message to the other
by acting locally on her own physical system and relying on correlations
she shares with the other party. This form of instantaneous communication
is called \emph{signalling}. In more precise terms, in a causal theory
it is not possible for a party to communicate the outcome of a local
measurement on her system to the other without exchanging physical
systems, classical communication included, as it is usually mediated
by electromagnetic signals \cite[theorem 1]{Chiribella-purification}.
\begin{thm}
In a causal theory it is impossible to have signalling without the
exchange of physical systems.
\end{thm}

\begin{proof}
Suppose we have two distant parties, Alice and Bob, who share a bipartite
state $\sigma_{\mathrm{AB}}$. Suppose Alice performs a local test
$\left\{ \mathcal{A}_{i}\right\} _{i\in X}$ on $\mathrm{A}$ and
Bob performs a local test $\left\{ \mathcal{B}_{j}\right\} _{j\in Y}$
on $\mathrm{B}$. Let us define the joint probability $p_{ij}:=\mathrm{tr}_{\mathrm{AB}}\left(\mathcal{A}_{i}\otimes\mathcal{B}_{j}\right)\sigma_{\mathrm{AB}}$
and the marginal probabilities as $p_{i}^{\left(\mathrm{A}\right)}:=\sum_{j\in Y}\mathrm{tr}_{\mathrm{AB}}\left(\mathcal{A}_{i}\otimes\mathcal{B}_{j}\right)\sigma_{\mathrm{AB}}$
and $p_{j}^{\left(\mathrm{B}\right)}:=\sum_{i\in X}\mathrm{tr}_{\mathrm{AB}}\left(\mathcal{A}_{i}\otimes\mathcal{B}_{j}\right)\sigma_{\mathrm{AB}}$.
Each party cannot acquire any information about the outcomes of the
other based only on its marginal probability. Indeed, let us examine
Alice's marginal probability $p_{i}^{\left(\mathrm{A}\right)}$ better.
Let $\rho_{\mathrm{A}}$ be the marginal state of $\sigma_{\mathrm{AB}}$
on system $\mathrm{A}$. 
\[
p_{i}^{\left(\mathrm{A}\right)}=\sum_{j\in Y}\mathrm{tr}_{\mathrm{AB}}\left(\mathcal{A}_{i}\otimes\mathcal{B}_{j}\right)\sigma_{\mathrm{AB}}=\left(u_{\mathrm{A}}\mathcal{A}_{i}\otimes\sum_{j\in Y}u_{\mathrm{B}}\mathcal{B}_{j}\right)\sigma_{\mathrm{AB}}=
\]
\[
=u_{\mathrm{A}}\mathcal{A}_{i}\otimes\mathrm{tr}_{\mathrm{B}}\sigma_{\mathrm{AB}}=\mathrm{tr}_{\mathrm{A}}\mathcal{A}_{i}\rho_{\mathrm{A}}
\]
We see that Alice's marginal probability does not depend on the test
performed by Bob at all, so she cannot get any information about the
outcome of Bob's test based only on her system. A similar reasoning
applies to Bob's party: he cannot gain any information about the outcome
of Alice's test.
\end{proof}
Since in a causal theory the order of composition coincides with the
order in which information flows, we can choose a later test according
to the result of a previous test. Suppose we perform a test $\left\{ \mathcal{C}_{i}\right\} _{i\in X}$
from $\mathrm{A}$ to $\mathrm{B}$ first. Depending on the outcome
$i$, then we perform different tests $\left\{ \mathcal{D}_{j_{i}}^{\left(i\right)}\right\} _{j_{i}\in Y_{i}}$
from $\mathrm{B}$ to $\mathrm{C}$. Here the superscript in round
brackets is aimed at highlighting the dependence of the test on the
outcome of the previous one. Let us make this concept more precise
with the following definition.
\begin{defn}
If $\left\{ \mathcal{C}_{i}\right\} _{i\in X}$ is a test from $\mathrm{A}$
to $\mathrm{B}$ and, for every $i$, $\left\{ \mathcal{D}_{j_{i}}^{\left(i\right)}\right\} _{j_{i}\in Y_{i}}$
is a test from $\mathrm{B}$ to $\mathrm{C}$, then the \emph{conditioned}
(or \emph{classically controlled})\emph{ test} is a test from $\mathrm{A}$
to $\mathrm{C}$ with outcomes $\left(i,j_{i}\right)\in Z:=\bigcup_{i\in X}\left\{ i\right\} \times Y_{i}$,
and events $\left\{ \mathcal{D}_{j_{i}}^{\left(i\right)}\circ\mathcal{C}_{i}\right\} _{\left(i,j_{i}\right)\in Z}$.
\end{defn}

The graphical representation is as usual.\[
\begin{aligned}\Qcircuit @C=1em @R=.7em @!R {    & \qw \poloFantasmaCn{\rA} &  \gate{\cD_{j_{i}}^{(i)} \circ \cC_{i}} & \qw \poloFantasmaCn{\rC} &\qw}\end{aligned}:= \begin{aligned}\Qcircuit @C=1em @R=.7em @!R {    & \qw \poloFantasmaCn{\rA} &  \gate{\cC_{i}} & \qw \poloFantasmaCn{\rB} &  \gate{\cD_{j_{i}}^{(i)}} & \qw \poloFantasmaCn{\rC} &\qw}\end{aligned}~.
\]Conditioning expresses the idea of choosing what to do at later steps
using the classical information about outcomes obtained at previous
steps. The test $\left\{ \mathcal{D}_{j_{i}}^{\left(i\right)}\circ\mathcal{C}_{i}\right\} $
is well-defined thanks to Causality: it satisfies the necessary condition
of eq.~\eqref{eq:necessity test}. Indeed
\[
\sum_{i}\sum_{j_{i}}u\mathcal{D}_{j_{i}}^{\left(i\right)}\circ\mathcal{C}_{i}=\sum_{i}u\mathcal{C}_{i}=u.
\]

A particular case of conditioning is randomisation.
\begin{defn}
If $\left\{ p_{i}\right\} _{i\in X}$ is a set of probabilities\footnote{Recall that a set of probabilities can be regarded as a test from
the trivial system to itself.} and, for every $i$, $\left\{ \mathcal{C}_{j_{i}}^{\left(i\right)}\right\} _{j_{i}\in Y_{i}}$
is a test from $\mathrm{A}$ to $\mathrm{B}$, we can construct the
\emph{randomised test} $\left\{ p_{i}\mathcal{C}_{j_{i}}^{\left(i\right)}\right\} _{i\in X,j_{i}\in Y_{i}}$,
which is a test from $\mathrm{A}$ to $\mathrm{B}$ whose events are
defined as\[
p_{i} \begin{aligned}\Qcircuit @C=1em @R=.7em @!R {   & \qw \poloFantasmaCn{\rA} &  \gate{\cC_{j_{i}}^{(i)}} & \qw \poloFantasmaCn{\rB} &\qw}\end{aligned}:=\begin{aligned}\Qcircuit @C=1em @R=.7em @!R {    & \qw \poloFantasmaCn{\rA} &\qw &\qw &  \gate{\cC_{j_{i}}^{(i)}} & \qw \poloFantasmaCn{\rB} &\qw \\ & \qw \poloFantasmaCn{\rI}  &  \gate{p_{i}} & \qw \poloFantasmaCn{\rI} &\qw  &\qw &\qw }\end{aligned}~.
\]
\end{defn}

In a randomised test we are performing a classical random process
and according to its outcome we apply a test $\left\{ \mathcal{C}_{j_{i}}^{\left(i\right)}\right\} _{j_{i}\in Y_{i}}$,
where $i$ is the outcome of the classical random process. The existence
of randomised tests tells us that a non-deterministic theory must
be convex, thus recovering one of the assumptions of the convex approach
to GPTs, but from a deeper principle, Causality.
\begin{prop}
If a causal theory is not deterministic, then for all systems $\mathrm{A}$
and $\mathrm{B}$, the sets $\mathsf{St}\left(\mathrm{A}\right)$,
$\mathsf{Eff}\left(\mathrm{A}\right)$ and $\mathsf{Transf}\left(\mathrm{A},\mathrm{B}\right)$
are convex.
\end{prop}

\begin{proof}
Let $p\in\left[0,1\right]$. If the theory is non-deterministic, $p$
is a state of the trivial system. Let $\left\{ \mathcal{C}_{i}\right\} _{i\in X}$
and $\left\{ \mathcal{D}_{j}\right\} _{j\in Y}$ be tests from $\mathrm{A}$
to $\mathrm{B}$. By randomisation, we can consider the test $\left\{ p\mathcal{C}_{i}\right\} _{i\in X}\cup\left\{ \left(1-p\right)\mathcal{D}_{j}\right\} _{j\in Y}$.
By coarse-graining, the convex combination $p\mathcal{C}_{i}+\left(1-p\right)\mathcal{D}_{j}$,
is still a transformation from $\mathrm{A}$ to $\mathrm{B}$. Taking
$\mathrm{A}$ or $\mathrm{B}$ equal to the trivial system, one has
the thesis for states and effects.
\end{proof}
Another important example of classically controlled tests are measure-and-prepare
tests.
\begin{defn}
A test is \emph{measure-and-prepare} if it is of the form $\left\{ \left|\rho_{i}\right)\left(a_{i}\right|\right\} _{i\in X}$,
where $\left\{ a_{i}\right\} _{i\in X}$ is an observation-test, and
$\rho_{i}$ is a deterministic state for every $i\in X$.

A channel $\mathcal{C}$ is \emph{measure-and-prepare} if it is the
coarse-graining of a measure-and-prepare test:
\[
\mathcal{C}=\sum_{i\in X}\left|\rho_{i}\right)\left(a_{i}\right|.
\]
\end{defn}

The idea behind a measure-and-prepare test is to perform an observation-test,
and to prepare the deterministic state $\rho_{i}$ if we get outcome
$i$ upon performing the observation-test.

The relationship between classical control and Causality is so tight
that a theory where all classically controlled tests are possible
is causal \cite[lemma 7]{Chiribella-purification} (see also \cite{Diagrammatic}).
\begin{prop}
A theory where every conditioned test is possible is causal.
\end{prop}

\begin{proof}
Suppose by contradiction that the theory is not causal, so there exist
two deterministic effects $u\neq u'$ for a system $\mathrm{A}$.
Now, take a generic state $\rho\in\mathsf{St}\left(\mathrm{A}\right)$.
By definition, there exists a preparation-test $\left\{ \rho_{i}\right\} $
such that $\rho_{i_{0}}=\rho$. Take the coarse-grained version $\left\{ \rho_{0},\rho_{1}\right\} $,
where $\rho_{0}:=\rho_{i_{0}}$, and $\rho_{1}:=\sum_{i\neq i_{0}}\rho_{i}$.
Now, consider the classically controlled test where one applies $u$
if $\rho_{0}$ is prepared, and $u'$ if $\rho_{1}$ is prepared.
This is a valid test, given by $\left\{ \left(u\middle|\rho_{0}\right),\left(u'\middle|\rho_{1}\right)\right\} $.
This is a test on the trivial system, therefore 
\[
\left(u\middle|\rho_{0}\right)+\left(u'\middle|\rho_{1}\right)=1.
\]
Now, $\rho_{0}+\rho_{1}$ is a deterministic state, and since $u'$
is a deterministic effect
\[
\left(u'\middle|\rho_{0}\right)+\left(u'\middle|\rho_{1}\right)=1.
\]
This implies $\left(u\middle|\rho_{0}\right)=\left(u'\middle|\rho_{0}\right)$.
Since $\rho_{0}=\rho$ is an arbitrary state of $\mathrm{A}$, one
has $u=u'$. The theory must be causal.
\end{proof}

\subsection{States of causal theories}

In causal theories, the norm of a state takes a particularly simple
form.\footnote{This result clearly extends to all elements of $\mathsf{St}_{+}\left(\mathrm{A}\right)$,
which are just a non-negative rescaling of states.}
\begin{prop}
In a causal theory, for a state $\rho$ we have $\left\Vert \rho\right\Vert =\mathrm{tr}\:\rho$.
\end{prop}

\begin{proof}
Clearly, for a state, we have $\left\Vert \rho\right\Vert =\sup_{a\in\mathsf{Eff}\left(\mathrm{A}\right)}\left(a\middle|\rho\right)$,
therefore 
\[
\mathrm{tr}\:\rho\leq\sup_{a\in\mathsf{Eff}\left(\mathrm{A}\right)}\left(a\middle|\rho\right)=\left\Vert \rho\right\Vert ,
\]
because $\mathrm{tr}$ is a deterministic effect. Now, for every effect
$a\in\mathsf{Eff}\left(\mathrm{A}\right)$, we also have
\[
\left(u\middle|\rho\right)=\left(a\middle|\rho\right)+\left(u-a\middle|\rho\right)\geq\left(a\middle|\rho\right),
\]
whence
\[
\left\Vert \rho\right\Vert =\sup_{a\in\mathsf{Eff}\left(\mathrm{A}\right)}\left(a\middle|\rho\right)\leq\mathrm{tr}\:\rho.
\]
One concludes that $\left\Vert \rho\right\Vert =\mathrm{tr}\:\rho$.
\end{proof}
For this reason, in a causal theory, a state is deterministic ($\mathrm{tr}\:\rho=1$)
if and only if it is normalised. The norm of a state is the probability
of preparing that state in some preparation-test. Indeed, since $\mathrm{tr}$
is deterministic, the probability in $\mathrm{tr}\:\rho$ comes only
from the randomness arising from the preparation of $\rho$. 
\begin{example}
In quantum theory, we have
\[
\left\Vert \rho\right\Vert =\mathrm{tr}\:\mathbf{1}\rho=\mathrm{tr}\:\rho.
\]
Therefore normalised states are density operators (the trace is equal
to 1).
\end{example}

For every (non-zero) state $\rho$ of a causal theory we can consider
the normalised state

\[
\overline{\rho}:=\frac{\rho}{\left\Vert \rho\right\Vert }.
\]
Suppose we have the preparation-test $\left\{ \rho_{i}\right\} $.
Clearly $\left\Vert \rho_{i}\right\Vert \leq1$ and one has equality
if and only if this is a single-outcome preparation-test, viz.\ a
deterministic state. In a causal theory, every sub-normalised state
$\rho_{i}$ can be written as $\rho_{i}=p_{i}\overline{\rho}_{i}$,
where $p_{i}=\left\Vert \rho_{i}\right\Vert \in\left[0,1\right]$
and $\overline{\rho}_{i}$ is a normalised state. Therefore every
preparation-test $\left\{ \rho_{i}\right\} _{i\in X}$ is a randomised
test: we have a classical source of randomness $\left\{ p_{i}\right\} _{i\in X}$,
where $p_{i}=\left\Vert \rho_{i}\right\Vert $, and according to the
classical outcome we prepare the deterministic state $\overline{\rho}_{i}$.

In conclusion, given the linearity of OPTs, in causal theories we
can simply restrict ourselves to normalised states. In the following
we will often drop the term ``normalised'' when talking about states,
and when we write about the ``state space'' we will mean the set
of normalised states $\mathsf{St}_{1}\left(\mathrm{A}\right)$ of
a system $\mathrm{A}$. The state space of a non-deterministic causal
theory is a convex set. Indeed, for any $p\in\left[0,1\right]$, consider
the state $p\rho_{0}+\left(1-p\right)\rho_{1}$, where $\rho_{0}$
and $\rho_{1}$ are two normalised states. $p\rho_{0}+\left(1-p\right)\rho_{1}$
is still normalised:
\[
\mathrm{tr}\left[p\rho_{0}+\left(1-p\right)\rho_{1}\right]=p\mathrm{tr}\:\rho_{0}+\left(1-p\right)\mathrm{tr}\:\rho_{1}=p+1-p=1.
\]
Clearly, normalised pure states are the extreme points of the state
space $\mathsf{St}_{1}\left(\mathrm{A}\right)$. We have an interesting
geometrical interpretation for the state space of any system $\mathrm{A}$.
It is the intersection of the hyperplane $\mathrm{tr}\:\xi=1$ (for
$\xi\in\mathsf{St}_{\mathbb{R}}\left(\mathrm{A}\right)$) with the
cone of states $\mathsf{St}_{+}\left(\mathrm{A}\right)$, as illustrated
in fig.~\ref{fig:causal cone}.
\begin{figure}
\begin{centering}
\includegraphics[bb=70bp 0bp 513bp 310bp,scale=0.8]{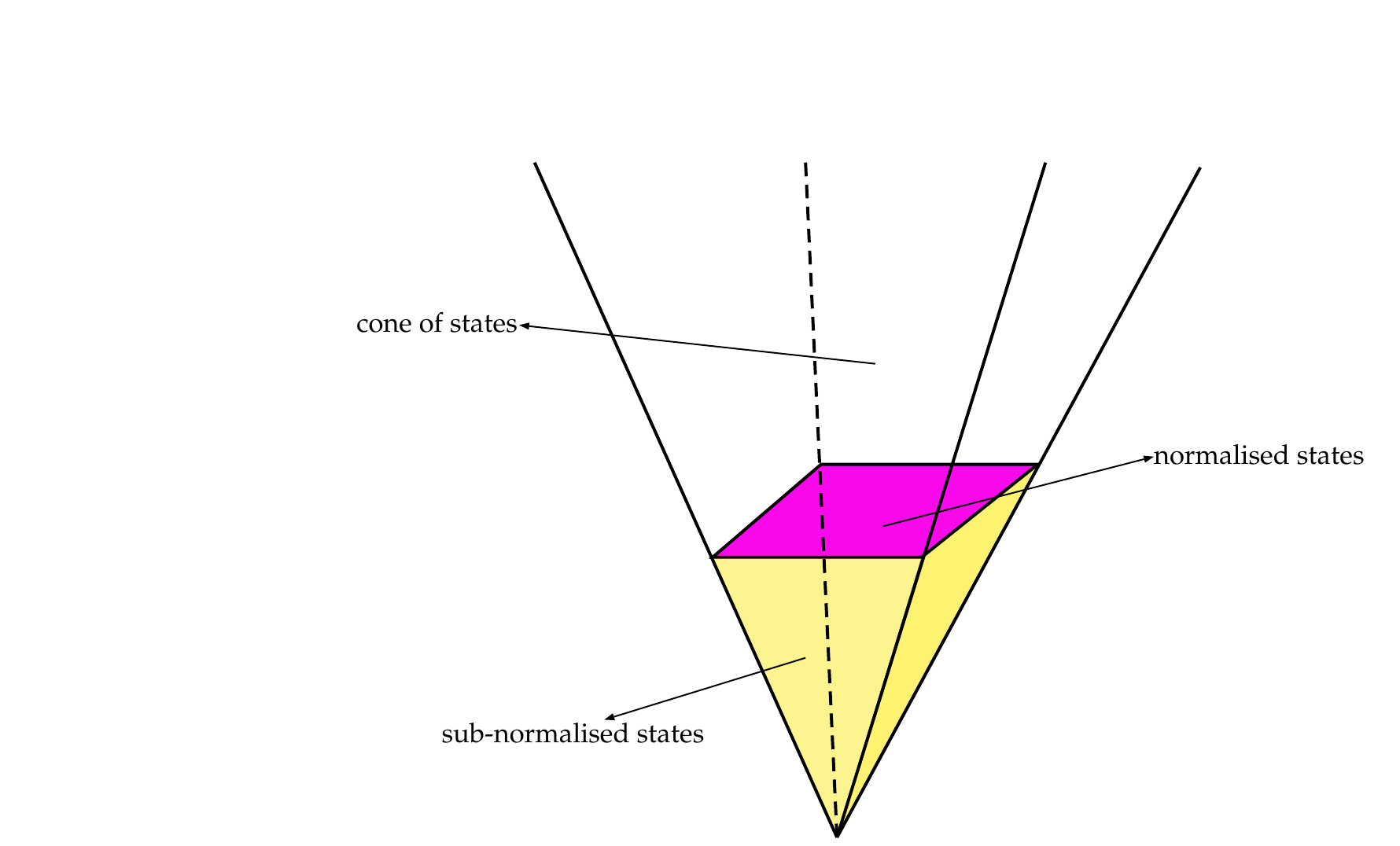}
\par\end{centering}
\caption{\label{fig:causal cone}The cone of states of a causal theory. The
set of normalised states (in purple) is given by the intersection
of the hyperplane $\mathrm{tr}\:\xi=1$, for $\xi$ in $\mathsf{St}_{\mathbb{R}}\left(\mathrm{A}\right)$,
with the cone of states $\mathsf{St}_{+}\left(\mathrm{A}\right)$
. Below that hyperplane are the sub-normalised states (in yellow).
The coloured part in the cone of states is the set of states $\mathsf{St}\left(\mathrm{A}\right)$.
The white part corresponds to super-normalised elements of $\mathsf{St}_{+}\left(\mathrm{A}\right)$,
which are non-physical.}
\end{figure}

Convex combinations of normalised states do not have only a mathematical
meaning, but can be also realised operationally. Suppose we have $\rho_{p}=p\rho_{0}+\left(1-p\right)\rho_{1}$,
where $\rho_{0},\rho_{1}\in\mathsf{St}_{1}\left(\mathrm{A}\right)$.
We can prepare $\rho_{p}$ by using the following procedure.
\begin{enumerate}
\item First of all, we perform a binary test in some arbitrary system with
outcomes $\left\{ 0,1\right\} $ and outcome probabilities $p_{0}=p$
and $p_{1}=1-p$.
\item If the outcome is $i$, then we prepare $\rho_{i}$. In this way,
we realise the preparation-test $\left\{ p_{0}\rho_{0},p_{1}\rho_{1}\right\} $.
Note that each state $p_{i}\rho_{i}$ is not normalised because it
is not deterministic: the state $\rho_{i}$ is prepared with probability
$p_{i}$.
\item Finally, we perform a coarse-graining over the outcomes, getting $\rho_{p}=p\rho_{0}+\left(1-p\right)\rho_{1}$.
\end{enumerate}
A coarse-graining of normalised states is a non-trivial convex combination
of them. Clearly pure states admit only trivial convex decompositions.
Every convex decomposition of a state $\rho$ reflects a particular
way of preparing it.
\begin{defn}
Let $\rho,\sigma\in\mathsf{St}_{1}\left(\mathrm{A}\right)$. We say
that $\sigma$ is \emph{contained} in $\rho$ if $\rho$ can be written
as $\rho=p\sigma+\left(1-p\right)\tau$, where $p\in\left(0,1\right]$
and $\tau\in\mathsf{St}_{1}\left(\mathrm{A}\right)$.
\end{defn}

In other words, $\sigma$ is contained in $\rho$ if it arises in
a convex decomposition of $\rho$. Clearly if $\rho$ is pure, only
$\rho$ can be contained in it. At the other extreme we have internal
states.
\begin{defn}
A state $\omega$ is \emph{internal} if every (normalised) state $\rho$
is contained in it.
\end{defn}

We will make use of the following definition too.
\begin{defn}
\label{def:equality upon input}Two transformations $\mathcal{A},\mathcal{A}'\in\mathsf{Transf}\left(\mathrm{A},\mathrm{B}\right)$
are \emph{equal upon input} of $\rho$, written as $\mathcal{A}=_{\rho}\mathcal{A}'$
if $\mathcal{A}\sigma=\mathcal{A}'\sigma$ for every state $\sigma$
contained in $\rho$.
\end{defn}

We conclude this subsection with the important definition of perfectly
distinguishable states, which are states that can be distinguished
in a single shot. It means that if the state is not known, it can
be identified with certainty.
\begin{defn}
The (normalised) states $\left\{ \rho_{i}\right\} _{i\in X}$ are
\emph{perfectly distinguishable} if there exists an observation-test
$\left\{ a_{i}\right\} _{i\in X}$, called the \emph{perfectly distinguishing
test}, such that
\[
\left(a_{i}\middle|\rho_{j}\right)=\delta_{ij}.
\]
\end{defn}

Perfectly distinguishable states might not exist in some systems of
a theory, as shown in the following example \cite{Objectivity}.
\begin{example}
Start with the state space of the classical trit, represented in fig.~\ref{fig:The-state-space},
with pure states $\alpha_{1}$, $\alpha_{2}$, $\alpha_{3}$.
\begin{figure}
\begin{centering}
\subfloat[\label{fig:The-state-space}The state space of a restricted trit coincides
with that of a classical trit.]{\begin{centering}
\includegraphics[scale=0.8]{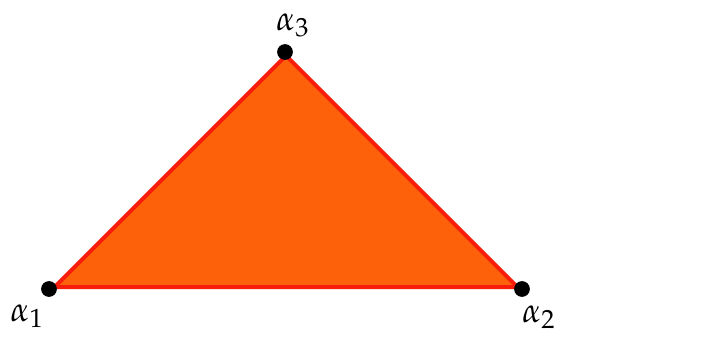}
\par\end{centering}
}\qquad{}\qquad{}\subfloat[\label{fig:A-section-of}A cross-section of the effect cone of the
classical trit (in orange) and of the restricted trit (in blue). The
dual cone is the same in both cases (in orange). The restriction on
effects is apparent.]{\begin{centering}
\includegraphics[scale=0.8]{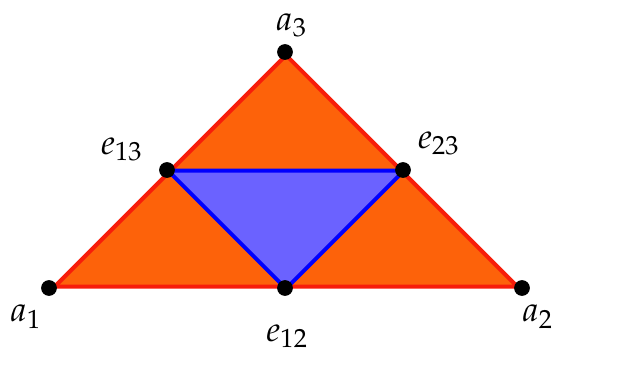}
\par\end{centering}
}
\par\end{centering}
\caption{The restricted trit}

\end{figure}
 They are perfectly distinguishable with observation-test $\left\{ a_{1},a_{2},a_{3}\right\} $:
$\left(a_{i}\middle|\alpha_{j}\right)=\delta_{ij}$. Instead of allowing
the full set of effects of classical theory, suppose that, for some
reasons, the most fine-grained effects that are allowed are $e_{ij}=\frac{1}{2}\left(a_{i}+a_{j}\right)$,
with $i\neq j$. A section of the dual cone (the same as the effect
cone of classical theory), and of the effect cone of the restricted
trit is represented in fig.~\ref{fig:A-section-of}. 

Since we have a smaller set of effects than the original classical
trit, we must check what happens to the state space. Indeed it may
happen that two states become tomographically indistinguishable because
there are not enough effects to witness their difference. However,
this is not the case of the restricted trit. The reason is that the
effects $e_{ij}$ are linearly independent, therefore they span exactly
the same effect vector space as the effects $a_{i}$, which is what
determines the tomographic power of a theory. Therefore the state
space of the restricted trit coincides with that of the classical
trit (cf.\ fig.~\ref{fig:The-state-space}). 

However, the restriction on the allowed effects has a dramatic consequence:
there are \emph{no} perfectly distinguishable pure states. The lack
of perfectly distinguishable \emph{pure} states determines the lack
of perfectly distinguishable states. Indeed, if $\rho_{1}$ were perfectly
distinguishable from $\rho_{2}$, all pure the states contained in
$\rho_{1}$ would be perfectly distinguishable from the pure states
contained in $\rho_{2}$. This is because if $\left(a\middle|\rho\right)=0$
(resp.\ $\left(a\middle|\rho\right)=1$) one has $a=_{\rho}0$ (resp.\ $a=_{\rho}u$).

First of all, let us show that $\left\{ \alpha_{1},\alpha_{2},\alpha_{3}\right\} $
are no longer perfectly distinguishable. Consider a generic effect
$e=\lambda_{12}e_{12}+\lambda_{13}e_{13}+\lambda_{23}e_{23}$, where
$\lambda_{ij}\geq0$. This effect could yield 0 on $\alpha_{2}$ and
$\alpha_{3}$ if and only if $\lambda_{12}=\lambda_{13}=\lambda_{23}=0$,
but this would be the zero effect, which cannot yield 1 on $\alpha_{1}$.
This means that the $\alpha_{i}$'s cannot be jointly perfectly distinguishable.
Maybe we can still find a pair of $\alpha_{i}$'s that are perfectly
distinguishable? The answer is again negative. To see it, take e.g.\ the
pair $\left\{ \alpha_{1},\alpha_{2}\right\} $ (for the others the
argument is the same). The only element in the effect cone that yields
1 on $\alpha_{1}$ and 0 on $\alpha_{2}$ is $2e_{13}$, but this
is \emph{not} a physical effect, because $u-2e_{13}=a_{2}$, which
is \emph{not} an effect. In other words, $2e_{13}$ cannot exist in
an observation-test of the form $\left\{ 2e_{13},u-2e_{13}\right\} $,
but all effects must be part of some observation-test! In conclusion,
the restricted trit has \emph{no} perfectly distinguishable states.
\end{example}

The example above is based on a theory where the no-restriction hypothesis
fails: the effect cone $\mathsf{Eff}_{+}\left(\mathrm{A}\right)$
was strictly contained in the dual cone $\mathsf{St}_{+}^{*}\left(\mathrm{A}\right)$.
If the two cones coincide, we can always prove that at least two perfectly
distinguishable pure states exist \cite{Objectivity}.
\begin{prop}
\label{prop:no-restriction distinguishable}In an unrestricted theory,
for every pure state $\psi_{1}$ there exists another pure state $\psi_{2}$
such that $\left\{ \psi_{1},\psi_{2}\right\} $ are perfectly distinguishable.
\end{prop}

\begin{proof}
Let $\psi_{1}$ be a pure state. The proof will consist of some steps.
In the first step, let us prove that there exists a non-trivial element
$f$ of the dual cone $\mathsf{St}_{+}^{*}\left(\mathrm{A}\right)$
such that $\left(f\middle|\psi_{1}\right)=0$. Note that being pure,
$\psi_{1}$ lies in some supporting hyperplane through the origin
of the cone $\mathsf{St}_{+}\left(\mathrm{A}\right)$ \cite{Boyd}.
Such a hyperplane must have equation $\left(f\middle|x\right)=0$
for all $x\in\mathsf{St}_{\mathbb{R}}\left(\mathrm{A}\right)$, where
$f$ is some non-trivial linear functional on $\mathsf{St}_{\mathbb{R}}\left(\mathrm{A}\right)$,
otherwise it would not pass through the origin (i.e.\ the null vector).
Being a supporting hyperplane, we can choose $f$ to be in the dual
cone $\mathsf{St}_{+}^{*}\left(\mathrm{A}\right)$ \cite{Boyd}. Thus
we have found $f\in\mathsf{St}_{+}^{*}\left(\mathrm{A}\right)$ such
that $\left(f\middle|\psi_{1}\right)=0$.

Let us consider the maximum of $f$ on the state space. Since $f$
is continuous and the state space is compact, it achieves its maximum
$\lambda^{*}$ on some state $\rho^{*}$. Note that $\lambda^{*}>0$,
otherwise $f$ would be the zero functional. Let us show that the
maximum is attained on some pure state. If $\rho^{*}$ is already
a pure state, there is nothing to prove. If it is not, consider a
refinement of $\rho^{*}$ in terms of pure states, $\rho^{*}=\sum_{i}p_{i}\psi_{i}$,
where $\left\{ p_{i}\right\} $ is a probability distribution. Apply
$f$ to $\rho^{*}$: 
\[
\lambda^{*}=\left(f\middle|\rho^{*}\right)=\sum_{i}p_{i}\left(f\middle|\psi_{i}\right).
\]
Clearly $\lambda^{*}\leq\max_{i}\left(f\middle|\psi_{i}\right)$,
but being $\lambda^{*}$ the maximum of $f$ , in fact $\lambda^{*}=\max\left(f\middle|\psi_{i}\right)$.
This means that there is a pure state $\psi_{2}$, chosen among these
$\psi_{i}$'s, on which $f$ attains its maximum.

Now consider the functional $a_{2}:=\frac{1}{\lambda^{*}}f$, which
takes values in the interval $\left[0,1\right]$ when applied to states.
Specifically $\left(a_{2}\middle|\psi_{2}\right)=1$ and $\left(a_{2}\middle|\psi_{1}\right)=0$.
By the no-restriction hypothesis, it is a valid effect, so we can
construct the observation-test $\left\{ a_{1},a_{2}\right\} $, where
$a_{1}:=u-a_{2}$, which distinguishes perfectly between $\psi_{1}$
and $\psi_{2}$.
\end{proof}
Even though the no-restriction hypothesis guarantees the existence
of perfectly distinguishable states, we do not wish to assume it for
its lack of operational motivation. Instead, in chapter~\ref{chap:Sharp-theories-with}
we will prove the existence of perfectly distinguishable states for
a class of OPTs from first principles.

Finally we give the following definition.
\begin{defn}
A set $\left\{ \rho_{i}\right\} _{i\in X}$ of perfectly distinguishable
states is \emph{maximal} if no state $\rho_{0}$ can be added such
that $\left\{ \rho_{i}\right\} _{i\in X}\cup\left\{ \rho_{0}\right\} $
are still perfectly distinguishable.
\end{defn}

If the $\rho_{i}$'s are pure states, instead of writing ``maximal
set of perfectly distinguishable states'', we will opt for the shorter
terminology of ``pure maximal set''.

\subsection{The group of reversible channels\label{subsec:The-group-of-reversible}}

Causality provides us with a new insight into the structure of reversible
channels of a system $\mathrm{A}$, and will enable us to prove some
interesting properties. First of all, we show that the set of channels
is compact \cite[corollary 30]{Chiribella-purification}.
\begin{prop}
\label{prop:channels closed}For every pair of systems $\mathrm{A}$
and $\mathrm{B}$ the set of channels $\mathsf{DetTransf}\left(\mathrm{A},\mathrm{B}\right)$
is compact.
\end{prop}

\begin{proof}
Let $\left\{ \mathcal{C}_{n}\right\} $ be a sequence of channels
that converges to some transformation $\mathcal{C}\in\mathsf{Transf}\left(\mathrm{A},\mathrm{B}\right)$.
Let us prove that $\mathcal{C}$ is a channel. We have $u_{\mathrm{B}}\mathcal{C}_{n}=u_{\mathrm{A}}$
for all $n\in\mathbb{N}$, because they are all channels. Then, for
every state $\rho\in\mathsf{St}\left(\mathrm{A}\right)$, we have
\[
\left(u_{\mathrm{B}}\middle|\mathcal{C}\middle|\rho\right)=\lim_{n\rightarrow+\infty}\left(u_{\mathrm{B}}\middle|\mathcal{C}_{n}\middle|\rho\right)=\left(u_{\mathrm{A}}\middle|\rho\right).
\]
Therefore we conclude that $u_{\mathrm{B}}\mathcal{C}=u_{\mathrm{A}}$,
which means that $\mathcal{C}$ is a channel, by proposition~\ref{prop:characterization channel}.
We conclude that every convergent sequence of channels converges to
a channel, by which we obtain the closure of $\mathsf{DetTransf}\left(\mathrm{A},\mathrm{B}\right)$.
$\mathsf{DetTransf}\left(\mathrm{A},\mathrm{B}\right)$ is bounded
because it is a subset of $\mathsf{Transf}\left(\mathrm{A},\mathrm{B}\right)$,
therefore it is compact in the norm topology. By the equivalence of
the norm and operational topologies, we can say that $\mathsf{DetTransf}\left(\mathrm{A},\mathrm{B}\right)$
is compact in both topologies.
\end{proof}
We can exploit the compactness of channels to prove that the group
of reversible channels is compact as well \cite[corollary 31]{Chiribella-purification}.
\begin{prop}
\label{prop:group compact}For every system $\mathrm{A}$, the group
$G_{\mathrm{A}}$ of reversible channels is compact.
\end{prop}

\begin{proof}
Consider a convergent sequence $\left\{ \mathcal{U}_{n}\right\} $
of reversible channels. By proposition~\ref{prop:channels closed},
the limit $\mathcal{U}$ is a channel on $\mathrm{A}$. Now consider
the sequence $\left\{ \mathcal{U}_{n}^{-1}\right\} $. Since the set
of channels is compact, there exists a subsequence $\left\{ \mathcal{U}_{n_{k}}^{-1}\right\} $
that converges to a channel $\mathcal{C}$. Now
\[
\mathcal{CU}=\lim_{n\rightarrow+\infty}\mathcal{U}_{n_{k}}^{-1}\mathcal{U}_{n_{k}}=\mathcal{I},
\]
and
\[
\mathcal{UC}=\lim_{n\rightarrow+\infty}\mathcal{U}_{n_{k}}\mathcal{U}_{n_{k}}^{-1}=\mathcal{I},
\]
where we have used lemma~\ref{lem:lemma product}, and the fact that
every subsequence of $\mathcal{U}_{n}$ converges to $\mathcal{U}$.
This proves that $G_{\mathrm{A}}$ is closed. Recalling the fact that
the set of channels is bounded, the group is compact (in both topologies).
\end{proof}
Finite groups are trivial examples of compact groups. Compact groups
enjoy a remarkable mathematical property: they admit a finite Haar
measure $h$, unique up to rescaling, namely a measure such that $h\left(G\right)<+\infty$,
and that is invariant under left and right action of the group\footnote{For a finite group, the Haar measure is simply the counting measure,
that counts the number of elements in a subset of the group.} \cite{Folland_real}. In other words, if $S$ is a Borel subset of
$G$, one has 
\[
h\left(\mathcal{U}S\right)=h\left(S\mathcal{U}\right)=h\left(S\right),
\]
for every $\mathcal{U}\in G$, where $\mathcal{U}S=\left\{ \mathcal{U}\mathcal{V}:\mathcal{V}\in S\right\} $
and $S\mathcal{U}=\left\{ \mathcal{V}\mathcal{U}:\mathcal{V}\in S\right\} $.
Since $h$ is a finite measure, it can be renormalised so that $h\left(G\right)=1$,
which we will always assume in the following. This measure will be
one of the key ingredients in the construction of invariant states,
an extremely important notion that will be used throughout this thesis.
\begin{defn}
A state $\chi\in\mathsf{St}\left(\mathrm{A}\right)$ is called \emph{invariant}
if $\mathcal{U}\chi=\chi$ for every $\mathcal{U}\in G_{\mathrm{A}}$.
\end{defn}

Since we are working in causal theories, in the following we will
consider only \emph{normalised} invariant states. Do they exist? The
answer is affirmative.
\begin{prop}
In every causal theory there is at least one invariant state.
\end{prop}

\begin{proof}
Take a pure state $\psi$, and consider the state
\[
\chi_{\psi}:=\int_{G}\mathcal{V}\psi\,\mathrm{d}\mathcal{V},
\]
where $\mathrm{d}\mathcal{V}$ is the Haar probability measure. $\chi_{\psi}$
is invariant, indeed for every reversible channel $\mathcal{U}$,
we have
\[
\mathcal{U}\chi_{\psi}=\int_{G}\mathcal{U}\mathcal{V}\psi\,\mathrm{d}\mathcal{V}=\int_{G}\mathcal{W}\psi\,\mathrm{d}\left(\mathcal{U}^{-1}\mathcal{W}\right)=\int_{G}\mathcal{W}\psi\,\mathrm{d}\mathcal{W}=\chi_{\psi},
\]
where we have set $\mathcal{W}:=\mathcal{U}\mathcal{V}$, and we have
exploited the invariance of the Haar probability measure. Therefore
at least one invariant state always exists.
\end{proof}
If $G$ is a finite, the definition of $\chi_{\psi}$ becomes $\chi_{\psi}=\frac{1}{\left|G\right|}\sum_{i=1}^{\left|G\right|}\mathcal{V}_{i}\psi$.

In general, there will be more than one invariant state, because $\chi_{\psi}$
depends on the choice of $\psi$, and different $\psi$'s may give
rise to different invariant states $\chi_{\psi}$. However, if $\psi$
and $\psi'$ are in the same orbit of the action of $G$, which means
$\psi'=\mathcal{U}\psi$, they generate the same invariant state.
Indeed,
\[
\chi_{\psi'}=\int_{G}\mathcal{V}\psi'\,\mathrm{d}\mathcal{V}=\int_{G}\mathcal{V}\mathcal{U}\psi\,\mathrm{d}\mathcal{V}=\int_{G}\mathcal{W}\psi\,\mathrm{d}\left(\mathcal{W}\mathcal{U}^{-1}\right)=\int_{G}\mathcal{W}\psi\,\mathrm{d}\mathcal{W}=\chi_{\psi},
\]
setting $\mathcal{W}:=\mathcal{V}\mathcal{U}$, and harnessing the
invariance of the Haar probability measure. If there is only one orbit
on the set of normalised pure states, the invariant state is unique.
When there is only one orbit, the action is called \emph{transitive}.
In practice it means that for every pair of pure states $\psi$, $\psi'$
there exists a reversible channel $\mathcal{U}$ such that $\psi'=\mathcal{U}\psi$.
\begin{prop}
\label{prop:uniqueness invariant}A causal theory with transitive
action has a unique invariant state, which is internal.
\end{prop}

\begin{proof}
We need to prove only the second part of the statement. Since the
invariant state is unique, we have $\chi=\int_{G}\mathcal{V}\psi\,\mathrm{d}\mathcal{V}$
for every pure state $\psi$, therefore, by linearity 
\begin{equation}
\chi=\int_{G}\mathcal{V}\rho\,\mathrm{d}\mathcal{V},\label{eq:invariant from rho}
\end{equation}
for every state $\rho$. Now, since we are in finite dimension, by
Carathéodory's theorem for convex geometry \cite{Caratheodory,Steinitz},
the integral~\eqref{eq:invariant from rho} is a finite convex combination
of reversible channels: $\chi=\sum_{i}p_{i}\mathcal{U}_{i}\rho$.
Now, apply $\mathcal{U}_{i_{0}}^{-1}$ to $\chi$, where $\mathcal{U}_{i_{0}}$
is one of the reversible channels in the convex combination. This
yields
\[
\chi=\mathcal{U}_{i_{0}}^{-1}\chi=p_{i_{0}}\rho+\sum_{i\neq i_{0}}p_{i}\mathcal{U}_{i_{0}}^{-1}\mathcal{U}_{i}\rho
\]
for every $\rho$. This shows that every state is contained in $\chi$,
whence $\chi$ is internal.
\end{proof}
The uniqueness of the invariant state does \emph{not} imply that the
action is transitive, as shown in the following example.
\begin{example}
Consider theory with a system whose state space is a diamond, as illustrated
in fig.~\ref{fig:rombo}.
\begin{figure}
\begin{centering}
\includegraphics[scale=0.8]{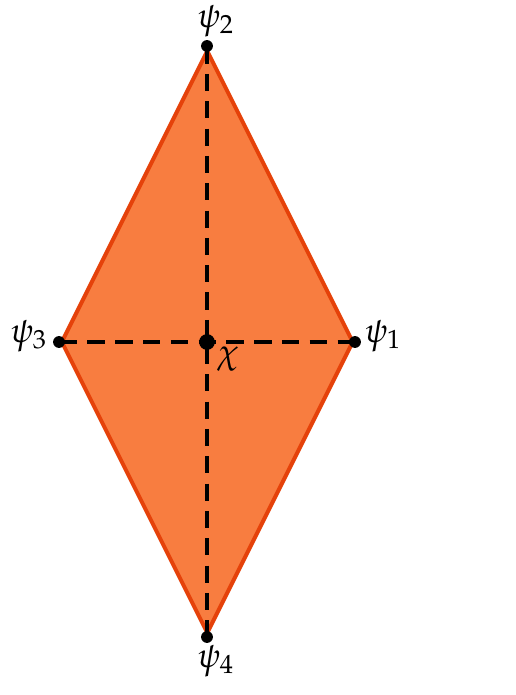}
\par\end{centering}
\caption{\label{fig:rombo}A state space where there is a unique invariant
state, the centre $\chi$ of the diamond, but the action of reversible
channels is not transitive on the pure states $\psi_{1}$, $\psi_{2}$,
$\psi_{3}$, and $\psi_{4}$. }

\end{figure}
 The pure states are its vertices $\psi_{1}$, $\psi_{2}$, $\psi_{3}$,
and $\psi_{4}$. Assuming that all symmetries of the state space are
allowed reversible channels, besides the identity we have the reflection
across the longer diagonal and the reflection across the shorter diagonal.
There are two orbits on the set of pure states: $\left\{ \psi_{1},\psi_{3}\right\} $
and $\left\{ \psi_{2},\psi_{4}\right\} $, so the action is \emph{not}
transitive. Yet there is a unique invariant state $\chi$, which is
the centre of the diamond.
\end{example}

Another interesting property of theories with a transitive action
of the group of reversible channels is that the set of pure states
is compact.
\begin{prop}
\label{prop:pure compactness}In a causal theory with transitive action,
the set of (normalised) pure states is compact.
\end{prop}

\begin{proof}
We only need to prove that the set of pure states is closed, being
a subset of a compact set, the set of states. Let $\left\{ \psi_{n}\right\} $
be a sequence of pure states converging to the state $\psi$. We have
to prove that $\psi$ is pure. Take a fixed pure state $\varphi_{0}$,
then, by the transitivity of the action, there exists a sequence $\left\{ \mathcal{U}_{n}\right\} $
of reversible channels such that $\psi_{n}=\mathcal{U}_{n}\varphi_{0}$.
By proposition~\ref{prop:group compact}, there exists a subsequence
$\left\{ \mathcal{U}_{n_{k}}\right\} $ converging to the reversible
channel $\mathcal{U}$. Hence
\[
\psi=\lim_{n\rightarrow+\infty}\psi_{n_{k}}=\lim_{n\rightarrow+\infty}\mathcal{U}_{n_{k}}\varphi_{0}=\mathcal{U}\varphi_{0}.
\]
We now that reversible channels send pure states to pure states \cite{Chiribella-Scandolo-entanglement},
so $\psi$ is a pure state too.
\end{proof}

\chapter{Resource theories\label{chap:Resource-theories}}

In recent years, thermodynamics far from the thermodynamic limit has
attracted remarkable attention \cite{delRio,Xuereb,Anders-thermo},
given the ever-increasing developments in the field of nanotechnology,
and the actual experimental realisations of microscopic systems. In
this new regime quantum effects become important, and must be taken
into account in the physical description. To this end, quantum resource
theories \cite{Quantum-resource-1,Quantum-resource-2,Gour-review}
have provided valuable tools, especially to study deterministic thermodynamics,
viz.\ non-fluctuating processes. The power of resource theories stems
from the fact that they describe thermodynamic transitions in a natural
way: instead of specifying the actual details of the ``mechanics''
implementing a thermodynamic transition between two quantum states,
they aim to study the convertibility criteria between states, namely
\emph{when} (and \emph{not} how) a state can be transformed into another.

The ideas behind resource theories are pretty general and independent
of the physical theory in which they are formulated \cite{Resource-theories,Resource-monoid,Resource-knowledge,Resource-currencies}.
The key concept is that the states of a physical system are to be
viewed as resources. Then one studies the conversion between different
resources by means of a restricted set of channels, which are singled
out by the physical setting as those that easy to implement. For this
reason, they are called \emph{free operations}.

As to resources, some of them are very abundant or easy to get, therefore
they are \emph{free}; others, instead, may be extremely rare or hard
to obtain, and these will be the most valuable. Thus, in the resource-theoretic
framework a hierarchy of resources according to their value is naturally
built in. This hierarchy is translated into mathematical terms as
a partial preorder induced by free operations. In this preorder, a
state is more valuable than another if from the former we can reach
a larger set of target states. In the light of this hierarchy, it
is important to find functions that assign a value to resources in
a way that it is consistent with the preorder. These are the \emph{resource
monotones} \cite{MHorodecki,Quantum-resource-1,Resource-theories,Quantum-resource-2,Gour-review},
and they can be very often identified with some thermodynamic potentials.

The subject of resource theories is extremely broad, both in quantum
theory and in GPTs. In this chapter we give a brief presentation of
the minimal necessary concepts to understand the contents and the
spirit of the rest of the thesis. We will start with the definition
of resource theory, and then we will introduce the resource preorder,
establishing a hierarchy among the states of a GPT. Finally, we define
resource monotones. We just mention that, besides the results presented
here for resource theories in GPTs, we have studied how the thermodynamic
limit emerges in the resource-theoretic approach \cite{EastThesis,East-article},
but they are not reported here, for they have no direct consequences
for the rest of the thesis.

\section{Resource theories of states}

In a great deal of situations, even in everyday life, one comes across
resources, so often that we do not even question the meaning of this
term. But what is actually a resource? To define a resource, one must
first specify a task. In this setting, an object is a resource if
it has some value with respect to performing that task. For example,
a chemical product is a resource if it can be used to obtain something
we want, for example it fertilises the soil.
\begin{example}
Ammonia has plenty of uses as a chemical product. In the process of
producing ammonia, $\mathrm{NH}_{3}$, we need hydrogen and nitrogen
as resources, according to the reaction
\[
\mathrm{N}_{2}+3\mathrm{H}_{2}\rightarrow2\mathrm{N}\mathrm{H}_{3}.
\]
Regarding all the chemical species of the above reaction as resources,
we understand an important feature of resources: they can be converted
into one another.
\end{example}

As the example above highlights, chemistry is the prototypical resource
theory, but there are plenty of examples even in physics, which is
the scope of our analysis. In particular, the advent of quantum information
processing has provided very many examples of new tasks that can be
performed only with access to quantum states, which, in our terminology,
become quantum resources. Think, for example, of quantum teleportation
\cite{Teleportation} or dense-coding \cite{Dense-coding}, which
are only possible if one has access to entangled states. For this
reason it is very natural to take quantum entanglement as a resource,
and indeed it was the first quantum resource to be studied extensively
\cite{Entanglement-entropy2}.

In this presentation, we want to go beyond quantum theory, and to
treat resources in general physical theories using the formalism of
causal GPTs. Indeed, as shown in \cite{Resource-theories,Resource-monoid,EastThesis,East-article},
many of the features of quantum resource theories \cite{Quantum-resource-1,Quantum-resource-2,Gour-review}
are in fact far more general in their scope. In the following we will
always assume that our resources are states, like in most of concrete
examples, but this need not be the case \cite{Resource-theories,Coherence-beyond-states,Gour-review,Thermal-capacity,Gour-super,Gour-entropy-channel}.

The idea behind a theory of resources is to identify which states
are useless as resources, because they are abundant, easy to obtain,
or of no use. These states are called \emph{free states}. Clearly,
all states of the trivial system, which are ``non-physical states'',
are free, and represent a ``void'' resource. For a generic system
$\mathrm{A}$, we have a partition of $\mathsf{St}_{1}\left(\mathrm{A}\right)$
into the set of free states $F$ and its complement, which contains
the true, costly, resources. As we highlighted above, the specification
of $F$ depends on the specific task we are considering. We can easily
have a state free for one task, and a true resource for another.

We will assume the set of free states to be topologically closed:
a state that can be arbitrarily well approximated by free states must
be free too. Finally, is natural to assume that the tensor product
of free states is another free state of the composite system. This
is because if two resources have zero value if taken on their own,
they have zero value also when taken together \cite{Resource-theories},
therefore $F_{\mathrm{AB}}\supseteq F_{\mathrm{A}}\otimes F_{\mathrm{B}}$,
where $F_{\mathrm{S}}$ denotes the set of free states of system $\mathrm{S}$.

Clearly resources can be manipulated by an agent who has a task to
perform, which means resources can be converted from one to another.
However, the specific task one is considering constrains the allowed
manipulation of resources: one is restricted to performing a subset
of all transformations of the theory. These restricted transformations
are usually regarded as being easy to implement, or realisable at
no cost. They are called \emph{free transformations}. Since we want
the agent to be able to \emph{choose} the operations they implement,
it is customary to focus only on deterministic transformations, i.e.\ channels,
where there is no randomness involved. This is what we will do. For
every input and output system $\mathrm{A}$ and $\mathrm{B}$ respectively,
we will consider a partition of channels $\mathsf{DetTransf}\left(\mathrm{A},\mathrm{B}\right)$
into a set of free channels $F_{\mathrm{A}\rightarrow\mathrm{B}}$
and a set of non-free channels. We will call $F_{\mathrm{A}\rightarrow\mathrm{B}}$
the set of \emph{free operations}. For similar reason to the case
of states, we take $F_{\mathrm{A}\rightarrow\mathrm{B}}$ to be topologically
closed. Given the definition of states in OPTs, when $\mathrm{A}$
is the trivial system, $F_{\mathrm{A}\rightarrow\mathrm{B}}$ defines
the free states of $\mathrm{B}$. Therefore there is no need to deal
with free operations and free states as separate objects \cite{Resource-theories,Gour-review}. 

When $\mathrm{B}$ is the trivial system, we will always take the
deterministic effect of $\mathrm{A}$ to be free. The reason why we
do this will become clear in the following section.

Now, what are the properties of the free operations of a physical
theory? They are summarised in the following requirements:
\begin{itemize}
\item the identity channel is free because doing nothing comes for free;
\item the $\mathtt{SWAP}$ is a free operation;
\item the sequential composition of free operations is a free operation;
\item the parallel composition of free operations is a free operation.
\end{itemize}
These requirements are quite intuitive in their meaning. From these
properties, it is obvious that the action of a free operation on a
free state yields a free state, because it is the sequential composition
of two free operations.

Now we are ready to collect all these remarks in the following definition
of a resource theory.
\begin{defn}
\label{def:A-resource-theory}A \emph{resource theory} on a causal
GPT is the specification, for every pair of systems $\mathrm{A}$
and $\mathrm{B}$, of a closed set of \emph{free operations} $F_{\mathrm{A}\rightarrow\mathrm{B}}\subseteq\mathsf{DetTransf}\left(\mathrm{A},\mathrm{B}\right)$.
This set satisfies the following properties:
\begin{itemize}
\item $\mathcal{\mathcal{I}}\in F_{\mathrm{A}\rightarrow\mathrm{A}}$;
\item $\mathtt{SWAP}\in F_{\mathrm{AB}\rightarrow\mathrm{BA}}$;
\item if $\mathcal{F}_{1}\in F_{\mathrm{A}\rightarrow\mathrm{B}}$ and $\mathcal{F}_{2}\in F_{\mathrm{B}\rightarrow\mathrm{C}}$,
then $\mathcal{F}_{2}\circ\mathcal{F}_{1}\in F_{\mathrm{A}\rightarrow\mathrm{C}}$;
\item if $\mathcal{F}_{1}\in F_{\mathrm{A}\rightarrow\mathrm{B}}$ and $\mathcal{F}_{2}\in F_{\mathrm{C}\rightarrow\mathrm{D}}$,
then $\mathcal{F}_{1}\otimes\mathcal{F}_{2}\in F_{\mathrm{AC}\rightarrow\mathrm{BD}}$;
\item $F_{\mathrm{A}\rightarrow\mathrm{I}}=\left\{ u_{\mathrm{A}}\right\} $.
\end{itemize}
\end{defn}

The interesting case, clearly, is when $F_{\mathrm{A}\rightarrow\mathrm{B}}$
is \emph{strictly} contained in the set of all channels from $\mathrm{A}$
and $\mathrm{B}$, because this means there is a true restriction
on the allowed manipulations of the system. It is indeed from this
restriction that the role of states as resources emerges, otherwise
the resource theory is trivial.

From a more mathematical perspective, the properties satisfied by
the set of free operations tell us that a resource theory specifies
a strict symmetric monoidal subcategory of the underlying OPT on which
it is defined \cite{Resource-theories,Resource-monoid}.
\begin{example}
The prototypical example of resource theory is the resource theory
of quantum entanglement \cite{Entanglement-entropy2}, where free
states are separable states, and free operations are LOCC channels\footnote{Recall that LOCC stands for ``Local Operations and Classical Communication''.}
\cite{LOCC1,LOCC2,Lo-Popescu}, namely channels that can be decomposed
as the action of local operations and classical communication exchanged
between the two parties. Here is an example:\[
\sum_{j_1,j_2,j_3}\begin{aligned} \Qcircuit @C=1em @R=.7em @!R { & \qw \poloFantasmaCn{\rA_0} & \gate{ \cA_{j_{1}} }\ar@{-->}[drr] & \qw & \qw \poloFantasmaCn{\rA_1}&\qw & \gate{\cA^{\left(j_{1},j_2\right)}_{j_3} } & \qw \poloFantasmaCn{\rA_2} &\qw \\ & \qw \poloFantasmaCn{\rB_0} & \qw &\qw & \gate{ \cB^{\left(j_{1}\right)}_{j_{2}} }\ar@{-->}[urr] & \qw &\qw \poloFantasmaCn{\rB_1} & \qw &\qw }\end{aligned} ~,
\]where dashed wires represent classical communication. Note that separable
states are exactly the states that can be prepared with a LOCC protocol.

Alternatively, one can take as free operations the largest set of
channels that send separable states to separable states. This is the
set of separable operations \cite{SEP}.
\end{example}

We saw that free states are a special kind of free operations, therefore
specifying free operations determines also free states: they are exactly
those operations with trivial input. Yet, sometimes it is possible,
or much easier, to specify only the free states. Can we recover a
full resource theory just from the structure of free states? In general
we will have a large freedom in the choice of free operations that
respect the requirements of definition~\ref{def:A-resource-theory}.
The largest set of free operations compatible with the set of free
states $F$ is the set of channels from $\mathrm{A}$ to $\mathrm{B}$
that send the free states of $\mathrm{A}$ to free states of $\mathrm{B}$:
\begin{equation}
F_{\mathrm{A}\rightarrow\mathrm{B}}^{\max}=\left\{ \mathcal{C}\in\mathsf{DetTransf}\left(\mathrm{A},\mathrm{B}\right):\mathcal{C}f\in F_{\mathrm{B}},\textrm{ where }f\in F_{\mathrm{A}}\right\} ,\label{eq:RNG}
\end{equation}
where $F_{\mathrm{S}}$ is the set of free states of system $\mathrm{S}$.
This is called the set of ``resource-non-generating operations'',
and it is the largest set of free operations, because every other
set of free operations satisfies the condition of eq.~\eqref{eq:RNG}.
However, this condition is not enough, because it says nothing about
the behaviour of these operations under parallel composition. Indeed,
from eq.~\eqref{eq:RNG} we know that the parallel composition of
two free operations in $F_{\mathrm{A}\rightarrow\mathrm{B}}^{\max}$
preserves the product of free states $\left(\mathcal{C}\otimes\mathcal{C}'\right)\left(f_{\mathrm{A}}\otimes f_{\mathrm{B}}\right)\in F_{\mathrm{AB}}$,
but the free states $F_{\mathrm{AB}}$ of the composite system $\mathrm{AB}$
may be also of a different form from product states. Therefore we
take as free operations the subset of $F_{\mathrm{A}\rightarrow\mathrm{B}}^{\max}$
such that arbitrary tensor products of its elements preserve all free
states of any composite system. These channels are also known as ``completely
resource-non-generating operations'' \cite{Gour-review}.

\section{The resource preorder\label{sec:The-resource-preorder}}

After giving the formal definition of resource theory, we want to
examine its implications. One of the main applications of the resource-theoretic
framework is to study state conversions when there is a restriction
on the allowed channels. A physical example is thermodynamics, where
one studies the thermodynamic convertibility of states under some
external constraints, e.g.\ the fact that the temperature is constant
or that no heat is exchanged in the thermodynamic transition. Thermodynamics
also provides an example of a situation in which, to derive thermodynamic
properties, such as relations between thermodynamic potentials, it
is not important to specify the actual details of the transition between
two states, but it is enough to know that such a transition is possible
\cite{giles2016mathematical}. This is the typical situation arising
when one considers functions of state, and this structure is beautifully
captured by the preorder one can set up in a resource theory.

The preorder arises when one wishes to better specify the value of
a resource. Indeed, in the previous section, we divided states into
two classes: free and costly ones. This is a rough classification,
because it is natural to expect that even among costly states there
will be an internal hierarchy, with some of them closer to free states.

As it often happens experimentally, the processing of resources consumes
or degrades them, even if such a processing is free. In this way,
from a precious resource, we end up with resources that are more and
more useless. The idea is that a resource is more valuable than another
if we can reach a larger set of resources by manipulating it. Therefore,
the most natural way to define a hierarchy on resources is to take
advantage of this idea about resource processing.
\begin{defn}
In a resource theory of states, we say that a state $\rho$ is \emph{more
valuable} than a state\footnote{Note that $\rho$ and $\sigma$ can be states of different systems.}
$\sigma$, and we write $\rho\succsim\sigma$, if there exists a free
operation $\mathcal{F}$ such that $\mathcal{F}\rho=\sigma.$
\end{defn}

We see that the hierarchy among resources is based on resource convertibility,
and note that we do not need to specify the actual details of $\mathcal{F}$,
but only that such an $\mathcal{F}$ exists. As already written, the
relation $\succsim$ on the set of states is a preorder.
\begin{prop}
The relation $\succsim$ is a preorder\footnote{Recall that a preorder is a relation that is reflexive and transitive.}.
\end{prop}

\begin{proof}
The relation $\succsim$ is reflexive. Indeed for any state $\rho$,
$\rho\succsim\rho$ because the identity channel is a free operation.
The relation is also transitive. Indeed, suppose we have $\rho\succsim\sigma$
and $\sigma\succsim\tau$. This means that there exist two free operations
$\mathcal{F}_{1}$ and $\mathcal{F}_{2}$ such that $\mathcal{F}_{1}\rho=\sigma$
and $\mathcal{F}_{2}\sigma=\tau$. Taking the sequential composition,
we have the free operation $\mathcal{F}_{2}\mathcal{F}_{1}$ such
that $\mathcal{F}_{2}\mathcal{F}_{1}\rho=\tau$, which means $\rho\succsim\tau$.
\end{proof}
In general, however, we \emph{cannot} conclude that, if $\rho\succsim\sigma$
and $\sigma\succsim\rho$, then $\rho=\sigma$. This only means that
it is possible to convert $\rho$ into $\sigma$ with a free operation
$\mathcal{F}$, and $\sigma$ into $\rho$ with a free operation $\mathcal{F}'$.
Note that this does \emph{not} even mean that $\mathcal{F}'=\mathcal{F}^{-1}$.
Nevertheless, if $\rho\succsim\sigma$ and $\sigma\succsim\rho$,
we can think of $\rho$ and $\sigma$ as \emph{equivalent}, and we
say that $\rho$ is \emph{as valuable as} $\sigma$. Indeed, it is
straightforward to see that we can define an equivalence relation
$\sim$, where $\rho\sim\sigma$ if $\rho\succsim\sigma$ and $\sigma\succsim\rho$.
Taking the quotient of the set of states modulo $\sim$, the preorder
$\succsim$ becomes a partial order $\succeq$ between equivalence
classes.

Sometimes, given two states $\rho$ and $\sigma$ we cannot find neither
a free operation converting $\rho$ into $\sigma$, nor a free operation
converting $\sigma$ into $\rho$. In this case $\rho$ and $\sigma$
are ``incomparable'': we cannot establish which of the two is the
more valuable. This means that the preorder on resources is only partial.

Let us prove that this preorder is compatible with the tensor product
of states. 
\begin{prop}
\label{prop:compatible preorder}Suppose $\rho\succsim\rho'$ and
$\sigma\succsim\sigma'$, then one has $\rho\otimes\sigma\succsim\rho'\otimes\sigma'$.
\end{prop}

\begin{proof}
By hypothesis there exist two free operations $\mathcal{F}_{1}$ and
$\mathcal{F}_{2}$ such that $\mathcal{F}_{1}\rho=\rho'$ and $\mathcal{F}_{2}\sigma=\sigma'$.
Then $\mathcal{F}_{1}\otimes\mathcal{F}_{2}$ is a free operation,
and 
\[
\left(\mathcal{F}_{1}\otimes\mathcal{F}_{2}\right)\left(\rho\otimes\sigma\right)=\rho'\otimes\sigma'.
\]
This implies $\rho\otimes\sigma\succsim\rho'\otimes\sigma'$.
\end{proof}
Therefore even the equivalence relation $\sim$ associated with the
preorder $\succsim$ is compatible with the parallel composition of
states.

Now we can prove that free states are indeed the least valuable states.
\begin{prop}
\label{prop:free states minima}Let $f\in\mathsf{St}_{1}\left(\mathrm{A}\right)$
be a free state. Then, for every state $\rho$ of any system $\mathrm{B}$,
one has $\rho\succsim f$.
\end{prop}

\begin{proof}
Let us consider the channel $\mathcal{F}=\left|f\right)_{\mathrm{A}}\left(u\right|_{\mathrm{B}}$,
which is a free operation because it is the sequential composition
of free operations. Then for every $\rho\in\mathsf{St}_{1}\left(\mathrm{B}\right)$,
we have
\[
\mathcal{F}\rho=f\mathrm{tr}\:\rho=f,
\]
so $\rho\succsim f$.
\end{proof}
An immediate consequence is that all free states (even of different
systems) are equivalent to each other. Indeed if $f'$ is another
free state, if we take $\rho=f'$, we have $f'\succsim f$. If instead
we take $\rho=f$, proposition~\ref{prop:free states minima} tells
us that $f\succsim f'$. In conclusion, $f\sim f'$.

Moreover, taking multiple copies of a free state does not increase
its value. This is because the tensor product of free states is still
a free state: we have $f\sim f^{\otimes n}$, for every $n\geq1$.
In summary, since there is no cost to prepare an arbitrary number
of copies of free states, having many copies is just like having a
single copy.
\begin{rem}
In the light of the resource preorder, we can understand why we required
the deterministic effect, which destroys resources, to be a free operation.
If $f$ is a free resource of system $\mathrm{A}$, there is a free
process $\mathcal{F}:\mathrm{I}\rightarrow\mathrm{A}$ that prepares
it from ``nothing''. According to our definition of the ``more
valuable'' relation, this means that the void resource, i.e.\ ``nothing'',
is more valuable than a real resource! This would be quite absurd
if we did not impose that there is also a free process destroying
$f$, thus implying that $f$ is more valuable than ``nothing''.
The conjunction of these two relations yields that a free state is
equivalent to ``nothing''. This highlights how useless free states
are.
\end{rem}

We have seen that free states of system $\mathrm{A}$ are particular
free operations, from the trivial system to $\mathrm{A}$. However,
if free operations are defined as having the same input and output
system, the only free state is the state of the trivial system, or
in other words, the theory does not have free states. In this case,
can we identify a set of states that can be introduced ``by hand''
as free states of the theory? Proposition~\ref{prop:free states minima}
states that free states are the minima of the preorder. We introduce
the following definition, which characterises states that mimic some
properties of free states, specifically the fact that they are the
minima of the preorder.
\begin{defn}
A state $\rho\in\mathsf{St}_{1}\left(\mathrm{A}\right)$ is \emph{almost
free} if for any state $\sigma\in\mathsf{St}_{1}\left(\mathrm{A}\right)$
we have $\sigma\succsim\rho$.
\end{defn}

We see that almost free states are all equivalent to each other: if
$\rho$ and $\rho'$ are almost free states, by definition $\rho\succsim\rho'$
and $\rho'\succsim\rho$, so $\rho\sim\rho'$. Clearly free states
are also almost free. If the theory has free states, almost free states
are free too. Indeed, take $\sigma$ to be a free state $f$, then
we have $f\succsim\rho$. This implies that there is a free operation
from $\mathrm{I}$ to $\mathrm{A}$ given by $\mathcal{F}f$, where
$\mathcal{F}$ is a free operation converting $f$ into $\rho$, such
that $\rho=\mathcal{F}f$. Hence $\rho$ is a free state too. Therefore,
talking about almost free states is meaningful in theories where there
are no free states.

Now, consider the almost free states that are stable under tensor
product, by which we mean that the product of two almost free states
is still an almost free state. These are the states that can be promoted
``by hand'' to free states of the theory. In conclusion, even if
a theory has no free states, if it has almost free states stable under
tensor product, we can introduce them as free states of the resource
theory \cite{Gour-review}: they are the minima of the resource preorder,
and all equivalent to one another. The only difference is that we
cannot prepare them with a free operation. We will use this fact in
subsection~\ref{subsec:The-RaRe-resource}, where we will add a free
state ``by hand'' that cannot emerge from the definition of free
operations.

\section{Resource monotones\label{sec:Resource-monotones}}

Having established a hierarchy among resources, it is sometimes useful
to have a direct way of quantifying the value of a resource by assigning
it a real number, its ``price''. This means translating the preorder
on states into the usual order on real numbers. To this end, we need
real-valued functions that respect the preorder on resources.
\begin{defn}
A real-valued function $M:\mathsf{St}_{1}\left(\mathrm{A}\right)\rightarrow\mathbb{R}$
is a \emph{resource monotone} (\emph{monotone} for short) if $\rho\succsim\sigma$
implies $M\left(\rho\right)\geq M\left(\sigma\right)$.
\end{defn}

In words, monotones assign a price to resources, consistent with their
value. Specifically, if $\rho\sim\sigma$, then $M\left(\rho\right)=M\left(\sigma\right)$.
Indeed, if $\rho\sim\sigma$, then $\rho\succsim\sigma$ and $\sigma\succsim\rho$,
thus $M\left(\rho\right)\geq M\left(\sigma\right)$ and $M\left(\sigma\right)\geq M\left(\rho\right)$,
whence one has $M\left(\rho\right)=M\left(\sigma\right)$.

A careful examination of the definition of monotones shows that they
have some tricky subtleties. Indeed, it is not possible to translate
the hierarchy of resources faithfully into the ordering of real numbers.
The main reason for such a difficulty is that we can only establish
a \emph{partial preorder} among resources, whereas we have a \emph{total
order} on real numbers. Indeed, we can have two incomparable states
$\rho$ and $\sigma$, but if $M$ is a monotone, we have either $M\left(\rho\right)\geq M\left(\sigma\right)$
or $M\left(\sigma\right)\geq M\left(\rho\right)$, because two real
numbers can always be compared.

According to the definition of resource monotones, $\rho\succsim\sigma$
implies $M\left(\rho\right)\geq M\left(\sigma\right)$, but the converse
implication in general does not hold. This means that the preorder
$\succsim$ is more fundamental than the order induced by monotones.
Indeed, if $M\left(\rho\right)\geq M\left(\sigma\right)$ we \emph{cannot}
conclude that $\rho\succsim\sigma$.

However, resource monotones are useful to detect \emph{non-convertibility}
of resources \cite{Resource-theories}. Recalling the definition,
$M\left(\rho\right)<M\left(\sigma\right)$ means $\rho\:\cancel{\succsim}\:\sigma$,
viz.\ there is no free operation converting $\rho$ into $\sigma$.
Similarly, if $M\left(\rho\right)=M\left(\sigma\right)$ we \emph{cannot}
conclude that $\rho\sim\sigma$. Indeed, a trivial monotone is a constant
function that assigns the same value to all resources, irrespective
of their place in the resource hierarchy. In this case, both equivalent
and inequivalent resources have the same value.

To obtain a full equivalence between the preorder on resources and
the ordering induced by monotones, we have to take more than one resource
monotone. In this respect, a family of monotones $\left\{ M_{i}\right\} _{i\in X}$
is said to be \emph{complete} if we have $\rho\succsim\sigma$ if
and only if $M_{i}\left(\rho\right)\geq M_{i}\left(\sigma\right)$
for every $i\in X$.
\begin{prop}[{\cite[proposition 5.2]{Resource-theories}}]
Every resource theory admits a complete family of monotones.
\end{prop}

\begin{proof}
Take $X$ to be state space of system $\mathrm{A}$: $X=\mathsf{St}_{1}\left(\mathrm{A}\right)$.
Then let us label monotones with a state $\tau$. For every resource
$\rho$, define $M_{\tau}\left(\rho\right)$ as
\[
M_{\tau}\left(\rho\right)=\begin{cases}
1 & \textrm{if }\rho\succsim\tau\\
0 & \textrm{if }\rho\:\cancel{\succsim}\:\tau
\end{cases}.
\]
Let us show that $M_{\tau}$ is a monotone for every state $\tau$.
Suppose $\rho\succsim\sigma$.
\begin{itemize}
\item If $\sigma\succsim\tau$, then by transitivity $\rho\succsim\tau$.
In this case we have $M_{\tau}\left(\rho\right)=M_{\tau}\left(\sigma\right)=1$,
whence $M_{\tau}\left(\rho\right)\geq M_{\tau}\left(\sigma\right)$.
\item If $\rho\:\cancel{\succsim}\:\tau$ and $\sigma\:\cancel{\succsim}\:\tau$,
then $M_{\tau}\left(\rho\right)=M_{\tau}\left(\sigma\right)=0$, whence
$M_{\tau}\left(\rho\right)\geq M_{\tau}\left(\sigma\right)$.
\item If $\rho\succsim\tau$ but $\sigma\:\cancel{\succsim}\:\tau$, $M_{\tau}\left(\rho\right)=1$
and $M_{\tau}\left(\sigma\right)=0$, and again $M_{\tau}\left(\rho\right)\geq M_{\tau}\left(\sigma\right)$.
\end{itemize}
This shows that $\left\{ M_{\tau}\right\} $ is a family of monotones.
Let us show that this family is also complete. To do that, we must
prove that if $M_{\tau}\left(\rho\right)\geq M_{\tau}\left(\sigma\right)$
for every state $\tau$, then $\rho\succsim\sigma$. Suppose, by contradiction
that $\rho\:\cancel{\succsim}\:\sigma$. Consider $\tau=\sigma$.
Then we have $M_{\sigma}\left(\rho\right)=0$ because $\rho\:\cancel{\succsim}\:\sigma$,
but $M_{\sigma}\left(\sigma\right)=1$ because $\sigma\succsim\sigma$,
therefore $M_{\sigma}\left(\rho\right)<M_{\sigma}\left(\sigma\right)$,
in contradiction with the hypothesis that $M_{\tau}\left(\rho\right)\geq M_{\tau}\left(\sigma\right)$
for every state $\tau$.
\end{proof}
Although we have managed to construct a complete family of resource
monotones for every resource theory, such a family is not so practical,
for it is indexed by the states themselves. 

It is useful to classify resource monotones into some categories according
to their behaviour under composition of resources \cite{MHorodecki}:
\begin{description}
\item [{Additive}] $M\left(\rho\otimes\sigma\right)=M\left(\rho\right)+M\left(\sigma\right)$,
for all states $\rho$ and $\sigma$.
\item [{Partially\ additive}] $M\left(\rho^{\otimes n}\right)=nM\left(\rho\right)$
for every state $\rho$, and every $n\geq1$.
\item [{Regularisable}] if $\lim_{n\rightarrow+\infty}\frac{1}{n}M\left(\rho^{\otimes n}\right)<+\infty$,
and the limit exists for every state $\rho$.
\end{description}
Clearly
\[
\textrm{Additive}\subseteq\textrm{Partially additive}\subseteq\textrm{Regularisable},
\]
as it is straightforward to check. Later in section~\ref{sec:Mixedness-monotones}
we will encounter some examples of additive monotones. Regularisable
monotones are fundamental when one wishes to study the thermodynamic
limit \cite{EastThesis,East-article}, because they can be extended
to the regime where there are infinitely many copies of a state (the
thermodynamic limit). Indeed in this regime, the quantities that matter
are densities, i.e.\ the value of the monotone per copy, and the
condition $\lim_{n\rightarrow+\infty}\frac{1}{n}M\left(\rho^{\otimes n}\right)<+\infty$
means that it is possible to define a density in the thermodynamic
limit, namely
\[
M^{\infty}\left(\rho\right):=\lim_{n\rightarrow+\infty}\frac{1}{n}M\left(\rho^{\otimes n}\right).
\]
$M^{\infty}$ is called the \emph{regularisation} of $M$.
\begin{prop}
Given a regularisable monotone $M$, its regularisation $M^{\infty}$
is a resource monotone too.
\end{prop}

\begin{proof}
Suppose $\rho\succsim\sigma$, then $\rho^{\otimes n}\succsim\sigma^{\otimes n}$
by proposition~\ref{prop:compatible preorder}. Therefore, since
$M$ is a monotone, $M\left(\rho^{\otimes n}\right)\geq M\left(\sigma^{\otimes n}\right)$.
Then we get $\frac{1}{n}M\left(\rho^{\otimes n}\right)\geq\frac{1}{n}M\left(\sigma^{\otimes n}\right)$,
and by taking the limit for $n\rightarrow+\infty$ of both sides we
get $M^{\infty}\left(\rho\right)\geq M^{\infty}\left(\sigma\right)$.
\end{proof}
For this reason, it is meaningful to give the following definition.
\begin{defn}
A monotone $M$ is \emph{regularised} if there exists a monotone $M'$,
such that, for every $\rho$
\[
M\left(\rho\right)=\lim_{n\rightarrow+\infty}\frac{1}{n}M'\left(\rho^{\otimes n}\right).
\]
\end{defn}

In words, a monotone is regularised if it is the regularisation of
another (regularisable) monotone. Regularised monotones $M$ are the
meaningful monotones in the thermodynamic limit, because they are
expressed as the density per particle of another monotone $M'$. Note
that partially additive monotones are not only regularisable, but
also regularised: it is enough to take $M'=M$. Indeed
\[
M\left(\rho\right)=\lim_{n\rightarrow+\infty}\frac{1}{n}M\left(\rho^{\otimes n}\right)=\lim_{n\rightarrow+\infty}\frac{1}{n}\cdot nM\left(\rho\right).
\]
Regularised monotones, which include partially additive monotones,
exhibit a remarkable property.
\begin{prop}
Let $M$ be a regularised monotone. If $f$ is a free state, $M\left(f\right)=0$.
Moreover, for every state $\rho$, $M\left(\rho\right)\geq0$.
\end{prop}

\begin{proof}
$M$ is regularised, therefore there exists a monotone $M'$ such
that $M\left(f\right)=\lim_{n\rightarrow+\infty}\frac{1}{n}M'\left(f^{\otimes n}\right)$.
Since $f$ is a free state, we have $f\sim f^{\otimes n}$, therefore
$M'\left(f^{\otimes n}\right)=M'\left(f\right)$. Hence
\[
M\left(f\right)=\lim_{n\rightarrow+\infty}\frac{1}{n}M'\left(f\right)=0,
\]
because $M'\left(f\right)$ is a constant. To conclude the proof,
recall that, for every state $\rho$, $\rho\succsim f$. Being $M$
a monotone, $M\left(\rho\right)\geq M\left(f\right)=0$.
\end{proof}
We motivated the introduction of the preorder on states with the fact
that some costly resources could be in fact close to free states.
We model this idea of closeness by introducing a distance from a state
to the set of free states.
\[
d\left(\rho\right):=\inf_{f\in F_{\mathrm{A}}}\left\Vert \rho-f\right\Vert 
\]
Here $\rho$ is a state of $\mathrm{A}$, and $f$ is a free state
of $\mathrm{A}$. Note that the infimum is achieved, so it is in fact
a minimum, because we have assumed the set of free states to be topologically
closed. This distance $d$ is a resource monotone. Indeed, suppose
$\rho\succsim\sigma$. This means there exists a free operation $\mathcal{F}$
such that $\mathcal{F}\rho=\sigma$. Then
\[
d\left(\sigma\right)=\inf_{f\in F_{\mathrm{A}}}\left\Vert \sigma-f\right\Vert =\inf_{f\in F_{\mathrm{A}}}\left\Vert \mathcal{F}\rho-f\right\Vert \leq\inf_{f\in F_{\mathrm{A}}}\left\Vert \mathcal{F}\rho-\mathcal{F}f\right\Vert ,
\]
because $\left\{ \mathcal{F}f\right\} $ is a subset of $F_{\mathrm{A}}$.
Now, the operational norm of a vector is non-increasing under a channel
by proposition~\ref{prop:norm non-increasing}. Therefore
\[
\left\Vert \mathcal{F}\rho-\mathcal{F}f\right\Vert =\left\Vert \mathcal{F}\left(\rho-f\right)\right\Vert \leq\left\Vert \rho-f\right\Vert .
\]
Taking the infimum, we get $\inf_{f\in F_{\mathrm{A}}}\left\Vert \mathcal{F}\rho-\mathcal{F}f\right\Vert \leq\inf_{f\in F_{\mathrm{A}}}\left\Vert \rho-f\right\Vert $.
In conclusion
\[
d\left(\sigma\right)\leq\inf_{f\in F_{\mathrm{A}}}\left\Vert \rho-f\right\Vert =d\left(\rho\right).
\]
This shows that $d$ is indeed a monotone. We have presented this
construction for the operational norm, but in fact it can be extended
to any function $D\left(\rho,\sigma\right)$ that is decreasing under
the action of channels \cite{Quantum-resource-2,Gour-review}.

\chapter{Sharp theories with purification\label{chap:Sharp-theories-with}}

In this chapter we introduce the theories we are going to study for
the rest of the thesis. These theories can be roughly characterised
as those admitting a level of description where all processes are
pure and reversible, and all measurements are sharp. For these reasons,
these theories are particularly appealing for the foundation of physics,
and in particular thermodynamics. The key axiom defining them\textemdash Purification\textemdash underpins
all dilation and extension theorems \cite{Chiribella-purification,Chiribella14},
so important in quantum theory, and from a thermodynamic perspective
it gives a formal guarantee that every observer can enlarge their
system in order to deal with an isolated one, where information is
maximal, and all evolutions are reversible. From a thermodynamic viewpoint,
somehow Purification can be regarded as dual to Causality: if Causality
allows one to go from larger to smaller systems, Purification enables
us to do the opposite. Moreover, mixed states, so important for thermodynamics
emerge in a different way: not only as ensembles, but also as marginals
of pure states. This is particularly appealing for the foundations
of thermodynamics, because one no longer needs to resort to fictitious
and subjective ensembles, but mixed states arise because one is tracing
out the degrees of freedom of the environment. A similar fact also
holds for the issue of irreversibility.

Informally, sharp theories with purification are causal theories satisfying
three additional axioms:
\begin{description}
\item [{Purity\ Preservation}] The composition of two pure transformations
is a pure transformation.
\item [{Pure\ Sharpness}] Every system has at least one pure sharp observable.
\item [{Purification}] Every state can be modelled as the marginal of a
pure state. Such a modelling is unique up to local reversible transformations.
\end{description}
These axioms will be presented more formally in section~\ref{sec:The-axioms-and},
where their first consequences, in particular of Purification, will
be examined. Sharp theories with purification enjoy some remarkable
properties that, in some sense, make them close to quantum theory.
In this chapter we will focus on those properties relevant to the
thermodynamic analysis of chapter~\ref{chap:Operational-thermodynamics}.
One of the key features is a state-effect duality, by which with every
normalised pure state we can associate a unique normalised pure effect,
the \emph{dagger} of the state, and vice versa \cite{QPL15,TowardsThermo}.
The second key result, which constitutes the high spot of this chapter,
is the diagonalisation theorem, which states that in these theories
every state can be \emph{diagonalised} \cite{QPL15,TowardsThermo},
viz.\ written as a convex combination of perfectly distinguishable
pure states, with unique coefficients \cite{TowardsThermo}. The fact
that with every state we can associate the probability distribution
of its eigenvalues allows us to introduce entropic functions, which
will be done in section~\ref{sec:Mixedness-monotones}. Furthermore,
the two key properties of sharp theories with purification allow us
to identify every element of the vector space of effects with a physical
observable. This will prove of fundamental importance in section~\ref{sec:Generalised-Gibbs-states}
to define thermal states, and to derive a lot of properties of a generalisation
of Shannon-von Neumann entropy (section~\ref{sec:Properties-of-Shannon-von}).

From these properties, others follow, but we do not report them in
detail in this chapter because they are not so crucial for thermodynamics.
We summarise them briefly here. A consequence of the state-effect
duality is that sharp theories with purification are (strongly) self-dual
\cite{HOI}, which enables us to extend the dagger to \emph{all} transformations
\cite{HOI}, not just to states and effects. The conjunction of strong
self-duality with the state-effect duality implies that these theories
satisfy the no-restriction hypothesis \cite{Chiribella-purification,Janotta-Lal}.
Moreover, we can prove that there exists a pure projector on every
face of the state space \cite{HOI}. The conjunction of these two
facts implies that sharp theories with purifications are Euclidean
Jordan algebras \cite{Alfsen-Shultz78,Alfsen-Shultz,Graydon-QPL,BarnumGraydonWilceCCEJA,HOI},
and that their interference is constrained at most to the second order
\cite{HOI}. However, note that not all Euclidean Jordan algebras
are sharp theories with purification.

We conclude this chapter by presenting two new examples of sharp theories
with purification. They are both theories constructed by imposing
some superselection rules to quantum theory, and by defining system
composition in such a way that the three defining axioms are satisfied.
The first example \cite{Purity} is a theory where every non-trivial
system is given by a pair of isomorphic quantum systems. This theory
in section~\ref{sec:Sufficiency-of-majorisation} will provide a
counterexample in which majorisation is \emph{not} sufficient to characterise
certain thermodynamic transitions. The second example \cite{TowardsThermo}
is an extension of classical theory in which some systems look classical
at the single-system level, but become entangled when composed. This
shows that classical theory can be embedded and treated as a sub-theory
of a sharp theory with purification, and that, ultimately, the results
we obtain for sharp theories with purification can be extended to
classical theory too, to provide an information-theoretic foundation
of classical thermodynamics.

\section{The axioms and their first consequences\label{sec:The-axioms-and}}

In this section we present the axioms defining sharp theories with
purification. These axioms are added on top of Causality, and will
single out a class of theories where everything is pure and reversible
at the fundamental level.

The first axiom to be added on top of Causality, Purity Preservation,
states that no information can leak to the environment when two pure
transformations are composed.
\begin{ax}[Purity Preservation \cite{Scandolo14}]
Sequential and parallel compositions of pure transformations yield
pure transformations.
\end{ax}

We consider Purity Preservation as a fundamental requirement to do
physics. Considering the theory as an algorithm to make deductions
about physical processes, Purity Preservation ensures that, when presented
with maximal information about two processes, the algorithm outputs
maximal information about their composition \cite{Scandolo14}. Purity
Preservation is very close to a slightly weaker axiom, Atomicity of
Composition, introduced by D'Ariano in \cite{D'Ariano}, and used
in the axiomatisation of \cite{Chiribella-informational}. However
Purity Preservation is stronger, in that it requires the preservation
of purity also for \emph{parallel} composition, and not just for sequential
composition like in D'Ariano's original axiom. An immediate consequence
of Purity Preservation is that the product of two pure states is pure,
a fact usually proved using Local Tomography \cite{Chiribella-purification}.
Notably, quaternionic quantum theory fails this principle, for the
product of two pure states is not pure in general \cite{Graydon-QPL,BarnumGraydonWilceCCEJA}.

The second axiom, Pure Sharpness, guarantees that every system possesses
at least one elementary property, in the sense of Piron \cite{PironBook}.
Recall that here we are not assuming the no-restriction hypothesis,
so the following axiom needs to be imposed.
\begin{ax}[Pure Sharpness \cite{QPL15}]
For every system there exists at least one pure effect occurring
with unit probability on some state.
\end{ax}

Pure Sharpness is reminiscent of the Sharpness axiom used in Hardy's
2011 axiomatisation \cite{Hardy-informational-2,hardy2013}, which
requires a one-to-one correspondence between pure states and effects
that distinguish maximal sets of states. A similar axiom also appeared
in works by Wilce \cite{Wilce-spectral,Royal-road-0,Royal-road},
where he stipulates that for every pure effect there exists a \emph{unique}
state on which it occurs with probability 1.

The two axioms above are satisfied by both classical and quantum theory.
Our last axiom, Purification, is precisely the one that characterises
all physical theories admitting a fundamental level of description
where all deterministic processes are pure and reversible. Essentially,
Purification expresses a strengthened version of the principle of
conservation of information \cite{Chiribella-educational,Scandolo14},
demanding not only that information be conserved, but also that randomness
can always be modelled as due to the presence of some inaccessible
degree of freedom. In its simplest form, Purification is phrased as
a requirement about \emph{causal} theories, where the marginal of
a bipartite state is defined in a canonical way. In this case, we
say that a state $\rho\in\mathsf{St}_{1}\left(\mathrm{A}\right)$
can be purified if there exists a pure state $\Psi\in\mathsf{PurSt}_{1}\left(\mathrm{AB}\right)$
that has $\rho$ as its marginal on system $\mathrm{A}$. In this
case, we call $\Psi$ a \emph{purification} of $\rho$ and $\mathrm{B}$
a \emph{purifying system}. The axiom is as follows.
\begin{ax}[Purification \cite{Chiribella-purification}]
Every state can be purified. Every two purifications of the same
state, with the same purifying system, differ by a reversible channel
on the purifying system.
\end{ax}

The second part of the axiom states that, if $\Psi,\Psi'\in\mathsf{PurSt}_{1}\left(\mathrm{AB}\right)$
are such that $\mathrm{tr}_{\mathrm{B}}\Psi_{\mathrm{AB}}=\mathrm{tr}_{\mathrm{B}}\Psi'_{\mathrm{AB}}$,
then\[ \begin{aligned}\Qcircuit @C=1em @R=.7em @!R { & \multiprepareC{1}{\Psi'} & \qw \poloFantasmaCn{\rA} & \qw \\ & \pureghost{\Psi'} & \qw \poloFantasmaCn{\rB} & \qw }\end{aligned}~=\!\!\!\!\begin{aligned}\Qcircuit @C=1em @R=.7em @!R { & \multiprepareC{1}{\Psi} & \qw \poloFantasmaCn{\rA} & \qw & \qw & \qw \\ & \pureghost{\Psi} & \qw \poloFantasmaCn{\rB} & \gate{\cU}& \qw \poloFantasmaCn{\rB}& \qw }\end{aligned}~, \]where
$\mathcal{U}$ is a reversible channel on $\mathrm{B}$.

Quantum theory, both on complex and real Hilbert spaces, satisfies
Purification. Recently also Spekkens' toy model\footnote{The original, non-convex, version of it \cite{Toy-model}.}
\cite{Toy-model} has been shown to satisfy Purification \cite{Disilvestro}.
Other non-trivial examples are fermionic quantum theory \cite{Fermionic1,Fermionic2},
and doubled quantum theory \cite{Purity}, presented in section~\ref{sec:Example:-doubled-quantum}.
Remarkably, even classical theory can be regarded as a sub-theory
of a larger physical theory where Purification is satisfied (see section~\ref{sec:Example:-extended-classical})
\cite{TowardsThermo}.
\begin{defn}[Sharp theories with purification]
A causal theory satisfying Purity Preservation, Pure Sharpness, and
Purification will be called a \emph{sharp theory with purification.}
\end{defn}

In the rest of the section we will outline the first consequences
of these axioms, especially of Purification.

\subsection{First consequences}

The easiest consequence of Purification is that reversible channels
act transitively on the set of pure states \cite{Chiribella-purification}.
\begin{prop}[Transitivity]
For any pair of pure states $\psi$, $\psi'$, there exists a reversible
channel $\mathcal{U}$ that $\psi'=\mathcal{U}\psi$.
\end{prop}

\begin{proof}
Every system $\mathrm{A}$ is a purifying system for the trivial system
$\mathrm{I}$. Then $\psi$ and $\psi'$ are two purifications of
the same deterministic state of the trivial system (which is the number
1), therefore they differ by a reversible channel on the purifying
system $\mathrm{A}$, which means $\psi'=\mathcal{U}\psi$.
\end{proof}
As a consequence, every finite-dimensional system possesses a unique
invariant state, which is an internal state (see proposition~\ref{prop:uniqueness invariant}).
Also, transitivity implies that the set of pure states is compact
for every system (proposition~\ref{prop:pure compactness}). This
is generally a non-trivial property\textemdash see \cite{Entropy-Barnum}
for a counterexample of a state space with a non-closed set of pure
states.

A crucial consequence of Purification is the \emph{steering property}.
\begin{thm}[Pure Steering]
Let $\rho\in\mathsf{St}_{1}\left(\mathrm{A}\right)$ and let $\Psi\in\mathsf{PurSt}_{1}\left(\mathrm{AB}\right)$
be a purification of $\rho$. Then $\sigma$ is contained in $\rho$
if and only if there exist an effect $b_{\sigma}$ on the purifying
system $\mathrm{B}$ and a non-zero probability $p$ such that\[ p\!\!\!\!\begin{aligned}\Qcircuit @C=1em @R=.7em @!R { & \prepareC{\sigma} & \qw \poloFantasmaCn{\rA} & \qw }\end{aligned}~=\!\!\!\!\begin{aligned}\Qcircuit @C=1em @R=.7em @!R { & \multiprepareC{1}{\Psi} & \qw \poloFantasmaCn{\rA} & \qw \\ & \pureghost{\Psi} & \qw \poloFantasmaCn{\rB} & \measureD{b_{\sigma}} }\end{aligned}~. \]
\end{thm}

\begin{proof}
Sufficiency is easy. Consider the observation-test $\left\{ b_{\sigma},u-b_{\sigma}\right\} $
such that\[ p\!\!\!\!\begin{aligned}\Qcircuit @C=1em @R=.7em @!R { & \prepareC{\sigma} & \qw \poloFantasmaCn{\rA} & \qw }\end{aligned}~=\!\!\!\!\begin{aligned}\Qcircuit @C=1em @R=.7em @!R { & \multiprepareC{1}{\Psi} & \qw \poloFantasmaCn{\rA} & \qw \\ & \pureghost{\Psi} & \qw \poloFantasmaCn{\rB} & \measureD{b_{\sigma}} }\end{aligned}~. \]Then\[\begin{aligned}\Qcircuit @C=1em @R=.7em @!R { & \prepareC{\rho} & \qw \poloFantasmaCn{\rA} & \qw }\end{aligned}~=\!\!\!\!\begin{aligned}\Qcircuit @C=1em @R=.7em @!R { & \multiprepareC{1}{\Psi} & \qw \poloFantasmaCn{\rA} & \qw \\ & \pureghost{\Psi} & \qw \poloFantasmaCn{\rB} & \measureD{u} }\end{aligned}~=\!\!\!\! \begin{aligned}\Qcircuit @C=1em @R=.7em @!R { & \multiprepareC{1}{\Psi} & \qw \poloFantasmaCn{\rA} & \qw \\ & \pureghost{\Psi} & \qw \poloFantasmaCn{\rB} & \measureD{b_{\sigma}}}\end{aligned}~+\!\!\!\!\begin{aligned}\Qcircuit @C=1em @R=.7em @!R { & \multiprepareC{1}{\Psi} & \qw \poloFantasmaCn{\rA} & \qw \\ & \pureghost{\Psi} & \qw \poloFantasmaCn{\rB} & \measureD{u-b_{\sigma}}}\end{aligned} ~= \]
\[
=p\!\!\!\!\begin{aligned}\Qcircuit @C=1em @R=.7em @!R { & \prepareC{\sigma} & \qw \poloFantasmaCn{\rA} & \qw }\end{aligned}+\left(1-p\right)\!\!\!\!\begin{aligned}\Qcircuit @C=1em @R=.7em @!R { & \prepareC{\tau} & \qw \poloFantasmaCn{\rA} & \qw }\end{aligned}~,
\]where $\tau$ is the state induced by applying the effect $u-b_{\sigma}$.
This proves that $\sigma$ is contained in $\rho$.

Conversely, if $\sigma$ is contained in $\rho$, it means that $\rho=p\sigma+\left(1-p\right)\tau$,
with $p\in\left(0,1\right)$. From an operational point of view, it
means that there is a preparation-test $\left\{ p\sigma,\left(1-p\right)\tau\right\} $,
of which $\rho$ is the coarse-graining. By Physicalisation of Readout,
there exists a normalised bipartite state $\Sigma\in\mathsf{St}_{1}\left(\mathrm{AX}\right)$,
and an observation-test $\left\{ c_{\sigma},c_{\tau}\right\} $ on
$\mathrm{X}$ such that\begin{equation}\label{eq:steering1} p\!\!\!\!\begin{aligned}\Qcircuit @C=1em @R=.7em @!R { & \prepareC{\sigma} & \qw \poloFantasmaCn{\rA} & \qw }\end{aligned}~=\!\!\!\!\begin{aligned}\Qcircuit @C=1em @R=.7em @!R { & \multiprepareC{1}{\Sigma} & \qw \poloFantasmaCn{\rA} & \qw \\ & \pureghost{\Sigma} & \qw \poloFantasmaCn{\rX} & \measureD{c_{\sigma}} }\end{aligned} \end{equation}and\begin{equation}\label{eq:steering2} \left(1-p\right)\!\!\!\!\begin{aligned}\Qcircuit @C=1em @R=.7em @!R { & \prepareC{\tau} & \qw \poloFantasmaCn{\rA} & \qw }\end{aligned}~=\!\!\!\!\begin{aligned}\Qcircuit @C=1em @R=.7em @!R { & \multiprepareC{1}{\Sigma} & \qw \poloFantasmaCn{\rA} & \qw \\ & \pureghost{\Sigma} & \qw \poloFantasmaCn{\rX} & \measureD{c_{\tau}} }\end{aligned}~. \end{equation}In
general, $\Sigma$ will not be pure, so let us take a purification
$\Phi\in\mathsf{PurSt}_{1}\left(\mathrm{AXC}\right)$ of $\Sigma$.
Note that $\Phi$ is a purification of $\rho$ too, indeed\[ \begin{aligned}\Qcircuit @C=1em @R=.7em @!R { & \multiprepareC{2}{\Phi} & \qw \poloFantasmaCn{\rA} & \qw \\ & \pureghost{\Phi} & \qw \poloFantasmaCn{\rX} & \measureD{u}\\ & \pureghost{\Phi} & \qw \poloFantasmaCn{\rC} & \measureD{u} }\end{aligned}~= \!\!\!\! \begin{aligned}\Qcircuit @C=1em @R=.7em @!R { & \multiprepareC{1}{\Sigma} & \qw \poloFantasmaCn{\rA} & \qw \\ & \pureghost{\Sigma} & \qw \poloFantasmaCn{\rX} & \measureD{c_{\sigma}}}\end{aligned}~+\!\!\!\!\begin{aligned}\Qcircuit @C=1em @R=.7em @!R { & \multiprepareC{1}{\Sigma} & \qw \poloFantasmaCn{\rA} & \qw \\ & \pureghost{\Sigma} & \qw \poloFantasmaCn{\rX} & \measureD{c_{\tau}}}\end{aligned} ~=\!\!\!\!\begin{aligned}\Qcircuit @C=1em @R=.7em @!R { & \prepareC{\rho} & \qw \poloFantasmaCn{\rA} & \qw }\end{aligned}~,\]having
used eqs.~\eqref{eq:steering1} and \eqref{eq:steering2}. Now, let
us show that we can induce $p\sigma$ by applying a suitable effect
on the purifying system of any purification $\Psi\in\mathsf{PurSt}_{1}\left(\mathrm{AB}\right)$
of $\rho$. Now, take the pure states $\beta\in\mathsf{PurSt}_{1}\left(\mathrm{B}\right)$,
$\xi\in\mathsf{PurSt}_{1}\left(\mathrm{X}\right)$, and $\gamma\in\mathsf{PurSt}_{1}\left(\mathrm{C}\right)$.
Then $\Phi\otimes\beta$ and $\Psi\otimes\xi\otimes\gamma$ are pure
states (by Purity Preservation), and they are purifications of $\rho$
with the same purifying system $\mathrm{XCB}$ (up to system swapping),
so\[ \begin{aligned}\Qcircuit @C=1em @R=.7em @!R { & \multiprepareC{2}{\Phi} & \qw \poloFantasmaCn{\rA} & \qw \\ & \pureghost{\Phi} & \qw \poloFantasmaCn{\rX} & \qw\\ & \pureghost{\Phi} & \qw \poloFantasmaCn{\rC} & \qw \\ & \prepareC{\beta} & \qw \poloFantasmaCn{\rB} & \qw }\end{aligned}~=\!\!\!\!\begin{aligned}\Qcircuit @C=1em @R=.7em @!R { & \multiprepareC{1}{\Psi} & \qw \poloFantasmaCn{\rA} & \qw &\qw &\qw \\ & \pureghost{\Psi} & \qw \poloFantasmaCn{\rB} & \multigate{2}{\cU} & \qw \poloFantasmaCn{\rX} &\qw \\  & \prepareC{\xi} & \qw \poloFantasmaCn{\rX} & \ghost{\cU} & \qw \poloFantasmaCn{\rC} &\qw \\ & \prepareC{\gamma} & \qw \poloFantasmaCn{\rC} & \ghost{\cU} & \qw \poloFantasmaCn{\rB} &\qw}\end{aligned}~. \]Then\[
p\!\!\!\!\begin{aligned}\Qcircuit @C=1em @R=.7em @!R { & \prepareC{\sigma} & \qw \poloFantasmaCn{\rA} & \qw }\end{aligned}~=\!\!\!\!\begin{aligned}\Qcircuit @C=1em @R=.7em @!R { & \multiprepareC{2}{\Phi} & \qw \poloFantasmaCn{\rA} & \qw \\ & \pureghost{\Phi} & \qw \poloFantasmaCn{\rX} & \measureD{c_{\sigma}}\\ & \pureghost{\Phi} & \qw \poloFantasmaCn{\rC} & \measureD{u} \\ & \prepareC{\beta} & \qw \poloFantasmaCn{\rB} & \measureD{u} }\end{aligned}~=\!\!\!\!\begin{aligned}\Qcircuit @C=1em @R=.7em @!R { & \multiprepareC{1}{\Psi} & \qw \poloFantasmaCn{\rA} & \qw &\qw &\qw \\ & \pureghost{\Psi} & \qw \poloFantasmaCn{\rB} & \multigate{2}{\cU} & \qw \poloFantasmaCn{\rX} &\measureD{c_{\sigma}} \\  & \prepareC{\xi} & \qw \poloFantasmaCn{\rX} & \ghost{\cU} & \qw \poloFantasmaCn{\rC} &\measureD{u} \\ & \prepareC{\gamma} & \qw \poloFantasmaCn{\rC} & \ghost{\cU} & \qw \poloFantasmaCn{\rB} &\measureD{u}}\end{aligned}~=:
\]
\[
=:\!\!\!\!\begin{aligned}\Qcircuit @C=1em @R=.7em @!R { & \multiprepareC{1}{\Psi} & \qw \poloFantasmaCn{\rA} & \qw \\ & \pureghost{\Psi} & \qw \poloFantasmaCn{\rB} & \measureD{b_{\sigma}}}\end{aligned},
\]where\[
\begin{aligned}\Qcircuit @C=1em @R=.7em @!R {&\qw \poloFantasmaCn{\rB} &\measureD{b_{\sigma}}}\end{aligned}~:=~\begin{aligned}\Qcircuit @C=1em @R=.7em @!R { & \qw & \qw \poloFantasmaCn{\rB} & \multigate{2}{\cU} & \qw \poloFantasmaCn{\rX} &\measureD{c_{\sigma}} \\  & \prepareC{\xi} & \qw \poloFantasmaCn{\rX} & \ghost{\cU} & \qw \poloFantasmaCn{\rC} &\measureD{u} \\ & \prepareC{\gamma} & \qw \poloFantasmaCn{\rC} & \ghost{\cU} & \qw \poloFantasmaCn{\rB} &\measureD{u}}\end{aligned}~.
\]This proves the theorem.
\end{proof}
Purification also enables us to link equality upon input (as defined
in definition~\ref{def:equality upon input}) to equality on purifications
\cite[theorem 7]{Chiribella-purification}, as an easy consequence
of Pure Steering.
\begin{prop}
\label{prop:purifications -> input} Let $\rho$ be a state of system
$\mathrm{A}$ and let $\Psi\in\mathsf{St}_{1}\left(\mathrm{AB}\right)$
be a purification of $\rho$. Then, for every pair of transformations
$\mathcal{A}$ and $\mathcal{A}'$, transforming $\mathrm{A}$ into
$\mathrm{C}$, if\[ \begin{aligned}\Qcircuit @C=1em @R=.7em @!R { & \multiprepareC{1}{\Psi} & \qw \poloFantasmaCn{\rA} & \gate{\mathcal A} & \qw \poloFantasmaCn{\rC}&\qw \\ & \pureghost{\Psi} & \qw \poloFantasmaCn{\rB} & \qw &\qw&\qw}\end{aligned} ~= \!\!\!\! \begin{aligned}\Qcircuit @C=1em @R=.7em @!R { & \multiprepareC{1}{\Psi} & \qw \poloFantasmaCn{\rA} & \gate{\mathcal A'} & \qw \poloFantasmaCn{\rC} &\qw \\ & \pureghost{\Psi} & \qw \poloFantasmaCn{\rB} & \qw &\qw &\qw}\end{aligned}~, \]then
$\mathcal{A}=_{\rho}\mathcal{A}'$.

If system $\mathrm{C}$ is trivial, then one has the full equivalence:
for every pair of effects $a$ and $a'$\[ \begin{aligned}\Qcircuit @C=1em @R=.7em @!R { & \multiprepareC{1}{\Psi} & \qw \poloFantasmaCn{\rA} & \measureD{a} \\ & \pureghost{\Psi} & \qw \poloFantasmaCn{\rB} & \qw }\end{aligned} ~= \!\!\!\! \begin{aligned}\Qcircuit @C=1em @R=.7em @!R { & \multiprepareC{1}{\Psi} & \qw \poloFantasmaCn{\rA} & \measureD{a'} \\ & \pureghost{\Psi} & \qw \poloFantasmaCn{\rB} & \qw}\end{aligned} \]if
and only if $a=_{\rho}a'$.
\end{prop}

Pure Steering guarantees the existence of pure dynamically faithful
states, in the following sense.
\begin{defn}
A state $\rho\in\mathsf{St}_{1}\left(\mathrm{AB}\right)$ is \emph{dynamically
faithful} on system $\mathrm{A}$ if for every system $\mathrm{C}$
and for every pair of transformations $\mathcal{A}$ and $\mathcal{A}'$
transforming $\mathrm{A}$ into $\mathrm{C}$ \[ \begin{aligned}\Qcircuit @C=1em @R=.7em @!R { & \multiprepareC{1}{\rho} & \qw \poloFantasmaCn{\rA} & \gate{\mathcal A} & \qw \poloFantasmaCn{\rC}&\qw \\ & \pureghost{\rho} & \qw \poloFantasmaCn{\rB} & \qw &\qw&\qw}\end{aligned} ~= \!\!\!\! \begin{aligned}\Qcircuit @C=1em @R=.7em @!R { & \multiprepareC{1}{\rho} & \qw \poloFantasmaCn{\rA} & \gate{\mathcal A'} & \qw \poloFantasmaCn{\rC} &\qw \\ & \pureghost{\rho} & \qw \poloFantasmaCn{\rB} & \qw &\qw &\qw}\end{aligned} \]implies
$\mathcal{A}=\mathcal{A}'$.
\end{defn}

Thanks to Pure Steering, we have the following characterisation.
\begin{prop}
\label{prop:faithful}A pure state $\Psi_{\mathrm{AB}}$ is dynamically
faithful on system $\mathrm{A}$ if and only if its marginal $\omega_{\mathrm{A}}$
on $\mathrm{A}$ is internal.
\end{prop}

\begin{proof}
The proof is an adaptation of the arguments of \cite[theorems 8, 9]{Chiribella-purification},
which is valid even without invoking the Local Tomography axiom used
therein \cite{Chiribella}.
\end{proof}
An indirect consequence of Pure Steering is a simple condition for
a set of transformations to be a test (cf.\ \cite[theorem 18]{Chiribella-purification}).
\begin{prop}
\label{prop:sufficientfortest}A set of transformations $\left\{ \mathcal{A}_{i}\right\} _{i=1}^{n}\subset\mathsf{Transf}\left(\mathrm{A},\mathrm{B}\right)$
is a test if and only if $\sum_{i=1}^{n}u_{\mathrm{B}}\mathcal{A}_{i}=u_{\mathrm{A}}$.
Specifically, a set of effects $\left\{ a_{i}\right\} _{i=1}^{n}$
is an observation-test if and only if $\sum_{i=1}^{n}a_{i}=u_{\mathrm{A}}$.
\end{prop}

A proof of this proposition that does not make use of Local Tomography
can be found in \cite{Scandolo-thesis}.

We will also make use of the following important consequences of Purification.
The first states that every channel admits a dilation to a reversible
channel \cite[subsection 9 A]{Chiribella-purification}.
\begin{prop}
\label{prop:dilation channels}For every channel $\mathcal{C}\in\mathsf{Transf}\left(\mathrm{A},\mathrm{B}\right)$
there is a reversible extension $\mathcal{V}\in\mathsf{Transf}\left(\mathrm{A},\mathrm{BE}\right)$
for some system $\mathrm{E}$, namely a channel $\mathcal{V}$ of
the form\begin{equation} \label{eq:reversible extension}\begin{aligned}\Qcircuit @C=1em @R=.7em @!R { & \qw \poloFantasmaCn{\rA} & \multigate{1}{\cV} & \qw \poloFantasmaCn{\rB} &\qw \\  &  &\pureghost{\cV} & \qw \poloFantasmaCn{\rE} & \qw }\end{aligned}~=\!\!\!\! \begin{aligned}\Qcircuit @C=1em @R=.7em @!R { & & \qw \poloFantasmaCn{\rA} & \multigate{1}{\cU} & \qw \poloFantasmaCn{\rB} &\qw \\ & \prepareC{\eta} & \qw \poloFantasmaCn{\rE'} &\ghost{\cU} & \qw \poloFantasmaCn{\rE} & \qw }\end{aligned}~,\end{equation}for
some system $\mathrm{E}'$, some pure state $\eta\in\mathsf{St}_{1}\left(\mathrm{E}'\right)$,
where $\mathcal{U}$ is a reversible channel, such that\[ \begin{aligned}\Qcircuit @C=1em @R=.7em @!R { & \qw \poloFantasmaCn{\rA} & \gate{\cC} & \qw \poloFantasmaCn{\rB} &\qw }\end{aligned} ~=~ \begin{aligned}\Qcircuit @C=1em @R=.7em @!R { & \qw \poloFantasmaCn{\rA} & \multigate{1}{\cV} & \qw \poloFantasmaCn{\rB} &\qw \\  &  &\pureghost{\cV} & \qw \poloFantasmaCn{\rE} & \measureD{u} }\end{aligned}~.\]

Moreover, if $\mathcal{V}$ and $\mathcal{V}'$ are two reversible
extensions of the same channel, one has\[ \begin{aligned}\Qcircuit @C=1em @R=.7em @!R { & \qw \poloFantasmaCn{\rA} & \multigate{1}{\cV'} & \qw \poloFantasmaCn{\rB} &\qw \\  &  &\pureghost{\cV'} & \qw \poloFantasmaCn{\rE} & \qw }\end{aligned}~=~\begin{aligned}\Qcircuit @C=1em @R=.7em @!R { & \qw \poloFantasmaCn{\rA} & \multigate{1}{\cV} & \qw \poloFantasmaCn{\rB} &\qw &\qw &\qw \\  &  &\pureghost{\cV} & \qw \poloFantasmaCn{\rE} & \gate{\cU'}& \qw \poloFantasmaCn{\rE}&\qw}\end{aligned}~,
\]where $\mathcal{U}'$ is a reversible channel on $\mathrm{E}$.
\end{prop}

This is the form Purification takes for channels; indeed it is possible
to prove that reversible channels are pure channels in theories with
purification \cite[corollary 28]{Chiribella-purification}, so $\mathcal{V}$
is also a \emph{pure} extension of the channel $\mathcal{C}$.

An immediate consequence is the following corollary.
\begin{cor}
\label{cor:pure map factorise}Suppose $\mathcal{V}$ is a channel
of the form~\eqref{eq:reversible extension}, and such that\[ \begin{aligned}\Qcircuit @C=1em @R=.7em @!R { & \qw \poloFantasmaCn{\rA} & \multigate{1}{\cV} & \qw \poloFantasmaCn{\rB} &\qw \\  &  &\pureghost{\cV} & \qw \poloFantasmaCn{\rE} & \measureD{u} }\end{aligned}~=~\begin{aligned}\Qcircuit @C=1em @R=.7em @!R { &  \qw \poloFantasmaCn{\rA} & \gate{\cU} & \qw \poloFantasmaCn{\rB} &\qw }\end{aligned}~,\]where
$\mathcal{U}$ is a reversible channel. Then\[
\begin{aligned}\Qcircuit @C=1em @R=.7em @!R { & \qw \poloFantasmaCn{\rA} & \multigate{1}{\cV} & \qw \poloFantasmaCn{\rB} &\qw \\  &  &\pureghost{\cV} & \qw \poloFantasmaCn{\rE} & \qw}\end{aligned}~=~\begin{aligned}\Qcircuit @C=1em @R=.7em @!R {  & \qw \poloFantasmaCn{\rA} & \gate{\cU} & \qw \poloFantasmaCn{\rB} &\qw \\ &&\prepareC{\eta} &\qw \poloFantasmaCn{\rE} &\qw}\end{aligned}~,
\]where $\eta$ is some pure state.
\end{cor}

\begin{proof}
Consider the channel\[
\begin{aligned}\Qcircuit @C=1em @R=.7em @!R { & \qw \poloFantasmaCn{\rA} & \multigate{1}{\cV'} & \qw \poloFantasmaCn{\rB} &\qw \\  &  &\pureghost{\cV'} & \qw \poloFantasmaCn{\rE} & \qw}\end{aligned}~=~\begin{aligned}\Qcircuit @C=1em @R=.7em @!R {  & \qw \poloFantasmaCn{\rA} & \gate{\cU} & \qw \poloFantasmaCn{\rB} &\qw \\  & &\prepareC{\eta'} &\qw \poloFantasmaCn{\rE}&\qw}\end{aligned}~,
\]for $\eta'$ pure; which is of the form~\eqref{eq:reversible extension},
where the joint reversible channel is $\mathcal{U}\otimes\mathcal{I}_{\mathrm{E}}$.
Clearly, we have\[
\begin{aligned}\Qcircuit @C=1em @R=.7em @!R { & \qw \poloFantasmaCn{\rA} & \multigate{1}{\cV'} & \qw \poloFantasmaCn{\rB} &\qw \\  &  &\pureghost{\cV'} & \qw \poloFantasmaCn{\rE} & \measureD{u} }\end{aligned}~=~\begin{aligned}\Qcircuit @C=1em @R=.7em @!R { &  \qw \poloFantasmaCn{\rA} & \gate{\cU} & \qw \poloFantasmaCn{\rB} &\qw }\end{aligned}~,
\]therefore, by proposition~\ref{prop:dilation channels}, there is
a reversible channel $\mathcal{U}'$ on $\mathrm{E}$ such that\[
\begin{aligned}\Qcircuit @C=1em @R=.7em @!R { & \qw \poloFantasmaCn{\rA} & \multigate{1}{\cV} & \qw \poloFantasmaCn{\rB} &\qw \\  &  &\pureghost{\cV} & \qw \poloFantasmaCn{\rE} & \qw }\end{aligned}~=\!\!\!\!\begin{aligned}\Qcircuit @C=1em @R=.7em @!R {  & & \qw \poloFantasmaCn{\rA} & \gate{\cU} & \qw \poloFantasmaCn{\rB} &\qw \\  &\prepareC{\eta'} &\qw \poloFantasmaCn{\rE}& \gate{\cU'} & \qw \poloFantasmaCn{\rE} &\qw}\end{aligned}~=~\begin{aligned}\Qcircuit @C=1em @R=.7em @!R {  & \qw \poloFantasmaCn{\rA} & \gate{\cU} & \qw \poloFantasmaCn{\rB} &\qw \\ &&\prepareC{\eta} &\qw \poloFantasmaCn{\rE} &\qw}\end{aligned},
\]where $\eta=\mathcal{U}'\eta'$ is a pure state because $\eta'$ is
pure.
\end{proof}
Now we can move to study less immediate consequences of the axioms,
which appeared in \cite{TowardsThermo}.

\section{State-effect duality\label{sec:State-effect-duality}}

In this section we derive the first important property of sharp theories
with purification: a duality between normalised pure states and normalised
pure effects. To do that, first we need some technical lemmas.

\subsection{Technical lemmas}

Pure Sharpness stipulates that for every system there is a pure effect
occurring with probability 1 on some state. We can easily show that
such a state must be pure \cite[lemma 26, theorem 7]{Chiribella-informational}.
\begin{prop}
\label{prop:uniqueness of state}Let $a$ be a normalised pure effect.
Then there exists a pure state $\alpha$ such that $\left(a\middle|\alpha\right)=1$.
If $\rho$ is another state such that $\left(a\middle|\rho\right)=1$,
then $\rho=\alpha$.
\end{prop}

Combining the above result with our Pure Sharpness axiom, we derive
the following proposition \cite[proposition 9]{QPL15}.
\begin{prop}
\label{prop:identifyingeffect}For every pure state $\alpha\in\mathsf{PurSt}_{1}\left(\mathrm{A}\right)$
there exists at least one pure effect $a\in\mathsf{PurEff}\left(\mathrm{A}\right)$
such that $\left(a\middle|\alpha\right)=1$.
\end{prop}

\begin{proof}
By Pure Sharpness, there exists at least one pure effect $a_{0}$
such that $\left(a_{0}\middle|\alpha_{0}\right)=1$, for some state
$\alpha_{0}$, where $\alpha_{0}$ is pure. Now, for a generic pure
state $\alpha$, by transitivity, there is a reversible channel $\mathcal{U}$
such that $\alpha=\mathcal{U}\alpha_{0}$. Hence, the effect $a:=a_{0}\mathcal{U}^{-1}$
is pure and $\left(a\middle|\alpha\right)=1$.
\end{proof}
In summary, for every normalised pure effect $a\in\mathsf{PurEff}_{1}\left(\mathrm{A}\right)$,
we can associate a \emph{unique} pure state $\alpha\in\mathsf{PurSt}_{1}\left(\mathrm{A}\right)$
with it such that $\left(a\middle|\alpha\right)=1$. A probabilistic
model like this has been dubbed ``sharp'' by Wilce \cite{Wilce2010,Wilce-spectral,Royal-road-0,Royal-road}.
Conversely, given a pure state $\alpha$, there always exists at least
one pure effect $a$ such that $\left(a\middle|\alpha\right)=1$.
This shows that there is a surjective correspondence between normalised
pure effects and normalised pure states. We will show in a while that
this correspondence is in fact bijective.

These results, presented here as propositions, were instead taken
by Selby and Coecke as fundamental requirements for the definition
of \emph{test structure}, providing an operational characterisation
of the Hermitian adjoint \cite{Selby-dagger}, used in a new reconstruction
of quantum theory \cite{Diagrammatic}.

\subsubsection{Probability balance of pure bipartite states}

Given a normalised state $\rho\in\mathsf{St}_{1}\left(\mathrm{A}\right)$,
we define the probability $p_{*}$ as the maximum probability that
a pure state can have in a convex decomposition of $\rho$, namely\footnote{Note that the maximum is well-defined because the set of pure states
is compact, thanks to transitivity.}
\[
p_{*}:=\max_{\alpha\in\mathsf{PurSt}_{1}\left(\mathrm{A}\right)}\left\{ p:\rho=p\alpha+\left(1-p\right)\sigma,\sigma\in\mathsf{St}_{1}\left(\mathrm{A}\right)\right\} .
\]
We call $p_{*}$ the \emph{maximum eigenvalue} of $\rho$, and say
that the \emph{pure} state $\alpha$ is the corresponding \emph{eigenstate}.
The reason for this terminology will become clear once we prove our
diagonalisation theorem.

A fundamental consequence of our axioms is that both marginals of
a bipartite state have the same maximum eigenvalue.
\begin{thm}
\label{thm:probability balance}Let $\Psi$ be a pure bipartite state
of system $\mathrm{AB}$, let $\rho_{\mathrm{A}}$ and $\rho_{\mathrm{B}}$
be its marginals on systems $\mathrm{A}$ and $\mathrm{B}$ respectively.
Then, $\rho_{\mathrm{A}}$ and $\rho_{\mathrm{B}}$ have the same
maximum eigenvalue, namely
\[
p_{*,\mathrm{A}}=p_{*,\mathrm{B}}=:p_{*},
\]
where $p_{*,\mathrm{A}}$ and $p_{*,\mathrm{B}}$ are the maximum
eigenvalues of $\rho_{\mathrm{A}}$ and $\rho_{\mathrm{B}}$ respectively.

Moreover, when $\rho_{\mathrm{A}}$ (or equivalently $\rho_{\mathrm{B}}$
) is decomposed as $\rho_{\mathrm{A}}=p_{*}\alpha+\left(1-p_{*}\right)\sigma$
for some \emph{pure} state $\alpha$ and some state $\sigma$, the
states $\alpha$ and $\sigma$ are perfectly distinguishable with
the observation-test $\left\{ a,u_{\mathrm{A}}-a\right\} $, where
$a$ is any pure effect such that $\left(a\middle|\alpha\right)=1$.
\end{thm}

\begin{proof}
The fact that both marginals have the same maximum eigenvalue was
proved in \cite[theorem 2, corollary 1]{QPL15}, and we will not report
the proof here for the sake of brevity.

Now, write $\rho_{\mathrm{A}}$ as
\begin{equation}
\rho_{\mathrm{A}}=p_{*}\alpha+\left(1-p_{*}\right)\sigma,\label{eq:rho*}
\end{equation}
where $\alpha$ is an eigenstate with maximum eigenvalue of $\rho$,
and $\sigma$ is possibly mixed. By \cite[proposition 11]{QPL15},
if $a$ is a pure effect such that $\left(a\middle|\alpha\right)=1$
we have $\left(a\middle|\rho_{\mathrm{A}}\right)=p_{*}$. Combining
this equality with eq.~\eqref{eq:rho*} we finally obtain 
\[
p_{*}=\left(a\middle|\rho_{\mathrm{A}}\right)=p_{*}+\left(1-p_{*}\right)\left(a\middle|\sigma\right),
\]
which implies $\left(a\middle|\sigma\right)=0$ (unless $p_{*}=1$,
but in this case the state $\rho_{\mathrm{A}}$ is pure). Hence, $\alpha$
and $\sigma$ are perfectly distinguishable with the test $\left\{ a,u_{\mathrm{A}}-a\right\} $.
\end{proof}
Now we have managed to decompose every given state into a mixture
of two perfectly distinguishable states. The probability balance has
a lot of other consequences. The first is that every non-trivial system
has at least two perfectly distinguishable pure states. To prove it,
however, first of all, we must note that for the invariant state $\chi$,
due to its invariance under the action of reversible channels, every
pure state is an eigenstate with maximum eigenvalue. Indeed, if $\chi$
is decomposed as 
\[
\chi=p_{*}\alpha+\left(1-p_{*}\right)\sigma,
\]
it can also be decomposed as 
\[
\chi=p_{*}\mathcal{U}\alpha+\left(1-p_{*}\right)\mathcal{U}\sigma,
\]
where $\mathcal{U}$ is a reversible channel. Owing to transitivity,
every pure state $\alpha'$ can be obtained as $\mathcal{U}\alpha$
for some suitable reversible channel, meaning that every pure state
is an eigenstate with maximum eigenvalue.
\begin{cor}
If $\mathrm{A}\neq\mathrm{I}$, then every pure state of $\mathrm{A}$
is perfectly distinguishable from some other pure state.
\end{cor}

\begin{proof}
The proof is an application of theorem~\ref{thm:probability balance}
to the case of the invariant state, and it has already appeared in
\cite[corollary 3]{QPL15}. Since for every pure state $\alpha$,
$\chi=p_{*}\alpha+\left(1-p_{*}\right)\sigma$, $\alpha$ is perfectly
distinguishable from $\sigma$, and from all the pure states contained
in $\sigma$. 
\end{proof}
It is quite remarkable that the existence of perfectly distinguishable
states pops out from the axioms, without being assumed from the start,
or without relying on mathematical assumptions such as the no-restriction
hypothesis (cf.\ proposition~\ref{prop:no-restriction distinguishable}).

Another consequence of the probability balance is the following.
\begin{cor}
\label{cor:pupstar}Let $\rho$ be a mixed state of system $\mathrm{A}$.
Then, the following are equivalent:
\begin{enumerate}
\item $\alpha$ is an eigenstate of $\rho$ with maximum eigenvalue $p_{*}$;
\item $\left(a\middle|\rho\right)=p_{*}$ for every pure effect $a$ such
that $\left(a\middle|\alpha\right)=1$.
\end{enumerate}
\end{cor}

\begin{proof}
By \cite[proposition 11]{QPL15}, we already know that $1\Rightarrow2$.
Let us prove the converse implication $2\Rightarrow1$. Suppose that
$\left(a\middle|\rho\right)=p_{*}$, where $a$ is a pure effect such
that $\left(a\middle|\alpha\right)=1$ on some pure state $\alpha$.
Let us show that $\alpha$ is an eigenstate of $\rho$ with maximum
eigenvalue $p_{*}$. Now, for every purification of $\rho$, say $\Psi\in\mathsf{PurSt}_{1}\left(\mathrm{AB}\right)$,
one has\begin{equation}\label{11} \begin{aligned}\Qcircuit @C=1em @R=.7em @!R { & \multiprepareC{1}{\Psi} & \qw \poloFantasmaCn{\rA} & \measureD{a} \\ & \pureghost{\Psi} & \qw \poloFantasmaCn{\rB} & \qw }\end{aligned}~=q\!\!\!\!\begin{aligned}\Qcircuit @C=1em @R=.7em @!R { & \prepareC{\beta} & \qw \poloFantasmaCn{\rB} & \qw }\end{aligned} ~, \end{equation}where
$\beta$ is a normalised state, pure by Purity Preservation. We have\[ q=q\!\!\!\!\begin{aligned}\Qcircuit @C=1em @R=.7em @!R { & \prepareC{\beta} & \qw \poloFantasmaCn{\rB} & \measureD{u} }\end{aligned} ~=\!\!\!\!\begin{aligned}\Qcircuit @C=1em @R=.7em @!R { & \multiprepareC{1}{\Psi} & \qw \poloFantasmaCn{\rA} & \measureD{a} \\ & \pureghost{\Psi} & \qw \poloFantasmaCn{\rB} & \measureD{u} }\end{aligned}~=\!\!\!\!\begin{aligned}\Qcircuit @C=1em @R=.7em @!R { & \prepareC{\rho} & \qw \poloFantasmaCn{\rA} & \measureD{a} }\end{aligned}~=p_{*}~. \]Hence
eq.~\eqref{11} becomes\begin{equation}\label{12} \begin{aligned}\Qcircuit @C=1em @R=.7em @!R { & \multiprepareC{1}{\Psi} & \qw \poloFantasmaCn{\rA} & \measureD{a} \\ & \pureghost{\Psi} & \qw \poloFantasmaCn{\rB} & \qw }\end{aligned}~=p_{*}\!\!\!\!\begin{aligned}\Qcircuit @C=1em @R=.7em @!R { & \prepareC{\beta} & \qw \poloFantasmaCn{\rB} & \qw }\end{aligned} ~. \end{equation}This
condition implies that $\beta$ is an eigenvector of the marginal
state $\widetilde{\rho}=\mathrm{tr}_{\mathrm{A}}\Psi$. The implication
$1\Rightarrow2$ guarantees that $\left(b\middle|\widetilde{\rho}\right)=p_{*}$,
for every pure effect $b$ such that $\left(b\middle|\beta\right)=1$.
The last condition implies an equation very similar to eq.~\eqref{12}:\begin{equation}\label{22} \begin{aligned}\Qcircuit @C=1em @R=.7em @!R { & \multiprepareC{1}{\Psi} & \qw \poloFantasmaCn{\rA} & \qw \\ & \pureghost{\Psi} & \qw \poloFantasmaCn{\rB} & \measureD{b} }\end{aligned}~=p_{*}\!\!\!\!\begin{aligned}\Qcircuit @C=1em @R=.7em @!R { & \prepareC{\alpha'} & \qw \poloFantasmaCn{\rA} & \qw }\end{aligned} ~. \end{equation}for
some pure state $\alpha'$. Hence, $\alpha'$ is an eigenstate of
$\rho$ with eigenvalue $p_{*}$. To conclude the proof, it is enough
to observe that $\alpha'=\alpha$. Indeed, combining eqs.~\eqref{12}
and \eqref{22} we have\[ \left( a \middle| \alpha'\right) = \frac{1}{p_*} \!\!\!\!\begin{aligned}\Qcircuit @C=1em @R=.7em @!R { & \multiprepareC{1}{\Psi} & \qw \poloFantasmaCn{\rA} & \measureD{a} \\ & \pureghost{\Psi} & \qw \poloFantasmaCn{\rB} & \measureD{b} }\end{aligned}~= \left( b\middle|\beta\right) = 1 , \]which
implies $\alpha'=\alpha$ by proposition~\ref{prop:uniqueness of state}.
\end{proof}
For every state we can also define another probability 
\begin{equation}
p^{*}:=\sup_{a\in\mathsf{PurEff}\left(\mathrm{A}\right)}\left(a\middle|\rho\right).\label{eq:p^*}
\end{equation}
We will see soon that the supremum in the definition of $p^{*}$ is
in fact a maximum. Note that for pure states $p^{*}=1$, thanks to
proposition~\ref{prop:identifyingeffect}. 

By corollary~\ref{cor:pupstar}, for every state one has the bound
$p_{*}\le p^{*}$. We can in fact prove that one has the equality.
\begin{prop}
\label{prop:pstar=00003Dpstar}For every state $\rho\in\mathsf{St}_{1}\left(\mathrm{A}\right)$
one has $p^{*}=p_{*}$.
\end{prop}

\begin{proof}
It is enough to show that $p^{*}\le p_{*}$. Pick a pure effect $a\in\mathsf{PurEff}\left(\mathrm{A}\right)$
such that $\left(a\middle|\rho\right)\neq0$. Such a pure effect exists
because any pure effect $a$ such that $\left(a\middle|\alpha\right)=1$,
where $\alpha$ is an eigenstate of $\rho$ with maximum eigenvalue
$p_{*}$, has the property $\left(a\middle|\rho\right)=p_{*}\neq0$.
Now consider a purification $\Psi\in\mathsf{PurSt}_{1}\left(\mathrm{AB}\right)$
of $\rho$, and define the pure state $\beta$ as\[ p\!\!\!\!\begin{aligned}\Qcircuit @C=1em @R=.7em @!R { & \prepareC{\beta} & \qw \poloFantasmaCn{\rB} & \qw }\end{aligned}~=\!\!\!\!\begin{aligned}\Qcircuit @C=1em @R=.7em @!R { & \multiprepareC{1}{\Psi} & \qw \poloFantasmaCn{\rA} & \measureD{a} \\ & \pureghost{\Psi} & \qw \poloFantasmaCn{\rB} & \qw }\end{aligned}~. \]Note
that $p$ is non-vanishing because it is given by $p=\left(a\middle|\rho\right)$.
So $\beta$ arises in a convex decomposition of the marginal of $\Psi$
on $\mathrm{B}$ with probability $p$. By construction $p\leq p_{*}$,
namely $\left(a\middle|\rho\right)\leq p_{*}$. Taking the supremum
over $a$, we finally obtain $p^{*}\le p_{*}$, thus proving that
$p^{*}=p_{*}$.
\end{proof}
As a consequence we have that $p^{*}$ is achieved by applying any
pure effect $a$ such that $\left(a\middle|\alpha\right)=1$ to $\rho$,
where $\alpha$ is an eigenstate of $\rho$ with maximum eigenvalue.
Therefore, the supremum in the definition of $p^{*}$ in eq.~\eqref{eq:p^*}
is in fact a maximum.

\subsubsection{Probability balance for purifications of the invariant state}

Recall that the invariant state $\chi\in\mathsf{St}_{1}\left(\mathrm{A}\right)$
can be written as 
\[
\chi=p_{*}\alpha+\left(1-p_{*}\right)\sigma
\]
for every \emph{pure} state $\alpha\in\mathsf{PurSt}_{1}\left(\mathrm{A}\right)$,
where $\sigma$ is a suitable state. In addition, we have the following
proposition.
\begin{prop}
\label{prop:steerpstar}Let $\Phi\in\mathsf{PurSt}_{1}\left(\mathrm{AB}\right)$
be a purification of the invariant state $\chi_{\mathrm{A}}$. Then,
for every pure state $\alpha\in\mathsf{PurSt}_{1}\left(\mathrm{A}\right)$
there exists a pure effect $b\in\mathsf{PurEff}\left(\mathrm{B}\right)$
such that\[ \begin{aligned}\Qcircuit @C=1em @R=.7em @!R { & \multiprepareC{1}{\Phi} & \qw \poloFantasmaCn{\rA} & \qw \\ & \pureghost{\Phi} & \qw \poloFantasmaCn{\rB} & \measureD{b} }\end{aligned}~=p_{*}\!\!\!\!\begin{aligned}\Qcircuit @C=1em @R=.7em @!R { & \prepareC{\alpha} & \qw \poloFantasmaCn{\rA} & \qw }\end{aligned} ~, \]where
$p_{*}$ is the maximum eigenvalue of $\chi_{\mathrm{A}}$. As a consequence,
every normalised pure effect $a\in\mathsf{PurEff}_{1}\left(\mathrm{A}\right)$
satisfies the condition $\left(a\middle|\chi_{A}\right)=p_{*}$.
\end{prop}

\begin{proof}
We know that every pure state $\alpha\in\mathsf{PurSt}_{1}\left(\mathrm{A}\right)$
is an eigenstate of $\chi_{\mathrm{A}}$. In the proof of corollary~\ref{cor:pupstar}
we saw that, if $a$ is a pure effect such that $\left(a\middle|\alpha\right)=1$,
and $\Phi\in\mathsf{PurSt}_{1}\left(\mathrm{AB}\right)$ is a purification
of $\chi_{\mathrm{A}}$, we have\[ \begin{aligned}\Qcircuit @C=1em @R=.7em @!R { & \multiprepareC{1}{\Phi} & \qw \poloFantasmaCn{\rA} & \measureD{a} \\ & \pureghost{\Phi} & \qw \poloFantasmaCn{\rB} & \qw }\end{aligned}~=p_{*}\!\!\!\!\begin{aligned}\Qcircuit @C=1em @R=.7em @!R { & \prepareC{\beta} & \qw \poloFantasmaCn{\rB} & \qw }\end{aligned} ~. \]In
the same proof we saw that, if $b$ is a pure effect such that $\left(b\middle|\beta\right)=1$,
we have\[ \begin{aligned}\Qcircuit @C=1em @R=.7em @!R { & \multiprepareC{1}{\Phi} & \qw \poloFantasmaCn{\rA} & \qw \\ & \pureghost{\Phi} & \qw \poloFantasmaCn{\rB} & \measureD{b} }\end{aligned}~=p_{*}\!\!\!\!\begin{aligned}\Qcircuit @C=1em @R=.7em @!R { & \prepareC{\alpha} & \qw \poloFantasmaCn{\rA} & \qw }\end{aligned} ~. \]

To conclude the proof, recall that there is a pure state $\alpha$
associated with every normalised pure effect $a$ such that $\left(a\middle|\alpha\right)=1$,
and that every pure state is an eigenstate of $\chi_{\mathrm{A}}$
with maximum eigenvalue. Then by corollary~\ref{cor:pupstar}, we
have $\left(a\middle|\chi_{\mathrm{A}}\right)=p_{*}$.
\end{proof}

\subsection{State-effect duality}

Using the results of the previous subsection, one can establish a
one-to-one correspondence between normalised pure states and normalised
pure effects. We refer to this correspondence as the \emph{dagger}
of states and effects. 
\begin{prop}
\label{prop:atmostoneeffect}For every pure state $\alpha\in\mathsf{PurSt}_{1}\left(\mathrm{A}\right)$
there is a unique (normalised) pure effect $a\in\mathsf{PurEff}_{1}\left(\mathrm{A}\right)$
such that $\left(a\middle|\alpha\right)=1$.
\end{prop}

\begin{proof}
The proof is essentially the same as in \cite[theorem 8]{Chiribella-informational},
even though we are assuming fewer axioms. Suppose that $a$ and $a'$
are two pure effects such that $\left(a\middle|\alpha\right)=\left(a'\middle|\alpha\right)=1$.
Then, let $\Phi\in\mathsf{PurSt}_{1}\left(\mathrm{AB}\right)$ be
a purification of the invariant state $\chi_{\mathrm{A}}$. By proposition~\ref{prop:steerpstar},
there exists a pure effect $b$ such that\begin{equation}\label{uffissima} \begin{aligned}\Qcircuit @C=1em @R=.7em @!R { & \multiprepareC{1}{\Phi} & \qw \poloFantasmaCn{\rA} & \qw \\ & \pureghost{\Phi} & \qw \poloFantasmaCn{\rB} & \measureD{b} }\end{aligned}~=p_{*}\!\!\!\!\begin{aligned}\Qcircuit @C=1em @R=.7em @!R { & \prepareC{\alpha} & \qw \poloFantasmaCn{\rA} & \qw }\end{aligned} ~, \end{equation}and
the two effects $a$ and $a'$ must satisfy
\begin{equation}
\left(a\middle|\chi_{\mathrm{A}}\right)=p_{*}=\left(a'\middle|\chi_{\mathrm{A}}\right).\label{eq:achiaprimchi}
\end{equation}
Now, let us define the pure states $\beta$ and $\beta'$ through
the relations\[ \begin{aligned}\Qcircuit @C=1em @R=.7em @!R { & \multiprepareC{1}{\Phi} & \qw \poloFantasmaCn{\rA} & \measureD{a} \\ & \pureghost{\Phi} & \qw \poloFantasmaCn{\rB} & \qw }\end{aligned}~ =:q\!\!\!\!\begin{aligned}\Qcircuit @C=1em @R=.7em @!R { & \prepareC{\beta} & \qw \poloFantasmaCn{\rB} & \qw }\end{aligned} ~ , \] \[ \begin{aligned}\Qcircuit @C=1em @R=.7em @!R { & \multiprepareC{1}{\Phi} & \qw \poloFantasmaCn{\rA} & \measureD{a'} \\ & \pureghost{\Phi} & \qw \poloFantasmaCn{\rB} & \qw }\end{aligned}~ =:q'\!\!\!\!\begin{aligned}\Qcircuit @C=1em @R=.7em @!R { & \prepareC{\beta'} & \qw \poloFantasmaCn{\rB} & \qw }\end{aligned} ~, \]where
$q$ and $q'$ are suitable probabilities. By applying the deterministic
effect on both sides and using eq.~\eqref{eq:achiaprimchi} one obtains
the equality $q=p_{*}=q'$. Hence, eqs.~\eqref{uffissima} and \eqref{eq:achiaprimchi}
lead to the equalities \[ \left( b \middle| \beta\right) = \frac {1} {p_*}\!\!\!\! \begin{aligned}\Qcircuit @C=1em @R=.7em @!R { & \multiprepareC{1}{\Phi} & \qw \poloFantasmaCn{\rA} & \measureD{a} \\ & \pureghost{\Phi} & \qw \poloFantasmaCn{\rB} & \measureD{b} }\end{aligned} ~= \left( a|\alpha\right) = 1 \]
\[
\left( b | \beta'\right) = \frac {1} {p_*} \!\!\!\! \begin{aligned}\Qcircuit @C=1em @R=.7em @!R { & \multiprepareC{1}{\Phi} & \qw \poloFantasmaCn{\rA} & \measureD{a'} \\ & \pureghost{\Phi} & \qw \poloFantasmaCn{\rB} & \measureD{b} }\end{aligned} = \left( a'\middle|\alpha\right) = 1 . \] By proposition~\ref{prop:uniqueness of state} we conclude that
$\beta$ and $\beta'$ must be equal. Recalling the definitions of
$\beta$ and $\beta'$, we then obtain the relation \[ \begin{aligned}\Qcircuit @C=1em @R=.7em @!R { & \multiprepareC{1}{\Phi} & \qw \poloFantasmaCn{\rA} & \measureD{a} \\ & \pureghost{\Phi} & \qw \poloFantasmaCn{\rB} & \qw }\end{aligned}~ = \!\!\!\! \begin{aligned}\Qcircuit @C=1em @R=.7em @!R { & \multiprepareC{1}{\Phi} & \qw \poloFantasmaCn{\rA} & \measureD{a'} \\ & \pureghost{\Phi} & \qw \poloFantasmaCn{\rB} & \qw }\end{aligned}~, \]which
implies $a=a'$ because the state $\Phi$ is dynamically faithful
on system $\mathrm{A}$ (proposition~\ref{prop:faithful}).
\end{proof}
Putting propositions~\ref{prop:uniqueness of state}, \ref{prop:identifyingeffect},
and \ref{prop:atmostoneeffect} together we finally obtain the desired
result.
\begin{thm}[State-effect duality]
\label{thm:duality}For every system $\mathrm{A}$ there exists a
bijective correspondence between normalised pure states and normalised
pure effects, called the \emph{dagger} and denoted by $\dagger:\mathsf{PurSt}_{1}\left(\mathrm{A}\right)\to\mathsf{PurEff}_{1}\left(\mathrm{A}\right)$.
The dagger satisfies the condition $\left(\alpha^{\dagger}\middle|\alpha\right)=1$,
for every $\alpha\in\mathsf{PurSt}_{1}\left(\mathrm{A}\right)$.
\end{thm}

Therefore, for every normalised pure state $\alpha$, $\alpha^{\dagger}$
denotes the associated pure effect such that $\left(\alpha^{\dagger}\middle|\alpha\right)=1$.
With a little abuse of notation we will denote the unique normalised
pure state associated with a normalised pure effect $a$ also with
a dagger. For example, if $a$ is a normalised pure effect, $a^{\dagger}$
is the unique normalised pure state such that $\left(a\middle|a^{\dagger}\right)=1$.

Note that here this duality is limited only to states and effects,
thus it is weaker than the dagger introduced in \cite{Selby-dagger},
however it was shown that it can be extended to all transformations
too \cite{HOI}.

An easy corollary of the state-effect duality is the following (cf.\ \cite[corollary 13]{Chiribella-informational}).
\begin{cor}[Transitivity on pure effects]
For every pair of pure normalised effects $a,a'\in\mathsf{PurEff}_{1}\left(\mathrm{A}\right)$,
there exists a reversible channel $\mathcal{U}$ on $\mathrm{A}$
such that $a'=a\mathcal{U}$.
\end{cor}

In other words, reversible channels act transitively on the set of
normalised pure effects too.

\section{No disturbance without information}

In sharp theories with purification, one can construct transformations
that are ``minimally-disturbing'' \cite{Chiribella-informational,Wehner-Pfister,Chiribella-Yuan2015}.
Besides being important in their own respect, these transformations
will provide the crucial ingredient to prove the diagonalisation theorem
in section~\ref{sec:Diagonalisation-of-states}.

The core result is the following proposition, from which important
and useful corollaries follow. Note that a similar result, but under
different axioms, was proved in \cite{Colleagues}.
\begin{prop}
\label{prop:non-disturbing measurement}Let $a$ be an effect such
that $\left(a\middle|\rho\right)=1$, for some $\rho\in\mathsf{St}_{1}\left(\mathrm{A}\right)$.
Then there exists a \emph{pure} transformation $\mathcal{T}$ on $\mathrm{A}$
such that $\mathcal{T}=_{\rho}\mathcal{I}$, where $\mathcal{I}$
is the identity, and $\left(u\middle|\mathcal{T}\middle|\sigma\right)\le\left(a\middle|\sigma\right)$,
for every state $\sigma\in\mathsf{St}_{1}\left(\mathrm{A}\right)$.
\end{prop}

\begin{proof}
The starting point of the proof is a result of \cite{Chiribella-informational},
which guarantees that every normalised effect $a\in\mathsf{Eff}_{1}\left(\mathrm{A}\right)$
can be written as 
\begin{equation}
a=u_{\mathrm{B}}\mathcal{A},\label{eq:effdecomp}
\end{equation}
where $\mathcal{A}$ is a \emph{pure} transformation from $\mathrm{A}$
to $\mathrm{B}$, and $\mathrm{B}$ is a suitable system. 

Now, let $\Psi\in\mathsf{PurSt}_{1}\left(\mathrm{AA}'\right)$ be
a purification of $\rho$. By eq.\ \eqref{eq:effdecomp}, we have\begin{equation}\label{eq:uno} \begin{aligned} \Qcircuit @C=1em @R=.7em @!R { & \multiprepareC{1}{\Psi} & \qw \poloFantasmaCn{\rA} & \gate{\mathcal A} & \qw \poloFantasmaCn{\rB} &\measureD{u} \\ & \pureghost{\Psi} & \qw \poloFantasmaCn{\rA'} & \qw &\qw &\qw }\end{aligned} ~= \!\!\!\! \begin{aligned}\Qcircuit @C=1em @R=.7em @!R { & \multiprepareC{1}{\Psi} & \qw \poloFantasmaCn{\rA} & \measureD{a} \\ & \pureghost{\Psi} & \qw \poloFantasmaCn{\rA'} & \qw }\end{aligned} \end{equation}Now,
since $\left(a\middle|\rho\right)=1$, we have $a=_{\rho}u_{\mathrm{A}}$.
Hence, proposition~\ref{prop:purifications -> input} implies\begin{equation}\label{eq:due} \begin{aligned}\Qcircuit @C=1em @R=.7em @!R { & \multiprepareC{1}{\Psi} & \qw \poloFantasmaCn{\rA} & \measureD{a} \\ & \pureghost{\Psi} & \qw \poloFantasmaCn{\rA'} & \qw }\end{aligned}~= \!\!\!\! \begin{aligned}\Qcircuit @C=1em @R=.7em @!R { & \multiprepareC{1}{\Psi} & \qw \poloFantasmaCn{\rA} & \measureD{u} \\ & \pureghost{\Psi} & \qw \poloFantasmaCn{\rA'} & \qw }\end{aligned}~. \end{equation}Combining
eqs.~\eqref{eq:uno} and \eqref{eq:due}, we obtain\[ \begin{aligned} \Qcircuit @C=1em @R=.7em @!R { & \multiprepareC{1}{\Psi} & \qw \poloFantasmaCn{\rA} & \gate{\mathcal A} & \qw \poloFantasmaCn{\rB} &\measureD{u} \\ & \pureghost{\Psi} & \qw \poloFantasmaCn{\rA'} & \qw &\qw &\qw }\end{aligned} ~=\!\!\!\! \begin{aligned}\Qcircuit @C=1em @R=.7em @!R { & \multiprepareC{1}{\Psi} & \qw \poloFantasmaCn{\rA} & \measureD{u} \\ & \pureghost{\Psi} & \qw \poloFantasmaCn{\rA'} & \qw }\end{aligned}~, \]meaning
that the two pure states $\left(\mathcal{A}\otimes\mathcal{I}_{\mathrm{A}'}\right)\Psi$
and $\Psi$ have the same marginal on system $\mathrm{A}'$. By the
uniqueness of purification, for fixed pure states $\alpha_{0}\in\mathsf{PurSt}_{1}\left(\mathrm{A}\right)$
and $\beta_{0}\in\mathsf{PurSt}_{1}\left(\mathrm{B}\right)$, there
must exist a reversible channel $\mathcal{U}$ on $\mathrm{AB}$,
such that\[ \begin{aligned} \Qcircuit @C=1em @R=.7em @!R { & \prepareC{\alpha_0} & \qw \poloFantasmaCn{\rA} & \qw & \qw & \multigate{1}{\mathcal U} & \qw \poloFantasmaCn{\rB} &\qw \\ & \multiprepareC{1}{\Psi} & \qw \poloFantasmaCn{\rA} & \gate{\mathcal A} & \qw \poloFantasmaCn{\rB} &\ghost{\mathcal U} & \qw \poloFantasmaCn{\rA} &\qw \\ & \pureghost{\Psi} & \qw \poloFantasmaCn{\rA'} & \qw &\qw &\qw &\qw &\qw} \end{aligned}~=\!\!\!\! \begin{aligned}\Qcircuit @C=1em @R=.7em @!R { & \prepareC{\beta_0} & \qw \poloFantasmaCn{\rB} & \qw \\ & \multiprepareC{1}{\Psi} & \qw \poloFantasmaCn{\rA} & \qw \\ & \pureghost{\Psi} & \qw \poloFantasmaCn{\rA'} & \qw }\end{aligned}~. \]By
applying $\beta_{0}^{\dagger}$ to both sides, we obtain\[ \begin{aligned} \Qcircuit @C=1em @R=.7em @!R { & \multiprepareC{1}{\Psi} & \qw \poloFantasmaCn{\rA} & \gate{\mathcal A} & \qw \poloFantasmaCn{\rB} & \gate{\mathcal P} & \qw \poloFantasmaCn{\rA} &\qw \\ & \pureghost{\Psi} & \qw \poloFantasmaCn{\rA'} & \qw &\qw &\qw &\qw &\qw}\end{aligned} ~= \!\!\!\! \begin{aligned}\Qcircuit @C=1em @R=.7em @!R { & \multiprepareC{1}{\Psi} & \qw \poloFantasmaCn{\rA} & \qw \\ & \pureghost{\Psi} & \qw \poloFantasmaCn{\rA'} & \qw }\end{aligned}~, \]where
$\mathcal{P}$ is the pure transformation defined as\[ \begin{aligned} \Qcircuit @C=1em @R=.7em @!R { & \qw \poloFantasmaCn{\rB} & \gate{\mathcal P} & \qw \poloFantasmaCn{\rA} &\qw }\end{aligned} ~:= \!\!\! \begin{aligned} \Qcircuit @C=1em @R=.7em @!R { & \prepareC{\alpha_0} & \qw \poloFantasmaCn{\rA} & \multigate{1}{\mathcal U} & \qw \poloFantasmaCn{\rB} &\measureD{\beta_0^{\dagger}} \\ && \qw \poloFantasmaCn{\rB} & \ghost{\mathcal U} & \qw \poloFantasmaCn{\rA} &\qw } \end{aligned}~. \]Let
us define the transformation $\mathcal{T}:=\mathcal{P}\mathcal{A}$,
which is pure by Purity Preservation. With this choice, we have\[ \begin{aligned} \Qcircuit @C=1em @R=.7em @!R { & \multiprepareC{1}{\Psi} & \qw \poloFantasmaCn{\rA} & \gate{\mathcal T} & \qw \poloFantasmaCn{\rA} &\qw \\ & \pureghost{\Psi} & \qw \poloFantasmaCn{\rA'} & \qw &\qw &\qw }\end{aligned} ~= \!\!\!\! \begin{aligned}\Qcircuit @C=1em @R=.7em @!R { & \multiprepareC{1}{\Psi} & \qw \poloFantasmaCn{\rA} & \qw \\ & \pureghost{\Psi} & \qw \poloFantasmaCn{\rA'} & \qw } \end{aligned}~, \]which
implies $\mathcal{T}=_{\rho}\mathcal{I}$ by proposition~\ref{prop:purifications -> input}.
Finally, for all states $\sigma\in\mathsf{St}\left(\mathrm{A}\right)$
we have the inequality 
\[
\left(u_{\mathrm{A}}\middle|\mathcal{T}\middle|\sigma\right)=\left(u_{\mathrm{A}}\middle|\mathcal{P}\mathcal{A}\middle|\sigma\right)\le\left(u_{\mathrm{B}}\middle|\mathcal{A}\middle|\sigma\right)=\left(a\middle|\sigma\right).
\]
Here, the inequality follows because all transformations are norm-non-increasing,
and in this case we are dealing with the norm of $\mathcal{A}\sigma$
under the action of the transformation $\mathcal{P}$. The last equality
follows from eq.~\eqref{eq:effdecomp}.
\end{proof}
Note that the pure transformation $\mathcal{T}$ is non-disturbing
on $\rho$ because it acts as the identity on $\rho$ and on all the
states contained in it. In other words, whenever we have a (possibly
mixed) effect occurring with unit probability on some state $\rho$,
we can always find a transformation that does not disturb $\rho$
(i.e.\ a non-disturbing non-demolition measurement). Being non-disturbing
means that $\mathcal{T}$ occurs with unit probability on all the
states contained in $\rho$. The other notable result of this proposition
is that the probability of $\mathcal{T}$ occurring on a generic state
$\sigma$ is less than or equal to the probability of the original
effect occurring on the same state.

The first consequence of proposition~\ref{prop:non-disturbing measurement}
is quite a technical result that is widely used in the rest of the
chapter and of the thesis.
\begin{cor}
\label{cor:conjecture proved}Let $\left\lbrace \alpha_{i}\right\rbrace _{i=1}^{n}$
be perfectly distinguishable \emph{pure} states. Then one always has
$\left(\alpha_{i}^{\dagger}\middle|\alpha_{j}\right)=\delta_{ij}$.
\end{cor}

\begin{proof}
Clearly we need only to prove that, for every $i,j\in\left\lbrace 1,\ldots,n\right\rbrace $,
whenever $j\neq i$, one has $\left(\alpha_{i}^{\dagger}\middle|\alpha_{j}\right)=0$.
Let $\left\lbrace a_{i}\right\rbrace _{i=1}^{n}$ be the perfectly
distinguishing test for the pure states $\left\lbrace \alpha_{i}\right\rbrace _{i=1}^{n}$.
Since $\left(a_{i}\middle|\alpha_{i}\right)=1$, by proposition~\ref{prop:non-disturbing measurement},
for every $i\in\left\lbrace 1,\ldots,n\right\rbrace $ there exists
a \emph{pure} transformation $\mathcal{A}_{i}$ not disturbing $\alpha_{i}$,
namely 
\begin{equation}
\mathcal{A}_{i}\alpha_{i}=\alpha_{i}.\label{eq:invariant conjecture}
\end{equation}
Instead, for all $j\neq i$, by proposition~\ref{prop:non-disturbing measurement}
one has
\[
\left(u\middle|\mathcal{A}_{i}\middle|\alpha_{j}\right)\leq\left(a_{i}|\alpha_{j}\right)=0,
\]
as $\left\lbrace a_{i}\right\rbrace _{i=1}^{n}$ is perfectly distinguishing.
This implies that, for all $j\neq i$, 
\begin{equation}
\mathcal{A}_{i}\alpha_{j}=0.\label{eq:zero conjecture}
\end{equation}
Now, evaluating the expression $\left(\alpha_{i}^{\dagger}\middle|\mathcal{A}_{i}\middle|\alpha_{i}\right)$,
and recalling eq.~\eqref{eq:invariant conjecture}, we get
\[
\left(\alpha_{i}^{\dagger}\middle|\mathcal{A}_{i}\middle|\alpha_{i}\right)=\left(\alpha_{i}^{\dagger}\middle|\alpha_{i}\right)=1.
\]
Since $\alpha_{i}^{\dagger}\mathcal{A}_{i}$ is a pure effect by Purity
Preservation, and it occurs with unit probability on the state $\alpha_{i}$,
by the state-effect duality it must be 
\begin{equation}
\alpha_{i}^{\dagger}\mathcal{A}_{i}=\alpha_{i}^{\dagger}.\label{eq:dagger invariant conjecture}
\end{equation}
Now, suppose $j\neq i$. By eqs.~\eqref{eq:dagger invariant conjecture}
and \eqref{eq:zero conjecture}, one has
\[
\left(\alpha_{i}^{\dagger}\middle|\alpha_{j}\right)=\left(\alpha_{i}^{\dagger}\middle|\mathcal{A}_{i}\middle|\alpha_{j}\right)=0.
\]
This concludes the proof.
\end{proof}
Note that, at this stage, this result does \emph{not} mean that the
effects $\left\lbrace \alpha_{i}^{\dagger}\right\rbrace _{i=1}^{n}$
make up an observation-test. This will be instead a consequence of
corollary~\ref{cor:dagger distinguishable}.

Using the existence of non-disturbing transformations we can also
give a sufficient condition for the perfect distinguishability of
states. The following condition, and especially its version for pure
states, form the core of the proof of the diagonalisation theorem,
and it is a more rigorous version of the construction used in \cite{QPL15}.
\begin{lem}
\label{lem:distinguishable}Let $\left\{ \rho_{i}\right\} _{i=1}^{n}$
be a set of normalised states. If there exists a set of effects\footnote{Not necessarily an observation-test or a subset of an observation-test.}
$\left\{ a_{i}\right\} _{i=1}^{n}$ such that $\left(a_{i}\middle|\rho_{i}\right)=1$
for all $i$, and $\left(a_{i}\middle|\rho_{j}\right)=0$ for all
$j>i$, then the states $\left\{ \rho_{i}\right\} _{i=1}^{n}$ are
perfectly distinguishable.
\end{lem}

\begin{proof}
By hypothesis, the binary observation-test $\left\{ a_{i},u-a_{i}\right\} $
distinguishes perfectly between $\rho_{i}$ and all the other states
$\rho_{j}$ with $j>i$. Equivalently, this observation-test distinguishes
perfectly between $\rho_{i}$ and the state $\widetilde{\rho}_{i}:=\frac{1}{n-i}\sum_{j>i}\rho_{j}$.
Specifically, $\left(u-a_{i}\middle|\widetilde{\rho}_{i}\right)=1$.
Applying proposition~\ref{prop:non-disturbing measurement}, we can
construct a pure transformation $\mathcal{A}_{i}^{\perp}$ such that
$\mathcal{A}_{i}^{\perp}=_{\widetilde{\rho}_{i}}\mathcal{I}$, and,
specifically,
\begin{equation}
\mathcal{A}_{i}^{\perp}\rho_{j}=\rho_{j}\label{eq:orthogonal non disturbing}
\end{equation}
for all $j>i$. Moreover, proposition~\ref{prop:non-disturbing measurement}
implies 
\[
\left(u\middle|\mathcal{A}_{i}^{\perp}\middle|\rho_{i}\right)\le\left(u-a_{i}\middle|\rho_{i}\right)=0,
\]
meaning that the transformation $\mathcal{A}_{i}^{\perp}$ never occurs
on the state $\rho_{i}$. Let us define the effect $a_{i,0}:=u-a_{i}-u\mathcal{A}_{i}^{\perp}$.
Note that this effect is well-defined, because $\left(a_{i,0}\middle|\sigma\right)\geq0$,
for all $\sigma\in\mathsf{St}_{1}\left(\mathrm{A}\right)$. Indeed,
by proposition~\ref{prop:non-disturbing measurement}, we have $\left(u\middle|\mathcal{A}_{i}^{\perp}\middle|\sigma\right)\le\left(u-a_{i}\middle|\sigma\right)$,
for all $\sigma\in\mathsf{St}_{1}\left(\mathrm{A}\right)$, whence
$\left(u-a_{i}-u\mathcal{A}_{i}^{\perp}\middle|\sigma\right)\geq0$.
Note that $a_{i,0}$ never occurs on the states $\rho_{k}$ with $k\ge i$.

Now, define the transformations $\mathcal{A}_{i}=\left|\rho_{i}\right)\left(a_{i}\right|$
and $\mathcal{A}_{i,0}=\left|\rho_{0}\right)\left(a_{i,0}\right|$,
where $\rho_{0}$ is a fixed normalised state. By proposition~\ref{prop:sufficientfortest},
the transformations $\left\{ \mathcal{A}_{i},\mathcal{A}_{i}^{\perp},\mathcal{A}_{i,0}\right\} $
form a test. Summarising the above observations, the test satisfies
the properties
\begin{equation}
\left\{ \begin{array}{l}
\mathcal{A}_{i}\rho_{i}=\rho_{i}\\
\mathcal{A}_{i}\rho_{j}=0\qquad\forall j>i\\
\mathcal{A}_{i}^{\perp}\rho_{i}=0\\
\mathcal{A}_{i}^{\perp}\rho_{j}=\rho_{j}\qquad\forall j>i\\
\mathcal{A}_{i,0}\rho_{k}=0\qquad\forall k\ge i.
\end{array}\right..\label{eq:propertiesA}
\end{equation}
By construction, the test distinguishes without error between the
state $\rho_{i}$ and all the states $\rho_{j}$ with $j>i$, in such
a way that the latter are not disturbed. Indeed, by construction $\mathcal{A}_{i}$
can only occur if the state is $\rho_{i}$, instead $\mathcal{A}_{i}^{\perp}$
never occurs on $\rho_{i}$, but it occurs with probability 1 if the
state is any of the $\rho_{j}$'s, with $j>i$, and it leaves them
unchanged. Finally, $\mathcal{A}_{i,0}$ never occurs on the states
$\rho_{k}$'s with $k\ge i$, so it does not play a role in the discrimination
process. Essentially, $\mathcal{A}_{i,0}$ only plays the role of
making $\left\{ \mathcal{A}_{i},\mathcal{A}_{i}^{\perp},\mathcal{A}_{i,0}\right\} $
a test.

Using the tests $\left\{ \mathcal{A}_{i},\mathcal{A}_{i}^{\perp},\mathcal{A}_{i,0}\right\} $
it is easy to construct a protocol that distinguishes perfectly between
the states $\left\{ \rho_{i}\right\} _{i=1}^{n}$. The protocol works
as follows: starting from $i=1$ perform the test $\left\{ \mathcal{A}_{i},\mathcal{A}_{i}^{\perp},\mathcal{A}_{i,0}\right\} $.
If the transformation $\mathcal{A}_{i}$ takes place, then the state
is $\rho_{i}$. If the transformation $\mathcal{A}_{i}^{\perp}$ takes
place, then perform the test $\left\{ \mathcal{A}_{i+1},\mathcal{A}_{i+1}^{\perp},\mathcal{A}_{i+1,0}\right\} $
(this can be done because $\mathcal{A}_{i}^{\perp}$ is non-disturbing).
Using this protocol, every state in the set $\left\{ \rho_{i}\right\} _{i=1}^{n}$
will be identified without error in at most $n$ steps. Overall, the
protocol is described by a test with $2n+1$ outcomes, corresponding
to the transformations 
\[
\mathcal{T}_{1}=\mathcal{A}_{1}
\]
\[
\mathcal{T}_{2}=\mathcal{A}_{2}\mathcal{A}_{1}^{\perp}
\]
\[
\vdots
\]
\[
\mathcal{T}_{n}=\mathcal{A}_{n}\mathcal{A}_{n-1}^{\perp}\ldots\mathcal{A}_{1}^{\perp}
\]
\[
\mathcal{T}_{n+1}=\mathcal{A}_{1,0}
\]
\[
\mathcal{T}_{n+2}=\mathcal{A}_{2,0}\mathcal{A}_{1}^{\perp}
\]
\[
\vdots
\]
\[
\mathcal{T}_{2n}=\mathcal{A}_{n,0}\mathcal{A}_{n-1}^{\perp}\ldots\mathcal{A}_{1}^{\perp}
\]
\[
\mathcal{T}_{2n+1}=\mathcal{A}_{n}^{\perp}\ldots\mathcal{A}_{1}^{\perp}
\]
To show that these transformations form a test, we use proposition~\ref{prop:sufficientfortest}:
$\left\lbrace \mathcal{T}_{i}\right\rbrace _{i=1}^{2n+1}$ is a test
if and only if $\sum_{i=1}^{2n+1}u\mathcal{T}_{i}=u$. An easy check
shows that this is the case.

To complete the proof, we need to construct a perfectly distinguishing
test $\left\lbrace e_{i}\right\rbrace _{i=1}^{n}$ for the states
$\left\lbrace \rho_{i}\right\rbrace _{i=1}^{n}$. By discarding the
output of the transformations $\left\lbrace \mathcal{T}_{i}\right\rbrace _{i=1}^{2n+1}$,
we get an observation-test $\left\lbrace t_{i}\right\rbrace _{i=1}^{2n+1}$
with $2n+1$ outcomes and effects $t_{i}:=u\mathcal{T}_{i}$. We claim
that the observation-test 
\begin{equation}
\left\lbrace e_{i}\right\rbrace _{i=1}^{n}=\left\lbrace t_{1},\ldots,t_{n-1},u-t_{1}-\ldots-t_{n-1}\right\rbrace \label{eq:observation-test}
\end{equation}
is perfectly distinguishing for the states $\left\lbrace \rho_{i}\right\rbrace _{i=1}^{n}$.
First of all, since $t_{1},\ldots,t_{n-1}$ coexist in a ($2n+1$)-outcome
test, the effect $u-t_{1}-\ldots-t_{n-1}$ is well-defined. Now let
us prove that the observation-test~\eqref{eq:observation-test} perfectly
distinguishes the states $\rho_{i}$'s. We start from $t_{1}=u\mathcal{A}_{1}$;
by~\eqref{eq:propertiesA} we get
\[
\left(t_{1}\middle|\rho_{j}\right)=\left(u\middle|\mathcal{A}_{1}\middle|\rho_{j}\right)=\delta_{1j}\left(u\middle|\rho_{j}\right)=\delta_{1j}.
\]
Now, if $i>1$, 
\[
t_{i}=u\mathcal{A}_{i}\mathcal{A}_{i-1}^{\perp}\ldots\mathcal{A}_{1}^{\perp}.
\]
If we wish to calculate $\left(t_{i}\middle|\rho_{j}\right)$, by
eq.~\eqref{eq:propertiesA}, $\rho_{j}$ is left invariant by all
the $\mathcal{A}_{k}^{\perp}$ with $k<j$. If $i\neq j$, then 
\[
\left(t_{i}\middle|\rho_{j}\right)=\left(u\middle|\mathcal{A}_{i}\mathcal{A}_{i-1}^{\perp}\ldots\mathcal{A}_{j}^{\perp}\middle|\rho_{j}\right)=0,
\]
again by eq.~\eqref{eq:propertiesA}. If, instead $j=i$, 
\[
\left(t_{i}\middle|\rho_{i}\right)=\left(u\middle|\mathcal{A}_{i}\middle|\rho_{i}\right)=\left(u\middle|\rho_{i}\right)=1.
\]
As a consequence of these results,
\[
\left(u-t_{1}-\ldots-t_{n-1}\middle|\rho_{j}\right)=\delta_{nj}.
\]
We conclude that $\left\lbrace e_{i}\right\rbrace _{i=1}^{n}$ is
really a perfectly distinguishing test, because we have $\left(e_{i}\middle|\rho_{j}\right)=\delta_{ij}$.
\end{proof}
In the case when the states are pure, and the effects in the statement
of lemma~\ref{lem:distinguishable} are the daggers of those pure
states, we can prove something stronger.
\begin{cor}
\label{cor:dagger distinguishable}Let $\left\{ \alpha_{i}\right\} _{i=1}^{n}$
be a set of normalised \emph{pure} states such that $\left(\alpha_{i}^{\dagger}\middle|\alpha_{j}\right)=0$
for all $j>i$, then the states $\left\{ \alpha_{i}\right\} _{i=1}^{n}$
are perfectly distinguishable, and the pure effects $\left\{ \alpha_{i}^{\dagger}\right\} _{i=1}^{n}$
coexist in an observation-test, which distinguishes the states $\left\{ \alpha_{i}\right\} _{i=1}^{n}$
perfectly.
\end{cor}

\begin{proof}
If we take $a_{i}:=\alpha_{i}^{\dagger}$, by lemma~\ref{lem:distinguishable},
we know that the states $\left\{ \alpha_{i}\right\} _{i=1}^{n}$ are
perfectly distinguishable. Referring to the proof of lemma~\ref{lem:distinguishable},
note that, since $\mathcal{A}_{i}^{\perp}$ is pure, we have 
\begin{equation}
\alpha_{j}^{\dagger}\mathcal{A}_{i}^{\perp}=\alpha_{j}^{\dagger}\qquad\forall j>i.\label{eq:usadopo}
\end{equation}
by a similar argument to the one in the proof of corollary~\ref{cor:conjecture proved}.
Indeed, the effect $\alpha_{j}^{\dagger}\mathcal{A}_{i}^{\perp}$
is pure by Purity Preservation, and satisfies 
\[
\left(\alpha_{j}^{\dagger}\middle|\mathcal{A}_{i}^{\perp}\middle|\alpha_{j}\right)=\left(\alpha_{j}^{\dagger}\middle|\alpha_{j}\right)=1,
\]
where we have used eq.~\eqref{eq:orthogonal non disturbing}. Let
us construct the perfectly distinguishing observation-test like in
the proof of lemma~\ref{lem:distinguishable}, by considering the
effects $t_{i}=u\mathcal{A}_{i}$. One has, recalling that $\mathcal{A}_{i}=\left|\alpha_{i}\right)\left(\alpha_{i}^{\dagger}\right|$,
\[
t_{1}=u\mathcal{A}_{1}=\alpha_{1}^{\dagger}
\]
\[
t_{2}=u\mathcal{A}_{2}\mathcal{A}_{1}^{\perp}=\alpha_{2}^{\dagger}A_{1}^{\perp}=\alpha_{2}^{\dagger}
\]
\[
\vdots
\]
\[
t_{n}=u\mathcal{A}_{n}\mathcal{A}_{n-1}^{\perp}\ldots\mathcal{A}_{1}^{\perp}=\alpha_{n}^{\dagger},
\]
having used eq.~\eqref{eq:usadopo}. This proves that the effects
$\left\{ \alpha_{i}^{\dagger}\right\} _{i=1}^{n}$ coexist in a ($2n+1$)-outcome
observation-test. As a consequence, as shown in the proof of lemma~\ref{lem:distinguishable},
we have that 
\[
\left\lbrace \alpha_{1}^{\dagger},\ldots,\alpha_{n-1}^{\dagger},u-\alpha_{1}^{\dagger}-\ldots-\alpha_{n-1}^{\dagger}\right\rbrace 
\]
is perfectly distinguishing, and specifically $\left(\alpha_{i}^{\dagger}\middle|\alpha_{j}\right)=\delta_{ij}$.
\end{proof}
As a consequence of this corollary, whenever some pure states are
perfectly distinguishable, their daggers coexist in an observation-test
that distinguishes them perfectly. Corollaries~\ref{cor:conjecture proved}
and \ref{cor:dagger distinguishable}, taken together, state that
a necessary and sufficient condition for the pure states $\left\lbrace \alpha_{i}\right\rbrace _{i=1}^{n}$
to be perfectly distinguishable is that their daggers satisfy $\left(\alpha_{i}^{\dagger}\middle|\alpha_{j}\right)=\delta_{ij}$,
for all $i,j\in\left\lbrace 1,\ldots,n\right\rbrace $.

\section{Diagonalisation of states\label{sec:Diagonalisation-of-states}}

This section represents the core of the whole chapter, for it introduces
a key tool that will have plenty of consequences throughout this thesis:
the diagonalisation of states. Since states are not density matrices,
clearly diagonalisation here has a different meaning. To understand
it, let us look at quantum theory from an operational angle. Note
that in the diagonalisation of density matrices, a quantum state $\rho$
of a $d$-dimensional Hilbert space is diagonalised as $\rho=\sum_{j=1}^{d}p_{j}\ket{j}\bra{j},$
where $\left\{ \ket{j}\right\} _{j=1}^{d}$ is an orthonormal basis,
and $\left\{ p_{j}\right\} _{j=1}^{d}$ is a probability distribution.
Since $\ket{j}\bra{j}$ represent orthogonal pure states, we understand
the operational meaning of diagonalisation in quantum theory: a state
is diagonalised when it is written as a convex combination of perfectly
distinguishable pure states. Indeed the $p_{j}$'s\textemdash the
eigenvalues of $\rho$\textemdash are the coefficients of a convex
combination, and the pure states $\left\{ \ket{j}\bra{j}\right\} $
are distinguished perfectly by the projective measurement $\left\{ \ket{j}\bra{j}\right\} $.

Therefore it is natural to extend the definition of diagonalisation
to GPTs as follows: a \emph{diagonalisation} of $\rho$ is a convex
decomposition of $\rho$ into perfectly distinguishable pure states.
The probabilities in such a convex decomposition will be called the
\emph{eigenvalues} of $\rho$, and the perfectly distinguishable pure
states the \emph{eigenstates} \cite{QPL15}.

Note that, while it is true that every state $\rho$ in GPTs can be
decomposed as a convex combination of pure states, the key point about
diagonalisation is that every state should be written as a convex
combination of \emph{perfectly distinguishable} pure states. This
is a non-trivial property, for example the square bit \cite{Barrett}
does \emph{not} satisfy it (see fig.~\ref{fig:square no diag}).
\begin{figure}
\begin{centering}
\includegraphics{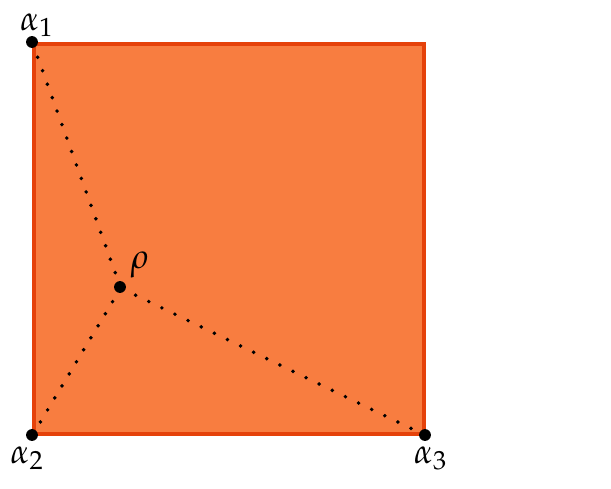}
\par\end{centering}
\caption{\label{fig:square no diag}The state $\rho$ is a non-trivial convex
combination of the pure states $\alpha_{1}$, $\alpha_{2}$, and $\alpha_{3}$,
which are \emph{not} jointly perfectly distinguishable. $\rho$ cannot
be diagonalised.}

\end{figure}

Rather than postulate the diagonalisation of all states like in \cite{Barnum-interference,Krumm-Muller,Colleagues},
here we \emph{derive} the diagonalisation of states from the axioms
of sharp theories with purification: Causality, Purity Preservation,
Pure Sharpness, and Purification. This result already appeared in
a preliminary form in \cite{QPL15}, here we recast it in a more rigorous
version, and we manage to prove the uniqueness of the eigenvalues
of states from the axioms of sharp theories with purification, a fact
that was only conjectured, and not proved in \cite{QPL15}.

We start with the actual diagonalisation theorem.
\begin{thm}
\label{thm:diago}In a sharp theory with purification, every state
of every (non-trivial) system can be diagonalised.
\end{thm}

The proof is a constructive procedure that returns a diagonalisation
of $\rho$ with the eigenvalues naturally listed in decreasing order,
namely $p_{i}\geq p_{i+1}$ for every $i$. In particular, one has
$p_{1}=p_{*}$, which justifies why we called $p_{*}$ the ``maximum
eigenvalue''.
\begin{proof}
In order to diagonalise the state $\rho$, it is enough to proceed
along the following steps:
\begin{enumerate}
\item Set $\rho_{1}=\rho$ and $p_{*,0}=0$.
\item For $i$ starting from $i=1$, decompose $\rho_{i}$ as $\rho_{i}=p_{*,i}\alpha_{i}+\left(1-p_{*,i}\right)\sigma_{i}$,
where $p_{*,i}$ is the maximum eigenvalue of $\rho_{i}$. Set $\rho_{i+1}=\sigma_{i}$,
$p_{i}=p_{*,i}\prod_{j<i}\left(1-p_{*,j}\right)$. If $p_{*,i}=1$,
then terminate, otherwise continue to the step $i+1$.
\end{enumerate}
Recall that theorem~\ref{thm:probability balance} guarantees the
condition $\left(\alpha_{i}^{\dagger}\middle|\sigma_{i}\right)=0$
at every step of the procedure. Since by construction every state
$\alpha_{j}$ with $j>i$ is contained in $\sigma_{i}$, we also have
$\left(\alpha_{i}^{\dagger}\middle|\alpha_{j}\right)=0$ for every
$j>i$. Hence, corollary~\ref{cor:dagger distinguishable} implies
that the states $\left\{ \alpha_{k}\right\} _{k=1}^{i}$, generated
by the first $i$ iterations of the protocol, are perfectly distinguishable,
for any $i$. For a finite-dimensional system, the procedure must
terminate in a finite number of iterations. Once the procedure has
been completed, the state $\rho$ is decomposed as $\rho=\sum_{i=1}^{r}p_{i}\alpha_{i}$
where $r$ is some positive (finite) integer, and $\left\{ \alpha_{i}\right\} _{i=1}^{r}$
are perfectly distinguishable pure states.
\end{proof}
Later in subsection~\ref{subsec:Uniqueness-of-the} we will show
that the vector of the eigenvalues $\mathbf{p}$, also called the
\emph{spectrum} of $\rho$, is uniquely determined by the state $\rho$,
which means that all diagonalisations of $\rho$ have the same eigenvalues.

Before moving forward, it is important to note that the eigenvalues
can be characterised as the outcome probabilities of a pure measurement
performed on the system.
\begin{cor}
\label{cor:computeegv}Let $\rho$ be a generic state, diagonalised
as $\rho=\sum_{i=1}^{r}p_{i}\alpha_{i}$. Then, one has $p_{i}=\left(\alpha_{i}^{\dagger}\middle|\rho\right)$,
for every $i\in\left\{ 1,\dots,r\right\} $.
\end{cor}

\begin{proof}
Immediate from the combination of theorem~\ref{thm:diago} and corollary~\ref{cor:conjecture proved},
because $\left(\alpha_{i}^{\dagger}\middle|\alpha_{j}\right)=\delta_{ij}$.
\end{proof}

\subsection{Diagonalisation of the invariant state}

Let us examine the properties of the diagonalisations of the invariant
states. In some respect they are special: first of all, all the eigenvalues
are equal.
\begin{prop}
\label{prop:diagonalization chi d-level}For every non-trivial system,
there exists a (strictly) positive integer $d$ such that
\begin{enumerate}
\item every diagonalisation of the invariant state consists of exactly $d$
pure states;
\item the eigenvalues of the invariant states are all equal to $\frac{1}{d}$.
\end{enumerate}
\end{prop}

\begin{proof}
Let $\chi=\sum_{i=1}^{r}p_{i}\alpha_{i}$ be a diagonalisation of
the invariant state $\chi$. By corollary~\ref{cor:computeegv},
$p_{i}=\left(\alpha_{i}^{\dagger}\middle|\chi\right)$, but by proposition~\ref{prop:steerpstar}
we have $\left(\alpha_{i}^{\dagger}\middle|\chi\right)=p_{*}$, whence
$p_{i}=p_{*}$ for every $i$. It follows that $p_{*}=\frac{1}{r}$.
Now consider another diagonalisation of $\chi$: $\chi=\sum_{i=1}^{r'}p'_{i}\alpha'_{i}$.
Repeating the same argument, we conclude that $p'_{i}=p_{*}=\frac{1}{r'}$.
This means that $r=r'=:d$.
\end{proof}
We will refer to $d$ as the \emph{dimension} of the system, for reasons
that will become clear soon.

Let us show that the set of states $\left\{ \alpha_{i}\right\} _{i=1}^{d}$
arising in any diagonalisation of the invariant state is \emph{maximal}.
\begin{prop}
\label{prop:maximal}Let $\chi$ be written as a uniform mixture of
pure states of the form $\chi=\frac{1}{d}\sum_{i=1}^{d}\alpha_{i}$,
where $d$ is the dimension of the system. Then
\begin{enumerate}
\item $\left\{ \alpha_{i}\right\} _{i=1}^{d}$ is a maximal set of perfectly
distinguishable pure states;
\item $\left\{ \alpha_{i}^{\dagger}\right\} _{i=1}^{d}$ is a pure observation-test.
\end{enumerate}
\end{prop}

\begin{proof}
Let us prove the two properties.
\begin{enumerate}
\item Suppose the invariant state is decomposed as $\chi=\frac{1}{d}\sum_{i=1}^{d}\alpha_{i}$.
Then, by proposition~\ref{prop:steerpstar} one has $\left(\alpha_{i}^{\dagger}\middle|\chi\right)=\frac{1}{d}$,
for any $i\in\left\lbrace 1,\ldots,d\right\rbrace $, and consequently
$\left(\alpha_{i}^{\dagger}\middle|\alpha_{j}\right)=\delta_{ij}$
. By corollary~\ref{cor:dagger distinguishable}, the states $\left\{ \alpha_{i}\right\} _{i=1}^{d}$
are perfectly distinguishable. Suppose by contradiction that this
is \emph{not} a maximal set; then we can add the pure state $\alpha_{d+1}$
so that $\left\{ \alpha_{i}\right\} _{i=1}^{d+1}$ is a pure maximal
set. If $\left\{ a_{i}\right\} _{i=1}^{d+1}$ is the observation-test
that distinguishes them, one must have $\left(a_{d+1}\middle|\chi\right)=\frac{1}{d}\sum_{i=1}^{d}\left(a_{d+1}\middle|\alpha_{i}\right)=0$.
But $\left(a_{d+1}\middle|\chi\right)=0$ implies $\left(a_{d+1}\middle|\rho\right)=0$
for every $\rho$, since every state is contained in the invariant
state, which is internal. This is in contradiction with the hypothesis
$\left(a_{d+1}\middle|\alpha_{d+1}\right)=1$. Hence, we have proved
that the set $\left\{ \alpha_{i}\right\} _{i=1}^{d}$ is maximal.
\item Let us prove that $\left\{ \alpha_{i}^{\dagger}\right\} _{i=1}^{d}$
is an observation-test, namely $\sum_{i=1}^{d}\alpha_{i}^{\dagger}=u$.
By propositions~\ref{prop:steerpstar} and \ref{prop:diagonalization chi d-level},
we have
\[
\sum_{i=1}^{d}\left(\alpha_{i}^{\dagger}\middle|\chi\right)=\sum_{i=1}^{d}\frac{1}{d}=1.
\]
Since $\chi$ is internal, this means that $\sum_{i=1}^{d}\left(\alpha_{i}^{\dagger}\middle|\rho\right)=1$
for every normalised state $\rho\in\mathsf{St}_{1}\left(\mathrm{A}\right)$,
whence $\sum_{i=1}^{d}\alpha_{i}^{\dagger}$ is the deterministic
effect $u$.
\end{enumerate}
\end{proof}
Propositions~\ref{prop:diagonalization chi d-level} and \ref{prop:maximal}
imply that the invariant state is a uniform mixture of the states
in a pure maximal set. Remarkably, the converse holds too:\emph{ every}
pure maximal set, mixed with equal weights, yields the invariant state.
\begin{prop}
\label{prop:diagonalization chi d-level 2}Let $\left\{ \alpha_{i}\right\} _{i=1}^{r}$
be a pure maximal set. Then one has $r=d$ and $\chi=\frac{1}{d}\sum_{i=1}^{d}\alpha_{i}$.
\end{prop}

\begin{proof}
We know that every pure state is an eigenstate of $\chi$ with maximum
eigenvalue. Specifically, we must have
\[
\chi=\frac{1}{d}\alpha_{1}+\frac{d-1}{d}\sigma_{1}
\]
for a state $\sigma_{1}$ that is perfectly distinguishable from $\alpha_{1}$
(cf.\ theorem~\ref{thm:probability balance}). For every $n<r$,
assume that the invariant state can be decomposed in the diagonalisation
process as 
\begin{equation}
\chi=\frac{1}{d}\left(\sum_{i=1}^{n}\alpha_{i}\right)+\frac{d-n}{d}\sigma_{n},\label{eq:inductionbasis}
\end{equation}
where the states $\left\{ \alpha_{i}\right\} _{i=1}^{n}\cup\left\{ \sigma_{n}\right\} $
are perfectly distinguishable, and we prove that a decomposition of
the same form can be found for $n+1$. To this purpose, we use the
relations
\begin{equation}
\left(\alpha_{i}\middle|\chi\right)=\frac{1}{d}\label{eq:same}
\end{equation}
for all $i\in\left\{ 1,\dots,r\right\} $, following from proposition~\ref{prop:steerpstar},
and valid for all normalised pure effects, and 
\begin{equation}
\left(\alpha_{i}^{\dagger}\middle|\alpha_{j}\right)=\delta_{ij}\label{eq:alfalfa}
\end{equation}
for all $i,j\in\left\{ 1,\dots,r\right\} $, following from the assumption
that the states $\left\lbrace \alpha_{i}\right\rbrace _{i=1}^{r}$
are perfectly distinguishable (cf.\ corollary~\ref{cor:conjecture proved}).
Eqs.~\eqref{eq:inductionbasis}, \eqref{eq:same}, and \eqref{eq:alfalfa}
yield the relation 
\[
\frac{1}{d}=\left(\alpha_{n+1}^{\dagger}\middle|\chi\right)=\frac{d-n}{d}\left(\alpha_{n+1}^{\dagger}\middle|\sigma_{n}\right),
\]
or, equivalently,
\begin{equation}
\left(\alpha_{n+1}^{\dagger}\middle|\sigma_{n}\right)=\frac{1}{d-n}.\label{eq:sigman}
\end{equation}
Hence, by proposition~\ref{prop:pstar=00003Dpstar}, the maximum
eigenvalue of $\sigma_{n}$ is greater than or equal to $\frac{1}{d-n}$.
In fact, it must be equal to $\frac{1}{d-n}$, because otherwise the
corresponding eigenstate $\alpha$ would lead to the contradiction,
recalling eq.~\eqref{eq:inductionbasis}:
\[
\frac{1}{d}=\left(\alpha^{\dagger}\middle|\chi\right)\ge\frac{d-n}{d}\left(\alpha^{\dagger}\middle|\sigma_{n}\right)>\frac{1}{d}.
\]
Hence, eq.~\eqref{eq:sigman} and corollary~\ref{cor:pupstar} imply
that $\alpha_{n+1}$ is an eigenstate of $\sigma_{n}$ with maximum
eigenvalue. Therefore, $\sigma_{n}$ can be decomposed as $\sigma_{n}=\frac{1}{d-n}\alpha_{n+1}+\frac{d-n-1}{d-n}\sigma_{n+1}$,
where the states $\alpha_{n+1}$ and $\sigma_{n+1}$ are perfectly
distinguishable. Inserting this relation into eq.~\eqref{eq:inductionbasis}
we obtain 
\[
\chi=\frac{1}{d}\left(\sum_{i=1}^{n+1}\alpha_{i}\right)+\frac{d-n-1}{d}\sigma_{n+1}.
\]
Now, since the states $\left\{ \alpha_{i}\right\} _{i=1}^{n}\cup\left\{ \sigma_{n}\right\} $
are perfectly distinguishable, so are the states $\left\{ \alpha_{i}\right\} _{i=1}^{n+1}\cup\left\{ \sigma_{n+1}\right\} $.
This proves the validity of eq.~\eqref{eq:inductionbasis} for every
$n\leq r$. To conclude the proof, consider eq.~\eqref{eq:inductionbasis}
for $n=r$. The condition that set $\left\{ \alpha_{i}\right\} _{i=1}^{r}$
is maximal implies that the state $\sigma_{r}$ should not arise in
the decomposition, otherwise the pure states contained in $\sigma_{r}$
would be perfectly distinguishable from the pure states $\left\{ \alpha_{i}\right\} _{i=1}^{r}$,
contradicting maximality. This is possible only if the corresponding
probability is zero, namely only if one has $r=d$.
\end{proof}
In summary, the above proposition guarantees that all pure maximal
sets of a system have the same cardinality, equal to $d$, and this
is why we called $d$ the dimension of the system. Moreover, any set
of $d$ perfectly distinguishable pure states is guaranteed to be
maximal. As a consequence, every state can have at most $d$ terms
in its diagonalisations. Then, clearly every diagonalisation $\rho=\sum_{i=1}^{r}p_{i}\alpha_{i}$,
where $p_{i}>0$ for all $i\in\left\{ 1,\ldots,r\right\} $, and $r\leq d$,
can be rewritten as $\rho=\sum_{i=1}^{d}p_{i}\alpha_{i}$ by completing
$\left\{ \alpha_{i}\right\} _{i=1}^{r}$ to a maximal set $\left\{ \alpha_{i}\right\} _{i=1}^{d}$,
and taking some of the eigenvalues to be zero. This means that the
spectrum of a state can always be taken to be a vector $\mathbf{p}$
with $d$ entries. This will play an important role in the next chapter. 

In other works \cite{Barnum-interference,Scandolo-thesis,QPL15,Krumm-Muller,Colleagues},
the above result and other properties of diagonalisations were derived
from the Strong Symmetry axiom \cite{Muller-self-duality,Barnum-interference},
stating that all pure maximal sets are connected by reversible channels
(see also subsection~\ref{subsec:Unrestricted-reversibility} for
its thermodynamic implications in sharp theories with purification).
Our result shows that the properties of diagonalisation can be derived
from a very different set of axiom: Causality, Purity Preservation,
Pure Sharpness, and Purification, and do not need Strong Symmetry,
unlike in our previous works \cite{Scandolo-thesis,QPL15}.

The diagonalisation of the invariant state induces a one-to-one correspondence
between maximal sets of perfectly distinguishable pure states and
\emph{pure sharp measurements} \cite{Chiribella-Yuan2014,Chiribella-Yuan2015},
which can be characterised as follows. 
\begin{defn}
An observation-test $\left\{ a_{i}\right\} _{i=1}^{n}$ is a \emph{pure
sharp measurement} if every effect $a_{i}$ is pure and normalised.
\end{defn}

Under the validity of our axioms, every pure sharp measurement can
be written as $\left\{ \alpha_{i}^{\dagger}\right\} _{i=1}^{n}$,
for some set of pure states $\left\{ \alpha_{i}\right\} _{i=1}^{n}$
(cf.\ theorem~\ref{thm:duality}). 
\begin{prop}
\label{prop:pure test chi}For every pure maximal set $\left\{ \alpha_{i}\right\} _{i=1}^{d}$,
the effects $\left\{ \alpha_{i}^{\dagger}\right\} _{i=1}^{d}$ form
a pure sharp measurement. Conversely, for every pure sharp measurement
$\left\{ \alpha_{i}^{\dagger}\right\} _{i=1}^{n}$, the states $\left\{ \alpha_{i}\right\} _{i=1}^{n}$
form a pure maximal set, and therefore $n=d$.
\end{prop}

\begin{proof}
Let $\left\{ \alpha_{i}\right\} _{i=1}^{d}$ be a pure maximal set.
By proposition~\ref{prop:diagonalization chi d-level 2}, we know
that $\frac{1}{d}\sum_{i=1}^{d}\alpha_{i}$ is a diagonalisation of
the invariant state $\chi$. Then, proposition~\ref{prop:maximal}
implies that $\left\{ \alpha_{i}^{\dagger}\right\} _{i=1}^{d}$ is
a pure sharp measurement.

Conversely, suppose that $\left\{ a_{i}\right\} _{i=1}^{n}$ is a
pure sharp measurement, then $a_{i}=\alpha_{i}^{\dagger}$ for some
state $\alpha_{i}$. By corollary~\ref{cor:conjecture proved}, we
know that $\left(\alpha_{i}^{\dagger}\middle|\alpha_{j}\right)=\delta_{ij}$,
and moreover $\left\{ a_{i}\right\} _{i=1}^{n}$ is an observation-test,
thus $\left\{ a_{i}\right\} _{i=1}^{n}$ distinguishes the states
$\left\{ \alpha_{i}\right\} _{i=1}^{n}$ perfectly. The states $\left\{ \alpha_{i}\right\} _{i=1}^{n}$
must form a pure maximal set. This can be proved by contradiction:
suppose the set $\left\{ \alpha_{i}\right\} _{i=1}^{n}$ is not maximal,
and extend it to a maximal set $\left\{ \alpha_{i}\right\} _{i=1}^{d}$.
Then, by the first part of this proof we have that $\left\{ \alpha_{i}^{\dagger}\right\} _{i=1}^{d}$
is an observation-test. By Causality, we then obtain 
\[
\sum_{i=1}^{d}\alpha_{i}^{\dagger}=u=\sum_{i=1}^{n}a_{i}=\sum_{i=1}^{n}\alpha_{i}^{\dagger},
\]
having used the equality $a_{i}=\alpha_{i}^{\dagger}$. In conclusion,
we have obtained the relation $\sum_{i=n+1}^{d}\alpha_{i}^{\dagger}=0$,
which can be satisfied only if $n=d$. Hence, the states $\left\{ \alpha_{i}\right\} _{i=1}^{d}$
form a pure maximal set.
\end{proof}
As a consequence, the product of two pure maximal sets is a pure maximal
set for the composite system. This property was called ``information
locality'' by Hardy \cite{Hardy-informational-2,hardy2013}, and
it has been recently shown, along with a weaker version of Purity
Preservation, to play a major role in the emergence of local classical
observers in GPTs \cite{Objectivity}. In words, the dimension of
a composite system $\mathrm{AB}$ is the product of the dimension
of the components: $d_{\mathrm{AB}}=d_{\mathrm{A}}d_{\mathrm{B}}$.
\begin{prop}[Information locality]
\label{prop:information locality}If $\left\{ \alpha_{i}\right\} _{i=1}^{d_{\mathrm{A}}}$
is a pure maximal set for system $\mathrm{A}$ and $\left\{ \beta_{j}\right\} _{j=1}^{d_{\mathrm{B}}}$
is a pure maximal set for system $\mathrm{B}$, then 
\[
\left\{ \alpha_{i}\otimes\beta_{j}\right\} _{i\in\left\{ 1,\dots,d_{\mathrm{A}}\right\} ,j\in\left\{ 1,\dots,d_{\mathrm{B}}\right\} }
\]
is a pure maximal set for the composite system $\mathrm{AB}$.
\end{prop}

\begin{proof}
By proposition~\ref{prop:pure test chi}, $\left\{ \alpha_{i}^{\dagger}\right\} _{i=1}^{d_{\mathrm{A}}}$
and $\left\{ \beta_{j}^{\dagger}\right\} _{j=1}^{d_{\mathrm{B}}}$
are two observation-tests for systems $\mathrm{A}$ and $\mathrm{B}$,
respectively. Now, the product of two observation-tests is an observation-test
(physically, corresponding to two measurements performed in parallel).
Hence, the product $\left\{ \alpha_{i}^{\dagger}\otimes\beta_{j}^{\dagger}\right\} _{i\in\left\{ 1,\dots,d_{\mathrm{A}}\right\} ,j\in\left\{ 1,\dots,d_{\mathrm{B}}\right\} }$
is an observation-test on the composite system $\mathrm{AB}$. Moreover,
each effect $\alpha_{i}^{\dagger}\otimes\beta_{j}^{\dagger}$ is pure,
due to Purity Preservation, and normalised. Using proposition~\ref{prop:pure test chi}
again, we obtain that $\left\{ \alpha_{i}\otimes\beta_{j}\right\} _{i\in\left\{ 1,\dots,d_{\mathrm{A}}\right\} ,j\in\left\{ 1,\dots,d_{\mathrm{B}}\right\} }$
is a pure maximal set. 
\end{proof}

\subsubsection{Diagonalisation of internal states}

As a side remark, we show here that internal states have exactly $d$
non-zero eigenvalues. In the quantum case, this amounts to saying
that internal states are full-rank density matrices. 
\begin{prop}
Every internal state $\omega$ has precisely $d$ non-vanishing eigenvalues
in every diagonalisation.
\end{prop}

\begin{proof}
Consider a complete state $\omega$ and one of its diagonalisations
$\omega=\sum_{i=1}^{r}p_{i}\alpha_{i}$, where $r\leq d$, and the
$p_{i}$'s are non-vanishing, for all $i\in\left\lbrace 1,\ldots,r\right\rbrace $.
Suppose by contradiction that $r<d$; this means that the states $\left\lbrace \alpha_{i}\right\rbrace _{i=1}^{r}$
do not form a pure maximal set, and therefore we can complete it by
adding $d-r$ states $\left\lbrace \alpha_{i}\right\rbrace _{i=r+1}^{d}$.
In this way we can rewrite the diagonalisation of $\omega$ as $\omega=\sum_{i=1}^{d}p_{i}\alpha_{i}$,
where $p_{i}=0$ for $i\in\left\lbrace r+1,\ldots,d\right\rbrace $,
and the states $\left\lbrace \alpha_{i}\right\rbrace _{i=1}^{d}$
are a pure maximal set. Take any $\alpha_{i}$ with $i\in\left\lbrace r+1,\ldots,d\right\rbrace $;
we have 
\[
0=p_{i}=\left(\alpha_{i}^{\dagger}\middle|\omega\right).
\]
On the other hand, $\omega$ is internal, therefore $\left(\alpha_{i}^{\dagger}\middle|\rho\right)=0$
for all states $\rho\in\mathsf{St}_{1}\left(\mathrm{A}\right)$, and
$i\in\left\lbrace r+1,\ldots,d\right\rbrace $. Hence $\left(\alpha_{i}^{\dagger}\middle|\alpha_{i}\right)=0$,
which is a contradiction. We conclude that $r=d$.
\end{proof}
Consequently, the pure states arising in every diagonalisation of
$\omega$ form a maximal set.

The converse also holds, and the proof has already appeared in \cite[corollary 19]{Chiribella-informational}
(and does not make use of the stronger axioms assumed therein).
\begin{prop}
Let $\left\lbrace \alpha_{i}\right\rbrace _{i=1}^{d}$ be a maximal
set of perfectly distinguishable pure states. Every convex combination
of the $\alpha_{i}$'s with all \emph{non-zero} coefficients yields
a complete state.
\end{prop}

\begin{proof}
Let $\omega$ be a mixture of the pure states $\left\lbrace \alpha_{i}\right\rbrace _{i=1}^{d}$
with $d$ non-zero probabilities $\omega=\sum_{i=1}^{d}p_{i}\alpha_{i}$,
where $p_{i}>0$. Consider the minimum eigenvalue $p_{\min}=\min_{i}\left\{ p_{i}\right\} $.
Then we can write $\omega=p_{\min}\chi+\left(1-p_{\min}\right)\sigma$,
where $\sigma$ is defined as
\[
\sigma:=\frac{1}{1-p_{\min}}\sum_{i=1}^{d}\left(p_{i}-\frac{p_{\min}}{d}\right)\alpha_{i},
\]
and it is well-defined because $p_{i}\geq\frac{p_{\min}}{d}$. Since
$\chi$ is contained in $\rho$, and $\chi$ is internal, we conclude
that $\rho$ is internal too.
\end{proof}

\subsubsection{Double stochasticity of the transition matrices}

Given two pure maximal sets, $\left\{ \alpha_{i}\right\} _{i=1}^{d}$
and $\left\{ \alpha'_{i}\right\} _{i=1}^{d}$, we call the matrix
$T_{ij}=\left(\alpha_{i}^{\dagger}\middle|\alpha'_{j}\right)$ a \emph{transition
matrix}. With this definition, we have the following result \cite[lemma 4]{QPL15}.
\begin{lem}
\label{lem:doubly-stochastic}In sharp theories with purification
all transition matrices are doubly stochastic\footnote{Recall that a doubly stochastic matrix is a matrix with non-negative
entries, in which every row and every column sum to 1.}.
\end{lem}

We do not report the proof of this lemma, since in section~\ref{sec:Majorisation-and-unital}
we will prove a stronger result (lemma~\ref{lem:channelmatrix}),
which will imply lemma~\ref{lem:doubly-stochastic}. Moreover, the
proof of lemma~\ref{lem:channelmatrix} will be virtually identical
to the proof of lemma~\ref{lem:doubly-stochastic}.

\subsection{Uniqueness of the diagonalisation\label{subsec:Uniqueness-of-the}}

Thanks to our axioms, the diagonalisation of a state is unique, up
to the obvious freedom arising in the presence of degeneracy among
the eigenvalues. This is a non-trivial consequence of the axioms:
notably \cite{Krumm-Muller,Krumm-thesis} exhibited examples of GPTs
where states can be diagonalised, but the same state can have more
than one diagonalisation, with different spectra.

To take degeneracy into account, given a diagonalisation $\rho=\sum_{i=1}^{r}p_{i}\alpha_{i}$
of $\rho$, we define the \emph{reduced spectrum} of $\rho$, as the
set of the \emph{distinct} eigenvalues of $\rho$, ordered in strictly
decreasing order $\lambda_{1}>\lambda_{2}>\ldots>\lambda_{s}>0$,
and we rewrite the diagonalisation as 
\[
\rho=\sum_{k=1}^{s}\lambda_{k}\Pi_{k},
\]
where
\[
\Pi_{k}:=\sum_{i:p_{i}=\lambda_{k}}\alpha_{i},
\]
and the sum is over the $\alpha_{i}$'s arising in the given diagonalisation
of $\rho$ whose eigenvalue is $\lambda_{k}$. When expressed in this
form, the diagonalisation is unique. Now we present the main theorem.
\begin{thm}
\label{thm:uniqueness diago}Let $\rho=\sum_{k=1}^{s}\lambda_{k}\Pi_{k}$
and $\rho=\sum_{l=1}^{s'}\lambda'_{l}\Pi'_{l}$ be two diagonalisations
of the same state. Then, one has $s=s'$, $\lambda_{k}=\lambda'_{k}$,
$\Pi_{k}=\Pi'_{k}$, for all $k\in\left\{ 1,\dots,s\right\} $.
\end{thm}

\begin{proof}
Let the two diagonalisations be $\rho=\sum_{i}p_{i}\alpha_{i}$ and
$\rho=\sum_{j}p'_{j}\alpha'_{j}$. First of all, let us prove that
$\lambda_{1}=\lambda'_{1}=p_{*}$. This is a non-trivial statement
to prove. Indeed, our diagonalisation algorithm of theorem~\ref{thm:diago}
outputs the first eigenvalue to be the maximum eigenvalue $p_{*}$,
but there might exist other diagonalisation algorithms yielding different
eigenvalues, none of which equal to $p_{*}$. Let us define the degeneracies
$d_{1}=\left|\left\{ i:p_{i}=\lambda_{1}\right\} \right|$ and $d'_{1}=\left|\left\{ j:p'_{j}=\lambda'_{1}\right\} \right|$,
and assume $d_{1}\ge d_{1}'$ without loss of generality. By definition,
we have for $i\in\left\{ 1,\ldots,d_{1}\right\} $
\[
\lambda_{1}=\left(\alpha_{i}^{\dagger}\middle|\rho\right)=\sum_{j}p'_{j}\left(\alpha_{i}^{\dagger}\middle|\alpha_{j}'\right)=\sum_{j}T_{ij}p'_{j}\le\lambda'_{1},
\]
having used the fact that the transition matrix $T_{ij}=\left(\alpha_{i}^{\dagger}\middle|\alpha_{j}'\right)$
is doubly stochastic for $i,j\in\left\{ 1,\ldots,d\right\} $, with
$d$ the dimension of the system. By a similar argument, with $\lambda_{1}$
and $\lambda'_{1}$ interchanged, finally we get the equality $\lambda_{1}=\lambda'_{1}$.
Since this applies to all diagonalisations of $\rho$, including those
obtained with the algorithm of theorem~\ref{thm:diago}, for which
$\lambda_{1}=p_{*}$, we conclude that $\lambda_{1}=\lambda_{1}'=p_{*}$.

The above relation implies the equality $\sum_{j=1}^{d_{1}'}\left(\alpha_{i}^{\dagger}\middle|\alpha_{j}'\right)=1$,
for every $i\in\left\{ 1,\dots,d_{1}\right\} $ or, equivalently,
$\left(\alpha_{i}^{\dagger}|\chi_{1}'\right)=\frac{1}{d_{1}'}$ for
all $i\in\left\{ 1,\dots,d_{1}\right\} $, where $\chi_{1}':=\frac{1}{d_{1}'}\Pi'_{1}$.
Note that $\frac{1}{d_{1}'}$ is the maximum eigenvalue of $\chi_{1}'$
because the states $\left\lbrace \alpha'_{j}\right\rbrace _{j=1}^{d'_{1}}$
are perfectly distinguishable, and we have just proved that the maximum
eigenvalue of a state arises in every diagonalisation. Thus corollary~\ref{cor:pupstar}
implies that $\alpha_{i}$ is an eigenstate with maximum eigenvalue.
In particular, choosing $i=1$ we obtain the decomposition
\[
\chi_{1}'=\frac{1}{d_{1}'}\alpha_{1}+\frac{d_{1}'-1}{d_{1}'}\sigma_{1},
\]
where $\sigma_{1}$ is a suitable state, perfectly distinguishable
from $\alpha_{1}$. We are now in the position to repeat the argument
in the proof of proposition~\ref{prop:diagonalization chi d-level 2}
for the states $\left\lbrace \alpha_{i}\right\rbrace _{i=1}^{d_{1}}$,
to find that $d_{1}=d_{1}'$ and 
\[
\chi_{1}'=\frac{1}{d_{1}}\sum_{i=1}^{d_{1}}\alpha_{1}=\frac{1}{d_{1}}\Pi_{1}.
\]
Hence, we proved the equality $\Pi_{1}'=\Pi_{1}$. We can now define
the state 
\[
\rho_{2}:=\frac{1}{1-d_{1}\lambda_{1}}\left(\rho-\lambda_{1}\Pi_{1}\right)=\frac{1}{1-d_{1}\lambda_{1}}\left(\sum_{k=2}^{s}\lambda_{k}\Pi_{k}\right)=\frac{1}{1-d_{1}\lambda_{1}}\left(\sum_{l=2}^{s'}\lambda'_{l}\Pi_{l}'\right).
\]
Repeating the above argument, we can prove the equalities $\lambda_{2}=\lambda_{2}'$
and $\Pi_{2}=\Pi_{2}'$. Once all distinct eigenvalues have been scanned,
the normalisation of the probability distribution implies the condition
$s=s'$.
\end{proof}
A very close result was proved by Wilce in the framework of probabilistic
models with conjugates and Jordan algebras \cite{Royal-road-0,Royal-road}. 

Theorem~\ref{thm:uniqueness diago} shows that the diagonalisation
is unique up to the choice of the eigenstates when we have degeneracy:
only then do we have the freedom of choice of the eigenstates relative
to degenerate eigenvalues, i.e.\ eigenvalues arising more than once
in a diagonalisation. See \cite{QPL15} for another proof of the uniqueness
of the eigenvalues based on majorisation (and also a further axiom).

Theorem~\ref{thm:uniqueness diago} implies that we can associate
a unique probability distribution with every state: its spectrum.
The spectrum of a state will be the basis to define entropies and
resource monotones for the thermodynamics of isolated systems.

\subsection{Extending the diagonalisation to arbitrary vectors\label{subsec:Extending-the-diagonalisation}}

The diagonalisation theorem, proved for normalised states, can be
easily extended to arbitrary elements of the vector space $\mathsf{St}_{\mathbb{R}}\left(\mathrm{A}\right)$.
\begin{prop}
For every system $\mathrm{A}$ and for every vector $\xi\in\mathsf{St}_{\mathbb{R}}\left(\mathrm{A}\right)$
there exist a unique set of $d$ real numbers $\left\{ x_{i}\right\} _{i=1}^{d}$
and a maximal set of perfectly distinguishable states $\left\{ \alpha_{i}\right\} _{i=1}^{d}$
such that 
\[
\xi=\sum_{i=1}^{d}x_{i}\alpha_{i}.
\]
\end{prop}

We omit the proof, which is the same as in \cite[corollary 21]{Chiribella-informational}.
Again the $x_{i}$'s are called the eigenvalues of $\xi$, and the
$\alpha_{i}$'s are called the eigenstates of $\xi$. A similar result
was obtained also in \cite{Krumm-Muller} under different axioms.

Note that, since $\xi$ is a generic vector of $\mathsf{St}_{\mathbb{R}}\left(\mathrm{A}\right)$,
the eigenvalues are arbitrary real numbers; if instead $\xi$ is in
the cone $\mathsf{St}_{+}\left(\mathrm{A}\right)$, the eigenvalues
are non-negative real numbers.

Using this result, we can prove that the operational norm of a vector
in $\mathsf{St}_{\mathbb{R}}\left(\mathrm{A}\right)$ coincides with
the 1-norm of its spectrum $\mathbf{x}$ \cite{HOI}.
\begin{prop}
Let $\xi\in\mathsf{St}_{\mathbb{R}}\left(\mathrm{A}\right)$ be diagonalised
as $\xi=\sum_{i=1}^{d}x_{i}\alpha_{i}$. Then $\left\Vert \xi\right\Vert =\sum_{i=1}^{d}\left|x_{i}\right|$.
\end{prop}

\begin{proof}
Let us separate the terms with non-negative eigenvalues from the terms
with negative eigenvalues, so that we can write $\xi=\xi_{+}-\xi_{-}$,
where $\xi_{+}:=\sum_{x_{i}\geq0}x_{i}\alpha_{i}$, and $\xi_{-}=\sum_{x_{i}<0}\left(-x_{i}\right)\alpha_{i}$.
Clearly, $\xi_{+},\xi_{-}\in\mathsf{St}_{+}\left(\mathrm{A}\right)$.
Recall the definition of the operational norm of a vector:
\[
\left\Vert \xi\right\Vert =\sup_{a\in\mathsf{Eff}\left(\mathrm{A}\right)}\left(a\middle|\xi\right)-\inf_{a\in\mathsf{Eff}\left(\mathrm{A}\right)}\left(a\middle|\xi\right).
\]
In order to achieve the supremum of $\left(a\middle|\xi\right)$ we
must have $\left(a\middle|\xi_{-}\right)=0$. Moreover, 
\[
\left(a\middle|\xi_{+}\right)=\sum_{x_{i}\geq0}x_{i}\left(a\middle|\alpha_{i}\right)\leq\sum_{x_{i}\geq0}x_{i}
\]
since $\left(a\middle|\alpha_{i}\right)\leq1$ for every $i$. The
supremum of $\left(a\middle|\xi_{+}\right)$ is achieved by $a=\sum_{x_{i}\geq0}\alpha_{i}^{\dagger}$.
Hence $\sup_{a}\left(a\middle|\xi\right)=\sum_{x_{i}\geq0}x_{i}$.
By a similar argument, one shows that $\inf_{a}\left(a\middle|\xi\right)=\sum_{x_{i}<0}x_{i}$.
Therefore
\[
\left\Vert \xi\right\Vert =\sum_{x_{i}\geq0}x_{i}+\sum_{x_{i}<0}\left(-x_{i}\right)=\sum_{i=1}^{d}\left|x_{i}\right|.
\]
\end{proof}
For $p\geq1$, the $p$-norm of a vector $\mathbf{x}\in\mathbb{R}^{d}$
is defined as $\left\Vert \mathbf{x}\right\Vert _{p}:=\left(\sum_{i=1}^{d}\left|x_{i}\right|^{p}\right)^{\frac{1}{p}}$,
thus we have $\left\Vert \xi\right\Vert =\left\Vert \mathbf{x}\right\Vert _{1}$,
where $\mathbf{x}$ is the spectrum of $\xi$.

\subsection{Extending the dagger map}

Thanks to the diagonalisation theorems, the dagger map $\dagger:\mathsf{PurSt}_{1}\left(\mathrm{A}\right)\to\mathsf{PurEff}_{1}\left(\mathrm{A}\right)$
can be extended to arbitrary vectors via the relation 
\[
\xi=\sum_{i=1}^{d}x_{i}\alpha_{i}\quad\longmapsto\quad\xi^{\dagger}:=\sum_{i=1}^{d}x_{i}\alpha_{i}^{\dagger}.
\]
Note that, since the diagonalisation is unique (up to degeneracy),
the vector $\xi^{\dagger}$ is well-defined, i.e.\ it does not depend
on the choice of the $\alpha_{i}$'s as long as they are eigenstates
of $\xi$. Writing $\xi$ like in theorem~\ref{thm:uniqueness diago},
$\xi=\sum_{k=1}^{s}\lambda_{k}\Pi_{k}$, to prove this fact, it suffices
to show the following.
\begin{prop}
\label{prop:well-defined}Let $\left\{ \alpha_{i}\right\} _{i=1}^{r}$
and $\left\{ \alpha_{j}'\right\} _{j=1}^{r}$ be two sets of perfectly
distinguishable pure states. Then, if $\sum_{i=1}^{r}\alpha_{i}=\sum_{j=1}^{r}\alpha_{j}'$,
one has $\sum_{i=1}^{r}\alpha_{i}^{\dagger}=\sum_{j=1}^{r}\alpha_{j}^{\prime\dagger}$.
\end{prop}

\begin{proof}
Let us extend $\left\{ \alpha_{i}\right\} _{i=1}^{r}$ and $\left\{ \alpha_{j}'\right\} _{j=1}^{r}$
to two pure maximal sets $\left\{ \alpha_{i}\right\} _{i=1}^{d}$
and $\left\{ \alpha_{j}'\right\} _{j=1}^{d}$. Then, the invariant
state has the two diagonalisations $\chi=\frac{1}{d}\sum_{i=1}^{d}\alpha_{i}$
and $\chi=\frac{1}{d}\sum_{j=1}^{d}\alpha'_{j}$ (proposition~\ref{prop:diagonalization chi d-level 2}).
Using this fact and the condition $\sum_{i=1}^{r}\alpha_{i}=\sum_{j=1}^{r}\alpha_{j}'$,
we obtain $\sum_{i=r+1}^{d}\alpha_{i}=\sum_{j=r+1}^{d}\alpha_{j}'$.
Hence, the invariant state can be decomposed as 
\[
\chi=\frac{1}{d}\left(\sum_{j=1}^{r}\alpha'_{j}+\sum_{i=r+1}^{d}\alpha_{i}\right).
\]
By proposition~\ref{prop:maximal}, this implies that the states
$\left\{ \alpha'_{j}\right\} _{j=1}^{r}\cup\left\{ \alpha_{i}\right\} _{i=r+1}^{d}$
form a maximal set. Now, the correspondence between maximal sets and
pure sharp measurements (proposition~\ref{prop:pure test chi}) implies
that the effects $\left\{ \alpha_{j}^{\prime\dagger}\right\} _{j=1}^{r}\cup\left\{ \alpha_{i}^{\dagger}\right\} _{i=r+1}^{d}$
form a measurement. Causality yields 
\[
\sum_{j=1}^{r}\alpha_{j}^{\prime\dagger}+\sum_{i=r+1}^{d}\alpha_{i}^{\dagger}=u.
\]
On the other hand, the normalisation of the measurement $\left\{ \alpha_{i}^{\dagger}\right\} _{i=1}^{d}$
reads 
\[
\sum_{i=1}^{d}\alpha_{i}^{\dagger}=u.
\]
Comparing the two equalities we obtain the desired relation $\sum_{i=1}^{r}\alpha_{i}^{\dagger}=\sum_{j=1}^{r}\alpha_{j}^{\prime\dagger}$.
\end{proof}
Similarly to what we did in section~\ref{sec:State-effect-duality},
with a little abuse of notation we will denote as $\dagger$ even
the inverse map, from $\mathsf{Eff}_{\mathbb{R}}\left(\mathrm{A}\right)$
to $\mathsf{St}_{\mathbb{R}}\left(\mathrm{A}\right)$. Now we are
ready to define observables, and to introduce a functional calculus
on them.

\subsection{Functional calculus on observables\label{subsec:Functional-calculus-on}}

Using Steering, one can convert the diagonalisation result for the
elements of $\mathsf{St}_{\mathbb{R}}\left(\mathrm{A}\right)$ into
a diagonalisation result for the elements of $\mathsf{Eff}_{\mathbb{R}}\left(\mathrm{A}\right)$
(see \cite{Chiribella-informational,Chiribella}, and, for a different
approach, \cite{Wilce-spectral,Colleagues,Royal-road-0,Royal-road}).
\begin{prop}
For every finite-dimensional system $\mathrm{A}$ and for every vector
$X\in\mathsf{Eff}_{\mathbb{R}}\left(\mathrm{A}\right)$ there exist
a set of $d$ real numbers $\left\{ x_{i}\right\} _{i=1}^{d}$, and
a pure maximal set of states $\left\{ \alpha_{i}\right\} _{i=1}^{d}$
such that
\[
X=\sum_{i=1}^{d}x_{i}\alpha_{i}^{\dagger}.
\]
\end{prop}

This result allows us to give a concrete characterisation of the norm
on the vector space of effects, along the same lines as we did for
the norm on the vector space of states.
\begin{prop}
Let $X\in\mathsf{Eff}_{\mathbb{R}}\left(\mathrm{A}\right)$ be diagonalised
as $X=\sum_{i=1}^{d}x_{i}\alpha_{i}^{\dagger}$. Then $\left\Vert X\right\Vert =\max_{i}\left|x_{i}\right|$.
\end{prop}

\begin{proof}
Recall the definition of the norm of a vector in $\mathsf{Eff}_{\mathbb{R}}\left(\mathrm{A}\right)$:
\[
\left\Vert X\right\Vert =\sup_{\rho\in\mathsf{St}\left(\mathrm{A}\right)}\left|\left(X\middle|\rho\right)\right|.
\]
Now, we have
\[
\left|\left(X\middle|\rho\right)\right|=\left|\sum_{i=1}^{d}x_{i}\left(\alpha_{i}^{\dagger}\middle|\rho\right)\right|\leq\sum_{i=1}^{d}\left|x_{i}\right|\left(\alpha_{i}^{\dagger}\middle|\rho\right).
\]
This is a sub-normalised ``convex'' combination of the $\left|x_{i}\right|$'s.
Therefore
\[
\sum_{i=1}^{d}\left|x_{i}\right|\left(\alpha_{i}^{\dagger}\middle|\rho\right)\leq\max_{i}\left|x_{i}\right|.
\]
In conclusion $\left|\left(X\middle|\rho\right)\right|\leq\max_{i}\left|x_{i}\right|$.
The bound is achieved when $\rho$ is the state $\alpha_{i_{0}}$
associated with $\max_{i}\left|x_{i}\right|=:x_{i_{0}}$. Therefore,
\[
\left\Vert X\right\Vert =\sup_{\rho\in\mathsf{St}\left(\mathrm{A}\right)}\left|\left(X\middle|\rho\right)\right|=\max_{i}\left|x_{i}\right|.
\]
\end{proof}
In other words, the norm of a vector $X\in\mathsf{Eff}_{\mathbb{R}}\left(\mathrm{A}\right)$
is $\left\Vert X\right\Vert =\left\Vert \mathbf{x}\right\Vert _{\infty}$,
where $\mathbf{x}$ is the spectrum of $X$, and $\left\Vert \mathbf{x}\right\Vert _{\infty}=\lim_{p\rightarrow+\infty}\left\Vert \mathbf{x}\right\Vert _{p}=\max_{i}\left|x_{i}\right|$.

Thanks to the diagonalisation theorem, all the elements of the vector
space $\mathsf{Eff}_{\mathbb{R}}\left(\mathrm{A}\right)$ can be regarded
as \emph{observables}, in a similar sense to the use of the term in
quantum theory. Indeed, given a diagonalisation $X=\sum_{i=1}^{d}x_{i}\alpha_{i}^{\dagger}$
one can think of the eigenvalues as the ``values'' associated with
the outcomes of the sharp measurement $\left\{ \alpha_{i}^{\dagger}\right\} _{i=1}^{d}$.
In this way, one can interpret
\[
\left\langle X\right\rangle _{\rho}:=\left(X\middle|\rho\right)=\sum_{i=1}^{d}x_{i}\left(\alpha_{i}^{\dagger}\middle|\rho\right)
\]
as the \emph{expectation value} of the observable $X$ on $\rho$,
because $\left(\alpha_{i}^{\dagger}\middle|\rho\right)$ are probabilities. 

Like in quantum theory, the spectral theorem allows one to define
a functional calculus on observables:\footnote{See also \cite{Colleagues,Royal-road-0,Royal-road} for a different
approach.} given an observable $X$ and a function $f:\mathbb{R}\to\mathbb{R}$,
one can define the observable 
\[
f\left(X\right):=\sum_{i=1}^{d}f\left(x_{i}\right)\alpha_{i}^{\dagger}.
\]
Note that the observable $f\left(X\right)$ is well-defined, because
the eigenvalues of $X$ are unique, and as a consequence of proposition~\ref{prop:well-defined}.
In particular, one can choose the observable $X$ to be the dagger
of a state $\rho=\sum_{i=1}^{d}p_{i}\alpha_{i}$, thus obtaining 
\[
f\left(\rho^{\dagger}\right)=\sum_{i=1}^{d}f\left(p_{i}\right)\alpha_{i}^{\dagger}.
\]
In the following chapter, we will take $f$ to be the logarithm function
on some base greater than 1, defining the ``surprisal observable'':
\[
-\log_{a}\rho^{\dagger}=-\sum_{i=1}^{d}\left(\log_{a}p_{i}\right)\alpha_{i}^{\dagger}.
\]
Therefore, Shannon-von Neumann entropy can be defined as the expectation
value of the surprisal observable
\[
S\left(\rho\right):=\left(-\log_{a}\rho^{\dagger}\middle|\rho\right).
\]
The functional calculus on observables will also be the basis to define
the generalised relative entropy in section~\ref{sec:Properties-of-Shannon-von}.

\section{Schmidt decomposition}

Using the diagonalisation theorem we can prove an operational version
of the Schmidt decomposition of pure bipartite states \cite{Nielsen-Chuang,Preskill}.
The intuitive content of the Schmidt decomposition is that for every
state of a bipartite system there exist two perfectly correlated pure
observation-tests on the component systems, a similar situation to
having conjugates \cite{Wilce-spectral,Royal-road-0,Royal-road}.
More formally, this property is stated by the following theorem.
\begin{thm}[Schmidt decomposition]
\label{thm:schmidt}Let $\Psi$ be a pure state of the composite
system $\mathrm{AB}$. Then, there exist a pure sharp measurement
$\left\{ a_{i}\right\} _{i=1}^{d_{\mathrm{A}}}$ on system $\mathrm{A}$,
and a pure sharp measurement $\left\{ b_{j}\right\} _{j=1}^{d_{\mathrm{B}}}$
on system $\mathrm{B}$ such that \begin{equation}\label{eq:schmidt} \begin{aligned}\Qcircuit @C=1em @R=.7em @!R { & \multiprepareC{1}{\Psi} & \qw \poloFantasmaCn{\rA} & \measureD{a_i} \\ & \pureghost{\Psi} & \qw \poloFantasmaCn{\rB} & \measureD{b_j} }\end{aligned}~=p_{i}\delta_{ij}\qquad \forall i \in \left\{1,\ldots, r\right\}, \end{equation}where
$r\le\min\left\{ d_{\mathrm{A}},d_{\mathrm{B}}\right\} $ is a suitable
integer, the \emph{Schmidt rank}, $\left\{ p_{i}\right\} _{i=1}^{r}$
is a probability distribution, with all non-vanishing elements.

Moreover, one has the diagonalisations $\rho_{\mathrm{A}}=\sum_{i=1}^{r}p_{i}a_{i}^{\dagger}$
and $\rho_{\mathrm{B}}=\sum_{i=1}^{r}p_{i}b_{i}^{\dagger}$, where
$\rho_{\mathrm{A}}$ and $\rho_{\mathrm{B}}$ are the marginals of
$\Psi$ on systems $\mathrm{A}$ and $\mathrm{B}$ respectively.
\end{thm}

\begin{proof}
Let $\rho_{\mathrm{A}}$ be the marginal of $\Psi$ on system $\mathrm{A}$
and let $\rho_{\mathrm{A}}=\sum_{i=1}^{r}p_{i}\alpha_{i}$ be a diagonalisation
of $\rho_{\mathrm{A}}$, where $p_{i}>0$ for all $i\in\left\lbrace 1,\ldots,r\right\rbrace $.
By Pure Steering, there exists an observation-test on $\mathrm{B}$,
call it $\left\{ \widetilde{b}_{i}\right\} _{i=1}^{r}$, such that\[ \begin{aligned}\Qcircuit @C=1em @R=.7em @!R { & \multiprepareC{1}{\Psi} & \qw \poloFantasmaCn{\rA} & \qw \\ & \pureghost{\Psi} & \qw \poloFantasmaCn{\rB} & \measureD{\widetilde {b}_i} }\end{aligned}~=p_{i}\!\!\!\!\begin{aligned}\Qcircuit @C=1em @R=.7em @!R { & \prepareC{\alpha_i} & \qw \poloFantasmaCn{\rA} & \qw }\end{aligned}~, \]
for every $i\in\left\{ 1,\ldots,r\right\} $. On the other hand, by
corollary~\ref{cor:computeegv}, the pure sharp measurement $\left\{ a_{i}\right\} _{i=1}^{d_{\mathrm{A}}}$,
where $d_{\mathrm{A}}\geq r$ and $a_{i}=\alpha_{i}^{\dagger}$ for
$i\in\left\lbrace 1,\ldots,r\right\rbrace $, induces pure states
on system $\mathrm{B}$, as follows \begin{equation}\label{eq:sterb} \begin{aligned}\Qcircuit @C=1em @R=.7em @!R { & \multiprepareC{1}{\Psi} & \qw \poloFantasmaCn{\rA} & \measureD{\alpha_i^\dagger} \\ & \pureghost{\Psi} & \qw \poloFantasmaCn{\rB} & \qw}\end{aligned}~=p_{i}\!\!\!\!\begin{aligned}\Qcircuit @C=1em @R=.7em @!R { & \prepareC{\beta_i} & \qw \poloFantasmaCn{\rB} & \qw }\end{aligned}~, \end{equation}where
each state $\beta_{i}$ is pure and normalised, for every $i\in\left\lbrace 1,\ldots,r\right\rbrace $.
Note that the right-hand side vanishes if $i\in\left\lbrace r+1,\ldots,d_{\mathrm{A}}\right\rbrace $.
Combining the two equations above, we obtain \[ \left( \widetilde{b}_j \middle| \beta_i\right)=\frac 1 {p_i}\!\!\!\!\begin{aligned}\Qcircuit @C=1em @R=.7em @!R { & \multiprepareC{1}{\Psi} & \qw \poloFantasmaCn{\rA} & \measureD{\alpha_i^\dagger} \\ & \pureghost{\Psi} & \qw \poloFantasmaCn{\rB} & \measureD{\widetilde{b}_j}}\end{aligned}=\frac {p_j} {p_i} \left(\alpha_i^\dagger \middle| \alpha_j\right)=\delta_{ij}, \]
for all $i,j\in\left\{ 1,\ldots,r\right\} $. Hence, the pure states
$\left\{ \beta_{i}\right\} _{i=1}^{r}$ are perfectly distinguishable.
This means that, if $\rho_{\mathrm{B}}$ is the marginal of $\Psi$
on system $\mathrm{B}$, the pure sharp measurement $\left\{ a_{i}\right\} _{i=1}^{d_{\mathrm{A}}}$
induces a diagonalisation of $\rho_{\mathrm{B}}$ in terms of the
states $\left\{ \beta_{i}\right\} _{i=1}^{r}$. Indeed $\sum_{i=1}^{r}p_{i}\beta_{i}=\rho_{\mathrm{B}}$
because $\sum_{i=1}^{d_{\mathrm{A}}}a_{i}=u$. By corollary~\ref{cor:conjecture proved},
we know that the effects $\left\{ \beta_{i}^{\dagger}\right\} _{i=1}^{r}$
are such that \[ \delta_{ij}=\left( \beta_j^\dagger	\middle| \beta_i\right)=\frac 1 {p_i}\!\!\!\!\begin{aligned}\Qcircuit @C=1em @R=.7em @!R { & \multiprepareC{1}{\Psi} & \qw \poloFantasmaCn{\rA} & \measureD{\alpha_i^\dagger} \\ & \pureghost{\Psi} & \qw \poloFantasmaCn{\rB} & \measureD{\beta_j^\dagger}}\end{aligned}~. \] 

Hence, choosing $\left\{ a_{i}\right\} _{i=1}^{d_{\mathrm{A}}}$ and
$\left\{ b_{j}\right\} _{j=1}^{d_{\mathrm{B}}}$ to be pure sharp
measurements with $a_{i}:=\alpha_{i}^{\dagger}$ and $b_{i}:=\beta_{i}^{\dagger}$,
for $i=\left\lbrace 1,\ldots,r\right\rbrace $, one obtains eq.~\eqref{eq:schmidt}.
To conclude the proof, recall that $\rho_{\mathrm{A}}=\sum_{i=1}^{r}p_{i}\alpha_{i}$
is a diagonalisation of $\rho_{\mathrm{A}}$. Moreover, $\rho_{\mathrm{B}}=\sum_{i=1}^{r}p_{i}\beta_{i}$
is a diagonalisation of $\rho_{\mathrm{B}}$ thanks to eq.~\eqref{eq:sterb}.
\end{proof}
Theorem~\ref{thm:schmidt} guarantees that the diagonalisations of
the two marginals of a pure bipartite state have the same non-vanishing
eigenvalues. Moreover, it implies that we can induce the pure states
in the diagonalisation of a state $\rho$ by applying suitable normalised
\emph{pure} effects on the purifying system of \emph{any} purification
of $\rho$.

What about general convex decompositions of $\rho$ into pure states
that are not diagonalisations? The following corollary guarantees
that they can always be induced by pure sharp measurements on the
purifying system of \emph{some} purification of $\rho$.
\begin{cor}
\label{cor:pure ensembles}Let $\rho=\sum_{i=1}^{n}\lambda_{i}\psi_{i}$
be a convex decomposition of $\rho\in\mathsf{St}_{1}\left(\mathrm{A}\right)$
into pure states, with $\lambda_{i}>0$ for every $i\in\left\{ 1,\ldots,n\right\} $.
Then there exist a purification $\Psi\in\mathsf{PurSt}_{1}\left(\mathrm{AB}\right)$
of $\rho$, with $d_{\mathrm{B}}\geq n$, and a pure sharp measurement
$\left\{ b_{j}\right\} _{j=1}^{d_{\mathrm{B}}}$ on $\mathrm{B}$,
such that\[
\lambda_{i}\!\!\!\!\begin{aligned}\Qcircuit @C=1em @R=.7em @!R { & \prepareC{\psi_i} & \qw \poloFantasmaCn{\rA} & \qw }\end{aligned}~=\!\!\!\!\begin{aligned}\Qcircuit @C=1em @R=.7em @!R { & \multiprepareC{1}{\Psi} & \qw \poloFantasmaCn{\rA} & \qw \\ & \pureghost{\Psi} & \qw \poloFantasmaCn{\rB} & \measureD{b_i} }\end{aligned}~,
\]for every $i\in\left\{ 1,\ldots,n\right\} $.
\end{cor}

\begin{proof}
Consider a system $\mathrm{X}$ of dimension $n$, and let $\left\{ \xi_{i}\right\} _{i=1}^{n}$
be a pure maximal set of $\mathrm{X}$. Consider now the state $\Sigma\in\mathsf{St}_{1}\left(\mathrm{AX}\right)$,
given by $\Sigma:=\sum_{i=1}^{n}\lambda_{i}\psi_{i}\otimes\xi_{i}$.
This is a diagonalisation of $\Sigma$, for the states $\left\{ \psi_{i}\otimes\xi_{i}\right\} _{i=1}^{n}$
are pure by Purity Preservation, and they are distinguished perfectly
by the observation-test $\left\{ u_{\mathrm{A}}\otimes\xi_{i}^{\dagger}\right\} _{i=1}^{n}$.
Now, let us consider a purification $\Psi\mathsf{\in PurSt}_{1}\left(\mathrm{AXC}\right)$
of $\Sigma$. This is clearly a purification of $\rho$ too, indeed\[
\begin{aligned}\Qcircuit @C=1em @R=.7em @!R { & \multiprepareC{2}{\Psi} & \qw \poloFantasmaCn{\rA} & \qw \\ & \pureghost{\Psi} & \qw \poloFantasmaCn{\rX} & \measureD{u} \\ & \pureghost{\Psi} & \qw \poloFantasmaCn{\rC} & \measureD{u}}\end{aligned}~=\!\!\!\!\begin{aligned}\Qcircuit @C=1em @R=.7em @!R { & \multiprepareC{1}{\Sigma} & \qw \poloFantasmaCn{\rA} & \qw \\ & \pureghost{\Sigma} & \qw \poloFantasmaCn{\rX} & \measureD{u}}\end{aligned}~=\sum_{i=1}^{n}\lambda_i\!\!\!\!\begin{aligned}\Qcircuit @C=1em @R=.7em @!R { & \prepareC{\psi_i} & \qw \poloFantasmaCn{\rA} & \qw \\ & \prepareC{\xi_i} & \qw \poloFantasmaCn{\rX} & \measureD{u}}\end{aligned}~=\!\!\!\!\begin{aligned}\Qcircuit @C=1em @R=.7em @!R { & \prepareC{\rho} & \qw \poloFantasmaCn{\rA} & \qw}\end{aligned}~.
\]Now, by theorem~\ref{thm:schmidt} applied to the purification $\Psi$
of $\Sigma$, there exists a pure sharp measurement\footnote{Clearly we have $d_{\mathrm{C}}\geq nd_{\mathrm{A}}$.}
$\left\{ c_{k}\right\} _{k=1}^{d_{\mathrm{C}}}$ on $\mathrm{C}$,
that induces the pure states in the diagonalisation of $\Sigma$:\[
\lambda_{i}\!\!\!\!\begin{aligned}\Qcircuit @C=1em @R=.7em @!R { & \prepareC{\psi_i} & \qw \poloFantasmaCn{\rA} & \qw \\ & \prepareC{\xi_i} & \qw \poloFantasmaCn{\rX} & \qw }\end{aligned}~=\!\!\!\!\begin{aligned}\Qcircuit @C=1em @R=.7em @!R { & \multiprepareC{2}{\Psi} & \qw \poloFantasmaCn{\rA} & \qw \\ & \pureghost{\Psi} & \qw \poloFantasmaCn{\rX} & \qw \\ & \pureghost{\Psi} & \qw \poloFantasmaCn{\rC} & \measureD{c_i}}\end{aligned}~,
\]for $i\in\left\{ 1,\ldots,n\right\} $. Now, take the pure sharp measurement
$\left\{ \xi_{i}^{\dagger}\right\} _{i=1}^{n}$ on $\mathrm{X}$,
which yields\[
\lambda_{i}\!\!\!\!\begin{aligned}\Qcircuit @C=1em @R=.7em @!R { & \prepareC{\psi_i} & \qw \poloFantasmaCn{\rA} & \qw}\end{aligned}~=\!\!\!\!\begin{aligned}\Qcircuit @C=1em @R=.7em @!R { & \multiprepareC{2}{\Psi} & \qw \poloFantasmaCn{\rA} & \qw \\ & \pureghost{\Psi} & \qw \poloFantasmaCn{\rX} & \measureD{\xi_i^{\dagger}} \\ & \pureghost{\Psi} & \qw \poloFantasmaCn{\rC} & \measureD{c_i}}\end{aligned}~,
\]for $i\in\left\{ 1,\ldots,n\right\} $. To complete the proof, note
that $\left\{ \xi_{i}^{\dagger}\otimes c_{k}\right\} _{i=1,}^{n}\phantom{}_{k=1}^{d_{\mathrm{C}}}$
is still a pure sharp measurement $\left\{ b_{j}\right\} $ on $\mathrm{XC}$
(by proposition~\ref{prop:information locality}), with $nd_{\mathrm{C}}\geq n$
effects, where $j$ runs on the pairs $\left(i,k\right)$. Now it
is enough to take the purifying system $\mathrm{B}$ to be $\mathrm{XC}$,
and to take $k=i$ in $\left\{ \xi_{i}^{\dagger}\otimes c_{k}\right\} _{i=1,}^{n}\phantom{}_{k=1}^{d_{\mathrm{C}}}$.
\end{proof}
We will use this corollary in the proof of theorem~\ref{thm:measurement =00003D preparation}.

\section{Example: doubled quantum theory\label{sec:Example:-doubled-quantum}}

In this section we present a new example of a sharp theory with purification
\cite{Purity}, called ``doubled quantum theory''. This theory will
provide a counterexample to thermodynamic convertibility and majorisation
in section~\ref{sec:Sufficiency-of-majorisation}.

Consider a theory where every non-trivial system is the direct sum
of two isomorphic quantum systems with Hilbert spaces $\mathcal{H}_{0}$
and $\mathcal{H}_{1}$, respectively. Physically, we can think of
the two Hilbert spaces as two superselection sectors. We associate
each ``doubled quantum system'' with a pair of isomorphic Hilbert
spaces $\left(\mathcal{H}_{0},\mathcal{H}_{1}\right)$, with $\mathcal{H}_{0}\approx\mathcal{H}_{1}$.
We define the states of the doubled quantum system to be of the form
\begin{equation}
\rho=p\rho_{0}\oplus\left(1-p\right)\rho_{1}\label{eq:state doubled}
\end{equation}
where $\rho_{0}$ and $\rho_{1}$ are two density matrices in the
two sectors and $p\in\left[0,1\right]$. The direct sum in eq.~\eqref{eq:state doubled}
means that there is no coherence between the two sectors.

Likewise, we define the effects to be all quantum effects of the form
$e=e_{0}\oplus e_{1}$, where $e_{0}$ and $e_{1}$ are two quantum
effects in the two sectors. The allowed channels from the input system
$\left(\mathcal{H}_{0},\mathcal{H}_{1}\right)$ to the output system
$\left(\mathcal{K}_{0},\mathcal{K}_{1}\right)$ are the quantum channels
(completely positive trace-preserving maps) that
\begin{enumerate}
\item send operators on $\mathcal{H}_{0}\oplus\mathcal{H}_{1}$ to operators
on $\mathcal{K}_{0}\oplus\mathcal{K}_{1}$;
\item map block-diagonal operators to block-diagonal operators.
\end{enumerate}
The set of allowed tests is defined as the set of quantum instruments
$\left\{ \mathcal{C}_{j}\right\} _{j\in X}$, where each quantum operation
$\mathcal{C}_{j}$ respects the two conditions above for channels.

This means that in the allowed unitary channels $\mathcal{U}\left(\cdot\right)=U\cdot U^{\dagger}$,
$U$ must be of the form $U=\left(U_{0}\oplus U_{1}\right)S^{k}$,
where $S$ is the unitary transformation that exchanges the two sectors
(it exists because they are isomorphic), $k\in\left\{ 0,1\right\} $,
and $U_{0}$ and $U_{1}$ are unitary transformations that act only
on $\mathcal{H}_{0}$ and $\mathcal{H}_{1}$, respectively. Therefore,
if $k=0$, there is no hopping of sector, and if $k=1$ the two sectors
are exchanged. 

\paragraph{Doubled quantum theory satisfies Causality, Pure Sharpness, and Purity
Preservation}

Causality is immediate: for every system, the only deterministic effect
is the identity matrix. Pure Sharpness is also immediate: every rank-one
projector is a pure sharp effect. As to Purity Preservation, note
that the only pure transformations are quantum operations of the single-Kraus
form $\mathcal{Q}\left(\cdot\right)=Q\cdot Q^{\dagger}$. Clearly,
the composition of two single-Kraus operations (both in parallel and
in sequence) is a single-Kraus operation. In other words, the composition
of two pure transformations is pure.

\subsection{Composite systems}

To study if doubled quantum theory satisfies Purification, it is necessary
to specify how systems compose in this theory.

The peculiarity of doubled quantum theory is the way systems are composed,
which is \emph{not} the intuitive way to compose systems with superselection
rules. The product of two doubled quantum systems $\left(\mathcal{H}_{0}^{\mathrm{A}},\mathcal{H}_{1}^{\mathrm{A}}\right)$
and $\left(\mathcal{H}_{0}^{\mathrm{B}},\mathcal{H}_{1}^{\mathrm{B}}\right)$
is the doubled quantum system $\left(\mathcal{H}_{0}^{\mathrm{AB}},\mathcal{H}_{1}^{\mathrm{AB}}\right)$,
with the two sectors defined by
\begin{equation}
\left\{ \begin{array}{l}
\mathcal{H}_{0}^{\mathrm{AB}}:=\left(\mathcal{H}_{0}^{\mathrm{A}}\otimes\mathcal{H}_{0}^{\mathrm{B}}\right)\oplus\left(\mathcal{H}_{1}^{\mathrm{A}}\otimes\mathcal{H}_{1}^{\mathrm{B}}\right)\\
\mathcal{H}_{1}^{\mathrm{AB}}:=\left(\mathcal{H}_{0}^{\mathrm{A}}\otimes\mathcal{H}_{1}^{\mathrm{B}}\right)\oplus\left(\mathcal{H}_{1}^{\mathrm{A}}\otimes\mathcal{H}_{0}^{\mathrm{B}}\right)
\end{array}\right..\label{eq:hhahhb}
\end{equation}
Note that the direct sum inside each sector does \emph{not} mean the
presence of additional sectors. We illustrate this with an example.
\begin{example}
Consider the composite system of two doubled qubits, corresponding
to $\mathcal{H}_{0}^{\mathrm{A}}\approx\mathcal{H}_{1}^{\mathrm{A}}\approx\mathcal{H}_{0}^{\mathrm{B}}\approx\mathcal{H}_{1}^{\mathrm{B}}\approx\mathbb{C}^{2}$.
An example of state of the composite system is the pure state (the
first index denotes the sector) 
\begin{equation}
\ket{\Psi}=\frac{1}{\sqrt{2}}\left(\ket{0,0}_{\mathrm{A}}\ket{0,0}_{\mathrm{B}}+\ket{1,0}_{\mathrm{A}}\ket{1,0}_{\mathrm{B}}\right),\label{eq:example0011}
\end{equation}
where $\left\{ \ket{0,0},\ket{0,1}\right\} $ is an orthonormal basis
for $\mathcal{H}_{0}$ and $\left\{ \ket{1,0},\ket{1,1}\right\} $
is an orthonormal basis for $\mathcal{H}_{1}$. Thus we see that there
\emph{is} coherence allowed between $\mathcal{H}_{0}^{\mathrm{A}}\otimes\mathcal{H}_{0}^{\mathrm{B}}$
and $\mathcal{H}_{1}^{\mathrm{A}}\otimes\mathcal{H}_{1}^{\mathrm{B}}$.
However, note that, when one of the two systems is traced out, the
remaining local state has the block-diagonal form $\rho=\frac{1}{2}\ket{0,0}\bra{0,0}\oplus\frac{1}{2}\ket{1,0}\bra{1,0}$.
This means that the coherence between the two terms in the state~\eqref{eq:example0011}
is invisible at the single-system level.
\end{example}

From a physical point of view, doubled quantum theory can be thought
of as ordinary quantum theory with a superselection rule on the \emph{total}
parity. Every system is split into two identical sectors of even and
odd parity, respectively. When systems are composed, the sectors are
grouped together based on the total parity, so that superpositions
between subspaces with the same parity are allowed.

Here we summarise the basic operational features of doubled quantum
theory concerning the composition of systems.

\paragraph{Doubled quantum theory violates Local Tomography}

An equivalent formulation of Local Tomography is that the dimension
of the vector space spanned by the states of a composite system is
equal to the product of the dimensions of the vector spaces spanned
by the states of the components \cite{Hardy-informational-1,Chiribella-purification,Hardy-informational-2},
in formula $D_{\mathrm{AB}}=D_{\mathrm{A}}D_{\mathrm{B}}$, where
$D$ is the dimension of the vector space of states. This is because,
if Local Tomography holds, one has $\mathsf{St}_{\mathbb{R}}\left(\mathrm{AB}\right)=\mathsf{St}_{\mathbb{R}}\left(\mathrm{A}\right)\otimes\mathsf{St}_{\mathbb{R}}\left(\mathrm{B}\right)$
(cf.\ remark~\ref{rem:local tomography}).

The equality $D_{\mathrm{AB}}=D_{\mathrm{A}}D_{\mathrm{B}}$ fails
to hold in doubled quantum theory, where the dimension of the global
vector space is strictly larger than the product of the dimensions
of the individual vector spaces. To see it, note that the block-diagonal
states of the form~\eqref{eq:state doubled} span a vector space
of dimension $D=2d^{2}$, where $d$ is the dimension of the Hilbert
spaces $\mathcal{H}_{0}$ and $\mathcal{H}_{1}$. Given two systems
$\mathrm{A}$ and $\mathrm{B}$, the product of the individual dimensions
is $D_{\mathrm{A}}D_{\mathrm{B}}=2d_{\mathrm{A}}^{2}\cdot2d_{\mathrm{B}}^{2}=\left(2d_{\mathrm{A}}d_{\mathrm{B}}\right)^{2}$.
On the other hand, each of the Hilbert spaces $\mathcal{H}_{0}^{\mathrm{AB}}$
and $\mathcal{H}_{1}^{\mathrm{AB}}$ in eq.~\eqref{eq:hhahhb} has
dimension $d_{\mathrm{AB}}=2d_{\mathrm{A}}d_{\mathrm{B}}$. Hence,
the vector space spanned by the states of the composite system has
dimension $D_{\mathrm{AB}}=2d_{\mathrm{AB}}^{2}=2\left(2d_{\mathrm{A}}d_{\mathrm{B}}\right)^{2}$,
that is, twice the dimension of the vector space spanned by product
states. This means that when systems are composed, genuinely new states
arise, that cannot be reduced to states of the two components.

\paragraph{Doubled quantum theory satisfies Purification }

A generic state of a system $\left(\mathcal{H}_{0},\mathcal{H}_{1}\right)$
can be diagonalised as
\[
\rho=\left(\sum_{j=1}^{d}\lambda_{j}\ket{\varphi_{j,0}}\bra{\varphi_{j,0}}\right)\oplus\left(\sum_{j=1}^{d}\mu_{j}\ket{\psi_{j,1}}\bra{\psi_{j,1}}\right),
\]
where $\left\{ \ket{\varphi_{j,0}}\right\} _{j=1}^{d}$ is an orthonormal
basis for $\mathcal{H}_{0}$ and $\left\{ \ket{\psi_{j,1}}\right\} _{j=1}^{d}$
is an orthonormal basis for $\mathcal{H}_{1}$. The state can be purified
e.g.\ by adding a copy of system $\left(\mathcal{H}_{0},\mathcal{H}_{1}\right)$.
Since the composite system has two superselection sectors, there will
be two types of purification: purifications in the even subspace $\mathcal{H}_{0}^{\mathrm{AB}}$
and purifications in the odd subspace $\mathcal{H}_{1}^{\mathrm{AB}}$.
A purification in the subspace $\mathcal{H}_{0}^{\mathrm{AB}}$ has
the form 
\[
\ket{\Psi_{0}}=\left(\sum_{j=1}^{d}\sqrt{\lambda_{j}}\ket{\varphi_{j,0}}\ket{\alpha_{j,0}}\right)+\left(\sum_{j=1}^{d}\sqrt{\mu_{j}}\ket{\psi_{j,1}}\ket{\beta_{j,1}}\right),
\]
where $\left\{ \ket{\alpha_{j,0}}\right\} _{j=1}^{d}$ is an orthonormal
basis for $\mathcal{H}_{0}$ and $\left\{ \ket{\beta_{j,1}}\right\} _{j=1}^{d}$
is an orthonormal basis for $\mathcal{H}_{1}$. A purification in
the subspace $\mathcal{H}_{1}^{\mathrm{AB}}$ will have the form 
\[
\ket{\Psi_{1}}=\left(\sum_{j=1}^{d}\sqrt{\lambda_{j}}\ket{\varphi_{j,0}}\ket{\alpha'_{j,1}}\right)+\left(\sum_{j=1}^{d}\sqrt{\mu_{j}}\ket{\psi_{j,1}}\ket{\beta'_{j,0}}\right)
\]
where $\left\{ \ket{\alpha'_{j,1}}\right\} _{j=1}^{d}$ is an orthonormal
basis for $\mathcal{H}_{1}$ and $\left\{ \ket{\beta'_{j,0}}\right\} _{j=1}^{d}$
is an orthonormal basis for $\mathcal{H}_{0}$. Note that any two
such purifications are equivalent under local unitary transformations:
indeed, one has $\ket{\Psi_{1}}=\left(\mathbf{1}\otimes U\right)\ket{\Psi_{0}}$,
where $U$ is the unitary matrix defined by
\[
U=\left(\sum_{j=1}^{d}\ket{\alpha'_{j,1}}\bra{\alpha_{j,0}}\right)+\left(\sum_{j=1}^{d}\ket{\beta'_{j,0}}\bra{\beta_{j,1}}\right).
\]
The same arguments apply to purifications within the same sector and
to purifications where the purifying system is not a copy of the original
system. In summary, every state can be purified and every two purifications
with the same purifying system are equivalent under local unitaries.

\section{Example: extended classical theory\label{sec:Example:-extended-classical}}

In this section we introduce another new example of a sharp theory
with purification \cite{TowardsThermo}. This example has a great
importance because it shows that classical theory, which does \emph{not}
satisfy Purification, can be regarded as part of a larger theory obeying
the Purification principle. Specifically, in this extended classical
theory there are some systems that look entirely classical at the
single-system level, but they compose in a different, coherent, way,
so as to save the validity of Purification. This example therefore
shows that all the results obtained above for sharp theories with
purification (e.g.\ diagonalisation, etc.) can be carried over to
classical theory, at least at the single-system level. Therefore classical
theory is not excluded by our treatment.

Extended classical theory will include classical and non-classical
systems, called \emph{coherent dits} (or \emph{codits} for short)\emph{,}
in analogy with the similar notion in quantum Shannon theory \cite{Harrow}.
The guiding idea is to entangle classical systems with each other,
to provide the desired purifications, while at the same time to keep
them classical at the single-system level. In principle, we could
have modelled classical theory as a sub-theory of quantum theory,
but clearly in this case we would observe interference at the level
of single systems, which we do not want. 

\subsection{Coherent composition of bits }

To understand how the construction of extended classical theory works,
let us illustrate the 2-dimensional case first. Recall that the state
of a classical bit can be represented using the density matrix formalism
as
\begin{equation}
\rho=p\ket{0}\bra{0}\oplus\left(1-p\right)\ket{1}\bra{1},\label{eq:cbitstate}
\end{equation}
where $p\in\left[0,1\right]$, and the direct sum sign is a reminder
that the off-diagonal elements are forbidden. The composite system
of two classical bits $\mathrm{A}$ and $\mathrm{B}$ is a 4-dimensional
classical system, which is represented as a quantum system with 4
superselection sectors. In formulas, if $\left(\mathcal{H}_{0}^{\mathrm{A}},\mathcal{H}_{1}^{\mathrm{A}}\right)$
represents the classical bit $\mathrm{A}$, with its 2 superselection
sectors $\mathcal{H}_{0}^{\mathrm{A}}=\mathrm{Span}\left\{ \ket{0}_{\mathrm{A}}\right\} $
and $\mathcal{H}_{1}^{\mathrm{A}}=\mathrm{Span}\left\{ \ket{1}_{\mathrm{A}}\right\} $,
and $\left(\mathcal{H}_{0}^{\mathrm{B}},\mathcal{H}_{1}^{\mathrm{B}}\right)$
is the classical system $\mathrm{B}$, the composite system is $\left(\mathcal{H}_{00}^{\mathrm{AB}},\mathcal{H}_{01}^{\mathrm{AB}},\mathcal{H}_{10}^{\mathrm{AB}},\mathcal{H}_{11}^{\mathrm{AB}}\right)$,
with 4 sectors of dimension 1. Here the composition is the usual one
in the presence of superselection sectors, namely
\[
\left\{ \begin{array}{l}
\mathcal{H}_{00}^{\mathrm{AB}}:=\mathcal{H}_{0}^{\mathrm{A}}\otimes\mathcal{H}_{0}^{\mathrm{B}}\\
\mathcal{H}_{01}^{\mathrm{AB}}:=\mathcal{H}_{0}^{\mathrm{A}}\otimes\mathcal{H}_{1}^{\mathrm{B}}\\
\mathcal{H}_{10}^{\mathrm{AB}}:=\mathcal{H}_{1}^{\mathrm{A}}\otimes\mathcal{H}_{0}^{\mathrm{B}}\\
\mathcal{H}_{11}^{\mathrm{AB}}:=\mathcal{H}_{1}^{\mathrm{A}}\otimes\mathcal{H}_{1}^{\mathrm{B}}
\end{array}.\right.
\]
Clearly a state of this classical composite system is of the form
\[
\rho_{\mathrm{AB}}=p_{00}\ket{0}\bra{0}\otimes\ket{0}\bra{0}\oplus p_{01}\ket{0}\bra{0}\otimes\ket{1}\bra{1}\oplus
\]
\[
\oplus p_{10}\ket{1}\bra{1}\otimes\ket{0}\bra{0}\oplus p_{11}\ket{1}\bra{1}\otimes\ket{1}\bra{1}
\]
where $\left\{ p_{ij}\right\} $, $i,j\in\left\{ 0,1\right\} $, is
a probability distribution.

Now, to construct extended classical theory, let us consider a single
2-dimensional system with the same states as the classical bit; they
are still of the form of eq.~\eqref{eq:cbitstate}. We will change
the way two classical bits compose, by imposing a superselection rule
given by the \emph{total} parity, in the same way we did for doubled
quantum theory. Therefore the composition of two classical bits $\left(\mathcal{H}_{0}^{\mathrm{A}},\mathcal{H}_{1}^{\mathrm{A}}\right)$
and $\left(\mathcal{H}_{0}^{\mathrm{B}},\mathcal{H}_{1}^{\mathrm{B}}\right)$,
where $\mathcal{H}_{0}^{\mathrm{A}}\approx\mathcal{H}_{1}^{\mathrm{A}}\approx\mathcal{H}_{0}^{\mathrm{B}}\approx\mathcal{H}_{1}^{\mathrm{B}}\approx\mathbb{C}$
yields the system $\left(\mathcal{H}_{0}^{\mathrm{AB}},\mathcal{H}_{1}^{\mathrm{AB}}\right)$
where
\begin{equation}
\left\{ \begin{array}{l}
\mathcal{H}_{0}^{\mathrm{AB}}:=\left(\mathcal{H}_{0}^{\mathrm{A}}\otimes\mathcal{H}_{0}^{\mathrm{B}}\right)\oplus\left(\mathcal{H}_{1}^{\mathrm{A}}\otimes\mathcal{H}_{1}^{\mathrm{B}}\right)\\
\mathcal{H}_{1}^{\mathrm{AB}}:=\left(\mathcal{H}_{0}^{\mathrm{A}}\otimes\mathcal{H}_{1}^{\mathrm{B}}\right)\oplus\left(\mathcal{H}_{1}^{\mathrm{A}}\otimes\mathcal{H}_{0}^{\mathrm{B}}\right)
\end{array}\right..\label{eq:composite cobits}
\end{equation}
Again, the direct sum inside each sector does \emph{not} denote an
internal superselection rule, so we allow coherence inside each sector,
which is invisible at the single-system level. This fact will be true
in all the direct sums we will write in this section about extended
classical theory.

Note that eq.~\eqref{eq:composite cobits} gives exactly the same
composition rule of doubled quantum theory, but restricted to 1-dimensional
Hilbert spaces. In the composition rule, the Hilbert spaces are grouped
together according to the residue classes modulo 2: if $k\in\left\{ 0,1\right\} $,
in $\mathcal{H}_{k}^{\mathrm{AB}}$ there is the direct sum of all
terms whose indices sum to $k$ modulo 2. Therefore
\[
\mathcal{H}_{0}^{\mathrm{AB}}=\mathrm{Span}\left\{ \ket{0}_{\mathrm{A}}\ket{0}_{\mathrm{B}},\ket{1}_{\mathrm{A}}\ket{1}_{\mathrm{B}}\right\} ;
\]
and
\[
\mathcal{H}_{1}^{\mathrm{AB}}=\mathrm{Span}\left\{ \ket{0}_{\mathrm{A}}\ket{1}_{\mathrm{B}},\ket{1}_{\mathrm{A}}\ket{0}_{\mathrm{B}}\right\} .
\]
Consequently, the pure states of the composite systems can be represented
as unit vectors either of the form ($\alpha,\beta\in\mathbb{C}$)
\[
\ket{\Phi_{0}}_{\mathrm{AB}}=\alpha\ket{0}_{\mathrm{A}}\ket{0}_{\mathrm{B}}+\beta\ket{1}\ket{1},
\]
or of the form 
\[
\ket{\Phi_{1}}_{\mathrm{AB}}=\alpha\ket{0}_{\mathrm{A}}\ket{1}_{\mathrm{B}}+\beta\ket{1}_{\mathrm{A}}\ket{0}_{\mathrm{B}}.
\]

Note that in the composite system the only allowed states are those
of the form $\rho=p\rho_{0}\oplus\left(1-p\right)\rho_{1}$, where
$p\in\left[0,1\right]$, $\rho_{0}$ is a density matrix on $\mathcal{H}_{0}^{\mathrm{AB}}$,
and $\rho_{1}$ is a density matrix on $\mathcal{H}_{1}^{\mathrm{AB}}$.
The allowed effects are of the form $e=e_{0}\oplus e_{1}$, where
$e_{0}$ is a quantum effect on $\mathcal{H}_{0}^{\mathrm{AB}}$,
and $e_{1}$ is a quantum effect on $\mathcal{H}_{1}^{\mathrm{AB}}$.
The allowed channels on this systems are those quantum channels that
respect the block-diagonal structure.
\begin{rem}
In the composition of the two classical bits we \emph{cannot} have
the full 4-dimensional quantum system
\[
\mathrm{Span}\left\{ \ket{0}_{\mathrm{A}}\ket{0}_{\mathrm{B}},\ket{1}_{\mathrm{A}}\ket{1}_{\mathrm{B}},\ket{0}_{\mathrm{A}}\ket{1}_{\mathrm{B}},\ket{1}_{\mathrm{A}}\ket{0}_{\mathrm{B}}\right\} .
\]
Indeed, if this were the case, an allowed (pure) state of the composite
system would be $\frac{1}{\sqrt{2}}\left(\ket{0}_{\mathrm{A}}+\ket{1}_{\mathrm{A}}\right)\ket{0}_{\mathrm{B}}$,
which, when system $\mathrm{B}$ is traced out, would yield the forbidden
state $\frac{1}{\sqrt{2}}\left(\ket{0}_{\mathrm{A}}+\ket{1}_{\mathrm{A}}\right)$.
Therefore, to keep the state of $\mathrm{A}$ and $\mathrm{B}$ classical,
we \emph{need} the presence of the two superselection sectors $\mathcal{H}_{0}^{\mathrm{AB}}$
and $\mathcal{H}_{1}^{\mathrm{AB}}$.
\end{rem}

With the above settings, it is easy to see that every state of a classical
bit can be purified. For example, the generic bit state $\rho=p\ket{0}\bra{0}\oplus\left(1-p\right)\ket{1}\bra{1}$
has the purification 
\[
\ket{\Psi_{0}}_{\mathrm{AB}}=\sqrt{p}\ket{0}_{\mathrm{A}}\ket{0}_{\mathrm{B}}+\sqrt{1-p}\ket{1}_{\mathrm{A}}\ket{1}_{\mathrm{B}}.
\]
In addition, it is possible to show that every two purifications of
the same state differ by a local unitary operation on the purifying
system. This includes, for example, the purification 
\[
\ket{\Psi_{1}}_{\mathrm{AB}}=\sqrt{p}\ket{0}_{\mathrm{A}}\ket{1}_{\mathrm{B}}+\sqrt{1-p}\ket{1}_{\mathrm{A}}\ket{0}_{\mathrm{B}},
\]
obtained from $\ket{\Psi_{0}}_{\mathrm{AB}}$ through the application
of a bit flip on system $\mathrm{B}$, which is an allowed transformation
because it preserves the block-diagonal structure of system $\mathrm{B}$.

Notice that the composition of two classical bits gives rise to a
4-dimensional system that is \emph{not} classical. Therefore, in the
theory we are constructing we will have classical and non-classical
systems with the same dimension (e.g.\ a classical 4-dimensional
system, and the coherent composition of 2 classical bits).

\subsection{Coherent composition of dits}

Let us generalise the results of the previous subsection to the coherent
composition of two classical dits. A classical dit can be represented
as a Hilbert space with $d$ superselection sectors of dimension 1:
$\left(\mathcal{H}_{0},\ldots,\mathcal{H}_{d-1}\right)$, where $\mathcal{H}_{k}=\mathrm{Span}\left\{ \ket{k}\right\} $.
The states of a classical dit are of the form
\[
\rho=\bigoplus_{k=0}^{d-1}p_{k}\ket{k}\bra{k},
\]
where $\left\{ p_{k}\right\} _{k=0}^{d-1}$ is a probability distribution.
Like for bits, the new composition rule for dits is based on residue
classes (modulo $d$): the composite of two classical dits $\left(\mathcal{H}_{0}^{\mathrm{A}},\ldots,\mathcal{H}_{d-1}^{\mathrm{A}}\right)$
and $\left(\mathcal{H}_{0}^{\mathrm{B}},\ldots,\mathcal{H}_{d-1}^{\mathrm{B}}\right)$
is the system $\left(\mathcal{H}_{0}^{\mathrm{AB}},\ldots,\mathcal{H}_{d-1}^{\mathrm{AB}}\right)$
with $d$ sectors of dimension $d$.
\[
\left\{ \begin{array}{l}
\mathcal{H}_{0}^{\mathrm{AB}}:=\left(\mathcal{H}_{0}^{\mathrm{A}}\otimes\mathcal{H}_{0}^{\mathrm{B}}\right)\oplus\left(\mathcal{H}_{1}^{\mathrm{A}}\otimes\mathcal{H}_{d-1}^{\mathrm{B}}\right)\oplus\ldots\oplus\left(\mathcal{H}_{d-1}^{\mathrm{A}}\otimes\mathcal{H}_{1}^{\mathrm{B}}\right)\\
\mathcal{H}_{1}^{\mathrm{AB}}:=\left(\mathcal{H}_{0}^{\mathrm{A}}\otimes\mathcal{H}_{1}^{\mathrm{B}}\right)\oplus\left(\mathcal{H}_{1}^{\mathrm{A}}\otimes\mathcal{H}_{0}^{\mathrm{B}}\right)\oplus\ldots\oplus\left(\mathcal{H}_{d-1}^{\mathrm{A}}\otimes\mathcal{H}_{2}^{\mathrm{B}}\right)\\
\vdots\\
\mathcal{H}_{d-1}^{\mathrm{AB}}:=\left(\mathcal{H}_{0}^{\mathrm{A}}\otimes\mathcal{H}_{d-1}^{\mathrm{B}}\right)\oplus\left(\mathcal{H}_{1}^{\mathrm{A}}\otimes\mathcal{H}_{d-2}^{\mathrm{B}}\right)\oplus\ldots\oplus\left(\mathcal{H}_{d-1}^{\mathrm{A}}\otimes\mathcal{H}_{0}^{\mathrm{B}}\right)
\end{array}\right..
\]
Note that in $\mathcal{H}_{k}^{\mathrm{AB}}$, for $k\in\left\{ 0,\ldots,d-1\right\} $
there is a direct sum of all the terms $\mathcal{H}_{j}^{\mathrm{A}}\otimes\mathcal{H}_{l}^{\mathrm{B}}$
such that $j+l\equiv k\mod d$. In this way
\[
\mathcal{H}_{k}^{\mathrm{AB}}=\mathrm{Span}\left\{ \ket{j}_{\mathrm{A}}\ket{k-j\mod d}_{\mathrm{B}}:j=0,\ldots,d-1\right\} .
\]
The states of the composite system $\mathrm{AB}$ are therefore density
matrices of the form
\begin{equation}
\rho_{\mathrm{AB}}=\bigoplus_{k=0}^{d-1}p_{k}\rho_{k}^{\mathrm{AB}},\label{eq:coherent composite}
\end{equation}
where $\rho_{k}^{\mathrm{AB}}$ is a density matrix of $\mathcal{H}_{k}^{\mathrm{AB}}$,
and $\left\{ p_{k}\right\} _{k=0}^{d-1}$ is a probability distribution.
Again, the allowed effects are of the form $e=\bigoplus_{k=0}^{d-1}e_{k}$,
where $e_{k}$ is a quantum effect on $\mathcal{H}_{k}^{\mathrm{AB}}$.
The allowed channels on this systems are those quantum channels that
respect the block-diagonal structure. 

\paragraph{The new composition of classical dits satisfies Purification}

It is easy to see that any state of a classical dit $\rho=\bigoplus_{j=0}^{d-1}p_{j}\ket{j}\bra{j}$
can be purified. Specifically, we can find a purification in every
sector of a composite system $\mathrm{AB}$, where $\mathrm{B}$ is
another classical dit. Indeed, for every sector $k$ in $\mathrm{AB}$,
$\rho$ can be purified as
\[
\ket{\Psi_{k}}_{\mathrm{AB}}=\sum_{j=0}^{d-1}\sqrt{p_{j}}\ket{j}_{\mathrm{A}}\ket{k-j\mod d}_{\mathrm{B}}.
\]
These purifications are all related to each other by a local unitary
on $\mathrm{B}$, which hops between the $d$ sectors in $\mathrm{B}$.

So far we have dealt with the composition of classical systems with
the same dimension. Let us define the new, coherent, composition of
a $d_{\mathrm{A}}$-dimensional classical system $\mathrm{A}$ with
a $d_{\mathrm{B}}$-dimensional classical system $\mathrm{B}$. The
standard way to compose system $\mathrm{A}$, given by $\left(\mathcal{H}_{0}^{\mathrm{A}},\ldots,\mathcal{H}_{d_{\mathrm{A}}-1}^{\mathrm{A}}\right)$,
and system $\mathrm{B}$, given by $\left(\mathcal{H}_{0}^{\mathrm{B}},\ldots,\mathcal{H}_{d_{\mathrm{B}}-1}^{\mathrm{B}}\right)$,
where each sector is 1-dimensional, is a system $\mathrm{AB}$ with
$d_{\mathrm{A}}d_{\mathrm{B}}$ 1-dimensional sectors $\left(\mathcal{H}_{00}^{\mathrm{AB}},\ldots,\mathcal{H}_{d_{\mathrm{A}}-1,d_{\mathrm{B}}-1}^{\mathrm{AB}}\right)$.

In this case, to define the new composition, we consider $\max\left\{ d_{\mathrm{A}},d_{\mathrm{B}}\right\} $
sectors, each of dimension $\min\left\{ d_{\mathrm{A}},d_{\mathrm{B}}\right\} $,
so that altogether the Hilbert space for $\mathrm{AB}$ will have
dimension $d_{\mathrm{A}}d_{\mathrm{B}}$. For concreteness, suppose
$d_{\mathrm{A}}\leq d_{\mathrm{B}}$; then the composite system will
be described by $\left(\mathcal{H}_{0}^{\mathrm{AB}},\ldots,\mathcal{H}_{d_{\mathrm{B}}-1}^{\mathrm{AB}}\right)$,
where
\[
\left\{ \begin{array}{l}
\mathcal{H}_{0}^{\mathrm{AB}}:=\left(\mathcal{H}_{0}^{\mathrm{A}}\otimes\mathcal{H}_{0}^{\mathrm{B}}\right)\oplus\left(\mathcal{H}_{1}^{\mathrm{A}}\otimes\mathcal{H}_{d_{\mathrm{B}}-1}^{\mathrm{B}}\right)\oplus\ldots\oplus\left(\mathcal{H}_{d_{\mathrm{A}}-1}^{\mathrm{A}}\otimes\mathcal{H}_{d_{\mathrm{B}}-d_{\mathrm{A}}+1}^{\mathrm{B}}\right)\\
\mathcal{H}_{1}^{\mathrm{AB}}:=\left(\mathcal{H}_{0}^{\mathrm{A}}\otimes\mathcal{H}_{1}^{\mathrm{B}}\right)\oplus\left(\mathcal{H}_{1}^{\mathrm{A}}\otimes\mathcal{H}_{0}^{\mathrm{B}}\right)\oplus\ldots\oplus\left(\mathcal{H}_{d_{\mathrm{A}}-1}^{\mathrm{A}}\otimes\mathcal{H}_{d_{\mathrm{B}}-d_{\mathrm{A}}+2}^{\mathrm{B}}\right)\\
\vdots\\
\mathcal{H}_{d_{\mathrm{B}}-1}^{\mathrm{AB}}:=\left(\mathcal{H}_{0}^{\mathrm{A}}\otimes\mathcal{H}_{d_{\mathrm{B}}-1}^{\mathrm{B}}\right)\oplus\left(\mathcal{H}_{1}^{\mathrm{A}}\otimes\mathcal{H}_{d_{\mathrm{B}}-2}^{\mathrm{B}}\right)\oplus\ldots\oplus\left(\mathcal{H}_{d_{\mathrm{A}}-1}^{\mathrm{A}}\otimes\mathcal{H}_{d_{\mathrm{B}}-d_{\mathrm{A}}}^{\mathrm{B}}\right)
\end{array}\right..
\]
Again, in $\mathcal{H}_{k}^{\mathrm{AB}}$, for $k\in\left\{ 0,\ldots,d_{\mathrm{B}}-1\right\} $
there is a direct sum of all the terms $\mathcal{H}_{j}^{\mathrm{A}}\otimes\mathcal{H}_{l}^{\mathrm{B}}$
such that $j+l\equiv k\mod d_{\mathrm{B}}$. In this way
\[
\mathcal{H}_{k}^{\mathrm{AB}}=\mathrm{Span}\left\{ \ket{j}_{\mathrm{A}}\ket{k-j\mod d_{\mathrm{B}}}_{\mathrm{B}}:j=0,\ldots,d_{\mathrm{A}}-1\right\} .
\]
The states are still of the form~\eqref{eq:coherent composite},
with $d=d_{\mathrm{B}}$.

\subsection{The other composites}

From the previous subsection we know that in extended classical theory,
the generic system we have encountered so far is made of $N$ superselection
sectors $\left(\mathcal{H}_{0},\ldots,\mathcal{H}_{N-1}\right)$,
each of which of dimension $n\leq N$. Note that this covers also
the usual $d$-dimensional classical systems, for which $N=d$, and
$n=1$. To complete the theory, we must specify how these generic
systems compose.

We \emph{define} the composition of two general systems of extended
classical theory to follow the same rules explained above for classical
systems. More specifically, consider $\left(\mathcal{H}_{0}^{\mathrm{A}},\ldots,\mathcal{H}_{N-1}^{\mathrm{A}}\right)$,
with sectors of dimension $n\leq N$, and $\left(\mathcal{H}_{0}^{\mathrm{B}},\ldots,\mathcal{H}_{M-1}^{\mathrm{B}}\right)$
with sectors of dimension $m\leq M$, and for concreteness suppose
$N\leq M$. The composite system $\mathrm{AB}$ will have $M=\max\left\{ N,M\right\} $
sectors:
\begin{equation}
\left\{ \begin{array}{l}
\mathcal{H}_{0}^{\mathrm{AB}}:=\left(\mathcal{H}_{0}^{\mathrm{A}}\otimes\mathcal{H}_{0}^{\mathrm{B}}\right)\oplus\left(\mathcal{H}_{1}^{\mathrm{A}}\otimes\mathcal{H}_{M-1}^{\mathrm{B}}\right)\oplus\ldots\oplus\left(\mathcal{H}_{N-1}^{\mathrm{A}}\otimes\mathcal{H}_{M-N+1}^{\mathrm{B}}\right)\\
\mathcal{H}_{1}^{\mathrm{AB}}:=\left(\mathcal{H}_{0}^{\mathrm{A}}\otimes\mathcal{H}_{1}^{\mathrm{B}}\right)\oplus\left(\mathcal{H}_{1}^{\mathrm{A}}\otimes\mathcal{H}_{0}^{\mathrm{B}}\right)\oplus\ldots\oplus\left(\mathcal{H}_{N-1}^{\mathrm{A}}\otimes\mathcal{H}_{M-N+2}^{\mathrm{B}}\right)\\
\vdots\\
\mathcal{H}_{M-1}^{\mathrm{AB}}:=\left(\mathcal{H}_{0}^{\mathrm{A}}\otimes\mathcal{H}_{M-1}^{\mathrm{B}}\right)\oplus\left(\mathcal{H}_{1}^{\mathrm{A}}\otimes\mathcal{H}_{M-2}^{\mathrm{B}}\right)\oplus\ldots\oplus\left(\mathcal{H}_{N-1}^{\mathrm{A}}\otimes\mathcal{H}_{M-N}^{\mathrm{B}}\right)
\end{array}\right.,\label{eq:general composition extended}
\end{equation}
where, again in $\mathcal{H}_{k}^{\mathrm{AB}}$ there is the direct
sum of $N$ terms $\mathcal{H}_{j}^{\mathrm{A}}\otimes\mathcal{H}_{l}^{\mathrm{B}}$
such that $j+l\equiv k\mod M$. We see that in this case each sector
$\mathcal{H}_{k}^{\mathrm{AB}}$ in the composite system has dimension
$nmN$. If we take $n=m=1$, we recover the coherent composition law
for classical systems. Note that, in general it is \emph{not} true
that $nmN\leq M$. Indeed, composing two systems arising from the
coherent composition of dits, the resulting system has $d$ sectors,
each of which of dimension $d^{3}$. Therefore, in the most general
system of extended classical theory there is no restriction on the
dimension of sectors, and it is just a system with $N\geq2$ isomorphic
superselection sectors $\left(\mathcal{H}_{0},\ldots,\mathcal{H}_{N-1}\right)$.
The same rule~\eqref{eq:general composition extended} still applies
to these systems.

Generic states are of the form $\rho=\bigoplus_{k=0}^{N-1}p_{k}\rho_{k}$,
where $\rho_{k}$ is a density matrix of the sector $\mathcal{H}_{k}$,
and $\left\{ p_{k}\right\} _{k=0}^{N-1}$ is a probability distribution.
Effects are of the form $e=\bigoplus_{k=0}^{N-1}e_{k}$, where $e_{k}$
is a quantum effect on $\mathcal{H}_{k}$. Finally, all transformations
between $\mathrm{A}$ and $\mathrm{B}$ are those quantum operations
from $\mathrm{A}$ to $\mathrm{B}$ that preserve the block-diagonal
structure.

Now we can finally show that extended classical theory is a sharp
theory with purification. It is straightforward to show that the theory
satisfies Causality, Purity Preservation and Pure Sharpness.

\paragraph{Extended classical theory satisfies Purification}

A state of $\left(\mathcal{H}_{0}^{\mathrm{A}},\ldots,\mathcal{H}_{N-1}^{\mathrm{A}}\right)$,
with sectors of dimension $n$, can be diagonalised as
\[
\rho=\bigoplus_{k=0}^{N-1}\left(\sum_{j=0}^{n-1}\lambda_{j,k}\ket{\varphi_{j,k}}\bra{\varphi_{j,k}}\right),
\]
where $\left\{ \lambda_{j,k}\right\} $ is a probability distribution,
and $\left\{ \ket{\varphi_{j,k}}\right\} $ is an orthonormal basis
of $\mathcal{H}_{k}^{\mathrm{A}}$, for every $k$. Then to obtain
a purification of $\rho$, it is enough to take the same system as
the purifying system $\mathrm{B}$. Specifically a purification in
the sector $\mathcal{H}_{l}^{\mathrm{AB}}$ is given by
\[
\ket{\Psi_{l}}_{\mathrm{AB}}=\sum_{k=0}^{N-1}\sum_{j=0}^{n-1}\sqrt{\lambda_{j,k}}\ket{\varphi_{j,k}}_{\mathrm{A}}\ket{\alpha_{j,l-k}}_{\mathrm{B}}
\]
where $\left\{ \ket{\alpha_{j,l-k}}\right\} _{j=0}^{n-1}$ is an orthonormal
basis of sector $\mathcal{H}_{l-k}^{\mathrm{B}}$.\footnote{Here, as above, $l-k$ is to be intended modulo $N$.}
Since all superselection sectors are isomorphic, one can convert a
purification on $\mathcal{H}_{l}^{\mathrm{AB}}$ into a purification
on $\mathcal{H}_{l'}^{\mathrm{AB}}$ by a local hopping unitary on
system $\mathrm{B}$.

\paragraph{Extended classical theory violates Local Tomography}

We can show that Local Tomography fails in general. For a system $\left(\mathcal{H}_{0}^{\mathrm{A}},\ldots,\mathcal{H}_{N-1}^{\mathrm{A}}\right)$
with sectors of dimension $n$, the dimension of $\mathsf{St}_{\mathbb{R}}\left(\mathrm{A}\right)$
is $D_{\mathrm{A}}=Nn^{2}$. Similarly for a system with $M$ sectors
of dimension $m$, we have $D_{\mathrm{B}}=Mm^{2}$. Now, by eq.~\eqref{eq:general composition extended}
we have $D_{\mathrm{AB}}=MN^{2}m^{2}n^{2}$. Therefore
\[
D_{\mathrm{AB}}=MN^{2}m^{2}n^{2}>MNm^{2}n^{2}=D_{\mathrm{A}}D_{\mathrm{B}},
\]
which means that all composites in extended classical theory violate
Local Tomography.

\chapter{Operational thermodynamics\label{chap:Operational-thermodynamics}}

After analysing the properties of sharp theories with purification
in great detail, in this chapter finally we move to the actual study
of thermodynamic properties of GPTs, with a special focus on sharp
theories with purification. We mainly examine the simplest instance
of thermodynamics, namely for systems with fixed energy, also known
as \emph{microcanonical thermodynamics}. Even this case will provide
us with a lot of foundational insights. In accordance with recent
thermodynamic results, we will use a resource-theoretic approach to
microcanonical thermodynamics, which will allow us to extend its scope
beyond classical and quantum theory. Here the relevant resource into
play is the purity of states, therefore a resource-theoretic treatment
of microcanonical thermodynamics involves setting up a resource theory
of purity. It turns out that there are three fairly natural choices
for such a resource theory, but in general they give rise to inequivalent
notions of resources \cite{Purity}, unlike in quantum theory \cite{Nicole}.

After a general treatment of microcanonical thermodynamics in causal
GPTs \cite{Purity} in sections~\ref{sec:The-microcanonical-framework}
and \ref{sec:Three-resource-theories}, the axioms of sharp theories
with purification are introduced in the thermodynamic analysis only
from section~\ref{sec:Majorisation-and-unital}. These theories will
be the subject of the rest of the chapter, even when not specified
explicitly. Majorisation will play a central role, for it gives a
necessary criterion for the thermodynamic conversion of states under
all three resource theories in sharp theories with purification. Requiring
it to be sufficient too\textemdash or, in other words, that the three
resource theories generate equivalent preorders\textemdash will result
in a non-trivial constraint on the dynamics of the theory, called
\emph{unrestricted reversibility} \cite{Purity}.

Given the interplay between the three resource theories of purity
and majorisation, we use the diagonalisation of states in sharp theories
with purification to define mixedness monotones as Schur-concave functions
of the spectrum of a state \cite{TowardsThermo}. We show that a large
class of ``spectral'' monotones coincide with the ``non-spectral''
definitions already put forward in the literature \cite{Entropy-Barnum,Entropy-Short,Entropy-Kimura,TowardsThermo}.
An important example of mixedness monotone is Shannon-von Neumann
entropy, which exhibits some properties close to its quantum counterpart
\cite{TowardsThermo}.

Moving beyond the setting of microcanonical thermodynamics, we use
the entropic machinery of sharp theories with purification to define
generalised Gibbs states using Jaynes' maximum entropy principle \cite{Jaynes1,Jaynes2},
thus introducing temperature in GPTs. We use these states to carry
out an operational derivation of Landauer's principle, showing also
how the bounds on energy dissipation can be overcome by using non-classical
correlations, witnessed by negative conditional entropy \cite{TowardsThermo}.

\section{The microcanonical framework\label{sec:The-microcanonical-framework}}

We start this chapter by examining the simplest example of thermodynamic
situation: an isolated thermodynamic system. In the usual statistical
mechanical treatment, this is described by the microcanonical ensemble
\cite{Huang,Kardar1}, extensively studied in classical and quantum
theory. Here we want to extend the microcanonical description to arbitrary
physical theories, to understand if all theories admit a sensible
microcanonical ensemble. Being an isolated system, the energy of a
microcanonical system is known.\footnote{Sometimes, for mathematical convenience, one tolerates a small uncertainty
$\Delta$ on the energy \cite{Huang}.} This of course restricts the microstates of the system to a subset
of the allowed microstates: a submanifold of the phase space in classical
theory, and a subspace of the Hilbert space in quantum theory. One
of the basic principles of statistical mechanics is that we can recover
the thermodynamic properties of microcanonical systems from a suitable
statistical mechanical state \cite{Kardar1}. What state do we assign
to isolated thermodynamic systems? The prescription comes from the
\emph{equal a priori probability postulate}.
\begin{ax}[Equal a priori probability postulate]
The state of a system with a macroscopic constraint is given by the
uniform mixture of all the microstates compatible with that constraint.\footnote{In the quantum case, some authors \cite{Kardar1} formulate this postulate
as the fact that the microcanonical state is an incoherent mixture
of the basis states of the subspace identified by the macroscopic
constraint.}
\end{ax}

Justifying this principle is generally hard, and it usually involves
arguments based on ergodic theory (see \cite[chapter 15]{Fasano-Marmi}
for a review of the approach), at least in classical theory, or more
recently, based on quantum entanglement \cite{Popescu-Short-Winter,Canonical-typicality}.

In the following we want to extend the microcanonical framework to
arbitrary causal theories\footnote{Here Causality is assumed for the sake of simplicity, but it is not
necessary, cf.~\cite{Purity}.}, without assuming any of the axioms of chapter~\ref{chap:Sharp-theories-with}.
The microcanonical framework requires the extension of two ingredients
to GPTs: the implementation of the macroscopic constraint, and the
equal a priori probability postulate. As to the latter, we do not
want to address the thorny issue of justifying it, which inevitably
requires specifying a lot of details about the physical theory under
examination, but rather to identify the conditions that allow one
to formulate this principle in a theory more general than classical
or quantum theory.

\subsection{Theories of systems with constraints}

We start by analysing how to define constrained systems in GPTs, and
we will take inspiration from quantum theory. Note that constrained
systems appear very often in physics and often as a result of conservation
laws. Think e.g.\ of a system of particles confined in a fixed volume,
or a system with conserved total angular momentum. In these scenarios,
it is always meaningful to ask ourselves about the state of ``minimum
information'' compatible with the macroscopic constraints, which
will behave as a sort of generalised microcanonical state. Note that
imposing constraints, at least for quantum and classical systems,
leads to the definition of an \emph{effective system} of smaller dimension
than the original one. Let us see if this idea works in GPTs too,
and if we can set up an effective OPT whose systems are these effective,
constrained, systems.

To have a clearer picture, let us briefly review how the microcanonical
constraint is usually imposed in quantum theory. In this case, we
know the Hamiltonian $H$, and the specific value $E$ of its spectrum.
This restricts the allowed quantum states to (mixtures of) the eigenstates
of $H$ in the eigenspace associated with $E$. These are exactly
the states $\rho$ for which $P_{E}\rho P_{E}=\rho$, where $P_{E}$
is the projector on the eigenspace with eigenvalue $E$. For example,
the system $\mathrm{S}$ could be an electron in a hydrogen atom,
in the absence of external fields. In general, the basis states of
the electron are labelled as $\ket{n,\ell,m_{\ell},m_{s}}$, where
$n$, $\ell$, $m_{\ell}$, and $m_{s}$ are the principal, orbital,
magnetic, and spin quantum number respectively. Suppose we know the
electron is in the ground state, corresponding to $n=1$ and $\ell=m_{\ell}=0$.
In this case, the subspace associated with the ground state is a two-dimensional
one, spanned by the ``spin-up'' and ``spin-down'' states $\ket{n=1,\ell=0,m_{\ell}=0,m_{s}=\frac{1}{2}}$
and $\ket{n=1,\ell=0,m_{\ell}=0,m_{s}=-\frac{1}{2}}$. This constrained
system can be regarded as an \emph{effective qubit}. Clearly, constraints
different from energy can be treated in a similar way. Another example
of an effective qubit is a single photon with wave vector $\mathbf{k}$:
it is enough to restrict ourselves to the two-dimensional space spanned
by the states $\ket{\mathbf{k},H,1}$ and $\ket{\mathbf{k},V,1}$,
corresponding to vertical and horizontal polarisation respectively.
In this case, we can see two constraints working together: a constraint
on the wave vector, and a constraint on the energy of the field, namely
the fact that we are dealing with a single photon.

The important feature of these examples of constrained systems is
that they arise from a \emph{linear} constraint $\mathcal{L}\left(\rho\right)=0$
on the space of density matrices. For instance, in the microcanonical
case described above, the linear map $\mathcal{L}$ is $\mathcal{L}\left(\cdot\right)=P_{E}\left(\cdot\right)P_{E}-\mathcal{I}\left(\cdot\right)$,
$\mathcal{I}$ being the identity channel.

Motivated by the analysis of the quantum case, we define a \emph{constrained
}(or \emph{effective}) \emph{system} in GPTs as the original system
$\mathrm{S}$ with some \emph{linear} maps $\left\{ \mathcal{L}_{i}\right\} _{i=1}^{n}$
implementing the constraints, where each $\mathcal{L}_{i}$ is an
element of $\mathsf{Transf}_{\mathbb{R}}\left(\mathrm{S}\right)$,
the vector space spanned by physical transformations (cf.\ section~\ref{sec:The-probabilistic-structure}):
\[
\mathrm{A}:=\left(\mathrm{S},\left\{ \mathcal{L}_{i}\right\} _{i=1}^{n}\right).
\]
Clearly now we need to specify the set of states, effects, and transformations
for the constrained system $\mathrm{A}$. The states of $\mathrm{A}$
are clearly those that satisfy the constraints enforced by the linear
maps $\left\{ \mathcal{L}_{i}\right\} _{i=1}^{n}$:
\[
\mathsf{St}\left(\mathrm{A}\right)=\left\{ \rho\in\mathsf{St}\left(\mathrm{S}\right):\mathcal{L}_{i}\left(\rho\right)=0,i=1,\dots,n\right\} .
\]
The transformations of the effective system $\mathrm{A}$ are those
transformations of $\mathrm{S}$ that send states of $\mathrm{A}$
to states of $\mathrm{A}$. The effects of $\mathrm{A}$ are just
the effects of the original system $\mathrm{S}$, restricted to the
states in $\mathsf{St}\left(\mathrm{A}\right)$, and possibly identified
if they become tomographically indistinguishable because of the reduced
number of states, as shown in example~\ref{exa:decoherence}.
\begin{rem}
In the classical case, a Hamiltonian is generally of the form $H\left(\mathbf{q},\mathbf{p}\right)=\frac{\left|\mathbf{p}\right|^{2}}{2m}+V\left(\mathbf{q}\right)$.
Thus imposing the microcanonical restriction $H\left(\mathbf{q},\mathbf{p}\right)=E$
leads to the constraint $\frac{\left|\mathbf{p}\right|^{2}}{2m}+V\left(\mathbf{q}\right)=E$
on $\mathbf{p}$ and $\mathbf{q}$, which is \emph{not} linear. From
this situation, one may be tempted to think that treating constraints
as linear in all physical theories is not the right way. In fact,
there is no contradiction to what we wrote above: the constraint $H\left(\mathbf{q},\mathbf{p}\right)=E$
is not linear on the phase space, i.e.\ the set of normalised pure
states of classical theory. However, our constraints are applied to
\emph{all} states, not just the pure ones, so the right space on which
to study the microcanonical restriction is the space of all probability
density functions $\rho\left(\mathbf{q},\mathbf{p}\right)$ on the
phase space. In that space, the microcanonical constraint becomes
that the probability density function $\rho\left(\mathbf{q},\mathbf{p}\right)$
be supported on the submanifold $H\left(\mathbf{q},\mathbf{p}\right)=E$,
which is a linear constraint on $\rho$.
\end{rem}

Clearly, the formalism of constrained systems can be applied to more
general situations than microcanonical thermodynamics. Example~\ref{exa:decoherence}
shows that even classical sub-theories of a given GPT can be studied
in the same way \cite{Objectivity}.
\begin{example}
\label{exa:decoherence}Given a pure maximal set $\left\{ \alpha_{i}\right\} _{i=1}^{d}$
of a system $\mathrm{S}$ of a causal theory, let $\boldsymbol{\alpha}$
be its convex hull: $\boldsymbol{\alpha}=\mathrm{Conv}\left\{ \alpha_{i}:i=1,\ldots,d\right\} $,
which represents the states of a classical sub-theory. Define the
decoherence on $\boldsymbol{\alpha}$ as a channel $D_{\boldsymbol{\alpha}}$
such that $D_{\boldsymbol{\alpha}}\rho\in\boldsymbol{\alpha}$ for
every state $\rho$, and $D_{\boldsymbol{\alpha}}\gamma=\gamma$ for
every state $\gamma$ in $\boldsymbol{\alpha}$ \cite{Objectivity}.
By this very definition, it is immediate to see that the classical
states (i.e.\ the states in $\boldsymbol{\alpha}$) are exactly those
for which $D_{\boldsymbol{\alpha}}\rho=\rho$. Hence the decoherence
$D_{\boldsymbol{\alpha}}$ can be viewed as implementing a constraint
on the states of a system, as explained above. What happens to the
effects of this constrained subsystem? According to the recipe presented
above, we simply restrict the effects of the original theory to the
effective system $\boldsymbol{\alpha}$. Since now there are fewer
states in $\boldsymbol{\alpha}$ than in the original theory, we must
identify those effects that are no longer tomographically distinct
on $\boldsymbol{\alpha}$. More precisely, we can introduce the following
equivalence relation on the original set of effects $\mathsf{Eff}\left(\mathrm{S}\right)$:
$e\sim_{\boldsymbol{\alpha}}f$ if $\left(e\middle|\gamma\right)=\left(f\middle|\gamma\right)$
for every state $\gamma$ in $\boldsymbol{\alpha}$. Therefore, the
set of effects of the classical sub-theory is the set of equivalence
classes $\mathsf{Eff}\left(\mathrm{S}\right)/\boldsymbol{\alpha}:=\mathsf{Eff}\left(\mathrm{S}\right)/\sim_{\boldsymbol{\alpha}}$.

To complete our analysis, let us show that $\mathsf{Eff}\left(\mathrm{S}\right)/\boldsymbol{\alpha}$
is actually the set of effects of some classical theory. Recall that
in classical theory, every element in the cone of effects arises as
a conical combination of the effects that distinguish the pure states.
In our setting this means checking that every element of $\mathsf{Eff}_{+}\left(\mathrm{S}\right)/\boldsymbol{\alpha}$
arises as a conical combination\footnote{Note that it is not hard to see that $\mathsf{Eff}_{+}\left(\mathrm{S}\right)/\boldsymbol{\alpha}$
is still a cone, with the sum and the multiplication by a scalar inherited
from $\mathsf{Eff}_{+}\left(\mathrm{S}\right)$.} of the equivalence classes $\left[a_{i}\right]$ of the effects that
distinguish the pure states $\alpha_{i}$ in $\boldsymbol{\alpha}$.
Consider a generic element $\xi$ in $\mathsf{Eff}_{+}\left(\mathrm{S}\right)$,
and let us show that it is in the same equivalence class as $\xi'=\sum_{i=1}^{d}\lambda_{i}a_{i}$,
where $\lambda_{i}=\left(\xi\middle|\alpha_{i}\right)$ for all $i$.
By linearity, to check the equivalence of two elements of $\mathsf{Eff}_{+}\left(\mathrm{S}\right)$,
it is enough to check that they produce the same numbers when applied
to all pure states $\alpha_{j}$. Now,
\[
\left(\xi'\middle|\alpha_{j}\right)=\sum_{i=1}^{d}\lambda_{i}\left(a_{i}\middle|\alpha_{j}\right)=\lambda_{j}=\left(\xi\middle|\alpha_{j}\right)
\]
This shows that the restricted effect cone $\mathsf{Eff}_{+}\left(\mathrm{S}\right)/\boldsymbol{\alpha}$
of the sub-theory is actually a classical effect cone, generated by
the effects that distinguish the pure states in $\boldsymbol{\alpha}$.
This shows that the formalism of constrained systems works well also
for the study of the emergence of classicality in GPTs.
\end{example}

So far, we have defined effective systems at the single-system level.
In order to complete our picture, we need to define the composition
of constrained systems too. Consider two effective systems $\mathrm{A}:=\left(\mathrm{S}_{\mathrm{A}},\left\{ \mathcal{L}_{\mathrm{A},i}\right\} _{i=1}^{n}\right)$
and $\mathrm{B}:=\left(\mathrm{S}_{\mathrm{B}},\left\{ \mathcal{L}_{\mathrm{B},j}\right\} _{j=1}^{m}\right)$.
A natural way to define the effective composite system $\mathrm{AB}$
is to select the states of the unconstrained composite system $\mathrm{S}_{\mathrm{A}}\mathrm{S}_{\mathrm{B}}$
that satisfy \emph{both} constraints\textemdash i.e.\ to select the
density matrices $\rho_{\mathrm{S}_{\mathrm{A}}\mathrm{S}_{\mathrm{B}}}$
such that 
\begin{equation}
\left\{ \begin{array}{l}
\left(\mathcal{L}_{\mathrm{A},i}\otimes\mathcal{I}_{\mathrm{S}_{\mathrm{B}}}\right)\left(\rho_{\mathrm{S}_{\mathrm{A}}\mathrm{S}_{\mathrm{B}}}\right)=0\\
\left(\mathcal{I}_{\mathrm{S}_{\mathrm{A}}}\otimes\mathcal{L}_{\mathrm{B},j}\right)\left(\rho_{\mathrm{S}_{\mathrm{A}}\mathrm{S}_{\mathrm{B}}}\right)=0
\end{array}\right.,\label{eq:compose constraints}
\end{equation}
for all $i\in\left\{ 1,\dots,n\right\} $, and all $j\in\left\{ 1,\dots,m\right\} $.

When the effective systems $\mathrm{A}$ and $\mathrm{B}$ result
from an energy constraint, the effective system $\mathrm{AB}$ describes
a system consisting of two parts, each with its own well-defined energy.
In the quantum case, the constraints~\eqref{eq:compose constraints}
amount to $\left(P_{E_{\mathrm{A}}}\otimes Q_{E_{\mathrm{B}}}\right)\rho_{\mathrm{S}_{\mathrm{A}}\mathrm{S}_{\mathrm{B}}}\left(P_{E_{\mathrm{A}}}\otimes Q_{E_{\mathrm{B}}}\right)=\rho_{\mathrm{S}_{\mathrm{A}}\mathrm{S}_{\mathrm{B}}}$,
where $E_{\mathrm{A}}$ and $E_{\mathrm{B}}$ are the energies of
the two local systems, and $P_{E_{\mathrm{A}}}$ and $Q_{E_{\mathrm{B}}}$
are the projectors on the corresponding eigenspaces. Note that this
differs from the case where in the composition of two microcanonical
systems only a restriction on the \emph{global} energy is imposed.
In this latter case, the states satisfy the \emph{weaker} condition
$\Pi_{E_{\mathrm{A}}+E_{\mathrm{B}}}\rho_{\mathrm{S}_{\mathrm{A}}\mathrm{S}_{\mathrm{B}}}\Pi_{E_{\mathrm{A}}+E_{\mathrm{B}}}=\rho_{\mathrm{S}_{\mathrm{A}}\mathrm{S}_{\mathrm{B}}}$,
where $\Pi_{E_{\mathrm{A}}+E_{\mathrm{B}}}$ is the projector on the
eigenspace of $H_{\mathrm{S}_{\mathrm{A}}}+H_{\mathrm{S}_{\mathrm{B}}}$
with eigenvalue $E_{\mathrm{A}}+E_{\mathrm{B}}$.

The reason why we adopt eq.~\eqref{eq:compose constraints} as the
rule of composition of effective systems is because we want to keep
the systems $\mathrm{A}$ and $\mathrm{B}$ as independently constrained
systems, which can be addressed separately on their own, even in a
composite setting, still retaining their status of effective systems.
This property is particularly important if we want to develop an operational
theory of effective systems, so that we can treat them as actual physical
systems, forgetting that they arise from constraints imposed on a
physical system. In the case of the effective qubits arising from
photon polarisation, this fact is particularly apparent. Consider
two photons of different spatial modes, with wave vectors $\mathbf{k}_{\mathrm{A}}$
and $\mathbf{k}_{\mathrm{B}}$. If we put a constraint on the \emph{total}
energy of the two photons, one of the allowed states would be
\[
\ket{\Psi}=\frac{1}{\sqrt{2}}\left(\ket{\mathbf{k}_{\mathrm{A}},H,2}\ket{\mathbf{k}_{\mathrm{B}},H,0}+\ket{\mathbf{k}_{\mathrm{A}},H,0}\ket{\mathbf{k}_{\mathrm{B}},H,2}\right),
\]
where the third entry denotes the number of photons, but this cannot
be interpreted as a state of two single photons, because the number
of photons is undefined. Imposing the constraint~\eqref{eq:compose constraints}
avoids this problem. For this reason, we will use the notation $\mathrm{AB}$
for effective systems defined by the constraint~\eqref{eq:compose constraints}.

In summary, given a theory and a set of constraints composed as in
eq.~\eqref{eq:compose constraints}, one can build a new \emph{effective
theory}, which consists only of effective systems. In the microcanonical
case, we can build an effective theory where every system has definite
energy, and where all subsystems of a composite systems have definite
energy too. For a given system $\mathrm{A}$ in such a theory, all
the states in $\mathsf{St}\left(\mathrm{A}\right)$ have the same
energy by construction. Likewise, all the transformations in $\mathsf{Transf}\left(\mathrm{A}\right)$
will be energy-preserving. For every pair of systems $\mathrm{A}$
and $\mathrm{B}$, the composite system $\mathrm{AB}$ consists of
two parts, each of which with its own, well-defined energy. The joint
transformations in $\mathsf{Transf}\left(\mathrm{AB}\right)$ will
be interpreted as operations that preserve the energy of the first
part \emph{and} the energy of the second part.

The advantage of this effective description is that we need not specify
the constraints: in principle, \emph{every} linear constraint can
fit into the framework. In this way, we can circumvent the thorny
issue of defining the notion of Hamiltonian in GPTs \cite{Hamiltonian-GPTs}
(cf.\ section~\ref{sec:Generalised-Gibbs-states}): in the effective
description, we can simply regard each effective system as a system
with trivial Hamiltonian, which assigns the same energy to all states
of the system. Moreover, since effective system have the same operational
structure of unconstrained ones, we can forget we are dealing with
constrained systems, and treat them as ordinary ones. This is why
we call them ``effective systems''. When dealing with microcanonical
systems, one should always bear in mind that they arise from a constraint
on the energy of the system, but for simplicity we will not mention
this explicitly henceforth.

\subsection{The principle of equal a priori probability}

Let us now move to the equal a priori probability principle, according
to which one should assign the same probability to all the microstates
of the system compatible with a given macroscopic constraint. But
what are the microstates? To understand it, let us see what happens
in classical and quantum theory. In classical theory, they are the
points in the phase space; in quantum theory they are vectors in the
Hilbert space. In both cases they are normalised pure states. Therefore
it is natural for us to define microstates as \emph{normalised pure
states}, representing those preparations of the system that are both
deterministic and maximally fine-grained. Then the principle of equal
a priori probability states that the system should be described by
a uniform mixture of all deterministic pure states satisfying the
constraint. For example, the microcanonical state of a finite-dimensional
quantum system at energy $E$ is described by the density matrix $\chi_{E}:=\int_{S_{E}}p_{E}\ket{\psi}\bra{\psi}\mathrm{d}\psi,$
where $S_{E}$ is the set of pure states in the eigenspace corresponding
to the eigenvalue $E$, and $p_{E}\mathrm{d}\psi$ is the uniform
probability measure over $S_{E}$. In the effective picture, where
we call this subspace $\mathrm{A}$, the microcanonical state is nothing
but the\emph{ maximally mixed state} $\chi_{\mathrm{A}}:=\int\ket{\psi}\bra{\psi}\mathrm{d}\psi$
where $\mathrm{d}\psi$ is the uniform probability measure over the
pure states of the system.

In GPTs the key problem is to define what we mean by ``equal a priori
probability'', by which we could define the microcanonical state
as 
\begin{equation}
\chi:=\int\psi\thinspace\mathrm{d}\psi.\label{eq:chi integrale}
\end{equation}
In quantum mechanics there is a canonical choice: the unitarily invariant
probability measure on the pure states of the system. The natural
extension to GPTs is to consider the probability measures that are
invariant under all reversible channels. The problem is, however,
that, in general, there may be more than one invariant probability
measure, as illustrated in the following example.
\begin{example}
\begin{figure}
\begin{centering}
\subfloat[\label{fig:semicerchio}]{\begin{centering}
\includegraphics[scale=0.8]{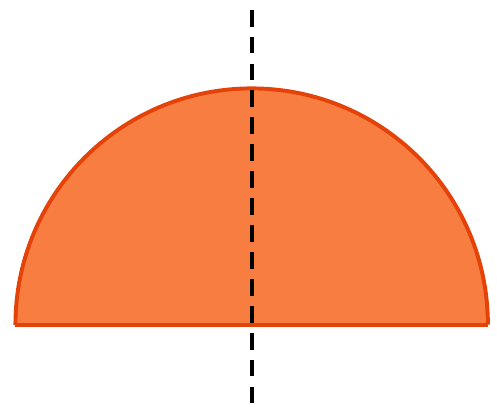}
\par\end{centering}
}\qquad{}\qquad{}\qquad{}\subfloat[\label{fig:cerchio}]{\begin{centering}
\includegraphics[scale=0.8]{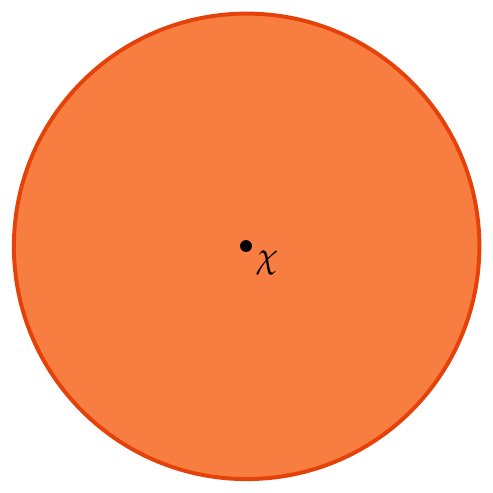}
\par\end{centering}
}
\par\end{centering}
\caption{Two different state spaces. In fig.~\ref{fig:semicerchio}, pure
states form a half-circle. Owing to the limited symmetry of the state
space, there is no canonical notion of equal a priori probability
on the set of pure states. As such, there are a lot of invariant states:
all the points of the state space on the symmetry axis. For the set
in fig.~\ref{fig:cerchio}, pure states form a circle, and the notion
of uniform probability distribution is uniquely defined. This means
that there is a unique microcanonical state $\chi$, which is the
only invariant state under all the symmetries of the disk.}
\end{figure}
Consider a system where the state space is a half-disk, like in fig.~\ref{fig:semicerchio}.
Pure states are the states on the half-circle, and they can be parametrised
by a polar angle $\vartheta\in\left[0,\pi\right]$. Now, reversible
channels must be symmetry transformations of the state space. For
the half-disk, the only symmetry transformations are the identity
and the reflection across the vertical symmetry axis (the black dashed
line in fig.~\ref{fig:semicerchio}). Hence, every probability density
function that assigns the same value to the points $\vartheta$ and
$\pi-\vartheta$ is guaranteed to be invariant under reversible channels.
This means that the ``equal a priori probability'' is \emph{not}
unique.

The situation is different if the state space of the system is a full
disk, as illustrated in fig.~\ref{fig:cerchio}, corresponding to
a rebit, a two-level system in real quantum mechanics \cite{Stuckelberg,Araki-real,Wootters-real,Hardy-real}.
In this case, every rotation of the disk is (at least in principle)
a reversible channel of the system. The invariant probability measure
is unique and given by the probability density function $p\left(\vartheta\right)=\frac{1}{2\pi}$. 
\end{example}

This example shows that there exist probabilistic theories where the
notion of ``equal a priori probability'' on pure states is not uniquely
defined. The physical implications of this is that in general there
is \emph{not} a unique microcanonical state. Indeed, by eq.~\eqref{eq:chi integrale},
$\mathrm{d}\psi$ is that invariant probability measure on pure states,
therefore different choices of $\mathrm{d}\psi$ correspond to different
microcanonical states. This is of course in contrast with the usual
statistical mechanical treatment, which assumes that the macroscopic
constraint is enough to determine the equilibrium state. In other
words, there are no further constraints coming from additional conserved
quantities that restrict the evolution of the system in the set of
pure states. In order to formulate the principle of equal a priori
probability, we introduce the following condition.
\begin{condition}
\label{cond:unique microcanonical}For every (finite-dimensional)
system there exists a \emph{unique} invariant probability measure
on normalised pure states.
\end{condition}

We saw that some systems do not admit a unique invariant probability
measure, like the one in fig.~\ref{fig:semicerchio}. However, are
we sure that at least one invariant probability measure always exist?
The answer is positive, and now we will show how to construct it.
Recall that the group of reversible channels $G$ is a compact group
(cf.\ subsection~\ref{subsec:The-group-of-reversible}), and that
compact groups admit a finite Haar measure $h$, which can be renormalised
so that $h\left(G\right)=1$ (see subsection~\ref{subsec:The-group-of-reversible}).
Now, there is a (continuous) group action of $G$ on normalised pure
states $\cdot:G_{\mathrm{A}}\times\mathsf{PurSt}_{1}\left(\mathrm{A}\right)\rightarrow\mathsf{PurSt}_{1}\left(\mathrm{A}\right)$,
given by $\mathcal{U}\psi$, for every $\mathcal{U}\in G_{\mathrm{A}}$,
and every $\psi\in\mathsf{PurSt}_{1}\left(\mathrm{A}\right)$. The
idea is to induce a probability measure on the set of normalised pure
states from the Haar probability measure on $G_{\mathrm{A}}$. To
do this, let us \emph{fix} a normalised pure state $\psi_{0}$ in
the action of $G_{\mathrm{A}}$, and consider the function $F_{\psi_{0}}:G_{\mathrm{A}}\rightarrow\mathsf{PurSt}_{1}\left(\mathrm{A}\right)$
such that $\mathcal{U}\mapsto\mathcal{U}\psi_{0}$. Since $F_{\psi_{0}}$
is continuous, we can induce a probability measure $\mu_{\psi_{0}}$
on $\mathsf{PurSt}_{1}\left(\mathrm{A}\right)$ (called the ``image
measure'' \cite{Folland_real}) by setting
\[
\mu_{\psi_{0}}\left(S\right):=h\left(F_{\psi_{0}}^{-1}\left(S\right)\right)
\]
for every Borel subset of $\mathsf{PurSt}_{1}\left(\mathrm{A}\right)$.
Let us show that this probability measure $\mu_{\psi_{0}}$ is invariant
under the action of $G$. We have $\mu_{\psi_{0}}\left(\mathcal{U}S\right)=h\left(F_{\psi_{0}}^{-1}\left(\mathcal{U}S\right)\right)$,
where $F_{\psi_{0}}^{-1}\left(\mathcal{U}S\right)=\left\{ \mathcal{V}\in G:\mathcal{V}\psi_{0}\in\mathcal{U}S\right\} $.
Now $F_{\psi_{0}}^{-1}\left(\mathcal{U}S\right)=\mathcal{U}F_{\psi_{0}}^{-1}\left(S\right)$.
To see it, take $\mathcal{V}\in G$ such that $\mathcal{V}\psi_{0}\in\mathcal{U}S$.
Now, this implies that $\mathcal{U}^{-1}\mathcal{V}\psi_{0}\in S$,
so $\mathcal{U}^{-1}\mathcal{V}\in F_{\psi_{0}}^{-1}\left(S\right)$,
and clearly $\mathcal{V}\in F_{\psi_{0}}^{-1}\left(\mathcal{U}S\right)$
can be written as $\mathcal{U}$ times an element in $F_{\psi_{0}}^{-1}\left(S\right)$.
Therefore $F_{\psi_{0}}^{-1}\left(\mathcal{U}S\right)\subseteq\mathcal{U}F_{\psi_{0}}^{-1}\left(S\right)$.
To show the other inclusion, take $\mathcal{V}\in G$ such that $\mathcal{V}\psi_{0}\in S$.
Clearly $\mathcal{U}\mathcal{V}\psi_{0}\in\mathcal{U}S$, so $F_{\psi_{0}}^{-1}\left(\mathcal{U}S\right)\supseteq\mathcal{U}F_{\psi_{0}}^{-1}\left(S\right)$,
and therefore $F_{\psi_{0}}^{-1}\left(\mathcal{U}S\right)=\mathcal{U}F_{\psi_{0}}^{-1}\left(S\right)$.
In conclusion
\[
\mu_{\psi_{0}}\left(\mathcal{U}S\right)=h\left(\mathcal{U}F_{\psi_{0}}^{-1}\left(S\right)\right)=h\left(F_{\psi_{0}}^{-1}\left(S\right)\right)=\mu_{\psi_{0}}\left(S\right),
\]
where we have used the invariance of the Haar probability measure
$h$. This shows that the image probability measure $\mu_{\psi_{0}}$
on $\mathsf{PurSt}_{1}\left(\mathrm{A}\right)$ is invariant under
reversible channels. This means that there \emph{always} exists an
invariant probability measure on normalised pure states.

Condition~\ref{cond:unique microcanonical} enforces the uniqueness
of such a measure. By its very definition $\mu_{\psi_{0}}$ depends
on the pure state $\psi_{0}$ we fix, so in general we expect that
different choices of $\psi_{0}$ will result in different invariant
probability measures on the set of normalised pure states. Let us
examine this issue in greater detail. Consider now $\psi'_{0}$ such
that $\psi'_{0}=\mathcal{U}\psi_{0}$ for some $\mathcal{U}\in G$.
What is the relationship between $\mu_{\psi_{0}}$ and $\mu_{\psi'_{0}}$?
To answer the question, we need to understand how to write $F_{\psi'_{0}}^{-1}\left(S\right)$
in terms of $F_{\psi_{0}}^{-1}\left(S\right)$ for any Borel set $S$.
We have $F_{\psi'_{0}}^{-1}\left(S\right)=F_{\psi_{0}}^{-1}\left(S\right)\mathcal{U}^{-1}$.
To see it, consider $\mathcal{V}\in G$ such that $\mathcal{V}\psi'_{0}\in S$,
which means $\mathcal{V}\mathcal{U}\psi_{0}\in S$. Therefore $\mathcal{V}$
can be written as an element of $F_{\psi_{0}}^{-1}\left(S\right)$
(i.e.\ $\mathcal{V}\mathcal{U}$) times $\mathcal{U}^{-1}$. This
shows $F_{\psi'_{0}}^{-1}\left(S\right)\subseteq F_{\psi_{0}}^{-1}\left(S\right)\mathcal{U}^{-1}$
. To prove the other inclusion, consider $\mathcal{V}\in G$ such
that $\mathcal{V}\psi_{0}\in S$, and let us show that $\mathcal{V}\mathcal{U}^{-1}$
is such that $\mathcal{V}\mathcal{U}^{-1}\psi'_{0}\in S$. This is
immediate, as $\psi'_{0}=\mathcal{U}\psi_{0}$. This proves that $F_{\psi'_{0}}^{-1}\left(S\right)\supseteq F_{\psi_{0}}^{-1}\left(S\right)\mathcal{U}^{-1}$,
so $F_{\psi'_{0}}^{-1}\left(S\right)=F_{\psi_{0}}^{-1}\left(S\right)\mathcal{U}^{-1}$.
Hence
\[
\mu_{\mathcal{U}\psi_{0}}\left(S\right)=h\left(F_{\mathcal{U}\psi_{0}}^{-1}\left(S\right)\right)=h\left(F_{\psi_{0}}^{-1}\left(S\right)\mathcal{U}^{-1}\right)=h\left(F_{\psi_{0}}^{-1}\left(S\right)\right)=\mu_{\psi_{0}}\left(S\right),
\]
where we have used the invariance properties of the Haar probability
measure once more. This shows that pure states in the same orbit of
the group action generate the same invariant probability measure.
Therefore having a unique orbit (i.e.\ a transitive action) is a
\emph{sufficient} condition to have a unique invariant probability
measure on normalised pure states (and hence a unique microcanonical
state).

This condition is \emph{necessary} too. To prove it, suppose there
is more than one orbit, and take a pure state $\psi_{0}$ in an orbit,
and another $\psi'_{0}$ in a different one. Let us show that the
induced probability measures are different. Take $S$ to be the orbit
$G\psi_{0}$ of $\psi_{0}$; since the action is continuous, this
is a closed set, hence a Borel set. Then $F_{\psi_{0}}^{-1}\left(G\psi_{0}\right)=\left\{ \mathcal{U}\in G:\mathcal{U}\psi_{0}\in G\psi_{0}\right\} =G$,
by definition of orbit of $\psi_{0}$. Hence $\mu_{\psi_{0}}\left(G\psi_{0}\right)=h\left(G\right)=1$.
On the other hand, for the same orbit $G\psi_{0}$, we have $F_{\psi'_{0}}^{-1}\left(G\psi_{0}\right)=\left\{ \mathcal{U}\in G:\mathcal{U}\psi'_{0}\in G\psi_{0}\right\} =\varnothing$,
because orbits are disjoint sets, so there are no elements in the
orbits of $\psi'_{0}$ also in the orbit of $\psi_{0}$. Hence $\mu_{\psi'_{0}}\left(G\psi_{0}\right)=h\left(\varnothing\right)=0$.
In conclusion, we have found a set to which the probability measures
generated by $\psi_{0}$ and $\psi'_{0}$ assign different values,
therefore they are different measures.\footnote{In fact, we can say even more: the probability measures associated
with different orbits are \emph{mutually singular} \cite{Folland_real}:
there exist two complementary subsets of $\mathsf{PurSt}_{1}\left(\mathrm{A}\right)$,
one with zero probability for one measure, and the other with zero
probability for the other. This means that the two probability measures
``live on disjoint sets''.}

To summarise: an invariant probability measure on $\mathsf{PurSt}_{1}\left(\mathrm{A}\right)$
always exists, and it is unique if and only if the action of reversible
channels is transitive on $\mathsf{PurSt}_{1}\left(\mathrm{A}\right)$.
We collect all these remarks in the following theorem.
\begin{thm}
\label{thm:unique probability}For every finite system $\mathrm{A}$,
the following are equivalent:
\begin{enumerate}
\item Condition~\ref{cond:unique microcanonical} is satisfied.
\item For every pair of normalised pure states, there exists a reversible
channel connecting them.
\end{enumerate}
\end{thm}

Note that enforcing the uniqueness of the invariant probability measure
on normalised pure states leads to a non-trivial requirement\textemdash transitivity\textemdash ,
which has appeared several times, either directly or indirectly, in
various reconstructions of quantum theory \cite{Hardy-informational-1,Chiribella-informational,Brukner,Masanes-physical-derivation,Hardy-informational-2,Masanes+all}.
Theorem~\ref{thm:unique probability} provides one more motivation
for this condition, this time from a completely different perspective,
a thermodynamic one. The transitivity requirement already appeared
in \cite{Chiribella-Scandolo-entanglement} as a requirement for a
``canonical resource theory of purity'' in the context of the duality
between pure-state entanglement and purity in GPTs.

Note that transitivity means that, at least in principle, the dynamics
allowed by the physical theory enable one to explore the entire set
of normalised pure states compatible with the macroscopic constraint
on the energy, a fact that is often assumed or remarked in textbook
presentations about the microcanonical ensemble \cite{Kardar1}.

We finish this subsection by noting that sharp theories with purification,
having a transitive action (see section~\ref{sec:The-axioms-and}),
satisfy condition~\ref{cond:unique microcanonical}. With these theories
we are on the right track to have a sensible microcanonical thermodynamics.

\subsection{The microcanonical state}

Once we know that condition~\ref{cond:unique microcanonical} holds
in a physical theory, we can use eq.~\eqref{eq:chi integrale} to
construct the microcanonical state, where now $\mathrm{d}\psi$ is
the unique invariant probability measure over the normalised pure
states of a system. The convexity of the state space guarantees that
the microcanonical state is indeed a state. Moreover, since the state
space is finite-dimensional, thanks to Carathéodory's theorem for
convex geometry \cite{Caratheodory,Steinitz}, it is possible to replace
the integral in eq.~\eqref{eq:chi integrale} with a \emph{finite}
sum. This means that the microcanonical state can (in principle) be
generated by picking pure states at random from a finite set.

The following proposition highlights two important properties of the
microcanonical state.
\begin{prop}
\label{prop:properties microcanonical}The microcanonical state $\chi$
satisfies the following two properties:
\begin{enumerate}
\item it is invariant under all reversible channels of the system;\label{enu:invariance chi}
\item it can be generated from any normalised pure state by a fixed random
reversible dynamic.\label{enu:RaRe chi}
\end{enumerate}
\end{prop}

\begin{proof}
Let us prove the two properties.
\begin{enumerate}
\item For every reversible channel $\mathcal{U}$, one has
\begin{equation}
\mathcal{U}\chi=\int\mathcal{U}\psi\thinspace\mathrm{d}\psi=\int\psi'\thinspace\mathrm{d}\left(\mathcal{U}^{-1}\psi'\right)=\int\psi'\thinspace\mathrm{d}\psi'=\chi,\label{eq:chi invariant reversible}
\end{equation}
where we have defined $\psi':=\mathcal{U}\psi$, and we have used
the invariance of the probability measure $\mathrm{d}\psi$.
\item Consider the transformation $\mathcal{T}=\int_{G}\mathcal{U}\thinspace\mathrm{d}\mathcal{U}$,
where $\mathrm{d}\mathcal{U}$ is the Haar probability measure on
the group of reversible channels. $\mathcal{T}$ is a channel because
\[
u\mathcal{T}=\int_{G}u\mathcal{U}\thinspace\mathrm{d}\mathcal{U}=u\int_{G}\thinspace\mathrm{d}\mathcal{U}=u.
\]
The channel $\mathcal{T}$ maps every normalised pure state $\psi$
to the microcanonical state: indeed, one has 
\[
\mathcal{T}\psi=\int_{G}\mathcal{U}\psi\thinspace\mathrm{d}\mathcal{U}=\chi,
\]
like in eq.~\eqref{eq:chi invariant reversible}. Since we are working
with finite-dimensional systems, the integral in the definition of
$\mathcal{T}$ can be replaced by a \emph{finite} convex combination
of reversible channels, by Carathéodory's theorem. Clearly, this writing
of $\mathcal{T}$ does \emph{not} depend on the pure state $\psi$
it is applied to. Therefore, there is a fixed convex combination of
reversible channels\textemdash a random reversible dynamic\textemdash such
that when it is applied to any pure state, it yields the microcanonical
state.
\end{enumerate}
\end{proof}
Property~\ref{enu:invariance chi} expresses the fact that the microcanonical
state is an \emph{equilibrium state}, in the sense that it does not
evolve under any of the reversible dynamics compatible with the macroscopic
constraint. Note that the notion of equilibrium here is different
from the notion of thermal equilibrium, which refers to interactions
with an external bath. This is rather a \emph{dynamical equilibrium}:
the probability assignments to pure states made in the definition
of the microcanonical state are stable under all possible evolutions
of the system.

By condition~\ref{cond:unique microcanonical} we are dealing with
GPTs with transitive action of reversible channels on pure states,
and they have a unique invariant state (cf.\ proposition~\ref{prop:uniqueness invariant}).
Thus property~\ref{enu:invariance chi} identifies microcanonical
states with the invariant states of those GPTs: looking for the microcanonical
state is equivalent to searching for the invariant state. In the following,
given the overall thermodynamic character of this chapter we will
stick to ``microcanonical states'' as terminology.
\begin{example}
In fig.~\eqref{fig:cerchio} we showed the state space of a rebit,
a disk, which admits a unique invariant probability density function
on pure states. In this case the microcanonical state is the unique
invariant state under all the symmetries of the disk, namely its centre.

In any sharp theory with purification, given any pure maximal set
$\left\{ \alpha_{i}\right\} _{i=1}^{d}$ the microcanonical state
is the invariant state, diagonalised as 
\begin{equation}
\chi=\frac{1}{d}\sum_{i=1}^{d}\alpha_{i}.\label{eq:microcanonical diagonalised}
\end{equation}
In quantum theory, the microcanonical state is the maximally mixed
state $\frac{1}{d}\mathbf{1}$.

It is worth noting that eq.~\eqref{eq:microcanonical diagonalised}
is often taken as the \emph{definition} of microcanonical state in
quantum statistical mechanics \cite{Huang,Kardar1}. According to
it, the microcanonical state is defined as the uniform mixture of
orthonormal basis states of the eigenspace of the fixed energy. Proposition~\ref{prop:diagonalization chi d-level 2}
shows that a similar definition is possible in every sharp theory
with purification. One might be tempted to use eq.~\eqref{eq:microcanonical diagonalised}
to define the microcanonical state in \emph{arbitrary} physical theories.
However, the fact that the state resulting from the uniform mixture
of the $\alpha_{i}$'s is independent of the choice of maximal set
is not guaranteed to hold in every theory. Moreover, in general there
is no relationship between the writing of eq.~\eqref{eq:microcanonical diagonalised}
and that of eq.~\eqref{eq:chi integrale}. For this reason, we prefer
to stick to eq.~\eqref{eq:chi integrale}, and define the microcanonical
state as the uniform mixture of \emph{all} pure states with a given
energy, rather than the uniform mixture of a particular pure maximal
set. From a physical angle, the uniform mixture of all pure states
represents the result of fully uncontrolled, but energy conserving
fluctuations in the experimental setup. This describes the situation
of a total lack of knowledge besides the knowledge of the value of
the energy: not even the ``energy eigenbasis''\textemdash the pure
maximal set $\left\{ \alpha_{i}\right\} _{i=1}^{d}$\textemdash is
known.
\end{example}

Property~\ref{enu:RaRe chi} instead refers to the fact that the
system can, at least in principle, be brought to equilibrium, so there
exists a sort of thermalisation process. Physically, the fact that
this process is a random reversible one can be interpreted as a situation
where the experimenter has no full control on the preparation, but
has control on the dynamics of the system, maybe through some classical
control fields \cite{Chiribella-Scandolo-entanglement}. In this picture,
property~\ref{enu:RaRe chi} guarantees that the experimenter can
prepare the microcanonical state by drawing the parameters of their
control fields at random. Further along this line, one can also imagine
scenarios where the randomisation occurs naturally as a result of
fluctuations of the fields. Property~\ref{enu:RaRe chi} is important
for the resource-theoretic approach (section~\ref{sec:Three-resource-theories}),
because it guarantees that the microcanonical state is ``easy to
prepare'', therefore it can be rightfully considered a free state.
Indeed we can reach the microcanonical state $\chi$ from \emph{any}
state $\rho$, even from mixed states. This follows from the linearity
of $\mathcal{T}$, because every $\rho$ can be written as a convex
combination of pure states; therefore $\mathcal{T}\rho=\chi$.

\subsection{Composition of microcanonical states}

At the level of single systems, condition~\ref{cond:unique microcanonical}
guarantees the existence of a microcanonical state. But how does the
microcanonical state behave under the composition of systems? It is
important to answer this question because, from the operational point
of view, it is natural to consider scenarios where the experimenter
has more than one system at their disposal. This aspect is not so
stressed in traditional textbook presentations of statistical mechanics,
where the composition of microcanonical systems (done fixing only
the global energy, therefore differently from us) is only mentioned
to show that in the thermodynamic limit if a composite system is in
the microcanonical state, so are its components \cite{Huang,Kardar1}.

Here we want to take a closer look at the composition of microcanonical
systems, according the rules of eq.~\eqref{eq:compose constraints}.
Composition is especially important in the context of resource theories,
where it is natural to ask how resources interact when combined together.
Think for example of the quantum resource theory of noisy operations
\cite{Local-information,Horodecki-Oppenheim,Nicole}. There microcanonical
states are treated as free. Since the experimenter can generate the
microcanonical states $\chi_{\mathrm{A}}$ and $\chi_{\mathrm{B}}$
at no cost, then they can generate the product state $\chi_{\mathrm{A}}\otimes\chi_{\mathrm{B}}$
at no cost too.

If we insist that microcanonical states are the \emph{only} free states
in a resource-theoretic treatment of microcanonical thermodynamics,
the product state $\chi_{\mathrm{A}}\otimes\chi_{\mathrm{B}}$ of
two microcanonical states must be the microcanonical state $\chi_{\mathrm{AB}}$
of the composite system $\mathrm{AB}$. This is consistent with the
intuitive interpretation of the microcanonical state as ``the state
of minimum information compatibly with the constraints''. In other
words, if one has minimum information on the parts of a system, then
one has minimum information about the whole. This is indeed the case
in quantum theory, where the product of two maximally mixed states
is maximally mixed: $\frac{1}{d_{\mathrm{A}}}\mathbf{1}\otimes\frac{1}{d_{\mathrm{B}}}\mathbf{1}=\frac{1}{d_{\mathrm{AB}}}\mathbf{1}$.

Recall that here we are dealing with effective systems, with their
corresponding constraints. For energy constraints, the composite of
two effective systems $\mathrm{A}$ and $\mathrm{B}$ is defined as
a system consisting of two parts, each constrained to a specific value
of the energy. Consistently with this interpretation, the microcanonical
state of system $\mathrm{AB}$ is the ``maximally mixed state''
in the set of states with fixed local energies. It is therefore natural
to require that minimum information about the parts should imply minimum
information about the whole.
\begin{condition}
\label{cond:informational equibrium}The microcanonical state of a
composite system is the product of the microcanonical states of its
components. In formula:
\begin{equation}
\chi_{\mathrm{A}}\otimes\chi_{\mathrm{B}}=\chi_{\mathrm{AB}},\label{eq:product chi}
\end{equation}
for every pair of effective systems $\mathrm{A}$ and $\mathrm{B}$.
\end{condition}

We call eq.~\eqref{eq:product chi} the\emph{ condition of informational
equilibrium}. Note that, again, here we are not referring to thermal
equilibrium between the two subsystems. This is clear from the fact
that we do not allow an energy flow between the two systems $\mathrm{A}$
and $\mathrm{B}$. Instead, we allow a flow of information, implemented
by the joint dynamics of the composite system $\mathrm{AB}$.

It is natural to ask which physical principles guarantee the condition
of informational equilibrium. One such principle is Local Tomography,
as shown in \cite{Chiribella-purification,Masanes-physical-derivation}.
However, Local Tomography is not necessary for informational equilibrium.
Indeed quantum theory on real Hilbert spaces violates Local Tomography,
but still satisfies the condition of informational equilibrium. As
already stressed, in our analysis we will \emph{not} assume Local
Tomography in our set of physical principles. Nevertheless, our principles
of sharp theories with purification guarantee the validity of the
condition of informational equilibrium, which is the really important
condition to set up a sensible theory of (microcanonical) thermodynamics.
\begin{example}
Sharp theories with purifications satisfy the condition of informational
equilibrium \cite{Purity}. To prove it, consider two systems $\mathrm{A}$
and $\mathrm{B}$, and pick two pure maximal sets for $\mathrm{A}$
and $\mathrm{B}$ respectively, say$\left\{ \alpha_{i}\right\} _{i=1}^{d_{\mathrm{A}}}$
and $\left\{ \beta_{j}\right\} _{j=1}^{d_{\mathrm{B}}}$. Then, the
product set $\left\{ \alpha_{i}\otimes\beta_{j}\right\} _{i\in\left\{ 1,\dots,d_{\mathrm{A}}\right\} ,j\in\left\{ 1,\dots,d_{\mathrm{B}}\right\} }$
is maximal for the composite system $\mathrm{AB}$, by proposition~\ref{prop:information locality}.
Using the decomposition~\eqref{eq:microcanonical diagonalised},
we obtain 
\[
\chi_{\mathrm{AB}}=\frac{1}{d_{\mathrm{AB}}}\sum_{i=1}^{d_{\mathrm{A}}}\sum_{j=1}^{d_{\mathrm{B}}}\alpha_{i}\otimes\beta_{j}=\frac{1}{d_{\mathrm{A}}d_{\mathrm{B}}}\left(\sum_{i=1}^{d_{\mathrm{A}}}\alpha_{i}\right)\otimes\left(\sum_{i=1}^{d_{\mathrm{B}}}\beta_{j}\right)=\chi_{\mathrm{A}}\otimes\chi_{\mathrm{B}},
\]
where we have used the information locality condition $d_{\mathrm{AB}}=d_{\mathrm{A}}d_{\mathrm{B}}$
(again proposition~\ref{prop:information locality}). In summary,
sharp theories with purification satisfy our two conditions for the
general microcanonical framework.
\end{example}

Now we are ready to extend the microcanonical framework from quantum
and classical theory to general physical theories. 
\begin{defn}
An OPT, interpreted as a theory of effective systems, is \emph{microcanonical}
if conditions~\ref{cond:unique microcanonical} and \ref{cond:informational equibrium}
are satisfied.
\end{defn}

Physically, a microcanonical theory is a theory where:
\begin{enumerate}
\item every system has a well-defined notion of uniform mixture of all pure
states;
\item uniform mixtures are stable under parallel composition of systems.
\end{enumerate}
Microcanonical theories provide the foundation for the definition
of three important resource theories, analysed in the following sections.

\section{Three microcanonical resource theories\label{sec:Three-resource-theories}}

In this section we study three different notions of state convertibility
in microcanonical theories. We adopt the resource-theoretic framework
of chapter~\ref{chap:Resource-theories}, where one fixes a set of
\emph{free operations}, closed under sequential and parallel composition.
As we saw in section~\ref{sec:The-resource-preorder}, the basic
question in the resource-theoretic framework is whether a given state
$\rho$ can be transformed into another state $\sigma$ by means of
free operations. This gives rise to a preorder between states, denoted
as\footnote{Here we adopt the notation $\succeq_{F}$ for the preorder, instead
of the notation $\succsim$ of section~\ref{sec:The-resource-preorder},
because of the analogy with the notation for majorisation, as we will
become clear in section~\ref{sec:Majorisation-and-unital}.} $\rho\succeq_{F}\sigma$, where $F$ is the set of free operations.
In the following we define three resource theories for microcanonical
thermodynamics.

\subsection{The RaRe resource theory\label{subsec:The-RaRe-resource}}

Our first resource theory takes random reversible channels \cite{Uhlmann1,Uhlmann2,Uhlmann3,Muller3D,Chiribella-Scandolo-entanglement}
as free operations.
\begin{defn}
A \emph{random reversible (RaRe) channel} on system $\mathrm{A}$
is a channel $\mathcal{R}$ of the form $\mathcal{R}=\sum_{i}p_{i}\mathcal{U}_{i}$,
where $\left\{ p_{i}\right\} $ is a probability distribution and
$\mathcal{U}_{i}$ is a reversible channel on system $\mathrm{A}$
for every $i$.
\end{defn}

Physically, RaRe channels are the operations that can be implemented
with limited control over the reversible dynamics of the system: if
the dynamics are subject to random fluctuations, the lack of control
on these fluctuations gives rise to a RaRe channel.

We already encountered a RaRe channel in proposition~\ref{prop:properties microcanonical}
as the channel enacting the ``thermalisation'' process, i.e.\ mapping
every pure state to the microcanonical state. From this point of view,
it is fairly natural to take RaRe channels as free operations. Moreover,
from a more mathematical perspective, RaRe channels have all the properties
required of free operations: the identity channel is RaRe, the sequential
composition of two RaRe channels is a RaRe channel, and so is the
parallel composition. We call the resulting resource theory the \emph{RaRe
resource theory} and we denote the corresponding preorder by $\succeq_{\mathrm{RaRe}}$.

Note that, in principle, the RaRe resource theory can be formulated
in \emph{every} GPT, even those that do \emph{not} satisfy conditions~\ref{cond:unique microcanonical}
and \ref{cond:informational equibrium}, because there is no microcanonical
state explicitly involved in the definition of RaRe channels. Such
generality, however, comes at a price: from the structure of RaRe
channels we \emph{cannot} infer free states by taking the input system
of a free operation to be the trivial system. Indeed, RaRe channels
have the same input and output system, so taking the input system
to be trivial makes the channel itself trivial. Therefore, strictly
speaking, the RaRe resource theory has no free states. As explained
in section~\ref{sec:The-resource-preorder}, if we want to add free
states, we have to look for almost free states. Now, assume the theory
on which RaRe channels are defined satisfies condition~\ref{cond:unique microcanonical}
and \ref{cond:informational equibrium}.

By proposition~\ref{prop:properties microcanonical} we know that
$\mathcal{T}\rho=\chi$ for every $\rho$, where $\mathcal{T}$ is
the RaRe thermalisation channel, and $\chi$ the microcanonical state.
This clearly means that $\rho\succeq_{\mathrm{RaRe}}\chi$ for every
$\rho$. Therefore all states equivalent to $\chi$ in the preorder
$\succeq_{\mathrm{RaRe}}$ will be the minima as resources. Now, for
a state $\rho$ to be equivalent to $\chi$, it must be $\chi\succeq_{\mathrm{RaRe}}\rho$.
This means that there exists a RaRe channel $\mathcal{R}$ such that
$\mathcal{R}\chi=\rho$. Since $\chi$ is an invariant state, $\mathcal{R}\chi=\chi$,
therefore there is only $\chi$ in its equivalence class. This means
that there is a unique almost free state, the microcanonical state
$\chi$. It is stable under tensor product, so the construction presented
in section~\ref{sec:The-resource-preorder} tells us that in microcanonical
theories we can rightfully consider the microcanonical state as a
free state, even though we cannot derive it from the structure of
free operations \cite{Gour-review}.

On the other extreme we have pure states, which are the maximum resource.
First of all, they are all equivalent resources thanks to transitivity:
reversible channels are a special case of RaRe channels. Then take
any pure state $\psi$, and any state $\rho$. $\rho$ can always
written as a (possibly trivial) convex combination of pure states
$\psi_{i}$: $\rho=\sum_{i}p_{i}\psi_{i}$. Because of the transitive
action of reversible channels on pure states, there exist some reversible
channels $\mathcal{U}_{i}$ such that $\psi_{i}=\mathcal{U}_{i}\psi$,
whence $\rho$ can be written as
\[
\rho=\sum_{i}p_{i}\mathcal{U}_{i}\psi.
\]
This means that $\psi\succeq_{\mathrm{RaRe}}\rho$ for any pure state
$\psi$ and any state $\rho$.

\subsection{The noisy resource theory}

Whilst the RaRe resource theory can be defined in every OPT, we now
discuss a second resource theory that can only be defined in physical
theories satisfying conditions~\ref{cond:unique microcanonical}
and \ref{cond:informational equibrium}. In this resource theory,
free operations mimic a thermalisation process, and are generated
by the following three processes:
\begin{enumerate}
\item bringing in an ancillary system in the microcanonical state;
\item letting the target system and the ``thermal'' ancilla jointly evolve
together;
\item removing the ``thermal'' ancilla.
\end{enumerate}
More precisely, these operations, usually called ``noisy'' \cite{Local-information,Horodecki-Oppenheim,Nicole},
are defined as follows.
\begin{defn}
\label{def:basic noisy}A channel $\mathcal{B}$, from system $\mathrm{A}$
to system $\mathrm{A}'$, is a \emph{basic noisy operation} if it
can be written as\footnote{Again, here Causality is assumed for convenience, it is possible to
put forward a definition of noisy operations also in non-causal theories
\cite{Purity}.}\begin{equation}\label{eq:noisy} \begin{aligned}\Qcircuit @C=1em @R=.7em @!R { & & \qw \poloFantasmaCn{\rA} & \gate{\cB} & \qw \poloFantasmaCn{\rA'} &\qw }\end{aligned} ~= \!\!\!\! \begin{aligned}\Qcircuit @C=1em @R=.7em @!R { & & \qw \poloFantasmaCn{\rA} & \multigate{1}{\cU} & \qw \poloFantasmaCn{\rA'} &\qw \\ & \prepareC{\chi} & \qw \poloFantasmaCn{\rE} &\ghost{\cU} & \qw \poloFantasmaCn{\rE'} & \measureD{u} }\end{aligned}~, \end{equation}where
$\mathrm{E}$ and $\mathrm{E}'$ are suitable systems such that $\mathrm{AE}\approx\mathrm{A'}\mathrm{E}'$,
and $\mathcal{U}$ is a reversible channel.
\end{defn}

Note that, if we take $\mathrm{A}$ to be the trivial system, harnessing
the invariance of the microcanonical state under reversible channels,
we obtain that the microcanonical state is a free state from the very
expression of basic noisy operations.

Definition~\ref{def:basic noisy}, unlike RaRe channels, does not
rely on the availability of external sources of randomness: all the
randomness is accounted for in the preparation of the microcanonical
state in the right-hand side of eq.~\eqref{eq:noisy}. In other words,
an external microcanonical state becomes the source of ``thermalisation''
for the target system, instead of the random fluctuations of reversible
dynamics in the case of the RaRe resource theory.

Definition~\ref{def:basic noisy} uses the term ``basic noisy operation''
instead of the customary ``noisy operation'' to account for a mathematical
subtlety which will arise in subsection~\ref{subsec:Inclusion-relations}:
the set of basic noisy operations is generally \emph{not} topologically
closed. In quantum theory, for example, there exist counterexamples
where the limit of a sequence of basic noisy operations is not a basic
noisy operation \cite{Shor}. It is then convenient to take the closure
of the set of basic noisy operations.
\begin{defn}
\label{def:noisy}A channel $\mathcal{N}$ is a\emph{ noisy operation}
if it is the limit of a sequence of basic noisy operations $\left\{ \mathcal{B}_{n}\right\} $.
\end{defn}

The set of noisy operations satisfies all the requirements for being
a set of free operations: the identity is a noisy operation, and the
parallel and sequential composition of two noisy operations are noisy
operations, thanks to the condition of informational equilibrium.
The resource theory where the set of free operations is the set of
noisy operations will be called the \emph{noisy resource theory}.
The corresponding preorder on states will be denoted by $\succeq_{\mathrm{noisy}}$.

\subsection{The unital resource theory}

In the third resource theory, the set of free operations includes
all the operations that transform microcanonical states into microcanonical
states. This is the largest set of free operations compatible with
$\chi$ being the free state, as described in chapter~\ref{chap:Resource-theories}.
These channels, called \emph{unital}, are the most general operations
that do not create resources out of free states. Mathematically, they
are defined as follows.
\begin{defn}
A channel $\mathcal{D}$ from system $\mathrm{A}$ to system $\mathrm{A}'$
is called \emph{unital} if $\mathcal{D}\chi_{\mathrm{A}}=\chi_{\mathrm{A}'}$.
\end{defn}

Unital channels are the operational generalisation of doubly stochastic
matrices in classical probability theory \cite{Olkin,Streater,Mendl-Wolf}.

Note that, by condition~\ref{cond:informational equibrium}, all
free states of a composite system are of the product form, therefore,
focusing on resource-non-generating operations is enough, they are
automatically completely resource-non-generating. If $\mathrm{A}$
is the trivial system, the preparation of the microcanonical state
is a (free) unital channel, so we recover that the microcanonical
state is free.

The set of unital channels enjoys all the properties required of a
set of free operations: the identity is a unital channel, and thanks
to the condition of informational equilibrium, the sequential and
parallel composition of unital channels is a unital channel. The resource
theory where free operations are unital channels will be called the\emph{
unital resource theory}. The corresponding preorder on states will
be denoted by $\succeq_{\mathrm{unital}}$.

\subsection{Inclusion relations\label{subsec:Inclusion-relations}}

Let us highlight the relations between the three sets of free operations
defined so far. First, RaRe channels are examples of unital channels.
This is clear because every RaRe channel can be written as a mixture
of reversible channels, each of which preserves the microcanonical
state.
\[
\mathcal{R}\chi=\sum_{i}p_{i}\mathcal{U}_{i}\chi=\sum_{i}p_{i}\chi=\chi
\]
Hence, we have the inclusion 
\begin{equation}
\mathsf{RaRe}\subseteq\mathsf{Unital}.\label{eq:RaRe unital}
\end{equation}
In classical probability theory, the inclusion is actually an equality,
as a consequence of Birkhoff's theorem \cite{Birkhoff,Olkin}. Remarkably,
in quantum theory there exist unital channels that are not random
unitary, meaning that the inclusion~\eqref{eq:RaRe unital} is generally
strict. The simplest example is due to Landau and Streater \cite{Streater}.

We have seen that all RaRe channels are unital. Noisy operations are
unital too.
\begin{prop}
Every noisy operation is unital.
\end{prop}

\begin{proof}
Suppose that $\mathcal{B}$ is a basic noisy operation, written as
in eq.~\eqref{eq:noisy}. Then, one has\[\begin{aligned}\Qcircuit @C=1em @R=.7em @!R { &\prepareC{\chi} & \qw \poloFantasmaCn{\rA} & \gate{\cB} & \qw \poloFantasmaCn{\rA'} &\qw }\end{aligned} ~= \!\!\!\! \begin{aligned}\Qcircuit @C=1em @R=.7em @!R { &\prepareC{\chi} & \qw \poloFantasmaCn{\rA} & \multigate{1}{\cU} & \qw \poloFantasmaCn{\rA'} &\qw \\ & \prepareC{\chi} & \qw \poloFantasmaCn{\rE} &\ghost{\cU} & \qw \poloFantasmaCn{\rE'} & \measureD{u} }\end{aligned} ~  = \!\!\!\! \begin{aligned}\Qcircuit @C=1em @R=.7em @!R { &\multiprepareC{1}{\chi} & \qw \poloFantasmaCn{\rA} & \multigate{1}{\cU} & \qw \poloFantasmaCn{\rA'} &\qw \\ & \pureghost{\chi} & \qw \poloFantasmaCn{\rE} &\ghost{\cU} & \qw \poloFantasmaCn{\rE'} & \measureD{u} } \end{aligned}~=
\]
\[  = \!\!\!\! \begin{aligned}\Qcircuit @C=1em @R=.7em @!R { &\multiprepareC{1}{\chi} & \qw \poloFantasmaCn{\rA'} &\qw \\ & \pureghost{\chi} & \qw \poloFantasmaCn{\rE'} & \measureD{u} } \end{aligned}~ = \!\!\!\! \begin{aligned}\Qcircuit @C=1em @R=.7em @!R { &\prepareC{\chi} & \qw \poloFantasmaCn{\rA'} &\qw \\ & \prepareC{\chi} & \qw \poloFantasmaCn{\rE'} & \measureD{u} } \end{aligned}~ = \!\!\!\! \begin{aligned}\Qcircuit @C=1em @R=.7em @!R { &\prepareC{\chi} & \qw \poloFantasmaCn{\rA'} &\qw }\end{aligned} ~,
\]having used the condition of informational equilibrium~\eqref{eq:product chi},
the invariance of the microcanonical state $\chi_{\mathrm{AE}}$ under
reversible channels.\footnote{To be precise, here we have used the fact that $\mathcal{U}\chi_{\mathrm{AE}}=\chi_{\mathrm{A}'\mathrm{E}'}$,
taking advantage of the fact that $\mathrm{AE}\approx\mathrm{A'}\mathrm{E}'$,
so $\mathrm{AE}$ and $\mathrm{A'}\mathrm{E}'$ can be treated for
all practical purposes as though they were the same system.} Hence, every basic noisy operation is unital. Taking the closure,
one gets that all noisy operations are unital.
\end{proof}
In summary, one has the inclusion
\begin{equation}
\mathsf{Noisy}\subseteq\mathsf{Unital}.\label{eq:noisy unital}
\end{equation}
 The inclusion is strict in quantum theory, where Haagerup and Musat
found examples of unital channels that cannot be realised as noisy
operations \cite{Haagerup-Musat}. 

It remains to understand the relation between RaRe channels and noisy
operations. In quantum theory, the set of noisy operations (strictly)
contains the set of RaRe channels as a proper subset \cite{Shor}.
In a generic theory, however, this inclusion relation need not hold.
\begin{example}
As a counterexample, consider a theory where, besides the $\mathtt{SWAP}$,
only local reversible channels are allowed, like in PR boxes \cite{PR-trivial}.
In this case, noisy operations on a given system $\mathrm{A}$ are
just reversible channels. Indeed\[
\begin{aligned}\Qcircuit @C=1em @R=.7em @!R { & & \qw \poloFantasmaCn{\rA} & \gate{\cB} & \qw \poloFantasmaCn{\rA} &\qw }\end{aligned} ~= \!\!\!\! \begin{aligned}\Qcircuit @C=1em @R=.7em @!R { & & \qw \poloFantasmaCn{\rA} & \gate{\cU} & \qw \poloFantasmaCn{\rA} &\qw \\ & \prepareC{\chi} & \qw \poloFantasmaCn{\rE} &\gate{\cV} & \qw \poloFantasmaCn{\rE} & \measureD{u} }\end{aligned}~=~\begin{aligned}\Qcircuit @C=1em @R=.7em @!R {  & \qw \poloFantasmaCn{\rA} & \gate{\cU} & \qw \poloFantasmaCn{\rA} &\qw}\end{aligned},
\]where we have used the invariance of $\chi$ and the fact that it
is a normalised state. In this theory, noisy operations are \emph{strictly}
contained in RaRe channels.
\end{example}

The inclusions~\eqref{eq:RaRe unital} and \eqref{eq:noisy unital}
are the most general result one can derive from the definitions alone.
The general situation is depicted in fig.~\ref{fig:general inclusions}.
\begin{figure}
\begin{centering}
\includegraphics[scale=0.8]{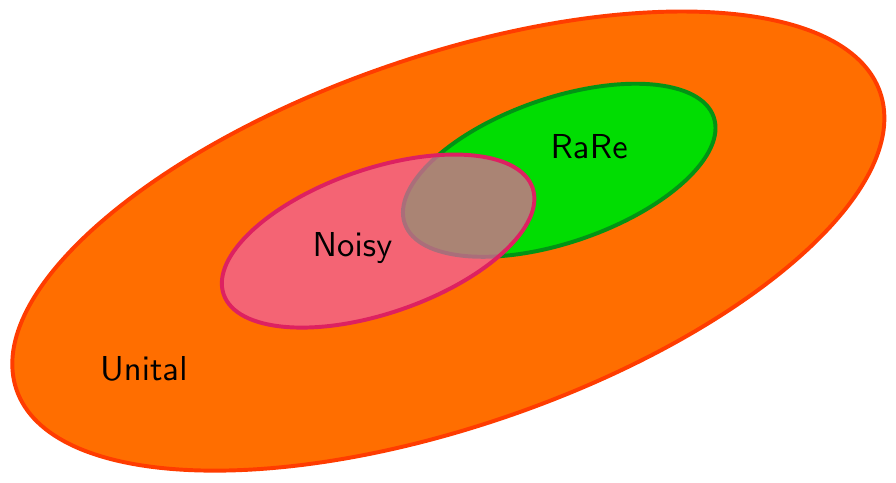}
\par\end{centering}
\caption{\label{fig:general inclusions}The most general inclusions between
the three sets of free operations. At this stage we cannot say anything
about the intersection between RaRe channels and noisy operations.}

\end{figure}
 To go further, we need to introduce axioms, which we will do in the
next subsection for sharp theories with purification.

\subsubsection{Inclusion relations in sharp theories with purification}

In sharp theories with purification, one can establish an inclusion
between RaRe channels and noisy operations, the same we have in quantum
theory. To obtain this result, we first restrict our attention to
\emph{rational} RaRe channels, i.e.\ RaRe channels of the form $\mathcal{R}=\sum_{i}p_{i}\mathcal{U}_{i}$
where each $p_{i}$ is a rational number. With this definition, we
have the following lemma.
\begin{lem}
\label{lem:RaRe noisy}In every sharp theory with purification, every
rational RaRe channel is a basic noisy operation.
\end{lem}

\begin{proof}
Let $\mathcal{R}$ be a rational RaRe channel, written as 
\[
\mathcal{R}=\sum_{i}\frac{n_{i}}{n}\mathcal{U}_{i},
\]
with $n_{i}\ge0$ and $\sum_{i}n_{i}=n$. Let $\mathrm{B}$ be an
$n$-dimensional system, and pick the pure maximal set $\left\{ \beta_{j}\right\} _{j=1}^{n}$.
Let $\mathcal{C}$ be the channel from $\mathrm{AB}$ to $\mathrm{A}$,
defined as 
\[
\mathcal{C}=\sum_{j=1}^{n}\mathcal{V}_{j}\otimes\beta_{j}^{\dagger},
\]
where $\left\{ \mathcal{V}_{j}\right\} _{j=1}^{n}$ are $n$ reversible
channels on $\mathrm{A}$, with $n_{1}$ of them equal to $\mathcal{U}_{1}$,
$n_{2}$ equal to $\mathcal{U}_{2}$, and so on. Since the theory
satisfies Purification, the channel $\mathcal{C}$ has a reversible
dilation \cite{Chiribella-purification,Chiribella14}, namely\[ \begin{aligned} \Qcircuit @C=1em @R=.7em @!R { & \qw \poloFantasmaCn{\rA} & \multigate{1}{\cC} & \qw \poloFantasmaCn{\rA} & \qw \\ & \qw \poloFantasmaCn{\rB} & \ghost{\cC} & & } \end{aligned}~ = ~ \begin{aligned} \Qcircuit @C=1em @R=.7em @!R { & \qw \poloFantasmaCn{\rA} & \multigate{2}{\cU} & \qw \poloFantasmaCn{\rA} & \qw \\ & \qw \poloFantasmaCn{\rB} & \ghost{\cU} & & \\ \prepareC{\gamma}& \qw \poloFantasmaCn{\rC} & \ghost{\cU} & \qw \poloFantasmaCn{\rC'} & \measureD{u} } \end{aligned} ~ , \]where
$\mathrm{C}$ and $\mathrm{C}'$ are suitable systems, $\gamma$ is
a pure state, and $\mathcal{U}$ is a reversible channel. Now, by
construction we have\[ \begin{aligned} \Qcircuit @C=1em @R=.7em @!R { & \qw \poloFantasmaCn{\rA} & \multigate{2}{\cU} & \qw \poloFantasmaCn{\rA} & \qw \\ \prepareC{\beta_j}& \qw \poloFantasmaCn{\rB} & \ghost{\cU} & & \\ \prepareC{\gamma}& \qw \poloFantasmaCn{\rC} & \ghost{\mathcal U} & \qw \poloFantasmaCn{\rC'} & \measureD{u} } \end{aligned} ~  = ~ \begin{aligned} \Qcircuit @C=1em @R=.7em @!R { & \qw \poloFantasmaCn{\rA} & \multigate{1}{\cC} & \qw \poloFantasmaCn{\rA} & \qw \\ \prepareC{\beta_j}& \qw \poloFantasmaCn{\rB} & \ghost{\cC} & & } \end{aligned} ~= ~\begin{aligned} \Qcircuit @C=1em @R=.7em @!R { & \qw \poloFantasmaCn{\rA} & \gate{\cV_j} & \qw \poloFantasmaCn{\rA} & \qw } \end{aligned}~,
\]for every $j\in\left\{ 1,\ldots,n\right\} $. Since $\mathcal{V}_{j}$
is reversible, the joint map must factorise (corollary~\ref{cor:pure map factorise}):\begin{equation}\label{eq:Uprog} \begin{aligned} \Qcircuit @C=1em @R=.7em @!R { & \qw \poloFantasmaCn{\rA} & \multigate{2}{\cU} & \qw \poloFantasmaCn{\rA} & \qw \\ \prepareC{\beta_j}& \qw \poloFantasmaCn{\rB} & \ghost{\cU} & & \\ \prepareC{\gamma}& \qw \poloFantasmaCn{\rC} & \ghost{\cU} & \qw \poloFantasmaCn{\rC'} & \qw } \end{aligned} ~= ~ \begin{aligned} \Qcircuit @C=1em @R=.7em @!R { & \qw \poloFantasmaCn{\rA} & \gate{\cV_j} & \qw \poloFantasmaCn{\rA} & \qw \\ & &\prepareC{\gamma_j} & \qw \poloFantasmaCn{\rC'} & \qw } \end{aligned} ~, \end{equation}for
some pure state $\gamma_{j}$ of system $\mathrm{C}'$. Composing
both sides with $\mathcal{V}_{j}^{-1}$ on the left, and with $\mathcal{U}^{-1}$
on the right we obtain\begin{equation}\label{eq:Uprogdag} \begin{aligned} \Qcircuit @C=1em @R=.7em @!R { & \qw \poloFantasmaCn{\rA} & \gate{\mathcal {V}_j^{-1}} & \qw \poloFantasmaCn{\rA} & \qw \\ &&\prepareC{\beta_j}& \qw \poloFantasmaCn{\rB} & \qw  \\ &&\prepareC{\gamma}& \qw \poloFantasmaCn{\rC} & \qw  } \end{aligned} ~ = ~ \begin{aligned} \Qcircuit @C=1em @R=.7em @!R { & \qw \poloFantasmaCn{\rA} & \multigate{2}{\cU^{-1}} & \qw \poloFantasmaCn{\rA} & \qw \\ & & \pureghost{\cU^{-1}} & \qw \poloFantasmaCn{\rB} & \qw \\ \prepareC{\gamma_j} & \qw \poloFantasmaCn{\rC'} & \ghost{\cU^{-1}} & \qw \poloFantasmaCn{\rC} & \qw } \end{aligned} ~. \end{equation}Now
we can feed the expression for $\gamma_{j}$ in terms of $\mathcal{U}$
from eq.~\eqref{eq:Uprog}  to the input of $\mathcal{U}^{-1}$ in
eq.~\eqref{eq:Uprogdag}, getting\begin{equation}\label{eq:doubleprog} \begin{aligned} \Qcircuit @C=1em @R=.7em @!R { & \qw \poloFantasmaCn{\rA} & \multigate{2}{\mathcal U} & \qw \poloFantasmaCn{\rA} & \qw & \qw &\qw &\qw \\ \prepareC{\beta_j}& \qw \poloFantasmaCn{\rB} & \ghost{\mathcal U} & & & \pureghost{\cU^{-1}} & \qw \poloFantasmaCn{\rB} & \qw \\ \prepareC{\gamma}& \qw \poloFantasmaCn{\rC} & \ghost{\mathcal U} & \qw \poloFantasmaCn{\rC'} & \qw & \ghost{\cU^{-1}} & \qw \poloFantasmaCn{\rC} & \qw \\ & \qw \poloFantasmaCn{\rA} & \qw& \qw & \qw & \multigate{-2}{\cU^{-1}} & \qw \poloFantasmaCn{\rA} & \qw } \end{aligned} ~ = ~ \begin{aligned} \Qcircuit @C=1em @R=.7em @!R { & \qw \poloFantasmaCn{\rA} & \gate{\cV_j} &\qw \poloFantasmaCn{\rA}&\qw \\ &&\prepareC{\beta_j}& \qw \poloFantasmaCn{\rB} & \qw  \\ &&\prepareC{\gamma}& \qw \poloFantasmaCn{\rC} & \qw  \\ & \qw \poloFantasmaCn{\rA} & \gate{\cV_j^{-1}} & \qw \poloFantasmaCn{\rA} & \qw } \end{aligned} ~ \end{equation}where
most of the right-hand side comes from the left-hand side of eq.~\eqref{eq:Uprogdag}.
At this point, we define the pure transformation\[
\begin{aligned} \Qcircuit @C=1em @R=.7em @!R { & \qw \poloFantasmaCn{\rA} & \multigate{2}{\mathcal P} &\qw \poloFantasmaCn{\rA}&\qw \\ & \qw \poloFantasmaCn{\rB} & \ghost{\mathcal P} & \qw \poloFantasmaCn{\rB} & \qw \\ & \qw \poloFantasmaCn{\rA} &\ghost{\mathcal P} & \qw \poloFantasmaCn{\rA} & \qw } \end{aligned} ~: =~ \begin{aligned} \Qcircuit @C=1em @R=.7em @!R { & \qw \poloFantasmaCn{\rA} & \multigate{2}{\mathcal U} & \qw \poloFantasmaCn{\rA} & \qw & \qw &\qw &\qw \\ & \qw \poloFantasmaCn{\rB} & \ghost{\mathcal U} & & & \pureghost{\cU^{-1}} & \qw \poloFantasmaCn{\rB} & \qw \\ \prepareC{\gamma}& \qw \poloFantasmaCn{\rC} & \ghost{\mathcal U} & \qw \poloFantasmaCn{\rC'} & \qw & \ghost{\cU^{-1}} & \qw \poloFantasmaCn{\rC} & \measureD{\gamma^\dag} \\ & \qw \poloFantasmaCn{\rA} & \qw& \qw & \qw & \multigate{-2}{\cU^{-1}} & \qw \poloFantasmaCn{\rA} & \qw } \end{aligned} ~.
\]By eq.~\eqref{eq:doubleprog} $\mathcal{P}$ satisfies the relation\[ \begin{aligned} \Qcircuit @C=1em @R=.7em @!R { & \qw \poloFantasmaCn{\rA} & \multigate{2}{\mathcal P} &\qw \poloFantasmaCn{\rA}&\qw \\ \prepareC{\beta_j}& \qw \poloFantasmaCn{\rB} & \ghost{\mathcal P} & \qw \poloFantasmaCn{\rB} & \qw \\ & \qw \poloFantasmaCn{\rA} &\ghost{\mathcal P} & \qw \poloFantasmaCn{\rA} & \qw } \end{aligned} ~ = ~ \begin{aligned} \Qcircuit @C=1em @R=.7em @!R { & \qw \poloFantasmaCn{\rA} & \gate{\cV_j} &\qw \poloFantasmaCn{\rA}&\qw \\ && \prepareC{\beta_j}& \qw \poloFantasmaCn{\rB} & \qw \\ & \qw \poloFantasmaCn{\rA} & \gate{\cV_j^{-1}} & \qw \poloFantasmaCn{\rA} & \qw } \end{aligned} ~, \]for
all values of $j$. Using this relation and the expression of $\chi_{\mathrm{B}}$
in terms of the $\beta_{x}$'s, we can reconstruct $\mathcal{R}$
from $\mathcal{P}$:\[ \begin{aligned} \Qcircuit @C=1em @R=.7em @!R { & \qw \poloFantasmaCn{\rA} & \multigate{2}{\mathcal P} &\qw \poloFantasmaCn{\rA}&\qw \\ \prepareC{\chi}& \qw \poloFantasmaCn{\rB} & \ghost{\mathcal P} & \qw \poloFantasmaCn{\rB} & \measureD{u} \\ \prepareC{\chi}& \qw \poloFantasmaCn{\rA} &\ghost{\mathcal P} & \qw \poloFantasmaCn{\rA} & \measureD{u} } \end{aligned} ~ = \frac 1n \sum_{j=1}^n ~ \begin{aligned} \Qcircuit @C=1em @R=.7em @!R { & \qw \poloFantasmaCn{\rA} & \multigate{2}{\mathcal P} &\qw \poloFantasmaCn{\rA}&\qw \\ \prepareC{\beta_j}& \qw \poloFantasmaCn{\rB} & \ghost{\mathcal P} & \qw \poloFantasmaCn{\rB} & \measureD{u} \\ \prepareC{\chi}& \qw \poloFantasmaCn{\rA} &\ghost{\mathcal P} & \qw \poloFantasmaCn{\rA} & \measureD{u} } \end{aligned} ~=
\]
\[
=\frac 1n \sum_{j=1}^n ~ \begin{aligned} \Qcircuit @C=1em @R=.7em @!R { &\qw \poloFantasmaCn{\rA} & \gate{\cV_j} &\qw \poloFantasmaCn{\rA}&\qw \\ \prepareC{\beta_j}& \qw \poloFantasmaCn{\rB} &\qw &\qw &\measureD{u} & \\ \prepareC{\chi} & \qw \poloFantasmaCn{\rA} & \gate{\cV_j^{-1}} & \qw \poloFantasmaCn{\rA} & \measureD{u} } \end{aligned} ~= \frac 1n \sum_{j=1}^n ~ \begin{aligned} \Qcircuit @C=1em @R=.7em @!R { &\qw \poloFantasmaCn{\rA} & \gate{\cV_j} &\qw \poloFantasmaCn{\rA}&\qw } \end{aligned} ~=
\]
\begin{equation}\label{eq:simplenoisy}
= \begin{aligned} \Qcircuit @C=1em @R=.7em @!R { &\qw \poloFantasmaCn{\rA} & \gate{\mathcal R} &\qw \poloFantasmaCn{\rA}&\qw } \end{aligned} ~, 
\end{equation}where we have used the fact that $\sum_{j=1}^{n}\mathcal{V}_{j}=\sum_{i}n_{i}\mathcal{U}_{i}$.
Finally, let us show that $\mathcal{P}$ is a channel. To this end,
it is enough to show that $u\mathcal{P}=u$ (see proposition~\ref{prop:characterization channel}).
This property is satisfied if and only if $\left(u\middle|\mathcal{P}\middle|\chi\right)=1$,
because every state lies in some convex decomposition of $\chi$ (see
proposition~\ref{prop:uniqueness invariant}). By the condition of
informational equilibrium and eq.~\eqref{eq:simplenoisy}, we have\[
\begin{aligned} \Qcircuit @C=1em @R=.7em @!R { \multiprepareC{2}{\chi}& \qw \poloFantasmaCn{\rA} & \multigate{2}{\mathcal P} &\qw \poloFantasmaCn{\rA}&\measureD{u} \\ \pureghost{\chi}& \qw \poloFantasmaCn{\rB} & \ghost{\mathcal P} & \qw \poloFantasmaCn{\rB} & \measureD{u} \\ \pureghost{\chi}& \qw \poloFantasmaCn{\rA} &\ghost{\mathcal P} & \qw \poloFantasmaCn{\rA} & \measureD{u} } \end{aligned}~=~\begin{aligned} \Qcircuit @C=1em @R=.7em @!R { \prepareC{\chi}& \qw \poloFantasmaCn{\rA} & \multigate{2}{\mathcal P} &\qw \poloFantasmaCn{\rA}&\measureD{u} \\ \prepareC{\chi}& \qw \poloFantasmaCn{\rB} & \ghost{\mathcal P} & \qw \poloFantasmaCn{\rB} & \measureD{u} \\ \prepareC{\chi}& \qw \poloFantasmaCn{\rA} &\ghost{\mathcal P} & \qw \poloFantasmaCn{\rA} & \measureD{u} }\end{aligned}~=~\begin{aligned} \Qcircuit @C=1em @R=.7em @!R { \prepareC{\chi} &\qw \poloFantasmaCn{\rA} & \gate{\mathcal R} &\qw \poloFantasmaCn{\rA}&\measureD{u} }\end{aligned}~=1~,
\]so $\mathcal{P}$ is a channel. Since every pure channel on a fixed
system (here $\mathrm{ABA}$) is reversible \cite{Chiribella-purification},
$\mathcal{P}$ is reversible. Hence, eq.\ \eqref{eq:simplenoisy}
shows that $\mathcal{R}$ is a basic noisy operation, with environment
$\mathrm{E}=\mathrm{BA}$, and reversible channel $\mathcal{P}$.
\end{proof}
In quantum theory, this statement is quite immediate, as pointed out
in \cite{Nicole}: a generic RaRe channel with rational probabilities
$\left\{ \frac{n_{i}}{n}\right\} _{i=1}^{r}$ and unitary gates $\left\lbrace U_{i}\right\rbrace _{i=1}^{r}$
can be realised as the basic noisy operation 
\[
\mathcal{B}\left(\rho\right):=\mathrm{tr}_{{\rm anc}}\left[U\left(\rho\otimes\frac{1}{n}\mathbf{1}\right)U^{\dagger}\right],
\]
where $\mathrm{tr}_{{\rm anc}}$ is the partial trace on the $n$-dimensional
system used as ancilla, and $U$ is the control-unitary gate
\[
U:=\sum_{j=1}^{n}V_{j}\otimes\ket{j}\bra{j},
\]
$\left\{ \ket{j}\right\} _{j=1}^{n}$ being an orthonormal basis for
the ancillary system, and $\left\lbrace V_{j}\right\rbrace _{j=1}^{n}$
being a list of unitary gates, $n_{1}$ of which are equal to $U_{1}$,
$n_{2}$ equal to $U_{2}$, and so on.

The proof of lemma~\ref{lem:RaRe noisy} shows that the situation
is in general more complicated in sharp theories with purification.
The reason is that the simple construction of quantum theory based
on control-unitary channels cannot be carried over. The analogue of
the control-unitary $U$ is a control-reversible transformation, which
performs a reversible transformation on the target system depending
on the state of a control system \cite{Control-reversible}. However,
in section~\ref{sec:Sufficiency-of-majorisation} we will show that
not every sharp theory with purification admits control-reversible
transformations. Specifically, the existence of control-reversible
transformations is equivalent to a non-trivial property of the dynamics,
which we will call ``unrestricted reversibility'' \cite{Purity}.
The non-trivial content of lemma~\ref{lem:RaRe noisy} is that the
inclusion of rational RaRe channels in noisy operations holds \emph{in
every sharp theory with purification}, without the need to assume
unrestricted reversibility or the existence of control-reversible
transformations. 

Now, since rational RaRe channels are dense in the set of RaRe channels,
and since the set of noisy operations is closed (see definition~\ref{def:noisy}),
we obtain the following theorem.
\begin{thm}
In every sharp theory with purification, RaRe channels are noisy operations.
\end{thm}

This theorem is important from a physical point of view, for it states
that the ``thermalisation'' RaRe channel in proposition~\ref{prop:properties microcanonical}
can be regarded as a physical ``thermalisation'' process in which
a system is put into contact with a ``thermal bath'' (the microcanonical
state), and left there to ``thermalise''.

Note that the inclusion of RaRe channels in the set of noisy operations
is generally strict: for example, in quantum theory there exist noisy
operations that are not RaRe channels \cite{Shor} . In summary, we
have the inclusions 
\begin{equation}
\mathsf{RaRe}\subseteq\mathsf{Noisy}\subseteq\mathsf{Unital},\label{eq:inclusions sharp}
\end{equation}
illustrated in fig.~\ref{fig:inclusions}.
\begin{figure}
\begin{centering}
\includegraphics[scale=0.8]{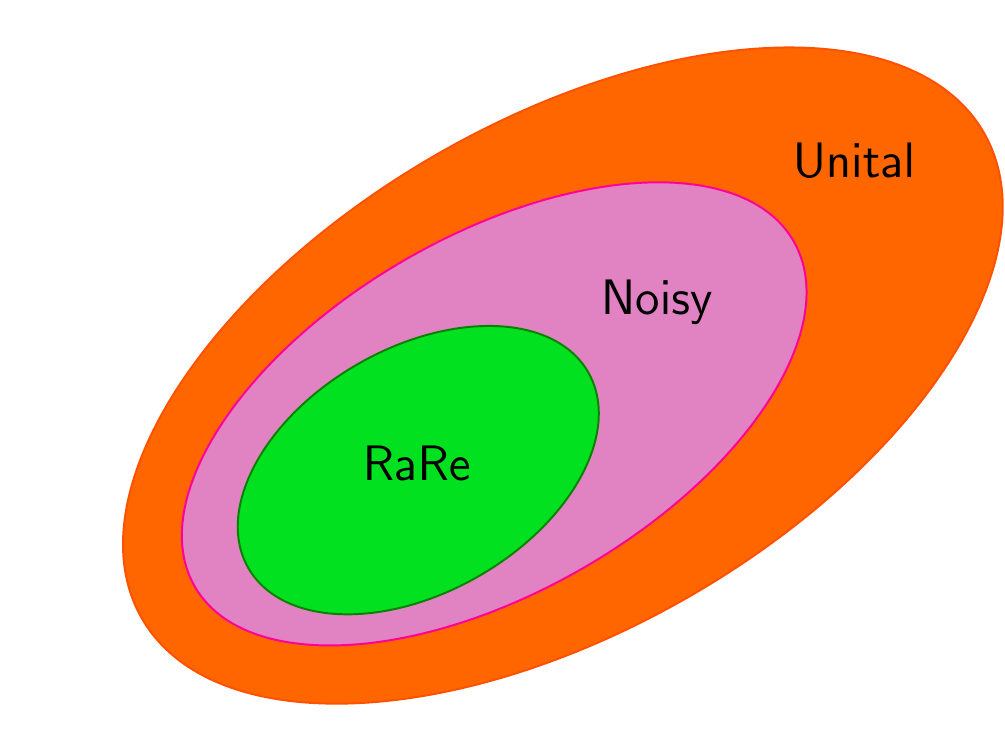}
\par\end{centering}
\caption{\label{fig:inclusions}Inclusion relations between the three sets
of free operations in sharp theories with purification.}

\end{figure}
 These inclusions imply the relations
\[
\rho\succeq_{\mathrm{RaRe}}\sigma\quad\Longrightarrow\quad\rho\succeq_{\mathrm{noisy}}\sigma\quad\Longrightarrow\quad\rho\succeq_{\mathrm{unital}}\sigma,
\]
valid for every pair of states $\rho$ and $\sigma$ of the \emph{same}
system.\footnote{Since RaRe channels have the same input and output system, we will
restrict ourselves to state convertibility in the same system, so
as to be able to compare the convertibility properties under \emph{all
three} resource theories.} Note that the unital relation $\succeq_{\mathrm{unital}}$ is the
weakest, i.e.\ the easiest to satisfy, whereas RaRe convertibility
is the strongest. Starting from RaRe convertibility we can already
characterise the states that are extremal in all the three preorders.
In subsection~\ref{subsec:The-RaRe-resource} we already saw that
$\rho\succeq_{\mathrm{RaRe}}\chi$ for every $\rho$. This clearly
implies that $\rho\succeq_{\mathrm{noisy}}\chi$ and $\rho\succeq_{\mathrm{unital}}\chi$,
so the microcanonical state is the resource with minimum value according
all the three preorders. This should not surprise, since it is a free
state (or can be regarded as such) for all the three resource theories.
At the other extreme, we also saw that for every pure state $\psi$
and every state $\rho$, $\psi\succeq_{\mathrm{RaRe}}\rho$, and that
pure states are all equivalent to each other in the RaRe resource
theory. By the inclusions~\eqref{eq:inclusions sharp}, these properties
carry over to the other two resource theories.

\section{Majorisation and unital channels\label{sec:Majorisation-and-unital}}

From this section on, we will always work in sharp theories with purification,
even when we do not state it explicitly. We know that in these theories
states can be diagonalised (see section~\ref{sec:Diagonalisation-of-states}),
so we wonder whether their spectra (i.e.\ the vectors of their eigenvalues)
play any role in determining the convertibility properties for the
three resource theories, like in quantum theory \cite{Nicole}. Indeed,
in general, determining when a state $\rho$ can be converted into
a state $\sigma$ is a hard task, because it involves checking \emph{all}
possible free operations in order to find out if a transition is possible.
Therefore a practical criterion that depends only on the two states
is highly desirable. Here, we will explore the role of majorisation
in this respect, which will provide us with a necessary and sufficient
condition for the unital preorder.

In a broad sense, unital channels are the generalisation of doubly
stochastic matrices, and in sharp theories with purification there
is a more explicit connection.
\begin{lem}
\label{lem:channelmatrix}Let $\mathcal{D}$ be a unital channel on
system $\mathrm{A}$, and let $\left\{ \alpha_{i}\right\} _{i=1}^{d}$
and $\left\{ \alpha_{i}'\right\} _{i=1}^{d}$ be two pure maximal
sets of $\mathrm{A}$. Then, the matrix $D$ with entries $D_{ij}:=\left(\alpha_{i}'^{\dagger}\middle|\mathcal{D}\middle|\alpha_{j}\right)$
is doubly stochastic.
\end{lem}

\begin{proof}
Every entry $D_{ij}$ is a probability, so it is non-negative. Moreover,
one has
\[
\sum_{i=1}^{d}D_{ij}=\sum_{i=1}^{d}\left(\alpha_{i}'^{\dagger}\middle|\mathcal{D}\middle|\alpha_{j}\right)=\left(u\middle|\mathcal{D}\middle|\alpha_{j}\right)=\mathrm{tr}\:\alpha_{j}=1,
\]
for all $j=1,\ldots,d$, where we have used the fact that the effects
$\left\{ \alpha_{i}'^{\dagger}\right\} _{i=1}^{d}$ form an observation-test
(proposition~\ref{prop:pure test chi}), and that $\mathcal{D}$
is a channel, so $u\mathcal{D}=u$. Finally, one has
\[
\sum_{j=1}^{d}D_{ij}=\sum_{j=1}^{d}\left(\alpha_{i}'^{\dagger}\middle|\mathcal{D}\middle|\alpha_{j}\right)=d\left(\alpha_{i}'^{\dagger}\middle|\mathcal{D}\middle|\chi\right)=d\left(\alpha_{i}'^{\dagger}\middle|\chi\right)=d\cdot\frac{1}{d}=1,
\]
for all $i=1,\ldots,d$, where we have used proposition~\ref{prop:diagonalization chi d-level 2},
and the fact that unital channels leave $\chi$ invariant. In conclusion
we have shown that $D$ is doubly-stochastic.
\end{proof}
This lemma is a generalisation of lemma~\ref{lem:doubly-stochastic},
which can be obtained by taking $\mathcal{D}$ to be the identity
channel $\mathcal{I}$.

Vice versa, every doubly stochastic matrix defines a unital channel.
\begin{lem}
\label{lem:matrixchannel}Let $D$ be a $d\times d$ doubly stochastic
matrix, and let $\left\{ \alpha_{i}\right\} _{i=1}^{d}$ and $\left\{ \alpha_{i}'\right\} _{i=1}^{d}$
be any two pure maximal sets of system $\mathrm{A}$. Then, the channel
defined by $\mathcal{D}:=\sum_{j=1}^{d}\left|\rho_{j}\right)\left(\alpha_{j}^{\dagger}\right|$,
with $\rho_{j}:=\sum_{i=1}^{d}D_{ij}\alpha'_{i}$, is unital.
\end{lem}

\begin{proof}
The transformation $\mathcal{D}$ is a measure-and-prepare channel:
it can be implemented by performing the observation-test $\left\{ \alpha_{j}^{\dagger}\right\} _{j=1}^{d}$
and by preparing the state $\rho_{j}$ conditionally on outcome $j$.
Moreover, one has 
\[
\mathcal{D}\chi=\sum_{j=1}^{d}\rho_{j}\left(\alpha_{j}^{\dagger}\middle|\chi\right)=\frac{1}{d}\sum_{j=1}^{d}\sum_{i=1}^{d}D_{ij}\alpha'_{i}=\frac{1}{d}\sum_{i=1}^{d}\alpha'_{i}=\chi,
\]
the third equality following from the definition of doubly stochastic
matrix, and the fourth from the diagonalisation of the state $\chi$
(proposition~\ref{prop:diagonalization chi d-level 2}).
\end{proof}
Lemmas~\ref{lem:channelmatrix} and \ref{lem:matrixchannel} establish
a direct connection between unital channels and doubly stochastic
matrices. Using this connection now we show that the ability to convert
states in the unital resource theory is completely determined by a
suitable majorisation criterion. Let us start by recalling the definition
of majorisation \cite{Olkin}.
\begin{defn}
Let $\mathbf{x}$ and $\mathbf{y}$ be two generic vectors in $\mathbb{R}^{d}$.
One says that $\mathbf{x}$ \emph{majorises} $\mathbf{y}$, denoted
${\bf x}\succeq{\bf y}$, if, when the entries of $\mathbf{x}$ and
$\mathbf{y}$ are rearranged in decreasing order, one has 
\[
\sum_{i=1}^{k}x_{i}\geq\sum_{i=1}^{k}y_{i}\quad k=1,\ldots,d-1
\]
and 
\[
\sum_{i=1}^{d}x_{i}=\sum_{i=1}^{d}y_{i}.
\]
\end{defn}

Majorisation can be equivalently characterised in terms of doubly
stochastic matrices: one has $\mathbf{x}\succeq\mathbf{y}$ if and
only if $\mathbf{y}=D\mathbf{x}$, where $D$ is a doubly stochastic
matrix \cite{Hardy-Littlewood-Polya1929,Olkin}. The idea is that
$\mathbf{y}\preceq\mathbf{x}$ if $\mathbf{y}$ is ``more random''
than $\mathbf{x}$.

In every sharp theory with purification, majorisation is a necessary
and sufficient condition for convertibility under unital channels.
\begin{thm}
\label{thm:unital-majorisation}Let $\rho$ and $\sigma$ be normalised
states, and let $\mathbf{p}$ and $\mathbf{q}$ be their spectra respectively.
The state $\rho$ can be converted into the state $\sigma$ by a unital
channel if and only if $\mathbf{p}\succeq\mathbf{q}$.
\end{thm}

\begin{proof}
Let $\rho=\sum_{j=1}^{d}p_{j}\alpha_{j}$ and $\sigma=\sum_{j=1}^{d}q_{j}\alpha'_{j}$
be diagonalisations of $\rho$ and $\sigma$ respectively. We first
show that $\rho\succeq_{\mathrm{unital}}\sigma$ implies $\mathbf{p}\succeq\mathbf{q}$.
Since $\sigma=\mathcal{D}\rho$, with $\mathcal{D}$ unital channel,
one has
\[
\sum_{j=1}^{d}q_{j}\alpha'_{j}=\sum_{j=1}^{d}p_{j}\mathcal{D}\alpha_{j}
\]
Applying $\alpha_{i}'^{\dagger}$ to both sides, we obtain
\begin{equation}
q_{i}=\sum_{j=1}^{d}p_{j}\left(\alpha_{i}'^{\dagger}\middle|\mathcal{D}\middle|\alpha_{j}\right)=:\sum_{j=1}^{d}D_{ij}p_{j},\label{eq:matrixD}
\end{equation}
where we have set $D_{ij}:=\left(\alpha_{i}'^{\dagger}\middle|\mathcal{D}\middle|\alpha_{j}\right)$.
Now, the $D_{ij}$'s are the entries of a doubly stochastic matrix
$D$ by lemma~\ref{lem:channelmatrix}. Hence, eq.~\eqref{eq:matrixD}
implies $\mathbf{p}\succeq\mathbf{q}$. 

To prove sufficiency, suppose that $\mathbf{p}\succeq\mathbf{q}$
and let $D$ be a doubly stochastic matrix such that $\mathbf{q}=D\mathbf{p}$.
Define the measure-and-prepare channel $\mathcal{D}:=\sum_{j=1}^{d}\left|\rho_{j}\right)\left(\alpha_{j}^{\dagger}\right|$,
with $\rho_{j}:=\sum_{i=1}^{d}D_{ij}\alpha'_{i}$. By construction,
one has
\[
\mathcal{D}\rho=\sum_{j=1}^{d}\rho_{j}\left(\alpha_{j}^{\dagger}\middle|\rho\right)=\sum_{i=1}^{d}\alpha'_{i}\sum_{j=1}^{d}D_{ij}p_{j}=\sum_{i=1}^{d}q_{i}\alpha'_{i}=\sigma.
\]
Now, the channel $\mathcal{D}$ is unital by lemma~\ref{lem:matrixchannel}.
Hence, $\rho$ can be converted into $\sigma$ by a unital channel.
\end{proof}
Note that since RaRe channels and noisy operations are special cases
of unital channels, majorisation is a \emph{necessary} condition for
convertibility in the RaRe and noisy resource theories. We will examine
when it is sufficient for RaRe channels in section~\ref{sec:Sufficiency-of-majorisation}.
So far we have proved
\begin{equation}
\rho\succeq_{\mathrm{RaRe}}\sigma\quad\Longrightarrow\quad\rho\succeq_{\mathrm{noisy}}\sigma\quad\Longrightarrow\quad\rho\succeq_{\mathrm{unital}}\sigma\quad\Longleftrightarrow\quad\mathbf{p}\succeq\mathbf{q}.\label{eq:chain implications}
\end{equation}
Now the relations $\psi\succeq_{F}\rho\succeq_{F}\chi$, where $F$
denotes any of the three sets of free operations, can be easily understood
in terms of majorisation. Indeed the spectrum of a pure state is $\left(\begin{array}{cccc}
1 & 0 & \ldots & 0\end{array}\right)^{T}$, whereas the spectrum of the microcanonical state is $\left(\begin{array}{ccc}
\frac{1}{d} & \ldots & \frac{1}{d}\end{array}\right)^{T}$, and it is straightforward to check that
\[
\left(\begin{array}{c}
1\\
0\\
\vdots\\
0
\end{array}\right)\succeq\mathbf{p}\succeq\left(\begin{array}{c}
\frac{1}{d}\\
\vdots\\
\frac{1}{d}
\end{array}\right).
\]

\subsection{Operational characterisation of the eigenvalues}

Besides giving a necessary and sufficient condition for the convertibility
under unital channels, majorisation also provides an operational characterisation
of the eigenvalues of a state: the eigenvalues are the least random
probability distribution that can be generated by pure observation-tests.
A similar result was also proved in GPTs satisfying other axioms in
\cite{Krumm-Muller,Krumm-thesis,Colleagues}.

To prove the main result we need a lemma on the structure of pure
observation-tests.
\begin{lem}
\label{lem:characterisation pure tests}Let $\left\{ a_{i}\right\} _{i=1}^{n}$
be a \emph{pure} observation-test. Then, for every $i\in\left\lbrace 1,\ldots,n\right\rbrace $,
$a_{i}=\lambda_{i}\alpha_{i}^{\dagger}$, for some pure state $\alpha_{i}$,
and $\lambda_{i}\in\left(0,1\right]$. Moreover $\sum_{i=1}^{n}\lambda_{i}=d$,
and $n\geq d$. One has $n=d$ if and only if $\left\{ a_{i}\right\} _{i=1}^{n}$
is a pure sharp measurement.
\end{lem}

\begin{proof}
Since we can extend the diagonalisation to elements of $\mathsf{Eff}_{\mathbb{R}}\left(\mathrm{A}\right)$,
we can write $a_{i}=\lambda_{i}\alpha_{i}^{\dagger}$, and there is
one term in the diagonalisation of $a_{i}$ because it is pure. Being
$a_{i}$ an effect, it must be $\lambda_{i}\in\left(0,1\right]$,
because $\lambda_{i}=\left(a_{i}\middle|\alpha_{i}\right)$. Now let
us prove that $\sum_{i=1}^{n}\lambda_{i}=d$. By Causality, $\sum_{i=1}^{n}\lambda_{i}\alpha_{i}^{\dagger}=u$.
Now consider 
\[
1=\mathrm{tr}\:\chi=\sum_{i=1}^{n}\lambda_{i}\left(\alpha_{i}^{\dagger}\middle|\chi\right)=\sum_{i=1}^{n}\lambda_{i}\cdot\frac{1}{d},
\]
whence $\sum_{i=1}^{n}\lambda_{i}=d$. Since $\lambda_{i}\leq1$,
we have 
\[
d=\sum_{i=1}^{n}\lambda_{i}\leq\sum_{i=1}^{n}1=n,
\]
so $n\geq d$. Now let us prove that $n=d$ if and only if $\left\{ a_{i}\right\} _{i=1}^{n}$
is a pure sharp measurement. Suppose $\left\{ a_{i}\right\} _{i=1}^{n}$
is a pure sharp measurement, then $n=d$ by proposition~\ref{prop:pure test chi}.
Conversely, suppose we know that $n=d$. Then in this case, the only
possibility of having $\sum_{i=1}^{n}\lambda_{i}=d$ is when $\lambda_{i}=1$
for every $i$. Therefore all the effects are normalised, and the
observation-test can be rewritten as $\left\{ \alpha_{i}^{\dagger}\right\} _{i=1}^{d}$
for some pure states $\left\{ \alpha_{i}\right\} _{i=1}^{d}$.
\end{proof}
Now we are ready to prove the main result. It is useful to introduce
the following notation: given two vectors $\mathbf{x}\in\mathbb{R}^{n}$,
and $\mathbf{y}\in\mathbb{R}^{m}$, let us define
\begin{equation}
\mathbf{x}\oplus\mathbf{y}:=\left(\begin{array}{c}
\mathbf{x}\\
\hline \mathbf{y}
\end{array}\right),\label{eq:direct sum vectors}
\end{equation}
which is a vector of $\mathbb{R}^{n+m}$. This operation is nothing
but appending $\mathbf{y}$ under $\mathbf{x}$ to create a larger
vector.
\begin{prop}
\label{prop:majorisation measurement}Consider a \emph{pure} observation-test
$\boldsymbol{a}=\left\lbrace a_{i}\right\rbrace _{i=1}^{n}$ and state
$\rho$. Let $\mathbf{q}_{\boldsymbol{a}}\in\mathbb{R}^{n}$ be the
vector with entries $q_{a,i}=\left(a_{i}\middle|\rho\right)$ for
$i\in\left\lbrace 1,\ldots,n\right\rbrace $. Then if $\mathbf{p}$
is the spectrum of $\rho$, define $\widetilde{\mathbf{p}}:=\mathbf{p}\oplus\mathbf{0}$,
where $\mathbf{0}$ is the $n-d$ dimensional null vector. Then $\mathbf{q}_{\boldsymbol{a}}\preceq\widetilde{\mathbf{p}}$.
\end{prop}

Note that, since by lemma~\ref{lem:characterisation pure tests}
$n\geq d$, the vector $\mathbf{q}_{\boldsymbol{a}}$ and the spectrum
$\mathbf{p}$ are not directly comparable in terms of majorisation,
because they have different dimensions, $\mathbf{q}_{\boldsymbol{a}}$
being the larger. Therefore, to circumvent the problem, we make $\mathbf{p}$
larger by attaching $n-d$ 0 entries. In this way we get $\widetilde{\mathbf{p}}$,
which is an $n$-dimensional vector like $\mathbf{q}_{\boldsymbol{a}}$.
Now the let us see the proof, whose lines will be close to \cite[lemma B.1]{Entropy-Short}
for quantum theory.
\begin{proof}
By lemma~\ref{lem:characterisation pure tests}, for each $a_{i}\in\left\lbrace a_{i}\right\rbrace _{i=1}^{n}$,
we have $a_{i}=\lambda_{i}\alpha_{i}^{\dagger}$, for some $0<\lambda_{i}\leq1$,
and for some pure state $\alpha_{i}$. Consider a diagonalisation
of $\rho=\sum_{j=1}^{d}p_{j}\alpha'_{j}$. We have 
\[
q_{a,i}:=\left(a_{i}\middle|\rho\right)=\sum_{j=1}^{d}p_{j}\left(a_{i}\middle|\alpha'_{j}\right)=\sum_{j=1}^{d}\lambda_{i}p_{j}\left(\alpha_{i}^{\dagger}\middle|\alpha'_{j}\right).
\]
Now, $M_{ij}:=\lambda_{i}\left(\alpha_{i}^{\dagger}\middle|\alpha'_{j}\right)$
are the entries of an $n\times d$ matrix $M$ such that $\mathbf{q}_{\boldsymbol{a}}=M\mathbf{p}$.
Clearly $M_{ij}\geq0$ for all $i=1,\ldots n$, and $j=1,\ldots,d$.
Calculating $\sum_{i=1}^{n}M_{ij}$, we have
\begin{equation}
\sum_{i=1}^{n}M_{ij}=\sum_{i=1}^{d}\left(\lambda_{i}\alpha_{i}^{\dagger}\middle|\alpha'_{j}\right)=\mathrm{tr}\:\alpha'_{j}=1,\label{eq:sum column M}
\end{equation}
whence the column of the matrix $M$ sum to 1. Now let us move to
$\sum_{j=1}^{d}M_{ij}$.
\begin{equation}
\sum_{j=1}^{d}M_{ij}=\lambda_{i}\sum_{j=1}^{d}\left(\alpha_{i}^{\dagger}\middle|\alpha'_{j}\right)=\lambda_{i}d\left(\alpha_{i}^{\dagger}\middle|\chi\right)=\lambda_{i}d\cdot\frac{1}{d}=\lambda_{i}\leq1\label{eq:sum row M}
\end{equation}
If $n=d$ we are done, because in this case $\lambda_{i}=1$ for every
$i$ (by lemma~\ref{lem:characterisation pure tests}), $M$ is doubly
stochastic and $\mathbf{q}_{\boldsymbol{a}}=M\mathbf{p}$, whence
the thesis.

Now, suppose $n>d$; we wish to construct an $n\times n$ doubly stochastic
matrix $D$ from $M$, such that we can write $q_{a,i}=\sum_{j=1}^{n}D_{ij}\widetilde{p}_{j}$,
where $\widetilde{\mathbf{p}}=\mathbf{p}\oplus\mathbf{0}$ is the
vector of probabilities defined as
\[
\widetilde{p}_{j}:=\begin{cases}
p_{j} & 1\leq j\leq d\\
0 & d+1\leq j\leq n
\end{cases}.
\]
Let us define $D$ as 
\[
D:=\left(\begin{array}{c|c}
\underbrace{M}_{d\textrm{ columns}} & \underbrace{\begin{array}{ccc}
\frac{1-\boldsymbol{\lambda}}{n-d} & \ldots & \frac{1-\boldsymbol{\lambda}}{n-d}\end{array}}_{n-d\textrm{ columns}}\end{array}\right),
\]
where the last $n-d$ columns are all equal to each other, with their
$i$th entry equal to $\frac{1-\lambda_{i}}{n-d}$.

Now, $D$ is doubly stochastic. Indeed each entry is non-negative,
because $\lambda_{i}\leq1$ for all $i\in\left\{ 1,\ldots,n\right\} $
and $n\geq d$. Furthermore, 
\[
\sum_{i=1}^{n}D_{ij}=\begin{cases}
\sum_{i=1}^{n}M_{ij} & 1\leq j\leq d\\
\frac{n-\sum_{i=1}^{n}\lambda_{i}}{n-d} & d+1\leq j\leq n
\end{cases}=1
\]
by eq.~\eqref{eq:sum column M}, and because $\sum_{i=1}^{n}\lambda_{i}=d$
(by lemma~\ref{lem:characterisation pure tests}). Finally 
\[
\sum_{j=1}^{n}D_{ij}=\sum_{j=1}^{d}M_{ij}+\sum_{j=d+1}^{n}\frac{1-\lambda_{i}}{n-d}=1,
\]
having used eq.~\eqref{eq:sum row M}. Clearly now we have $q_{a,i}=\sum_{j=1}^{n}D_{ij}\widetilde{p}_{j}$,
because, by construction of $\widetilde{\mathbf{p}}$ and $D$, 
\[
\mathbf{q}_{\boldsymbol{a}}=\left(\begin{array}{c|c}
M & \begin{array}{ccc}
\frac{1-\boldsymbol{\lambda}}{n-d} & \ldots & \frac{1-\boldsymbol{\lambda}}{n-d}\end{array}\end{array}\right)\left(\begin{array}{c}
\mathbf{p}\\
\hline \mathbf{0}
\end{array}\right).
\]
Therefore $\mathbf{q}_{\boldsymbol{a}}\preceq\widetilde{\mathbf{p}}$.
\end{proof}

\section{Mixedness monotones\label{sec:Mixedness-monotones}}

In the previous section we saw that the majorisation criterion determines
whether a state is more resourceful than another in the unital resource
theory, and it is a necessary condition in the other two. To be more
quantitative, let us introduce monotones (cf.\ section~\ref{sec:Resource-monotones}).
We noted that pure states are the most valuable states in all three
resource theories (see subsection~\ref{subsec:Inclusion-relations}),
so it is natural to call the three resource theories ``resource theories
of purity''. A \emph{purity monotone under the free operations $F$}
for system $\mathrm{A}$ is a function $P:\mathsf{St}_{1}\left(\mathrm{A}\right)\to\mathbb{R}$
satisfying the condition $P\left(\rho\right)\ge P\left(\sigma\right)$
if $\rho\succeq_{F}\sigma$ \cite{Chiribella-Scandolo-entanglement}.

In sharp theories with purification, thanks to theorem~\ref{thm:unital-majorisation},
unital purity monotones have a complete mathematical characterisation
in terms of Schur-convex functions.
\begin{defn}
A function $f:\mathbb{R}^{d}\to\mathbb{R}$ is called \emph{Schur-convex}
if $\mathbf{x}\preceq\mathbf{y}$ implies $f\left(\mathbf{x}\right)\leq f\left(\mathbf{y}\right)$.
\end{defn}

These functions assign a greater real number to more ordered vectors.
The following proposition is therefore quite natural.
\begin{prop}
\label{prop:unital monotones}A function on the state space $P:\mathsf{St}_{1}\left(\mathrm{A}\right)\to\mathbb{R}$
is a unital purity monotone if and only if $P\left(\rho\right)=f\left(\mathbf{p}\right)$,
where $\mathbf{p}$ is the spectrum of $\rho$ and $f:\mathbb{R}^{d}\to\mathbb{R}$
is a Schur-convex function.
\end{prop}

\begin{proof}
Theorem~\ref{thm:unital-majorisation} shows that the convertibility
of states under unital channels is fully captured by their eigenvalues.
Consequently, a unital monotone will be a function only of the eigenvalues
of a state: there exists a function $f:\mathbb{R}^{d}\to\mathbb{R}$
such that $P\left(\rho\right)=f\left(\mathbf{p}\right)$, for every
normalised state $\rho$. Now, suppose that $\mathbf{p}$ and $\mathbf{q}$
are two probability distributions satisfying $\mathbf{p}\succeq\mathbf{q}$.
Then, theorem~\ref{thm:unital-majorisation} implies that there is
a unital channel transforming the state $\rho=\sum_{i=1}^{d}p_{i}\alpha_{i}$
into the state $\sigma=\sum_{i=1}^{d}q_{i}\alpha_{i}$, for any pure
maximal set $\left\{ \alpha_{i}\right\} _{i=1}^{d}$. As a result,
we obtain the relation 
\[
f\left(\mathbf{p}\right)=P\left(\rho\right)\geq P\left(\sigma\right)=f\left(\mathbf{q}\right).
\]
This means that $f$ is Schur-convex.

Conversely, given a Schur-convex function $f$ one can define a function
$P_{f}$ on the state space, as $P_{f}\left(\rho\right):=f\left(\mathbf{p}\right)$,
$\mathbf{p}$ being the spectrum of $\rho$. By theorem~\ref{thm:unital-majorisation},
if $\rho\succeq_{\mathrm{unital}}\sigma$, then $\mathbf{p}\succeq\mathbf{q}$,
where $\mathbf{p}$ and $\mathbf{q}$ are the spectra of the two states,
respectively. Therefore
\[
P_{f}\left(\rho\right)=f\left(\mathbf{p}\right)\geq f\left(\mathbf{q}\right)=P_{f}\left(\sigma\right),
\]
where we have used the fact that $f$ is Schur-convex. This shows
that $P_{f}$ is a unital purity monotone.
\end{proof}
Since majorisation is a necessary condition for convertibility in
all the three resource theories, $P_{f}\left(\rho\right)=f\left(\mathbf{p}\right)$,
with $\mathbf{p}$ is the spectrum of $\rho$, and $f$ Schur-convex,
is a purity monotone in \emph{all} the three resource theories. Hence
these functions will simply be called ``purity monotones'', without
the need to specify the resource theory. The fact that we have not
proved the sufficiency of majorisation for RaRe channels means that
there might exist RaRe purity monotones that are not generated by
Schur-convex functions.
\begin{rem}
\label{rem:remark monotone invariant }Note that the purity monotones
arising from Schur-convex functions are invariant under reversible
channels: $P_{f}\left(\rho\right)=P_{f}\left(\mathcal{U}\rho\right)$.
This is because $\rho$ and $\mathcal{U}\rho$ have the same spectrum.
Indeed if $\rho=\sum_{i=1}^{d}p_{i}\alpha_{i}$ is a diagonalisation
of $\rho$, then $\mathcal{U}\rho=\sum_{i=1}^{d}p_{i}\mathcal{U}\alpha_{i}$
is a diagonalisation of $\mathcal{U}\rho$, because the pure states
$\mathcal{U}\alpha_{i}$ are distinguished by the observation-test
$\left\{ \alpha_{i}^{\dagger}\mathcal{U}^{-1}\right\} _{i=1}^{d}$.
Then obviously $\rho$ and $\mathcal{U}\rho$ have the same spectrum.
\end{rem}

In thermodynamics the relevant quantities for isolated systems are
entropies, which are often presented popularly as ``measures of disorder''.
From this perspective, entropies are the opposite of purity monotones,
which instead assign a higher value to ordered states. To comply with
the standard thermodynamic treatment, we define \emph{mixedness monotones},
and we will stick to them in the following presentation.
\begin{defn}
A \emph{mixedness monotone} for system $\mathrm{A}$ is a function
$M:\mathsf{St}_{1}\left(\mathrm{A}\right)\rightarrow\mathbb{R}$ such
that $-M$ is a purity monotone.
\end{defn}

In other words, if $\rho\succeq_{F}\sigma$, we have $M\left(\rho\right)\leq M\left(\sigma\right)$.
In the same spirit, one defines Schur-concave functions.
\begin{defn}
A function $f:\mathbb{R}^{d}\to\mathbb{R}$ is called \emph{Schur-concave}
if $-f$ is Schur-convex.
\end{defn}

This means that for a Schur-concave function, if $\mathbf{x}\preceq\mathbf{y}$
then $f\left(\mathbf{x}\right)\geq f\left(\mathbf{y}\right)$.

Viewing proposition~\ref{prop:unital monotones} in the light of
mixedness monotones, given a Schur-concave function $f$, one can
generate a mixedness monotone $M_{f}\left(\rho\right):=f\left(\mathbf{p}\right)$,
where $\rho$ is a state, and $\mathbf{p}$ its spectrum. In the case
of the unital resource theory, mixedness monotones are \emph{all}
generated by Schur-concave functions.

A slightly more restrictive notion is that of \emph{generalised entropy}.
\begin{defn}
For every system $\mathrm{A}$, let $M:\mathsf{St}_{1}\left(\mathrm{A}\right)\to\mathbb{R}$
be a mixedness monotone. We say that $M$ is a \emph{generalised entropy}
if it is additive, that is 
\[
M\left(\rho_{\mathrm{A}}\otimes\sigma_{\mathrm{B}}\right)=M\left(\rho_{\mathrm{A}}\right)+M\left(\sigma_{\mathrm{B}}\right)
\]
for all $\rho_{\mathrm{A}}\in\mathsf{St}_{1}\left(\mathrm{A}\right)$,
and all $\sigma_{\mathrm{B}}\in\mathsf{St}_{1}\left(\mathrm{B}\right)$.
\end{defn}

By proposition~\ref{prop:unital monotones}, some generalised entropies
can be obtained from particular Schur-concave functions.
\begin{cor}
\label{cor:generalised entropies}Let $f:\mathbb{R}^{d}\rightarrow\mathbb{R}$
be a Schur-concave function for all $d$, satisfying the additivity
property $f\left({\bf p}\otimes{\bf q}\right)=f\left({\bf p}\right)+f\left({\bf q}\right)$,
where ${\bf p}\otimes{\bf q}$ denotes the Kronecker product. Then,
the corresponding mixedness monotone $M_{f}\left(\rho\right)=f\left(\mathbf{p}\right)$,
where $\mathbf{p}$ is the spectrum of $\rho$, is a generalised entropy.
\end{cor}

Again, in the case of the unital resource theory, all generalised
entropies are of the form of corollary~\ref{cor:generalised entropies}.
Let us see some examples.
\begin{example}
An important example of additive Schur-concave functions are Rényi
entropies \cite{Renyi} 
\[
H_{\alpha}\left({\bf p}\right)=\frac{1}{1-\alpha}\log_{a}\left(\sum_{i=1}^{d}p_{i}^{\alpha}\right),
\]
with $a>1$ and $\alpha\geq0$, where $\mathbf{p}$ is a vector of
probabilities\footnote{This means a vector representing a probability distribution.}.
If we take the limit $\alpha\rightarrow1$ we recover \emph{Shannon-von
Neumann entropy} \cite{vonNeumann,Shannon}:
\[
H\left({\bf p}\right)=-\sum_{i=1}^{d}p_{i}\log_{a}p_{i}=\lim_{\alpha\to1}H_{\alpha}\left({\bf p}\right).
\]
Particularly important cases are when $\alpha=0$, and $\alpha\rightarrow+\infty$.
In this cases,
\[
H_{0}\left({\bf p}\right)=\log_{a}\left|\mathrm{supp}\:\mathbf{p}\right|,
\]
where $\left|\mathrm{supp}\:\mathbf{p}\right|$ denotes the number
of the non-vanishing entries of $\mathbf{p}$. Moreover,
\[
H_{\infty}\left(\mathbf{p}\right):=\lim_{\alpha\rightarrow+\infty}H_{\alpha}\left({\bf p}\right)=-\log_{a}p_{\mathrm{max}},
\]
where $p_{\mathrm{max}}$ denotes the maximum entry of $\mathbf{p}$.
Rényi entropies are decreasing in $\alpha$, so
\[
H_{\infty}\left(\mathbf{p}\right)\leq H_{\alpha}\left({\bf p}\right)\leq H_{0}\left({\bf p}\right)
\]
for $\alpha\geq0$. For this reason, $H_{\infty}$ is also known as
the \emph{min-entropy}, and $H_{0}$ as the \emph{max-entropy}.

Now we can define the generalised Rényi entropies as $S_{\alpha}\left(\rho\right):=H_{\alpha}\left({\bf p}\right)$,
for $\alpha\in\left[0,+\infty\right]$, where $\mathbf{p}$ is the
spectrum of $\rho$. In particular $S\left(\rho\right)=H\left(\mathbf{p}\right)$
is the generalised Shannon-von Neumann entropy. Note that one has
the obvious bounds 
\[
0\le S_{\alpha}\left(\rho\right)\le\log_{a}d,
\]
for every state $\rho$, and every $\alpha\in\left[0,+\infty\right]$,
where $d$ is the dimension of the system. The lower bound is achieved
by any pure state, and the upper bound by the microcanonical state
$\chi$.
\end{example}

\subsection{Preparation and measurement monotones}

In some cases, it is possible to connect generalised entropies, defined
on the spectra of states, to measures of the minimum randomness we
can extract from a state by pure measurements, or of the minimum randomness
needed to prepare that state. To this end, it is useful to identify
a particular class of Schur-concave functions, which we call \emph{reducible}.
\begin{defn}
A Schur-concave function $f$ is called \emph{reducible} if for every
vector $\mathbf{x}\in\mathbb{R}^{d}$ one has 
\[
f\left(\mathbf{x}\oplus\mathbf{0}\right)=f\left(\mathbf{x}\right),
\]
where $\mathbf{0}$ is a null vector with any number of components.
\end{defn}

Here we have used the same notation as in eq.~\eqref{eq:direct sum vectors}
to denote appending some zero entries to a vector. In words, for a
reducible Schur-concave function the zero entries of a vector do not
matter, it only sees the non-vanishing ones. Examples of reducible
Schur-concave functions are Rényi entropies (and hence Shannon-von
Neumann entropy). Not all Schur-concave functions are reducible though,
as shown in the following example.
\begin{example}
Given a vector of probabilities $\mathbf{p}$ of dimension $d$, consider
the function $V\left(\mathbf{p}\right)=\frac{1}{d}\left(1-\sum_{i=1}^{d}p_{i}^{2}\right)$,
arising from the variance of $\mathbf{p}$. $V$ is Schur-concave,
but it is \emph{not} reducible. Indeed, consider the vectors $\mathbf{p}=\left(\begin{array}{cc}
\frac{1}{2} & \frac{1}{2}\end{array}\right)^{T}$, and $\widetilde{\mathbf{p}}=\left(\begin{array}{ccc}
\frac{1}{2} & \frac{1}{2} & 0\end{array}\right)^{T}$; we have $V\left(\mathbf{p}\right)=\frac{1}{4}$, but $V\left(\widetilde{\mathbf{p}}\right)=\frac{1}{6}$,
whence $V$ is not reducible.
\end{example}

It is worth noting that the marginals of a pure bipartite state have
the same value for any reducible Schur-concave function. So if $\Psi$
is a pure state of $\mathrm{AB}$ and $\rho_{\mathrm{A}}$ and $\rho_{\mathrm{B}}$
its marginals on $\mathrm{A}$ and $\mathrm{B}$ respectively, one
has
\[
M_{f}\left(\rho_{\mathrm{A}}\right)=M_{f}\left(\rho_{\mathrm{B}}\right),
\]
for every reducible Schur-concave function $f$. This easily follows
from the generalised Schmidt decomposition (theorem~\ref{thm:schmidt}),
which ensures that the marginals of a pure bipartite state have the
same non-vanishing eigenvalues. These non-vanishing eigenvalues are
the only ones seen by a reducible Schur-concave function.

In every sharp theory with purification, the mixedness monotones arising
from reducible Schur-concave functions have a nice characterisation
in terms of optimal measurements, or, dually, in terms of optimal
ensemble decompositions. Let us start from two definitions.
\begin{defn}
Given a reducible Schur-concave function $f$, the \emph{measurement
monotone} $M_{f}^{{\rm meas}}$ of a state $\rho\in\mathsf{St}_{1}\left(\mathrm{A}\right)$
is defined as
\[
M_{f}^{{\rm meas}}\left(\rho\right):=\inf_{\boldsymbol{a}}f\left({\bf q}\right),
\]
where the infimum is taken over all pure observation-tests $\boldsymbol{a}:=\left\lbrace a_{i}\right\rbrace $
of system $\mathrm{A}$, and $q_{i}:=\left(a_{i}\middle|\rho\right)$. 
\end{defn}

\begin{defn}
The \emph{preparation monotone} $M_{f}^{{\rm prep}}$ is defined as
\[
M_{f}^{{\rm prep}}\left(\rho\right):=\inf_{\sum_{i}\lambda_{i}\psi_{i}=\rho}f\left(\boldsymbol{\lambda}\right),
\]
where the infimum is over all convex decompositions $\sum_{i}\lambda_{i}\psi_{i}$
of the state $\rho$ in terms of pure states $\psi_{i}$.
\end{defn}

In words, the measurement monotone $M_{f}^{{\rm meas}}$ is the smallest
amount of mixedness (as measured by the function $f$) present in
the probability distributions generated by pure observation-tests
on $\rho$. Dually, the preparation monotone $M_{f}^{{\rm prep}}$
is the smallest amount of mixedness necessary to prepare $\rho$ as
an ensemble of pure states. Why do we have to take the infimum? To
understand it, consider the following example for a preparation monotone.
\begin{example}
Consider a pure state $\psi$ and a preparation monotone $M_{f}^{{\rm prep}}$.
Clearly, convex decompositions of $\psi$ into pure states are all
trivial, i.e.\ of the form $\sum_{i}\lambda_{i}\psi$. Comparing
the probability distribution $\left\{ \lambda_{i}\right\} $ with
the extremal one $\boldsymbol{\delta}$, we have
\[
\boldsymbol{\lambda}\preceq\left(\begin{array}{c}
1\\
0\\
\vdots\\
0
\end{array}\right)=\boldsymbol{\delta}.
\]
Therefore $f\left(\boldsymbol{\lambda}\right)\geq f\left(\boldsymbol{\delta}\right)$;
however this is a spurious effect, because the pure states involved
in a convex realisation of $\psi$ are exactly the same, and all equal
to $\psi$! The presence of the infimum cancels these spurious effects,
and we have $M_{f}^{{\rm prep}}\left(\psi\right)=f\left(\boldsymbol{\delta}\right)$.
\end{example}

Let us clarify why we need $f$ to be reducible with an example from
quantum theory.
\begin{example}
Consider a state $\rho=p\ket{0}\bra{0}+\left(1-p\right)\ket{1}\bra{1}$
of a qutrit, where $\left\{ \ket{0},\ket{1},\ket{2}\right\} $ is
the computational basis, and $p\in\left(0,1\right)$. Clearly if we
use the spectral measurement $\left\{ \ket{0}\bra{0},\ket{1}\bra{1},\ket{2}\bra{2}\right\} $,
in the evaluation of the measurement monotone we obtain the vector
${\bf q}=\left(\begin{array}{ccc}
p & 1-p & 0\end{array}\right)^{T}$. Now take the pure POVM $\left\{ \ket{0}\bra{0},\ket{1}\bra{1},\frac{1}{2}\ket{2}\bra{2},\frac{1}{2}\ket{2}\bra{2}\right\} $,
which gives rise to the vector ${\bf q}'=\left(\begin{array}{cccc}
p & 1-p & 0 & 0\end{array}\right)^{T}$. Clearly there is no difference in the randomness one can extract
from the former and the latter measurement when performed on $\rho$.
However, if $f$ is non-reducible, there is no guarantee that $f\left(\mathbf{q}\right)=f\left(\mathbf{q}'\right)$.

Similarly, in the evaluation of the preparation monotone, the writing
$\rho=p\ket{0}\bra{0}+\left(1-p\right)\ket{1}\bra{1}$ gives rise
to $\boldsymbol{\lambda}=\left(\begin{array}{cc}
p & 1-p\end{array}\right)^{T}$. However, one can also write $\rho$ as $\rho=p\ket{0}\bra{0}+\left(1-p\right)\ket{1}\bra{1}+0\ket{2}\bra{2},$
giving rise to $\boldsymbol{\lambda}'=\left(\begin{array}{ccc}
p & 1-p & 0\end{array}\right)^{T}$. Again, there is no more randomness involved in the latter preparation
than in the former, so we should have $f\left(\boldsymbol{\lambda}\right)=f\left(\boldsymbol{\lambda}'\right)$.
\end{example}

The definitions of measurement and preparation monotones can be put
forward in any causal GPT, since they involve very primitive elements
of GPTs, such as the convex structure, and the fact that observation-tests
yield probabilities when performed on states. Specifically \cite{Entropy-Barnum,Entropy-Short,Entropy-Kimura}
defined preparation and measurement Shannon entropy in GPTs.

Despite the name we used, in general GPTs, or even in microcanonical
ones, $M_{f}^{{\rm meas}}$ and $M_{f}^{{\rm prep}}$ lack a clear
interpretation as mixedness monotones, unless further assumptions
are made on $f$ (or on the theory). For instance, taking $f$ concave
makes $M_{f}^{{\rm meas}}$ a RaRe mixedness monotone \cite{Chiribella-Scandolo-entanglement}.
However, in sharp theories with purification we find that $M_{f}^{{\rm meas}}$
and $M_{f}^{{\rm prep}}$ are actual mixedness monotones.
\begin{thm}
\label{thm:measurement =00003D preparation}In every sharp theory
with purification one has 
\[
M_{f}^{{\rm meas}}\left(\rho\right)=M_{f}^{{\rm prep}}\left(\rho\right)=M_{f}\left(\rho\right),
\]
 for every reducible Schur-concave function $f$ and for every state
$\rho$.
\end{thm}

\begin{proof}
Let us prove that $M_{f}^{{\rm meas}}$ coincides with $M_{f}$. Let
$\rho=\sum_{i=1}^{d}p_{i}\alpha_{i}$ be a diagonalisation of $\rho\in\mathsf{St}_{1}\left(\mathrm{A}\right)$.
If we take the pure sharp measurement $\left\lbrace \alpha_{i}^{\dagger}\right\rbrace _{i=1}^{d}$,
we have $\left(\alpha_{i}^{\dagger}\middle|\rho\right)=p_{i}$. Hence,
\[
M_{f}^{{\rm meas}}\left(\rho\right)\le f\left({\bf p}\right)=M_{f}\left(\rho\right).
\]
To prove the converse inequality, recall proposition~\ref{prop:majorisation measurement}:
for every pure observation-test $\left\lbrace a_{i}\right\rbrace $,
one has $\mathbf{q}\preceq\widetilde{\mathbf{p}}$, where $\mathbf{q}$
is the vector of probabilities $q_{i}=\left(a_{i}\middle|\rho\right)$
and $\widetilde{\mathbf{p}}$ is the vector of the eigenvalues of
$\rho$ (with additional zeros appended, if needed). Since $f$ is
Schur-concave, we have $f\left({\bf q}\right)\ge f\left(\widetilde{\mathbf{p}}\right)$
and, taking the infimum over all pure measurements 
\[
M_{f}^{{\rm meas}}\left(\rho\right)\ge f\left(\widetilde{\mathbf{p}}\right)=f\left(\mathbf{p}\right)=M_{f}\left(\rho\right),
\]
where we have used the fact that $f$ is reducible. Summarising, we
have obtained the equality $M_{f}^{{\rm meas}}=M_{f}$. 

We now prove the equality $M_{f}^{{\rm prep}}=M_{f}$. By definition,
we have 
\[
M_{f}^{{\rm prep}}\left(\rho\right)\le f\left(\mathbf{p}\right)=M_{f}\left(\rho\right),
\]
because the diagonalisation is a special case of pure-state decomposition.
The converse inequality follows from Pure Steering. By corollary~\ref{cor:pure ensembles},
every convex decomposition of $\rho=:\rho_{\mathrm{A}}$ into pure
states can be induced by a pure sharp measurement applied to a suitable
purification of $\rho$ (which in general depends on the particular
decomposition considered). Now take the decomposition $\rho_{\mathrm{A}}=\sum_{i}\lambda_{i}\alpha_{i}$
of $\rho_{\mathrm{A}}$ into pure states, and let $\Psi\in\mathsf{PurSt}_{1}\left(\mathrm{AB}\right)$
be an associated purification as per corollary~\ref{cor:pure ensembles}.
Then there is a \emph{pure} sharp measurement $\left\{ b_{i}\right\} $
on system $\mathrm{B}$ that will induce the states\footnote{Here we are allowing some of the $\lambda_{i}$'s to be zero; this
is why we are taking the same index $i$ for the effects in the pure
sharp measurement $\left\{ b_{i}\right\} $, unlike in corollary~\ref{cor:pure ensembles}.} $\left\{ \lambda_{i}\alpha_{i}\right\} $:\[ \begin{aligned}\Qcircuit @C=1em @R=.7em @!R { & \multiprepareC{1}{\Psi} & \qw \poloFantasmaCn{\rA} & \qw \\ & \pureghost{\Psi} & \qw \poloFantasmaCn{\rB} & \measureD{b_i}}\end{aligned}~=\lambda_{i}\!\!\!\!\begin{aligned}\Qcircuit @C=1em @R=.7em @!R { & \prepareC{\alpha_i} & \qw \poloFantasmaCn{\rA} & \qw }\end{aligned}~. \]Discarding
system $\mathrm{A}$ on both sides we obtain 
\[
\left(b_{i}\middle|\rho_{\mathrm{B}}\right)=\lambda_{i},
\]
where $\rho_{\mathrm{B}}$ is the marginal state on system $\mathrm{B}$.
In other words, $\boldsymbol{\lambda}$ is the vector of the outcome
probabilities for the pure observation-test $\left\{ b_{i}\right\} $.
By definition of measurement monotone, we must have 
\begin{equation}
f\left(\boldsymbol{\lambda}\right)\ge M_{f}^{{\rm meas}}\left(\rho_{\mathrm{B}}\right)=M_{f}\left(\rho_{\mathrm{B}}\right).\label{eq:first step preparation}
\end{equation}
Since $f$ is reducible, for all the purifying systems used to induce
pure-state decompositions of $\rho$, $M_{f}\left(\rho_{\mathrm{B}}\right)=M_{f}\left(\rho_{\mathrm{A}}\right)=M_{f}\left(\rho\right)$.
Now, taking the infimum over all pure-state decompositions of $\rho_{\mathrm{A}}=\rho$
in eq.~\eqref{eq:first step preparation}, we obtain the desired
inequality
\[
M_{f}^{{\rm prep}}\left(\rho\right)\geq M_{f}\left(\rho\right).
\]
\end{proof}
This result, in particular the equality between the measurement and
preparation max-entropies $S_{0}$, was linked, in the presence of
Strong Symmetry \cite{Muller-self-duality,Barnum-interference}, to
the absence of higher-order interference, i.e.\ to the lack of irreducible
behaviour in the interference pattern obtained from three (or more)
slits \cite{Colleagues}. Therefore, this points out that sharp theories
with purification and Strong Symmetry (cf.\ also subsection~\ref{subsec:Unrestricted-reversibility})
do not have higher-order interference. In fact, it was proved that
Strong Symmetry is not necessary to prove the lack of higher-order
interference \cite{HOI}.

\section{Properties of Shannon-von Neumann entropy\label{sec:Properties-of-Shannon-von}}

The first property follows from theorem~\ref{thm:measurement =00003D preparation}:
since the measurement Shannon-von Neumann entropy was proved to be
concave in \cite{Entropy-Barnum,Entropy-Short,Entropy-Kimura}, it
follows that the ``spectral'' Shannon-von Neumann entropy we defined
here is concave as too. This means that
\[
S\left(\sum_{i}p_{i}\rho_{i}\right)\geq\sum_{i}p_{i}S\left(\rho_{i}\right),
\]
where $\left\{ p_{i}\right\} $ is a probability distribution.

As seen in subsection~\ref{subsec:Functional-calculus-on}, Shannon-von
Neumann entropy can be expressed as
\[
S\left(\rho\right)=\left(-\log_{a}\rho^{\dagger}\middle|\rho\right),
\]
meaning that $S\left(\rho\right)$ is the expectation value of the
\emph{surprisal observable} $-\log_{a}\rho^{\dagger}$. This alternative
formulation is useful because it suggests a generalisation of the
relative entropy to sharp theories with purification.
\begin{defn*}
Let $\rho$ and $\sigma$ be two normalised states. The \emph{relative
entropy} of $\rho$ to $\sigma$ is 
\[
S\left(\rho\parallel\sigma\right):=\left(\log_{a}\rho^{\dagger}-\log_{a}\sigma^{\dagger}\middle|\rho\right).
\]
\end{defn*}
The key property of the relative entropy is Klein's inequality (cf.\ also
\cite{Scandolo-thesis,Colleagues}.
\begin{lem}[Klein's inequality]
Let $\rho$ and $\sigma$ be two normalised states. One has $S\left(\rho\parallel\sigma\right)\geq0$
and $S\left(\rho\parallel\sigma\right)=0$ if and only if $\rho=\sigma$.
\end{lem}

\begin{proof}
The proof follows similar lines to the quantum case (see e.g.\ \cite{Nielsen-Chuang}).
Let $\rho=\sum_{i=1}^{d}p_{i}\alpha_{i}$ and $\sigma=\sum_{i=1}^{d}q_{i}\alpha'_{i}$
be diagonalisations of $\rho$ and $\sigma$. Now, let us compute
$S\left(\rho\parallel\sigma\right)$ explicitly. Assume that all the
eigenvalues of $\rho$ and $\sigma$ are non-zero, for the result
in the general case can be obtained by using the continuity of the
logarithm, and by taking limits suitably. Hence, 
\[
\left(\log\rho^{\dagger}\middle|\rho\right)=\sum_{i=1}^{d}p_{i}\log_{a}p_{i},
\]
and 
\[
\left(\log_{a}\sigma^{\dagger}\middle|\rho\right)=\sum_{i,j=1}^{d}\left(\alpha_{i}'^{\dagger}\middle|\alpha_{j}\right)p_{j}\log_{a}q_{i}=\sum_{i,j=1}^{d}T_{ij}p_{j}\log_{a}q_{i},
\]
where $T_{ij}:=\left(\alpha_{i}'^{\dagger}\middle|\alpha_{j}\right)$
are the entries of a doubly stochastic matrix (lemma~\ref{lem:doubly-stochastic}).
Then 
\begin{equation}
S\left(\rho\parallel\sigma\right)=\sum_{j=1}^{d}p_{j}\left(\log_{a}p_{j}-\sum_{i=1}^{d}T_{ij}\log_{a}q_{i}\right)\geq\sum_{j=1}^{d}p_{j}\left(\log_{a}p_{j}-\log_{a}r_{j}\right),\label{eq:KL bound}
\end{equation}
having used the concavity of the logarithm, and having set $r_{j}:=\sum_{i=1}^{d}T_{ij}q_{i}$.
The right-hand side of the last inequality is the classical relative
entropy $D\left(\mathbf{p}\parallel\mathbf{r}\right)$. Since $D\left(\mathbf{p}\parallel\mathbf{r}\right)$
is always non-negative, we obtain the bound 
\[
S\left(\rho\parallel\sigma\right)\ge D\left(\mathbf{p}\parallel\mathbf{r}\right)\geq0.
\]

The fact that $S\left(\rho\parallel\sigma\right)=0$ if $\rho=\sigma$
is obvious from the very definition of relative entropy. Let us prove
the converse. Since the classical relative entropy vanishes if and
only if $\mathbf{p}=\mathbf{r}$, the condition $S\left(\rho\parallel\sigma\right)=0$
implies $p_{j}=\sum_{i=1}^{d}T_{ij}q_{i}$, for all $i\in\left\{ 1,\ldots,d\right\} $.
Inserting this equality into eq.~\eqref{eq:KL bound} we obtain the
relation 
\[
0=\sum_{j=1}^{d}p_{j}\left[\log\left(\sum_{i=1}^{d}T_{ij}q_{i}\right)-\sum_{i=1}^{d}T_{ij}\log q_{i}\right].
\]
Since the logarithm is a strictly concave function, the equality implies
that $T$ is a permutation matrix. Hence, we have $T_{ij}=\delta_{i,\pi\left(j\right)}$,
for a suitable permutation $\pi$. Recalling the definition of $T$,
we obtain 
\[
T_{ij}=\left(\alpha_{i}'^{\dagger}\middle|\alpha_{j}\right)=\delta_{i,\pi\left(j\right)},
\]
which in turn implies $\alpha_{i}=\alpha'_{\pi\left(i\right)}$, for
all $i\in\left\{ 1,\ldots,d\right\} $ due to the state-effect duality.
In conclusion, we have obtained
\[
\rho=\sum_{i=1}^{d}p_{i}\alpha_{i}=\sum_{i=1}^{d}q_{\pi\left(i\right)}\alpha'_{\pi\left(i\right)}=\sigma.
\]
\end{proof}
Like in quantum theory, this version of Klein's inequality allows
one to prove a number of important properties. The easiest application
is the subadditivity of Shannon-von Neumann entropy.
\begin{prop}[Subadditivity]
Let $\rho_{\mathrm{AB}}$ be a bipartite state of system $\mathrm{AB}$,
and let $\rho_{\mathrm{A}}$ and $\rho_{\mathrm{B}}$ be its marginals
on system $\mathrm{A}$ and $\mathrm{B}$ respectively. Shannon-von
Neumann is \emph{subadditive}, namely 
\[
S\left(\rho_{\mathrm{AB}}\right)\leq S\left(\rho_{\mathrm{A}}\right)+S\left(\rho_{\mathrm{B}}\right).
\]
The equality holds if and only if $\rho_{\mathrm{AB}}$ is a product
state.
\end{prop}

The proof follows from the application of Klein's inequality to the
states $\rho:=\rho_{\mathrm{A}}\otimes\rho_{\mathrm{B}}$ and $\sigma:=\rho_{\mathrm{AB}}$.
The subadditivity of the entropy guarantees that the \emph{mutual
information}, defined as 
\[
I\left(\mathrm{A};\mathrm{B}\right)_{\rho_{\mathrm{AB}}}:=S\left(\rho_{\mathrm{A}}\right)+S\left(\rho_{\mathrm{B}}\right)-S\left(\rho_{\mathrm{AB}}\right)
\]
is a non-negative quantity, and vanishes if and only if $\rho_{\mathrm{AB}}$
is a product state. Therefore, the mutual information can be used
as a measure of correlations. On the other hand, the \emph{conditional
entropy} 
\[
S\left(\mathrm{A}|\mathrm{B}\right):=S\left(\rho_{\mathrm{AB}}\right)-S\left(\rho_{\mathrm{B}}\right)
\]
can be negative, thanks to Purification, because $\mathrm{AB}$ can
be in a pure state, $S\left(\rho_{\mathrm{AB}}\right)=0$, while $\rho_{\mathrm{B}}$
can be mixed, so $S\left(\rho_{\mathrm{B}}\right)>0$. Like in quantum
theory, the negativity of conditional entropy can be exploited for
novel thermodynamic protocols \cite{Negative-entropy}, as we will
see in subsection~\ref{subsec:Memories}.

Another consequence of Klein's inequality is the \emph{triangle inequality}.
\begin{prop}[Triangle inequality]
For every bipartite state $\rho_{\mathrm{AB}}$ one has 
\[
S\left(\rho_{\mathrm{AB}}\right)\geq\left|S\left(\rho_{\mathrm{A}}\right)-S\left(\rho_{\mathrm{B}}\right)\right|,
\]
where $\rho_{\mathrm{A}}$ and $\rho_{\mathrm{B}}$ are the marginals
of $\rho_{\mathrm{AB}}$ on $\mathrm{A}$ and $\mathrm{B}$ respectively.
\end{prop}

Its proof is the same as in the quantum case (see e.g.\ \cite{Nielsen-Chuang}),
and we will omit it for brevity. Combining subadditivity and the triangle
inequality, one obtains the bound 
\[
\left|S\left(\rho_{\mathrm{A}}\right)-S\left(\rho_{\mathrm{B}}\right)\right|\leq S\left(\rho_{\mathrm{AB}}\right)\leq S\left(\rho_{\mathrm{A}}\right)+S\left(\rho_{\mathrm{B}}\right),
\]
valid in all sharp theories with purification.

\subsection{The second law lemma\label{subsec:The-second-law-lemma}}

Using the mathematical properties of Shannon-von Neumann entropy we
can prove a physical one that arises naturally when considering a
system evolving jointly with its environment. Assume the system and
the surrounding environment are uncorrelated at the initial time.
Consistently with Purification, here we assume that, by suitably enlarging
the environment, the interaction can be modelled by a reversible channel
$\mathcal{U}$. We denote the initial states of the system and the
environment by $\rho_{\mathrm{S}}$ and $\rho_{\mathrm{E}}$ respectively,
so the initial state of the composite system is $\rho_{\mathrm{SE}}=\rho_{\mathrm{S}}\otimes\rho_{\mathrm{E}}$.
Primed states will denote the states after the interaction.\[
\begin{aligned} \Qcircuit @C=1em @R=.7em @!R { \prepareC{\rho_\rS}& \qw \poloFantasmaCn{\rS} & \multigate{1}{\mathcal U} &\qw \poloFantasmaCn{\rS}&\qw \\ \prepareC{\rho_\rE}& \qw \poloFantasmaCn{\rE} & \ghost{\mathcal U} & \qw \poloFantasmaCn{\rE} & \qw}\end{aligned}
\]

The result of the interaction is typically to create correlations
between the system and the environment, thus increasing the mutual
information from the initial zero value to a final non-zero value.
The creation of correlations can be equivalently phrased as an increase
of the sum of the system and environment entropies. Indeed, the positivity
of the mutual information gives the bound
\[
0\le I\left(\mathrm{S};\mathrm{E}\right)_{\rho'_{\mathrm{SE}}}=S\left(\rho'_{\mathrm{S}}\right)+S\left(\rho'_{\mathrm{E}}\right)-S\left(\rho'_{\mathrm{SE}}\right)=
\]
\[
=S\left(\rho'_{\mathrm{S}}\right)+S\left(\rho'_{\mathrm{E}}\right)-S\left(\rho_{\mathrm{SE}}\right)=S\left(\rho'_{\mathrm{S}}\right)+S\left(\rho'_{\mathrm{E}}\right)-S\left(\rho_{\mathrm{S}}\right)-S\left(\rho_{\mathrm{E}}\right),
\]
the second equality coming from the fact that reversible channels
do not change the entropy (see remark~\ref{rem:remark monotone invariant }),
so $S\left(\rho'_{\mathrm{SE}}\right)=S\left(\rho_{\mathrm{SE}}\right)$.
The resulting bound 
\begin{equation}
S\left(\rho'_{\mathrm{S}}\right)+S\left(\rho'_{\mathrm{E}}\right)\ge S\left(\rho_{\mathrm{S}}\right)+S\left(\rho_{\mathrm{E}}\right)\label{eq:almost2law}
\end{equation}
is sometimes regarded as an elementary instance of the second law
of thermodynamics \cite{Preskill}. It is important, however, not
to confuse the sum of the entropies $S\left(\rho'_{\mathrm{S}}\right)+S\left(\rho'_{\mathrm{E}}\right)$
with the total entropy $S\left(\rho'_{\mathrm{SE}}\right)$, which
remains \emph{unchanged} due to the reversibility of the global evolution.
The best reading of eq.~\eqref{eq:almost2law} is probably that a
decrease in the entropy of the system must be accompanied by an increase
of the entropy of the environment. Following Reeb and Wolf \cite{Reeb-Wolf}
we will refer to eq.~\eqref{eq:almost2law} as the \emph{second law
lemma}.

Operationally, the second law lemma is the statement that uncorrelated
systems can only become more correlated as a result of reversible
interactions. The interesting part of it is that ``correlations''
here are measured in terms of entropies: the existence of an entropic
measure of correlations is a non-trivial consequence of the axioms.

\section{Generalised Gibbs states\label{sec:Generalised-Gibbs-states}}

After analysing microcanonical thermodynamics in GPTs, and in particular
in sharp theories with purification, it is time to move beyond it,
and begin studying the role of temperature. In classical and quantum
statistical mechanics, we know that temperature, and the associated
equilibrium state, the \emph{canonical ensemble}, are tightly related
to the Hamiltonian of the system. Here we encounter the first difficulty:
defining the Hamiltonian in GPTs is a rather thorny issue \cite{Hamiltonian-GPTs}.
To circumvent the problem, instead of determining which observable
is the Hamiltonian, let us instead focus on the fact that the Hamiltonian
is some observable, and recall that in subsection~\ref{subsec:Functional-calculus-on},
we learnt a few things about observables in sharp theories with purification.
Now, suppose we are given an observable $H$, which we may think of
as the Hamiltonian of the system. However, for the following derivation,
this identification is not at all necessary, and our treatment provides
an immediate way to see generalisations of the canonical ensemble
beyond the Hamiltonian case even in quantum theory. 

In the following, since Shannon-von Neumann entropy can be seen as
the expectation value of the surprisal observable, we will make use
of it to derive the form of thermal states following Jaynes' maximum
entropy principle \cite{Jaynes1,Jaynes2} 

Now, suppose that the only information we have about the state of
system $\mathrm{A}$ is the expectation value of a certain observable
$H$ (e.g.\ the Hamiltonian). Which state should we assign to the
system? The maximum entropy principle posits that, among the states
with the given expectation value, we should choose the one that maximises
Shannon-von Neumann entropy, namely the state $\rho_{\mathrm{max}}$
such that
\begin{equation}
\rho_{\mathrm{max}}=\arg\max\left\{ S\left(\rho\right):\left\langle H\right\rangle _{\rho}=E\right\} .\label{eq:argmax}
\end{equation}
Before proceeding, let us make a brief remark about the choice of
the base $a$ for the logarithm in the definition of Shannon-von Neumann
entropy. The argument in the following can be carried out using any
base $a>1$ for the logarithm in the definition of entropy. However,
if we want the entropy to have a thermodynamic meaning, we must recover
thermodynamic predictions, such as the correct calculation of entropy
differences for the ideal gas. An argument by von Neumann \cite{vonNeumann},
later extended to GPTs \cite{Krumm-thesis,Colleagues}, shows that
one must take the natural logarithm. With this in mind, let us determine
the state that maximises Shannon-von Neumann entropy in sharp theories
with purification, subject to the constraint $\left\langle H\right\rangle =E$.
Like in quantum theory, it turns out that there is a one-parameter
family of states: \emph{Gibbs states}. They are of the form 
\begin{equation}
\gamma_{\beta}:=\frac{\mathrm{e}^{-\beta H^{\dagger}}}{\mathrm{tr}\:\mathrm{e}^{-\beta H^{\dagger}}},\label{eq:Gibbs concise}
\end{equation}
where $\beta\in\left[-\infty,+\infty\right]$, and the value of the
parameter $\beta$ is a function of $E$. The expression in eq.~\eqref{eq:Gibbs concise}
means in fact
\[
\gamma_{\beta}=\frac{1}{Z}\sum_{i=1}^{d}\mathrm{e}^{-\beta E_{i}}\varphi_{i},
\]
with $Z:=\sum_{i=1}^{d}\mathrm{e}^{-\beta E_{i}}$, where the $E_{i}$'s
are the eigenvalues of $H$, and each $\varphi_{i}$ is a pure state
such that $\left(H\middle|\varphi_{i}\right)=E_{i}$, namely the corresponding
eigenstate. How can we check that the Gibbs states are exactly the
solution to the maximisation problem in eq.~\eqref{eq:argmax}? Instead
of solving it directly, let us show that the entropy of Gibbs states
is higher than of any other state with the same expectation value
$E$ for $H$. To this end, let us calculate the expectation value
of $H$ on a Gibbs state:
\[
E\left(\beta\right):=\left\langle H\right\rangle _{\gamma_{\beta}}=-\frac{\mathrm{d}}{\mathrm{d}\beta}\ln Z.
\]
This can assume all values between $E_{\mathrm{min}}$ and $E_{\mathrm{max}}$
(the minimum and maximum eigenvalue of $H$). Now, in the maximisation
problem~\eqref{eq:argmax}, $E$ is fixed, rather than a function
of $\beta$. It is instead $\beta$ to be a function of $E$. In other
words, we have to invert the function $E\left(\beta\right)$. Now,
if $E_{\mathrm{min}}<E_{\mathrm{max}}$, namely $H$ is not fully
degenerate, the function $E\left(\beta\right)$ is invertible for
all $\beta\in\left[-\infty,+\infty\right]$ \cite[lemma 9]{Reeb-Wolf}.
If $H$ is fully degenerate, the Gibbs state is the microcanonical
state $\chi=\frac{1}{d}\sum_{i=1}^{d}\varphi_{i}$, for every $\beta$,
a case we treated in the previous sections. Therefore, it is not restrictive
to assume $E_{\mathrm{min}}<E_{\mathrm{max}}$. In this case, there
is an inverse function $\beta\left(E\right)$, and now we are ready
to prove that the Gibbs state $\gamma_{\beta\left(E\right)}$ is the
maximum entropy state with expectation value $E$. The proof is based
on an argument by Preskill \cite{Preskill}. First, note that the
entropy of a Gibbs state is
\begin{equation}
S\left(\gamma_{\beta\left(E\right)}\right)=\beta\left(E\right)E+\ln Z.\label{eq:entropy Gibbs}
\end{equation}
Then use Klein's inequality between any $\rho$ with $\left\langle H\right\rangle _{\rho}=E$,
and $\gamma_{\beta\left(E\right)}$. 
\[
0\le S\left(\rho\parallel\gamma_{\beta\left(E\right)}\right)=\left(\ln\rho^{\dagger}-\ln\gamma_{\beta\left(E\right)}^{\dagger}\middle|\rho\right)=\left(\ln\rho^{\dagger}+\beta\left(E\right)H+u\ln Z\middle|\rho\right)=
\]
\[
=-S\left(\rho\right)+\beta\left(E\right)E+\ln Z=-S\left(\rho\right)+S\left(\gamma_{\beta\left(E\right)}\right),
\]
where we have used the fact that $S\left(\rho\right)=-\left(\ln\rho^{\dagger}\middle|\rho\right)$,
that $E=\left\langle H\right\rangle _{\rho}=\left(H\middle|\rho\right)$,
and eq.~\eqref{eq:entropy Gibbs}. This yields the bound 
\[
S\left(\rho\right)\le S\left(\gamma_{\beta\left(E\right)}\right),
\]
for every $\rho$ such that $\left\langle H\right\rangle _{\rho}=E$.

Motivated by this characterisation, we regard the Gibbs state $\gamma_{\beta\left(E\right)}$
as the equilibrium state of a system with fixed expectation value
$E$ of the observable $H$. In the following we will focus on the
case where $H$ is the ``energy of the system''. In this case, we
will write the parameter $\beta$ as $\beta=\frac{1}{kT}$, where
$k$ is Boltzmann's constant, and $T$ is interpreted as the ``temperature''.
Consistently, we will regard $\gamma_{\beta}$ as the ``equilibrium
state at temperature $T$''. Since the system is finite-dimensional,
in principle we allow \emph{negative} temperatures \cite{Negative-temperatures1,Negative-temperatures2},
corresponding to a situation in which higher energy levels are favoured.
However, the interpretation of negative $\beta$'s as sensible inverse
\emph{thermodynamic} temperatures has been recently questioned \cite{No-negative-temperatures},
so in the following treatment we will always assume temperatures to
be non-negative.

\section[An operational Landauer's principle]{An operational derivation of Landauer's principle}

The entropic tools constructed from the axioms allow us to prove an
operational version of Landauer's principle, based on a recent argument
by Reeb and Wolf \cite{Reeb-Wolf}. The scenario considered here is
that of a system $\mathrm{S}$ that interacts reversibly with an environment
$\mathrm{E}$, initially in the equilibrium state at temperature $T$,
i.e.\ in a Gibbs state. Like in subsection~\ref{subsec:The-second-law-lemma},
primed states and quantities denote the states and the quantities
after the interaction.\[
\begin{aligned} \Qcircuit @C=1em @R=.7em @!R { \prepareC{\rho_\rS}& \qw \poloFantasmaCn{\rS} & \multigate{1}{\mathcal U} &\qw \poloFantasmaCn{\rS}&\qw \\ \prepareC{\gamma_\beta}& \qw \poloFantasmaCn{\rE} & \ghost{\mathcal U} & \qw \poloFantasmaCn{\rE} & \qw}\end{aligned}
\]

In this context, Landauer's principle amounts to the statement that
a decrease in the entropy of the system must be accompanied by an
increase in the expected energy of the environment. More formally,
we have the following theorem, exactly equal to the result by Reeb
and Wolf \cite{Reeb-Wolf}, but where quantum entropies are replaced
by our definition of Shannon-von Neumann entropy \cite{TowardsThermo}.
\begin{thm}
Suppose that the system and the environment are initially in the product
state $\rho_{\mathrm{SE}}=\rho_{\mathrm{S}}\otimes\gamma_{\mathrm{E},\beta}$
where $\gamma_{\mathrm{E},\beta}$ is Gibbs state at inverse temperature
$\beta$. Let $H_{\mathrm{E}}$ be the energy observable of the environment.
After a reversible interaction $\mathcal{U}$, the system and the
environment will satisfy the equality
\begin{equation}
\left\langle H_{\mathrm{E}}\right\rangle '-\left\langle H_{\mathrm{E}}\right\rangle =kT\left(S\left(\rho_{\mathrm{S}}\right)-S\left(\rho'_{\mathrm{S}}\right)+I\left(\mathrm{S};\mathrm{E}\right)_{\rho'_{\mathrm{SE}}}+S\left(\rho'_{\mathrm{E}}\parallel\gamma_{\mathrm{E},\beta}\right)\right),\label{eq:Landauer}
\end{equation}
where $\left\langle H_{\mathrm{E}}\right\rangle =\left(H_{\mathrm{E}}\middle|\gamma_{\mathrm{E},\beta}\right)$
and $\left\langle H_{\mathrm{E}}\right\rangle '=\left(H_{\mathrm{E}}\middle|\rho'_{\mathrm{E}}\right)$
are the expectation values of the environment energy at the initial
and final times respectively.
\end{thm}

The proof is identical to the quantum case \cite{Reeb-Wolf}, and
we do not report it for brevity. The key point is, again, Klein's
inequality, which implies that the last two terms in the right-hand
side of eq.~\eqref{eq:Landauer} are always non-negative, and therefore
one has the bound 
\begin{equation}
\left\langle H_{\mathrm{E}}\right\rangle '-\left\langle H_{\mathrm{E}}\right\rangle \ge kT\left(S\left(\rho_{\mathrm{S}}\right)-S\left(\rho'_{\mathrm{S}}\right)\right),\label{eq:landauerbound}
\end{equation}
stating that it is impossible to reduce the entropy of the system
without heating up the environment. Furthermore, the equality condition
in Klein's inequality implies that the lower bound~\eqref{eq:landauerbound}
is attained if and only if
\begin{enumerate}
\item the system and the environment remain uncorrelated after the interaction;
\item the environment remains in the equilibrium state.
\end{enumerate}
Note that in this case $\left\langle H_{\mathrm{E}}\right\rangle '=\left\langle H_{\mathrm{E}}\right\rangle $,
and there is in fact no dissipation. This is because eq.~\eqref{eq:Landauer}
prescribes $S\left(\rho_{\mathrm{S}}\right)=S\left(\rho'_{\mathrm{S}}\right)$,
so the entropy of the state cannot decrease, and no real erasure takes
place here.

\subsection{The role of memories and negative conditional entropy\label{subsec:Memories}}

Combining the settings of \cite{Negative-entropy} and \cite{Reeb-Wolf},
we can explore the thermodynamic meaning of \emph{negative} conditional
entropy in sharp theories with purification. We will see that non-classical
correlations, witnessed by negative conditional entropy, can be harnessed
to overcome Landauer's bound \eqref{eq:landauerbound}: we can reduce
the entropy of the system without heating up the environment. 

Suppose there is a system $\mathrm{M}$, the \emph{memory}, which
contains some information about the system $\mathrm{S}$ because it
is correlated with it. Now we can consider the case where the composite
system $\mathrm{SME}$ is in the state $\rho_{\mathrm{SM}}\otimes\gamma_{\mathrm{E},\beta}$,
with $\gamma_{\mathrm{E},\beta}$ a Gibbs state.\[
\begin{aligned} \Qcircuit @C=1em @R=.7em @!R { \multiprepareC{1}{\rho_{\mathrm{SM}}}& \qw \poloFantasmaCn{\rS} & \multigate{2}{\mathcal U} &\qw \poloFantasmaCn{\rS}&\qw \\ \pureghost{\rho_{\mathrm{SM}}}& \qw \poloFantasmaCn{\rM} & \ghost{\mathcal U} &\qw \poloFantasmaCn{\rM}&\qw \\  \prepareC{\gamma_\beta}& \qw \poloFantasmaCn{\rE} & \ghost{\mathcal U} & \qw \poloFantasmaCn{\rE} & \qw}\end{aligned}
\]We will also assume that the overall reversible evolution does not
increase the entropy of the memory, viz.\ $S\left(\rho'_{\mathrm{M}}\right)\leq S\left(\rho_{\mathrm{M}}\right)$
\cite{Negative-entropy}. It is immediate to see that the second law
lemma (subsection~\ref{subsec:The-second-law-lemma}) takes the form
\[
S\left(\rho'_{\mathrm{SM}}\right)-S\left(\rho_{\mathrm{SM}}\right)+S\left(\rho'_{\mathrm{E}}\right)-S\left(\rho_{\mathrm{E}}\right)=I\left(\mathrm{SM};\mathrm{E}\right)_{\rho'_{\mathrm{SME}}}
\]
which allows one to reformulate Landauer's principle as follows \cite{Reeb-Wolf}.
\begin{prop}
\label{prop:landauer negative}Suppose that the system, the memory,
and the environment are initially in the state $\rho_{\mathrm{SME}}=\rho_{\mathrm{SM}}\otimes\gamma_{\mathrm{E},\beta}$
with $\gamma_{\mathrm{E},\beta}$ a Gibbs state. After a reversible
interaction $\mathcal{U}$ such that $S\left(\rho'_{\mathrm{M}}\right)\leq S\left(\rho_{\mathrm{M}}\right)$,
the system, the memory, and the environment satisfy the equality 
\[
\left\langle H_{\mathrm{E}}\right\rangle '-\left\langle H_{\mathrm{E}}\right\rangle =kT\left(S\left(\mathrm{S}|\mathrm{M}\right)_{\rho_{\mathrm{SM}}}-S\left(\mathrm{S}|\mathrm{M}\right)_{\rho'_{\mathrm{SM}}}+\right.
\]
\[
\left.+S\left(\rho_{\mathrm{M}}\right)-S\left(\rho'_{\mathrm{M}}\right)+I\left(\mathrm{SM};\mathrm{E}\right)_{\rho'_{\mathrm{SME}}}+S\left(\rho'_{\mathrm{E}}\parallel\gamma_{\mathrm{E},\beta}\right)\right).
\]
\end{prop}

Given that $S\left(\rho'_{\mathrm{M}}\right)\leq S\left(\rho_{\mathrm{M}}\right)$,
this means again that 
\begin{equation}
\left\langle H_{\mathrm{E}}\right\rangle '-\left\langle H_{\mathrm{E}}\right\rangle \geq kT\left(S\left(\mathrm{S}|\mathrm{M}\right)_{\rho_{\mathrm{SM}}}-S\left(\mathrm{S}|\mathrm{M}\right)_{\rho'_{\mathrm{SM}}}\right)\label{eq:landauermemory}
\end{equation}
Comparing this with eq.~\eqref{eq:landauerbound}, we notice the
presence of conditional entropies which may be \emph{negative}. In
the particular case of $\mathrm{SM}$ in a \emph{pure} state, $S\left(\mathrm{S}|\mathrm{M}\right)_{\rho_{\mathrm{SM}}}=-S\left(\rho_{\mathrm{S}}\right)$
due to Schmidt decomposition (theorem~\ref{thm:schmidt}). Then eq.~\eqref{eq:landauermemory}
reads 
\begin{equation}
\left\langle H_{\mathrm{E}}\right\rangle '-\left\langle H_{\mathrm{E}}\right\rangle \geq kT\left(S\left(\rho'_{\mathrm{S}}\right)-S\left(\rho_{\mathrm{S}}\right)\right),\label{eq:landauermemorypure}
\end{equation}
where the roles of the initial and final states are swapped with respect
to eq.~\eqref{eq:landauerbound}. Now we will show that this swapping
allows us to perform the erasure of a mixed state of $\mathrm{S}$
towards a fixed pure state of $\mathrm{S}$ at \emph{no} thermodynamic
cost, thus overcoming Landauer's bound \eqref{eq:landauerbound}.

Suppose $\mathrm{SM}$ is initially in a \emph{pure entangled} state
$\Psi\in\mathsf{PurSt}_{1}\left(\mathrm{SM}\right)$; in this case
$\rho_{\mathrm{S}}$ is mixed \cite{Chiribella-Scandolo-entanglement},
and we have $S\left(\rho_{\mathrm{M}}\right)=S\left(\rho_{\mathrm{S}}\right)>0$.
Suppose we want to erase $\rho_{\mathrm{S}}$ to a fixed pure state
$\psi_{0}$ of $\mathrm{S}$. Now, let us consider the joint reversible
evolution of $\mathrm{SME}$ to be $\mathcal{U}_{\mathrm{SM}}\otimes\mathcal{I}_{\mathrm{E}}$,
where $\mathcal{U}_{\mathrm{SM}}$ is the reversible channel mapping
$\Psi$ to $\psi_{0}\otimes\varphi_{0}$, where $\varphi_{0}$ is
some pure state of the memory $\mathrm{M}$ ($\mathcal{U}_{\mathrm{SM}}$
exists thanks to transitivity). Clearly this reversible evolution
respects the hypotheses of proposition~\ref{prop:landauer negative}
because $0=S\left(\rho'_{\mathrm{M}}\right)<S\left(\rho_{\mathrm{M}}\right)$,
so it performs the erasure of $\rho_{\mathrm{S}}$ to $\psi_{0}$.
Let us evaluate its thermodynamic cost $\left\langle H_{\mathrm{E}}\right\rangle '-\left\langle H_{\mathrm{E}}\right\rangle $.
Since initially the environment is uncorrelated in the state $\gamma_{\mathrm{E},\beta}$,
and the evolution is $\mathcal{U}_{\mathrm{SM}}\otimes\mathcal{I}_{\mathrm{E}}$,
we have $\rho'_{\mathrm{E}}=\gamma_{\mathrm{E},\beta}$, so the erasure
occurs at zero thermodynamic cost. Note that eq.~\eqref{eq:landauermemorypure}
is satisfied, indeed its left-hand side vanishes, while its right-hand
side is negative and equal to $-kTS\left(\rho_{\mathrm{S}}\right)$.

Again, pure-state entanglement, captured by the negativity of the
conditional entropy, is a resource even in sharp theories with purification,
in that it allows us to overcome Landauer's principle, and to perform
erasure at no thermodynamic cost (cf.\ the quantum case in \cite{Negative-entropy}). 

\section[Majorisation and unrestricted reversibility]{Sufficiency of majorisation and unrestricted reversibility\label{sec:Sufficiency-of-majorisation}}

After the above digression about entropies in sharp theories with
purification, and their role in determining the equilibrium state
when the expectation value of an observable is known, it is time to
go back to the three resource theories of section~\ref{sec:Three-resource-theories},
and to deal with the last question we left unanswered: is majorisation
sufficient for the convertibility under RaRe channels? If not, when
is it so?

The answer is \emph{negative} \cite{Purity}. First, let us try to
give an intuitive reason for why it is so, which will guide us in
the search for an actual counterexample. Note that majorisation is
merely a condition on the spectra of states, and carries no information
about the dynamics allowed by the theory. Instead, RaRe convertibility
is all about the dynamics: if a theory does not have enough reversible
dynamics, a state could majorise another without a RaRe channel transforming
the former into the latter. So, a priori majorisation and RaRe convertibility
might not related, and it is instead surprising that in quantum theory
majorisation is sufficient to characterise the RaRe preorder. If we
want to look for a counterexample, we need to focus on theories where
there is a ``restriction'' on the allowed reversible channels. For
this reason, let us focus on doubled quantum theory, presented in
section~\ref{sec:Example:-doubled-quantum}, where reversible channels
are constrained by the parity superselection rule.

Consider the following states of a doubled qubit (the first index
denotes the sector $\mathcal{H}_{0}$ or $\mathcal{H}_{1}$): 
\begin{equation}
\rho=\frac{1}{2}\left(\ket{0,0}\bra{0,0}+\ket{0,1}\bra{0,1}\right)\label{eq:rho}
\end{equation}
and 
\begin{equation}
\sigma=\frac{1}{2}\ket{0,0}\bra{0,0}\oplus\frac{1}{2}\ket{1,0}\bra{1,0},\label{eq:sigma}
\end{equation}
where $\left\{ \ket{0,0},\ket{0,1}\right\} $ is an orthonormal basis
for $\mathcal{H}_{0}$ and $\left\{ \ket{1,0},\ket{1,1}\right\} $
an orthonormal basis for $\mathcal{H}_{1}$. The key point here is
that the state $\rho$ is fully contained in one sector (the even
parity sector), while the state $\sigma$ is a mixture of two states
in two different sectors. 

The two states have the same spectrum, and therefore they are equivalent
in terms of majorisation of their spectra. However, there is no RaRe
channel transforming one state into the other. To see this, we use
the following lemmas.
\begin{lem}
\label{lem:equivalent states RaRe}If any two states $\rho$ and $\sigma$
are equivalent under RaRe channels, there exists a reversible channel
$\mathcal{U}$ such that $\sigma=\mathcal{U}\rho$.
\end{lem}

The proof can be found in \cite{Muller3D}.
\begin{lem}
No unitary matrices $U$ in doubled quantum theory are such that $\sigma=U\rho U^{\dagger}$,
where $\rho$ and $\sigma$ the states of eqs.\ \eqref{eq:rho} and
\eqref{eq:sigma}.
\end{lem}

\begin{proof}
The proof is by contradiction. Suppose that one has $\sigma=U\rho U^{\dagger}$,
for some unitary matrix $U$. Then, define the vectors $\ket{\varphi_{0}}:=U\ket{0,0}$
and $\ket{\varphi_{1}}:=U\ket{0,1}$. With this definition, we have
\[
U\rho U^{\dagger}=\frac{1}{2}\left(\ket{\varphi_{0}}\bra{\varphi_{0}}+\ket{\varphi_{1}}\bra{\varphi_{1}}\right).
\]
Now, $U\rho U^{\dagger}$ must be an allowed state in doubled quantum
theory. This means that there are only two possibilities: either $\ket{\varphi_{0}}$
and $\ket{\varphi_{1}}$ belong to the same sector, or they do not.
But $\sigma$ is a mixture of states in both sectors. Hence, $\ket{\varphi_{0}}$
and $\ket{\varphi_{1}}$ must belong to different sectors, if the
relation $U\rho U^{\dagger}=\sigma$ is to hold. At this point, there
are only two possibilities: either
\[
\left\{ \begin{array}{l}
U\ket{0,0}=\ket{0,0}\\
U\ket{0,1}=\ket{1,0}
\end{array}\right.
\]
or 
\[
\left\{ \begin{array}{l}
U\ket{0,0}=\ket{1,0}\\
U\ket{0,1}=\ket{0,0}
\end{array}\right..
\]
However, neither of these conditions can be satisfied by a unitary
matrix in doubled quantum theory: every unitary matrix satisfying
either condition would map the valid state $\ket{0,+}=\frac{1}{\sqrt{2}}\left(\ket{0,0}+\ket{0,1}\right)$
into the \emph{invalid} state $\frac{1}{\sqrt{2}}\left(\ket{0,0}+\ket{1,0}\right)$,
which is forbidden by the parity superselection rule.
\end{proof}
Since unitary channels are the only reversible channels in doubled
quantum theory, we conclude that no RaRe channel can convert $\rho$
into $\sigma$ (and vice versa), despite being equivalent in terms
of majorisation.\footnote{Note that a very similar counterexample can be constructed in extended
classical theory using the coherent composition of classical bits.
However, in this theory, at the level of single classical systems,
majorisation is still sufficient for RaRe convertibility.} Summarising: in general, in sharp theories with purification, majorisation
is \emph{not} sufficient for the convertibility via RaRe channels.
Clearly, this means that in general, in sharp theories with purification,
the three resource theories give rise to inequivalent preorders on
resources.

If instead majorisation is sufficient for RaRe convertibility, the
three preorders become equivalent: this is the missing step to close
the chain of implications~\eqref{eq:chain implications}. In fact,
the sufficiency of majorisation is equivalent to the equivalence of
the three resource preorders.

\subsection{Unrestricted reversibility\label{subsec:Unrestricted-reversibility}}

The condition for the equivalence of the three resource theories can
be expressed in three, mutually equivalent ways, corresponding to
three axioms independently introduced by different authors.
\begin{ax}[Permutability \cite{Hardy-informational-2,hardy2013}]
Every permutation of every pure maximal set can be implemented by
a reversible channel.
\end{ax}

\begin{ax}[Strong Symmetry \cite{Muller-self-duality,Barnum-interference}]
For every two pure maximal sets, there exists a reversible channel
that converts the states in one set into the states in the other.
\end{ax}

\begin{ax}[Reversible Controllability \cite{Control-reversible}]
For every pair of systems $\mathrm{A}$ and $\mathrm{B}$, every
pure maximal set $\left\{ \alpha_{i}\right\} _{i=1}^{d}$ of system
$\mathrm{A}$ and every set of reversible channels $\left\{ \mathcal{U}_{i}\right\} _{i=1}^{d}$
on system $\mathrm{B}$, not necessarily distinct, there exists a
reversible channel $\mathcal{U}$ on the composite system $\mathrm{AB}$
such that\[
\begin{aligned} \Qcircuit @C=1em @R=.7em @!R { & \prepareC{\alpha_i} & \qw \poloFantasmaCn{\rA} & \multigate{1}{\mathcal U} & \qw \poloFantasmaCn{\rA} & \qw \\ & & \qw \poloFantasmaCn{\rB} & \ghost{\mathcal U} & \qw \poloFantasmaCn{\rB} & \qw } \end{aligned} = \begin{aligned} \Qcircuit @C=1em @R=.7em @!R { & \prepareC{\alpha_i} & \qw \poloFantasmaCn{\rA} & \qw &\qw &\qw \\ & & \qw \poloFantasmaCn{\rB} & \gate{\mathcal U_i} & \qw \poloFantasmaCn{\rB} &\qw} \end{aligned}
\]for every $i\in\left\{ 1,\dots,d\right\} $.
\end{ax}

This last axiom states the possibility of implementing control-reversible
transformations in a globally reversible fashion.

Permutability, Strong Symmetry, and Reversible Controllability are
logically distinct requirements. For example, Strong Symmetry implies
Permutability, but the converse is not true in general, as shown in
the following example.
\begin{example}
Consider the square bit \cite{Barrett}. Here the state space is a
square, and the pure states are its vertices. The group of reversible
channels is the symmetry group of the square, which is the dihedral
group $D_{4}$. Every pair of vertices is a set of perfectly distinguishable
pure states. Fig.~\ref{fig:square} shows the situation for the pure
states 
\[
\alpha_{1}=\left(\begin{array}{c}
-1\\
1\\
1
\end{array}\right)\qquad\alpha_{2}=\left(\begin{array}{c}
-1\\
-1\\
1
\end{array}\right)\qquad\alpha_{3}=\left(\begin{array}{c}
1\\
-1\\
1
\end{array}\right),
\]
where the third component gives the normalisation.
\begin{figure}
\begin{centering}
\includegraphics{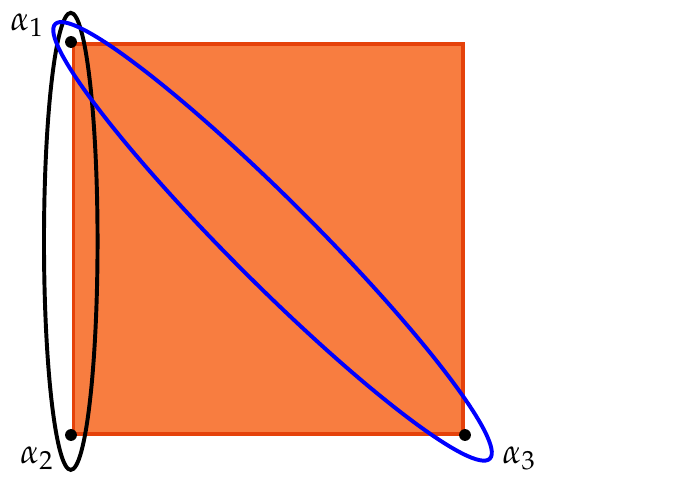}
\par\end{centering}
\caption{\label{fig:square}Normalised states of the square bit. The two sets
$\left\{ \alpha_{1},\alpha_{2}\right\} $ (circled in black) and $\left\{ \alpha_{1},\alpha_{3}\right\} $
(circled in blue) consist of perfectly distinguishable pure states.
Permutability holds, because every permutation of every pair of perfectly
distinguishable pure states can be implemented by a reversible channel,
corresponding to a symmetry of the square. However, no reversible
channel can transform $\alpha_{2}$ into $\alpha_{3}$ while leaving
$\alpha_{1}$ unchanged. Hence, Strong Symmetry cannot hold for the
square bit.}
\end{figure}
 The pure observation-test $\left\{ a_{1},a_{2}\right\} $, where
\[
a_{1}=\frac{1}{2}\left(\begin{array}{ccc}
0 & 1 & 1\end{array}\right)\qquad a_{2}=\frac{1}{2}\left(\begin{array}{ccc}
0 & -1 & 1\end{array}\right),
\]
is the perfectly distinguishing test for the two sets $\left\{ \alpha_{1},\alpha_{2}\right\} $
and $\left\{ \alpha_{1},\alpha_{3}\right\} $.

Now, since every set of perfectly distinguishable pure states has
two elements, the only non-trivial permutation of the elements of
such a set is the transposition. This permutation can be implemented
by considering the reflection across the axis of the segment connecting
the two points. Hence the square bit satisfies Permutability. On the
other hand, the square bit does \emph{not} satisfy Strong Symmetry.
Consider the two maximal sets $\left\{ \alpha_{1},\alpha_{2}\right\} $
and $\left\{ \alpha_{1},\alpha_{3}\right\} $. There are no reversible
channels mapping the former to the latter because no symmetries of
the square map a side to a diagonal.
\end{example}

Although different in general, Permutability, Strong Symmetry, and
Reversible Controllability become equivalent in sharp theories with
purification.
\begin{prop}
\label{prop:permstrong}In every sharp theory with purification, Permutability,
Strong Symmetry, and Reversible Controllability are equivalent requirements.
\end{prop}

\begin{proof}
The implication ``Strong Symmetry $\Rightarrow$ Permutability''
follows immediately from the definitions. The implication ``Strong
Symmetry $\Rightarrow$ Reversible Controllability'' was proved by
Lee and Selby \cite{Control-reversible} using Causality, Purification,
and the property that the product of two pure states is pure, which
is guaranteed by Purity Preservation. Hence, we only need to prove
the implications ``Permutability $\Rightarrow$ Strong Symmetry''
and ``Reversible Controllability $\Rightarrow$ Strong Symmetry''. 

Let us prove that Permutability implies Strong Symmetry. The first
part of the proof is similar to \cite[theorem 30]{Hardy-informational-2}.
Consider two maximal sets of perfectly distinguishable pure states
$\left\{ \varphi_{i}\right\} _{i=1}^{d}$ and $\left\{ \psi_{i}\right\} _{i=1}^{d}$.
Assuming Permutability, we will show that there exists a reversible
channel $\mathcal{U}$ such that $\psi_{i}=\mathcal{U}\varphi_{i}$,
for all $i=1,\ldots,d$. First of all, note that the states $\left\{ \varphi_{i}\otimes\psi_{j}\right\} $
are pure (by Purity Preservation) and perfectly distinguishable. Then
Permutability implies there exists a reversible channel $\mathcal{U}$
such that for all $i=1,\ldots,d$ \cite{hardy2013}\[
\begin{aligned} \Qcircuit @C=1em @R=.7em @!R { & \prepareC{\varphi_i} & \qw \poloFantasmaCn{\rA} & \multigate{1}{\mathcal U} & \qw \poloFantasmaCn{\rA} & \qw \\ & \prepareC{\psi_1} & \qw \poloFantasmaCn{\rA} & \ghost{\mathcal U} & \qw \poloFantasmaCn{\rA} & \qw } \end{aligned} ~= \!\!\!\! \begin{aligned} \Qcircuit @C=1em @R=.7em @!R { & \prepareC{\varphi_1} & \qw \poloFantasmaCn{\rA} &\qw \\ & \prepareC{\psi_i} & \qw \poloFantasmaCn{\rA} &\qw } \end{aligned}~.
\]Applying the pure effect $\varphi_{1}^{\dagger}$ to both sides of
the equation we obtain\begin{equation}\label{eq:connectP} \begin{aligned} \Qcircuit @C=1em @R=.7em @!R { & \prepareC{\varphi_i} & \qw \poloFantasmaCn{\rA} & \gate{\mathcal P} & \qw \poloFantasmaCn{\rA} & \qw } \end{aligned} ~=\!\!\!\! \begin{aligned} \Qcircuit @C=1em @R=.7em @!R { & \prepareC{\psi_i} & \qw \poloFantasmaCn{\rA} &\qw } \end{aligned}~, \end{equation}with\[ \begin{aligned} \Qcircuit @C=1em @R=.7em @!R { &\qw \poloFantasmaCn{\rA} & \gate{\mathcal P} & \qw \poloFantasmaCn{\rA} &\qw } \end{aligned} ~:=\!\!\!\! \begin{aligned} \Qcircuit @C=1em @R=.7em @!R { & & \qw \poloFantasmaCn{\rA} & \multigate{1}{\mathcal U} & \qw \poloFantasmaCn{\rA} & \measureD{\varphi_1^\dag} \\ & \prepareC{\psi_1} & \qw \poloFantasmaCn{\rA} & \ghost{\mathcal U} & \qw \poloFantasmaCn{\rA} & \qw } \end{aligned}~. \]By
construction, $\mathcal{P}$ is pure (by Purity Preservation, and
because reversible channels are pure), and occurs with probability
1 on all the states $\left\{ \varphi_{i}\right\} _{i=1}^{d}$. Moreover,
the diagonalisation $\chi=\frac{1}{d}\sum_{i=1}^{d}\varphi_{i}$ implies
that $\mathcal{P}$ occurs with probability 1 on every state because
$\left(u\middle|\mathcal{P}\middle|\chi\right)=1$, and $\chi$ is
internal. Since $\mathcal{P}$ is a pure deterministic transformation
on $\mathrm{A}$, it must be reversible \cite{Chiribella-purification}.
Hence, eq.~\eqref{eq:connectP} proves that the states $\left\{ \varphi_{i}\right\} _{i=1}^{d}$
can be reversibly transformed into the states $\left\{ \psi_{i}\right\} _{i=1}^{d}$.
In short, Permutability implies Strong Symmetry. 

Let us prove now that Reversible Controllability implies Strong Symmetry.
Let $\left\{ \varphi_{i}\right\} _{i=1}^{d}$ and $\left\{ \psi_{i}\right\} _{i=1}^{d}$
be two pure maximal sets of a generic system $\mathrm{A}$. Since
reversible channels act transitively on pure states, for every $i\in\left\lbrace 1,\ldots,d\right\rbrace $,
one can find a reversible channel $\mathcal{U}_{i}$ that maps $\psi_{1}$
into $\psi_{i}$, in formula $\mathcal{U}_{i}\psi_{1}=\psi_{i}.$
Moreover, Reversible Controllability implies that we can find a reversible
channel $\mathcal{U}$ such that\begin{equation}\label{eq:2cont} \begin{aligned} \Qcircuit @C=1em @R=.7em @!R { & \prepareC{\varphi_i} & \qw \poloFantasmaCn{\rA} & \multigate{1}{\mathcal U} & \qw \poloFantasmaCn{\rA} & \qw \\ & & \qw \poloFantasmaCn{\rA} & \ghost{\mathcal U} & \qw \poloFantasmaCn{\rA} & \qw } \end{aligned} ~=\!\!\!\! \begin{aligned} \Qcircuit @C=1em @R=.7em @!R { & \prepareC{\varphi_i} & \qw \poloFantasmaCn{\rA} & \qw & \qw &\qw \\ & & \qw \poloFantasmaCn{\rA} & \gate{\mathcal U_i} & \qw \poloFantasmaCn{\rA} & \qw } \end{aligned} \end{equation}for
every $i\in\left\lbrace 1,\ldots,d\right\rbrace $. Likewise, for
every $i\in\left\lbrace 1,\ldots,d\right\rbrace $, one can always
find a reversible channel $\mathcal{V}_{i}$ that transforms $\varphi_{i}$
into $\varphi_{1}$, in formula $\mathcal{V}_{i}\varphi_{i}=\varphi_{1}.$
And again, one can find a reversible channel $\mathcal{V}$ such that
\begin{equation}\label{eq:4cont} \begin{aligned} \Qcircuit @C=1em @R=.7em @!R { & \qw \poloFantasmaCn{\rA} & \multigate{1}{\mathcal V} & \qw \poloFantasmaCn{\rA} & \qw \\ \prepareC{\psi_i} & \qw \poloFantasmaCn{\rA} & \ghost{\mathcal V} & \qw \poloFantasmaCn{\rA} & \qw } \end{aligned}~ = \!\!\!\!\begin{aligned} \Qcircuit @C=1em @R=.7em @!R { & & \qw \poloFantasmaCn{\rA} & \gate{\mathcal V_i} & \qw \poloFantasmaCn{\rA} &\qw \\ & \prepareC{\psi_i} & \qw \poloFantasmaCn{\rA} & \qw & \qw &\qw } \end{aligned} \end{equation}
for every $i\in\left\lbrace 1,\ldots,d\right\rbrace $. Combining
eqs.~\eqref{eq:2cont} and \eqref{eq:4cont} with the definition
of $\mathcal{U}_{i}$ and $\mathcal{V}_{i}$, we obtain\[ \begin{aligned} \Qcircuit @C=1em @R=.7em @!R { & \prepareC{\varphi_i} & \qw \poloFantasmaCn{\rA} & \multigate{1}{\mathcal U} & \qw \poloFantasmaCn{\rA} & \multigate{1}{\mathcal V} & \qw \poloFantasmaCn{\rA} &\qw \\ & \prepareC{\psi_1} & \qw \poloFantasmaCn{\rA} & \ghost{\mathcal U} & \qw \poloFantasmaCn{\rA} & \ghost{\mathcal V } &\qw \poloFantasmaCn{\rA} &\qw } \end{aligned}~ = \!\!\!\!\begin{aligned} \Qcircuit @C=1em @R=.7em @!R { & \prepareC{\varphi_1} & \qw \poloFantasmaCn{\rA} & \qw \\ & \prepareC{\psi_i} & \qw \poloFantasmaCn{\rA} & \qw } \end{aligned} \]for
every $i$. Hence, one has\begin{equation}\label{eq:connectP2} \begin{aligned} \Qcircuit @C=1em @R=.7em @!R { & \prepareC{\varphi_i} & \qw \poloFantasmaCn{\rA} & \gate{\mathcal P} & \qw \poloFantasmaCn{\rA} & \qw } \end{aligned} ~=\!\!\!\! \begin{aligned} \Qcircuit @C=1em @R=.7em @!R { & \prepareC{\psi_i} & \qw \poloFantasmaCn{\rA} &\qw } \end{aligned}~, \end{equation}with\[ \begin{aligned} \Qcircuit @C=1em @R=.7em @!R { &\qw \poloFantasmaCn{\rA} & \gate{\mathcal P} & \qw \poloFantasmaCn{\rA} &\qw } \end{aligned} ~:=\!\!\!\! \begin{aligned} \Qcircuit @C=1em @R=.7em @!R { & & \qw \poloFantasmaCn{\rA} & \multigate{1}{\mathcal U} & \qw \poloFantasmaCn{\rA} & \multigate{1}{\mathcal V} & \qw \poloFantasmaCn{\rA} & \measureD{\varphi_1^\dag} \\ & \prepareC{\psi_1} & \qw \poloFantasmaCn{\rA} & \ghost{\mathcal U} & \qw \poloFantasmaCn{\rA} & \ghost{\mathcal V} & \qw \poloFantasmaCn{\rA} & \qw } \end{aligned}~. \]
By the same argument used in the first part of the proof, we conclude
that $\mathcal{P}$ is a reversible channel. Hence, eq.~\eqref{eq:connectP2}
implies that the set $\left\{ \varphi_{i}\right\} _{i=1}^{d}$ can
be reversibly converted into the set $\left\{ \psi_{i}\right\} _{i=1}^{d}$.
In short, Reversible Controllability implies Strong Symmetry.
\end{proof}
Since Permutability, Strong Symmetry, and Reversible Controllability
are equivalent in the present context, we conflate them into a single
notion.
\begin{defn}
A sharp theory with purification has \emph{unrestricted reversibility}
if the theory satisfies Permutability, or Strong Symmetry, or Reversible
Controllability.
\end{defn}

The fact that three desirable properties of GPTs become equivalent
under our axioms gives further evidence that the axioms of sharp theories
with purification capture an important structure of physical theories.

Examples of sharp theories with purification and unrestricted reversibility
are quantum theory on real and complex Hilbert space.

\subsection{When the three resource theories are equivalent}

Now we characterise exactly when the RaRe, Noisy, and Unital Resource
theories are equivalent in terms of state convertibility. Owing to
the inclusions $\mathsf{RaRe}\subseteq\mathsf{Noisy}\subseteq\mathsf{Unital}$,
a sufficient condition for the equivalence is that the convertibility
under unital channels implies the convertibility under RaRe channels,
or in other words, that majorisation is sufficient for RaRe convertibility.
The characterisation is as follows.
\begin{thm}
\label{thm:unitalrare}In every sharp theory with purification, the
following statements are equivalent:
\begin{enumerate}
\item the RaRe, noisy, and unital Resource theories are equivalent in terms
of state convertibility;\label{enu:equivalence}
\item the theory has unrestricted reversibility.
\end{enumerate}
\end{thm}

\begin{proof}
To prove the implication $2\Rightarrow1$, it is enough to show that
unrestricted reversibility implies that majorisation is sufficient
for the RaRe preorder \cite{QPL15}. Consider two states $\rho$ and
$\sigma$, diagonalised as $\rho=\sum_{i=1}^{d}p_{i}\alpha_{i}$ and
$\sigma=\sum_{i=1}^{d}q_{i}\alpha'_{i}$. Suppose $\mathbf{p}\succeq\mathbf{q}$,
then $\mathbf{q}=D\mathbf{p}$ for some doubly stochastic matrix $D$.
By Birkhoff\textquoteright s theorem \cite{Birkhoff,Olkin} $D$ can
be written as a convex combination of permutation matrices $D=\sum_{k}\lambda_{k}\Pi_{k}$,
where the $\Pi_{k}$'s are permutation matrices, and $\left\{ \lambda_{k}\right\} $
is a probability distribution. Therefore $\mathbf{q}=\sum_{k}\lambda_{k}\Pi_{k}\mathbf{p}$;
specifically, this means that $q_{i}=\sum_{k}\lambda_{k}\sum_{j=1}^{d}\left(\Pi_{k}\right)_{ij}p_{j}$
for every $i\in\left\{ 1,\ldots,d\right\} $. Therefore, we have 
\begin{equation}
\sigma=\sum_{i=1}^{d}q_{i}\alpha'_{i}=\sum_{i=1}^{d}\sum_{k}\lambda_{k}\sum_{j=1}^{d}\left(\Pi_{k}\right)_{ij}p_{j}\alpha'_{i}=\sum_{k}\lambda_{k}\sum_{j=1}^{d}p_{j}\sum_{i=1}^{d}\left(\Pi_{k}\right)_{ij}\alpha'_{i}\label{eq:majorisation-RaRe}
\end{equation}
Now, $\sum_{i=1}^{d}\left(\Pi_{k}\right)_{ij}\alpha'_{i}$ is a pure
state, given by $\alpha'_{\pi_{k}\left(j\right)}$, for a suitable
permutation $\pi_{k}\in S_{d}$, the symmetric group with $d$ elements.
By unrestricted reversibility, the permutation $\pi_{k}$ is implemented
by a reversible channel $\mathcal{V}_{k}$. Moreover, by unrestricted
reversibility there exists a reversible channel $\mathcal{U}$ such
that $\mathcal{U}\alpha_{j}=\alpha'_{j}$ for every $j\in\left\{ 1,\ldots,d\right\} $.
Defining $\mathcal{U}_{k}=\mathcal{V}_{k}\mathcal{U}$, we then have
\[
\mathcal{U}_{k}\alpha_{j}=\mathcal{V}_{k}\alpha'_{j}=\alpha'_{\pi_{k}\left(j\right)}=\sum_{i=1}^{d}\left(\Pi_{k}\right)_{ij}\alpha'_{i},
\]
which, combined with eq.~\eqref{eq:majorisation-RaRe}, yields
\[
\sigma=\sum_{k}\lambda_{k}\sum_{j=1}^{d}p_{j}\mathcal{U}_{k}\alpha_{j}=\sum_{k}\lambda_{k}\mathcal{U}_{k}\rho.
\]
Hence, $\rho\succeq_{\mathrm{RaRe}}\sigma$.

To prove the implication $1\Rightarrow2$, we show that condition~\ref{enu:equivalence}
implies the validity of Strong Symmetry. Let $\left\{ \alpha_{i}\right\} _{i=1}^{d}$
and $\left\{ \alpha'_{i}\right\} _{i=1}^{d}$ be two pure maximal
sets, and let $\left\{ p_{i}\right\} _{i=1}^{d}$ be a probability
distribution, with $p_{1}>p_{2}>\ldots>p_{d}>0$. Consider the two
states $\rho$ and $\sigma$ defined by $\rho=\sum_{i=1}^{d}p_{i}\alpha_{i}$,
and $\sigma=\sum_{i=1}^{d}p_{i}\alpha'_{i}$. Since the two states
$\rho$ and $\sigma$ have the same eigenvalues, the majorisation
criterion guarantees that $\rho$ can be converted into $\sigma$,
and $\sigma$ into $\rho$, by a unital channel (theorem~\ref{thm:unital-majorisation}).
Now, our hypothesis is that convertibility under unital channels implies
convertibility under RaRe channels. The mutual convertibility of $\rho$
and $\sigma$ under RaRe channels implies that there exists a reversible
channel $\mathcal{U}$ such that $\sigma=\mathcal{U}\rho$ (lemma~\ref{lem:equivalent states RaRe}).
Applying the effect $\alpha_{1}'^{\dagger}$ to both sides of the
equality $\sigma=\mathcal{U}\rho$, we obtain 
\[
p_{1}=\left(\alpha_{1}'^{\dagger}\middle|\sigma\right)=\sum_{j=1}^{d}p_{j}\left(\alpha_{1}'^{\dagger}\middle|\mathcal{U}\middle|\alpha_{j}\right)=\sum_{j=1}^{d}D_{1j}p_{j}\le p_{1},
\]
having used the fact that $D_{ij}:=\left(\alpha_{i}'^{\dagger}\middle|\mathcal{U}\middle|\alpha_{j}\right)$
are the entries of a doubly stochastic matrix (lemma~\ref{lem:channelmatrix}).
The above condition is satisfied only if $\left(\alpha_{1}'^{\dagger}\middle|\mathcal{U}\middle|\alpha_{1}\right)=1$.
By the state-effect duality (theorem~\ref{thm:duality}), this condition
is equivalent to the condition 
\begin{equation}
\mathcal{U}\alpha_{1}=\alpha'_{1}.\label{eq:ualpha}
\end{equation}
Now, decompose the states $\rho$ and $\sigma$ as $\rho=p_{1}\alpha_{1}+\left(1-p_{1}\right)\rho_{1}$
and $\sigma=p_{1}\alpha'_{1}+\left(1-p_{1}\right)\sigma_{1}$, where
\[
\rho_{1}:=\frac{\sum_{i=2}^{d}p_{i}\alpha_{i}}{\sum_{i=2}^{d}p_{i}}
\]
and 
\[
\sigma_{1}:=\frac{\sum_{i=2}^{d}p_{i}\alpha'_{i}}{\sum_{i=2}^{d}p_{i}}.
\]
Combining eq.~\eqref{eq:ualpha} with the equality $\mathcal{U}\rho=\sigma$,
we obtain the condition $\mathcal{U}\rho_{1}=\sigma_{1}$. Applying
to $\rho_{1}$ and $\sigma_{1}$ the same argument we used for $\rho$
and $\sigma$, we obtain the equality $\mathcal{U}\alpha_{2}=\alpha'_{2}$.
Iterating the procedure $d-1$ times, we finally obtain the equality
$\mathcal{U}\alpha_{i}=\alpha'_{i}$ for every $i\in\left\lbrace 1,\dots,d\right\rbrace $.
Hence, every two pure maximal sets are connected by a reversible channel.
\end{proof}
Theorem~\ref{thm:unitalrare} gives necessary and sufficient conditions
for the equivalence of the three resource theories of microcanonical
thermodynamics. In addition, it provides a thermodynamic motivation
for the condition of unrestricted reversibility: the equivalence of
the three resource theories of purity. Again, thermodynamics constrains
the underlying structure of a physical theory. We see that from this
argument we can rule out doubled quantum theory and extended classical
theory: in those these theories, majorisation is \emph{not} sufficient,
and they do \emph{not} satisfy unrestricted reversibility.

The results of this section can be summed up in the following theorem:
\begin{thm}
In every sharp theory with purification and unrestricted reversibility,
the following are equivalent
\begin{enumerate}
\item $\rho\succeq_{\mathrm{RaRe}}\sigma$
\item $\rho\succeq_{\mathrm{noisy}}\sigma$
\item $\rho\succeq_{\mathrm{unital}}\sigma$
\item $\mathbf{p}\succeq\mathbf{q}$
\end{enumerate}
for arbitrary normalised states $\rho$ and $\sigma$, where $\mathbf{p}$
and $\mathbf{q}$ are the spectra of $\rho$ and $\sigma$, respectively.
\end{thm}

\begin{proof}
The implications $1\Rightarrow2$ and $2\Rightarrow3$ follow from
the inclusions~\eqref{eq:inclusions sharp}. The implication $3\Rightarrow4$
follows from theorem~\ref{thm:unital-majorisation}. The implication
$4\Rightarrow1$ follows from theorem~\ref{thm:unitalrare}.
\end{proof}
Theorem~\ref{thm:unitalrare} tells us that the RaRe, noisy, and
unital resource theories are all equivalent in terms of state convertibility.
It is important to stress that the equivalence holds despite the fact
that the three sets of operations are generally different. Since the
preorders $\succeq_{\mathrm{RaRe}}$, $\succeq_{\mathsf{\mathrm{noisy}}}$,
and $\succeq_{\mathrm{unital}}$ coincide, we can say that the RaRe,
noisy, and unital resource theories define the same notion of resource,
which one may rightfully call ``\emph{purity}''. Accordingly, we
will talk about ``the resource theory of purity'' without specifying
the set of free operations. 

An important consequence of the equivalence is that the RaRe, noisy,
and unital resource theories have the same quantitative measures of
resourcefulness.
\begin{prop}
Let $P:\mathsf{St}_{1}\left(\mathrm{A}\right)\to\mathbb{R}$ be a
real-valued function on the state space of system $\mathrm{A}$. If
$P$ is a monotone under one of the sets $\mathsf{RaRe}$, $\mathsf{Noisy}$
and $\mathsf{Unital}$, then it is a monotone under all the other
sets.
\end{prop}

In the light of proposition~\ref{prop:unital monotones}, this means
that \emph{all} purity monotones in \emph{all three} resource theories
are of the form $P\left(\rho\right)=f\left(\mathbf{p}\right)$, where
$f$ is a Schur-convex function, and $\mathbf{p}$ the spectrum of
$\rho$. Dually, all mixedness monotones are of the form $M\left(\rho\right)=f\left(\mathbf{p}\right)$,
where this time $f$ is a Schur-concave function.

\section{Entanglement-thermodynamics duality}

We conclude the chapter by showing that sharp theories with purification
and unrestricted reversibility exhibit a fundamental duality between
the resource theory of purity and the resource theory of pure bipartite
entanglement \cite{Chiribella-Scandolo-entanglement}, where free
operations are LOCC channels \cite{LOCC1,LOCC2,Lo-Popescu}. The content
of the duality is that a pure bipartite state is more entangled than
another if and only if the marginal states of the latter are purer
than the marginal states of the former. More formally, the duality
can be stated as follows \cite{Chiribella-Scandolo-entanglement}.
\begin{defn}
A theory satisfies the \emph{entanglement-thermodynamics duality}
if, for every pair of systems $\mathrm{A}$ and $\mathrm{B}$, and
every pair of pure states $\Phi,\Psi\in\mathsf{PurSt}_{1}\left(\mathrm{AB}\right)$
the following are equivalent:
\begin{enumerate}
\item $\Psi$ can be converted into $\Phi$ by local operations and classical
communication\footnote{Note that classical communication between the two parties can be easily
modelled as an example of classical control, allowed by Causality:
Bob chooses what to do based on the classical outcome of a previous
test performed by Alice.};
\item the marginal of $\Phi$ on system $\mathrm{A}$ can be converted into
the marginal of $\Psi$ on system $\mathrm{A}$ by a RaRe channel;
\item the marginal of $\Phi$ on system $\mathrm{B}$ can be converted into
the marginal of $\Psi$ on system $\mathrm{B}$ by a RaRe channel.
\end{enumerate}
\end{defn}

The duality can be illustrated by the diagrams
\[
\begin{CD}\Psi@>\textrm{LOCC}>>\Phi\\
@V\mathrm{tr}_{\mathrm{B}}VV@VV\mathrm{tr}_{\mathrm{B}}V\\
\rho_{\mathrm{A}}@<\textrm{RaRe}<<\rho'_{\mathrm{A}}
\end{CD}\qquad\qquad\begin{CD}\Psi@>\textrm{LOCC}>>\Phi\\
@V\mathrm{tr}_{\mathrm{A}}VV@VV\mathrm{tr}_{\mathrm{A}}V\\
\rho_{\mathrm{B}}@<\textrm{RaRe}<<\rho'_{\mathrm{B}}
\end{CD}\quad.
\]

Our earlier work \cite{Chiribella-Scandolo-entanglement} showed that
the entanglement-thermodynamics duality can be proved from four axioms:
Causality, Purity Preservation, Purification, and Local Exchangeability\textemdash the
last defined as follows.
\begin{ax}[Local Exchangeability \cite{Chiribella-Scandolo-entanglement}]
For every pair of systems $\mathrm{A}$ and $\mathrm{B}$, and for
every pure state $\Psi\in\mathsf{PurSt}_{1}\left(\mathrm{AB}\right)$
there exist two channels $\mathcal{C}\in\mathsf{Transf}\left(\mathrm{A},\mathrm{B}\right)$
and $\mathcal{D}\in\mathsf{Transf}\left(\mathrm{B},\mathrm{A}\right)$,
which in general depend on $\Psi$, such that \[ \begin{aligned}\Qcircuit @C=1em @R=.7em @!R { &\multiprepareC{1}{\Psi} & \qw \poloFantasmaCn{\rA} &\gate{\cC} & \qw \poloFantasmaCn{\rB} & \qw \\ & \pureghost{\Psi} & \qw \poloFantasmaCn{\rB} & \gate{\cD} & \qw \poloFantasmaCn{\rA} & \qw } \end{aligned} ~=\!\!\!\! \begin{aligned}\Qcircuit @C=1em @R=.7em @!R { &\multiprepareC{1}{\Psi} & \qw \poloFantasmaCn{\rA} &\multigate{1}{\tt SWAP} & \qw \poloFantasmaCn{\rB} & \qw \\ & \pureghost{\Psi} & \qw \poloFantasmaCn{\rB} & \ghost{\tt SWAP} & \qw \poloFantasmaCn{\rA} & \qw } \end{aligned} ~. \]
\end{ax}

Since Causality, Purity Preservation, and Purification are already
assumed among our axioms, proving the entanglement-thermodynamics
duality is reduced to proving the validity of Local Exchangeability.
\begin{prop}
Every sharp theory with purification and unrestricted reversibility
satisfies Local Exchangeability.
\end{prop}

\begin{proof}
Let $\Psi\in\mathsf{PurSt}_{1}\left(\mathrm{AB}\right)$ be a generic
pure state and let $\rho_{\mathrm{A}}$ and $\rho_{\mathrm{B}}$ be
its marginal states, diagonalised as $\rho_{\mathrm{A}}=\sum_{i=1}^{r}p_{i}\alpha_{i}$
and $\rho_{\mathrm{B}}=\sum_{i=1}^{r}p_{i}\beta_{i}$, where $p_{i}>0$
for all $i=1,\dots,r$, and $r\leq\min\left\lbrace d_{\mathrm{A}},d_{\mathrm{B}}\right\rbrace $.
Here we are invoking Schmidt decomposition (theorem~\ref{thm:schmidt}),
by which the marginals of a pure bipartite state have the same spectrum
up to vanishing elements. Now, we extend the set of eigenstates of
$\rho_{\mathrm{A}}$ and $\rho_{\mathrm{B}}$ to two pure maximal
sets. Without loss of generality assume $d_{\mathrm{A}}\leq d_{\mathrm{B}}$.
By Permutability, there must exist a reversible channel $\mathcal{U}\in\mathsf{Transf}\left(\mathrm{BA},\mathrm{AB}\right)$
such that\footnote{Strictly speaking $\mathrm{BA}\neq\mathrm{AB}$, whereas Permutability
refers to states of the same system. This can be easily accommodated
by inserting the $\mathtt{SWAP}$ channel suitably.} $\mathcal{U}\left(\beta_{1}\otimes\alpha_{i}\right)=\alpha_{1}\otimes\beta_{i}$
for every $i\in\left\{ 1,\ldots,d_{\mathrm{A}}\right\} $. Similarly,
there must exist a reversible channel $\mathcal{V}\in\mathsf{Transf}\left(\mathrm{BA},\mathrm{AB}\right)$
such that $\mathcal{V}\left(\beta_{i}\otimes\alpha_{1}\right)=\alpha_{i}\otimes\beta_{1}$
for every $i\in\left\{ 1,\ldots,d_{\mathrm{A}}\right\} $ At this
point, we define the pure transformations\[ \begin{aligned} \Qcircuit @C=1em @R=.7em @!R { &\qw \poloFantasmaCn{\rA} & \gate{\mathcal P} & \qw \poloFantasmaCn{\rB} &\qw } \end{aligned} ~:=\!\!\!\! \begin{aligned} \Qcircuit @C=1em @R=.7em @!R { & \prepareC{\beta_1} & \qw \poloFantasmaCn{\rB} & \multigate{1}{\mathcal U} & \qw \poloFantasmaCn{\rA} & \measureD{\alpha_1^\dag} \\ & & \qw \poloFantasmaCn{\rA} & \ghost{\mathcal U} & \qw \poloFantasmaCn{\rB} & \qw } \end{aligned}~, \]\[\begin{aligned} \Qcircuit @C=1em @R=.7em @!R { &\qw \poloFantasmaCn{\rB} & \gate{\mathcal Q} & \qw \poloFantasmaCn{\rA} &\qw } \end{aligned} ~:=\!\!\!\! \begin{aligned} \Qcircuit @C=1em @R=.7em @!R { & & \qw \poloFantasmaCn{\rB} & \multigate{1}{\mathcal V} & \qw \poloFantasmaCn{\rA} & \qw \\ & \prepareC{\alpha_1} & \qw \poloFantasmaCn{\rA} & \ghost{\mathcal V} & \qw \poloFantasmaCn{\rB} & \measureD{\beta_1^\dag} } \end{aligned}~. \]and
the pure state\[ \begin{aligned} \Qcircuit @C=1em @R=.7em @!R { & \multiprepareC{1}{\Psi'} & \qw \poloFantasmaCn{\rB} &\qw \\ & \pureghost{\Psi'} & \qw \poloFantasmaCn{\rA} & \qw }\end{aligned} ~:= \!\!\!\! \begin{aligned}\Qcircuit @C=1em @R=.7em @!R { & \multiprepareC{1}{\Psi} & \qw \poloFantasmaCn{\rA} & \gate{\mathcal P} & \qw \poloFantasmaCn{\rB} &\qw \\ & \pureghost{\Psi} & \qw \poloFantasmaCn{\rB} & \gate{\mathcal Q} & \qw \poloFantasmaCn{\rA} &\qw }\end{aligned}~, \]where
the purity of $\mathcal{P}$, $\mathcal{Q}$, and $\Psi'$ follows
from Purity Preservation. Like in the proof of proposition~\ref{prop:permstrong},
we can prove that $\mathcal{P}$ and $\mathcal{Q}$ are in fact channels,
so $u_{\mathrm{B}}\mathcal{P}=u_{\mathrm{A}}$ and $u_{\mathrm{A}}\mathcal{Q}=u_{\mathrm{B}}$.
Hence $\Psi'$ and ${\tt SWAP}\Psi$ have the same marginals. Then,
the uniqueness of purification applied to both systems $\mathrm{A}$
and $\mathrm{B}$ (viewed as purifying systems of one another) implies
that there exist two reversible channels $\mathcal{W}_{\mathrm{A}}$
and $\mathcal{W}_{\mathrm{B}}$ such that\[ \begin{aligned}\Qcircuit @C=1em @R=.7em @!R { &\multiprepareC{1}{\Psi} & \qw \poloFantasmaCn{\rA} &\multigate{1}{\tt SWAP} & \qw \poloFantasmaCn{\rB} & \qw \\ & \pureghost{\Psi} & \qw \poloFantasmaCn{\rB} & \ghost{\tt SWAP} & \qw \poloFantasmaCn{\rA} & \qw } \end{aligned}~=\!\!\!\!\begin{aligned}\Qcircuit @C=1em @R=.7em @!R { & \multiprepareC{1}{\Psi'} & \qw \poloFantasmaCn{\rB} & \gate{\mathcal W_\rB} & \qw \poloFantasmaCn{\rB} &\qw \\ & \pureghost{\Psi'} & \qw \poloFantasmaCn{\rA} & \gate{\mathcal W_\rA} & \qw \poloFantasmaCn{\rA} &\qw }\end{aligned} ~= \]\[ =\!\!\!\!\begin{aligned}\Qcircuit @C=1em @R=.7em @!R { & \multiprepareC{1}{\Psi} & \qw \poloFantasmaCn{\rA} & \gate{\mathcal P} & \qw \poloFantasmaCn{\rB} & \gate{\mathcal W_\rB} & \qw \poloFantasmaCn{\rB} &\qw \\ & \pureghost{\Psi} & \qw \poloFantasmaCn{\rB} & \gate{\mathcal Q} & \qw \poloFantasmaCn{\rA} & \gate{\mathcal W_\rA} & \qw \poloFantasmaCn{\rA} &\qw }\end{aligned}~. \]Hence,
we have shown that there exist two local \emph{pure} channels $\mathcal{C}:=\mathcal{W}_{\mathrm{B}}\mathcal{P}$
and $\mathcal{D}:=\mathcal{W}_{\mathrm{A}}\mathcal{Q}$ that reproduce
the action of the swap channel on the state $\Psi$.
\end{proof}
Note that Local Exchangeability is implemented by \emph{pure} channels
in sharp theories with purification and unrestricted reversibility.

To sum up, every sharp theory with purification and unrestricted reversibility
satisfies the entanglement-thermodynamics duality. As a consequence
of the duality, mixedness monotones, characterised at the end of the
previous section, are in one-to-one correspondence with measures of
pure bipartite entanglement. For example, Shannon-von Neumann entropy
of the marginals of a pure bipartite state can be regarded as the
\emph{entanglement entropy} \cite{Entanglement-entropy1,Entanglement-entropy2,Janzing2009},
an entropic measure of entanglement that is playing an increasingly
important role in quantum field theory \cite{Ryu1,Ryu2} and condensed
matter \cite{Area-law}. 

In passing, we also mention that the validity of Local Exchangeability
implies that every state admits a \emph{symmetric purification} \cite[theorem 3]{Chiribella-Scandolo-entanglement},
in the following sense.
\begin{defn}
Let $\rho$ be a state of system $\mathrm{A}$ and let $\Psi$ be
a pure state of $\mathrm{AA}$. We say that $\Psi$ is a \emph{symmetric
purification} of $\rho$ if\[ \begin{aligned}\Qcircuit @C=1em @R=.7em @!R { & \multiprepareC{1}{\Psi} & \qw \poloFantasmaCn{\rA} & \qw \\ & \pureghost{\Psi} & \qw \poloFantasmaCn{\rA} & \measureD{u}}\end{aligned}~=\!\!\!\!\begin{aligned}\Qcircuit @C=1em @R=.7em @!R {& \prepareC {\rho} & \qw \poloFantasmaCn{\rA} & \qw }\end{aligned}~, \]and\[ \begin{aligned}\Qcircuit @C=1em @R=.7em @!R { & \multiprepareC{1}{\Psi} & \qw \poloFantasmaCn{\rA} & \measureD{u} \\ &\pureghost{\Psi} &\qw \poloFantasmaCn{\rA} &\qw}\end{aligned}~=\!\!\!\!\begin{aligned}\Qcircuit @C=1em @R=.7em @!R {& \prepareC {\rho} & \qw \poloFantasmaCn{\rA} & \qw }\end{aligned}~. \] 
\end{defn}

With the above notation, we have the following.
\begin{prop}
In every sharp theory with purification and unrestricted reversibility,
every state admits a symmetric purification.
\end{prop}

\chapter{Conclusions and outlook\label{chap:Conclusions}}

In this DPhil thesis, the first doctoral thesis on this topic, we
studied the foundations of thermodynamics and statistical mechanics
using the tool of general probabilistic theories. In particular, we
focused on microcanonical thermodynamics, describing systems where
the energy is known and fixed. For the first time, we extended the
microcanonical framework to arbitrary physical theories, beyond the
known and well-studied cases of classical and quantum theory. Being
a general paradigm, one expects thermodynamics to be valid in all
physical theories, yet we discovered that, in order to have a physically
meaningful microcanonical thermodynamics, we need to impose two requirements
on the underlying theory. These two requirements are:
\begin{enumerate}
\item for every system there is a unique invariant probability measure;\label{enu:MC1}
\item the product of microcanonical states is still a microcanonical state.\label{enu:MC2}
\end{enumerate}
Requirement~\ref{enu:MC1} implies that for every system, once the
energy is fixed, the microcanonical state is uniquely determined.
This is because that invariant probability measure is used to define
the microcanonical state. Requirement~\ref{enu:MC2} expresses the
stability of the equilibrium state under parallel composition, and
it is similar in spirit to the notion of \emph{complete passivity}
for quantum thermal states \cite{Pusz,Lenard,Sparaciari-passive}.
In other words, there is no ``activation process'' when composing
equilibrium states that could bring us out of equilibrium (and therefore
extract work).

Specifically, requirement~\ref{enu:MC1} is fully equivalent to one
of the conditions (transitivity) that appeared in several reconstructions
of quantum theory from first principles: the fact that for every pair
of pure states there exists a reversible channel connecting them \cite{Hardy-informational-1,Chiribella-informational,Brukner,Masanes-physical-derivation,Hardy-informational-2,Masanes+all}.
Therefore, our results offer a new perspective on transitivity, this
time from thermodynamics: transitivity is a necessary condition to
have a well-posed microcanonical thermodynamics. Since thermodynamics
is not considered a fundamental theory, but rather emergent from an
underlying theory through the paradigm of statistical mechanics, we
usually expect that it is the underlying theory to constrain the thermodynamic
behaviour. Instead, in our case we found the opposite: the reasonable
thermodynamic desiderata of requirements~\ref{enu:MC1} and \ref{enu:MC2}
constrain the underlying theory.

For theories satisfying requirements~\ref{enu:MC1} and \ref{enu:MC2}
we set up a resource-theoretic treatment, where the microcanonical
state was taken as the only free state, and we chose three different
sets of free operations, with similar definitions to quantum theory:
\begin{enumerate}
\item random reversible (RaRe) channels;
\item noisy operations;
\item unital channels.
\end{enumerate}
We studied these three resource theories in a class of physical theories,
sharp theories with purification, which we propose as an axiomatic
foundation for statistical mechanics. The fundamental feature of these
theories is that they admit a level of description where all states
are pure, and all evolutions reversible. By this we mean that every
mixed state can be written as the marginal of a pure state of a larger
system, and that every channel can be written as discarding one of
the outputs of a bipartite reversible channel. Therefore, in these
theories ``impurity'', partial information (viz.\ probabilistic
mixtures), and irreversibility arise because of discarding. There
is no need for the presence of an external agent who assigns probabilities
subjectively or performs some coarse-graining. For this reason, sharp
theories with purification are particularly suitable as candidate
theories to solve the remaining tension between the pure and deterministic
character of the fundamental dynamics, and the subjectivity of statistical
ensembles. In these theories, mixed states, in particular thermal
states, and their associated probabilities arise because in the thermodynamic
description one is tracing out some degrees of freedom, typically
those of the environment. This opens the way to the derivation of
equilibrium states from typicality arguments based on entanglement
\cite{Popescu-Short-Winter,Canonical-typicality,Dahlsten-typicality,Muller-blackhole},
since it is not hard to see that every sharp theory with purification
\emph{must} have entangled states \cite{Chiribella-Scandolo-entanglement}.

With these properties, sharp theories with purification appear to
be reasonably close to quantum theory. Indeed the convex examples
known so far \cite{Stuckelberg,Araki-real,Wootters-real,Hardy-real,Fermionic1,Fermionic2,TowardsThermo,Purity}
are variations on quantum theory, obtained by imposing superselection
rules \cite{Fermionic1,Fermionic2,TowardsThermo,Purity}, or by considering
real, instead of complex, amplitudes \cite{Stuckelberg,Araki-real,Wootters-real,Hardy-real}.
In fact, it is possible to prove that all sharp theories with purification
are Euclidean Jordan algebras, and that they can exhibit at most second-order
interference \cite{HOI,Sorkin1,Sorkin2,Barnum-interference,Lee-Selby-interference}.
Besides these examples, we showed that, quite surprisingly, even classical
theory can be extended to a sharp theory with purification. In this
extension, at the level of single systems, classical systems look
perfectly classical, and have all the properties of classical systems.
What changes is the way they compose, because we need to have entangled
states in composite systems: the composition of two classical systems
is no longer a classical system.

After introducing sharp theories with purification, we studied their
properties in relation to thermodynamics. The first is a state-effect
duality, by which with every normalised pure state we can associate
a unique normalised pure effect (and vice versa) that occurs with
unit probability on that pure state. In \cite{HOI} this was proved
to be the stepping stone for the definition of the dagger of all transformations.
The second important property is that all states can be diagonalised,
i.e.\ written as a convex combination of perfectly distinguishable
pure states, with unique coefficients, the eigenvalues of the state.
Finally, sharp theories with purification satisfy requirements~\ref{enu:MC1}
and \ref{enu:MC2}, so they admit a well-defined microcanonical thermodynamics.
In these theories it is therefore possible to introduce the three
resource theories.

We showed that the sets of free operations obey the same inclusion
relations as in quantum theory: the set of RaRe channels is included
in the set of noisy operations, which is in turn included in the set
of unital channels. In addition, the convertibility of states under
unital channels is fully described by majorisation on the spectra
of states, thanks to the diagonalisation theorem. This allowed us
to find mixedness monotones aplenty: they are Schur-concave functions
on the spectrum of states. We were able to prove that, for a large
class of them, their definition on the spectrum coincides with two
definitions given in the GPT literature \cite{Entropy-Barnum,Entropy-Short,Entropy-Kimura},
based on pure measurements and pure preparations respectively.

Among mixedness monotones, we focused on Shannon-von Neumann entropy,
and we used it to define the generalised Gibbs states through Jaynes'
maximum entropy principle \cite{Jaynes1,Jaynes2}: we fixed the expectation
value of an energy observable, and determined the state that maximises
the Shannon-von Neumann entropy with that constraint, promoting it
to an equilibrium state. This was the only part of the thesis where
we departed from microcanonical thermodynamics and we explored the
role of temperature. This was instrumental in proving an operational
version of Landauer's principle in sharp theories with purification,
where we linked the reduction of entropy in the system to the heat
dissipated into the environment, following the approach by Reeb and
Wolf \cite{Reeb-Wolf}.

We showed that, if we want the three resource theories to be equivalent
and define the same preorder on states, the axioms of sharp theories
with purification are not enough, and we \emph{must} add a further
principle, unrestricted reversibility, expressing the richness of
the reversible dynamics of the theory. Again, from a thermodynamic
requirement we derived a constraint on the underlying theory. This
constraint restricts the set of allowed theories even further. For
instance, doubled quantum theory and the coherent composites of extended
classical theory are ruled out, and we get even closer to quantum
theory. This could be an indication that quantum theory is eventually
the only theory supporting a physically sensible thermodynamics.

From unrestricted reversibility, we proved that the three resource
theories obey a duality with the resource theory of entanglement \cite{Chiribella-Scandolo-entanglement}.
This connects the entanglement of pure bipartite states with the purity
of their marginals: a pure state is more entangled than another if
and only if the marginals of the latter are purer than the marginals
of the former. In this way entanglement becomes a fertile ground for
the foundations of thermodynamics, at least in the microcanonical
setting. 

The results of this thesis are only the surface of a deep operational
structure, where thermodynamic and information-theoretic features
are interwoven at the level of fundamental principles. The work initiated
here clearly still has a lot of potentialities for further exploration.
For instance, we still lack an operational derivation of strong subadditivity
of Shannon-von Neumann entropy \cite{Strong-subadditivity}, and of
the monotonicity of the relative entropy under the action of channels
\cite{Petz}. The proof of these results is notoriously difficult
even in ordinary quantum theory, but the motivation is extremely strong,
for they are the key to the derivation of the second law of thermodynamics
\cite{Preskill} and of its quantum generalisations \cite{2ndlaws},
not to mention the consequences for information processing. An operational
derivation of these results will shed a new light on quantum theory
too, highlighting the principles leading to strong subadditivity,
which at present are hidden behind the technical character of the
existing proofs.

Another area of future research is the completion of the characterisation
of microcanonical thermodynamics in sharp theories with purification.
The aspects related to the thermodynamic limit are currently under
investigation \cite{East-article}, but it is worth exploring other
sides, such as catalysis, and the role of entropies in the single-shot
work extraction \cite{Nicole,Dahlsten-extractable}.

Clearly, the whole area of canonical thermodynamics, viz.\ for systems
at a fixed \emph{temperature}, rather than fixed energy, in general
probabilistic theories is still largely unexplored. In this thesis,
we briefly touched on it with the derivation of the generalised Gibbs
states using the maximum entropy principle. However, we did not study
the possible resource theories one can introduce in this setting \cite{Athermality1,Horodecki-Oppenheim-2,Athermality2,Gibbs-preserving-maps},
nor did we study the existence of the canonical state in full generality,
viz.\ without assuming the axioms of sharp theories with purification.
To this end, a promising way to derive the canonical state is to harness
the idea of complete passivity \cite{Pusz,Lenard}, and enforce it
in arbitrary physical theories. Alternatively, we can pursue the derivation
of equilibrium states (including the canonical ones) in sharp theories
with purification based on typicality and entanglement, along the
lines of \cite{Popescu-Short-Winter,Dahlsten-typicality,Muller-blackhole}.
Results in this direction would bring further evidence that sharp
theories with purification provide the appropriate ground for the
construction of a well-founded statistical mechanics.

As we noted, the essence of sharp theories with purification is entanglement
and the possibility of purifying states and transformations. What
about classical theory? Is classical thermodynamics well founded?
The answer to this question is not obvious nor easy. Clearly, our
experience and the physical results, both theoretical and experimental,
tell us that this is the case. Although we do it all the time, and
it works very well, in classical theory there is no principle that
\emph{formally} justifies enlarging an open system to recover the
isolated picture. The missing principle we are looking for is Purification
\cite{Chiribella-purification}, which is quintessentially thermodynamic.
Indeed, when in thermodynamics we model an open system as part of
an isolated system, where the other part has been discarded or neglected,
it is precisely what happens in a theory where Purification is at
work. There, all non-reversible channels can be seen as a reversible
evolution in a larger system (the isolated system, where evolution
is assumed to be reversible), part of which has been discarded. Instead
in classical theory, whatever is mixed and irreversible, it stays
so, irrespective of how much we enlarge the system. This is precisely
due to the lack of entanglement. Then why does thermodynamics work
in classical theory?

A possible answer might be that classical thermodynamics works because
Nature is ultimately quantum, where Purification holds. Therefore
the formal underpinnings of classical thermodynamics may be found
in quantum thermodynamics, of which it is a sub-theory. Another possible
answer might come from our results: classical theory admits an extension
to a sharp theory with purification with actual classical systems
among its systems. This extended classical theory offers a new possibility
for the foundations of classical statistical mechanics, allowing one
to view classical ensembles, at least from a formal point of view,
as arising from joint pure entangled states. At the same time, it
allows us to export the results and the proof techniques of sharp
theories with purification to classical theory. This motivates the
following conjecture.
\begin{conjecture*}
Every theory with a ``well-behaved'' thermodynamics can be extended
to a sharp theory with purification.
\end{conjecture*}
As we currently lack a formal definition of ``well-behaved'' thermodynamics,
our conjecture is not a mathematical statement for the time being,
but rather an open research programme. Addressing this programme directly
will be rather hard, and will mean rigorously formulating a set of
desiderata about thermodynamics, from which to \emph{derive} the requirements
that the underlying physical theory should meet. An example of thermodynamic
desiderata is provided by Lieb and Yngvason's axioms \cite{Lieb-Yngvason},
recently revisited from a quantum information perspective \cite{Weilenmann1,Weilenmann2},
which capture the fundamental structures underpinning the second law
of thermodynamics. Connecting general probabilistic theories with
Lieb and Yngvason's desiderata is a promising route to approach our
conjecture, and produce new results in the axiomatic foundations of
thermodynamics.

\chapter*{\markboth{ACKNOWLEDGEMENTS}{ACKNOWLEDGEMENTS}Acknowledgements}

I am really grateful to a great number of people in these three years
as a DPhil student. They have been great years, in terms of scientific
and personal growth, and I can truly say I have learnt a lot during
this time. All the people I have met, in some way or another, have
taught me something, therefore they all deserve a thank you. Here
I will mention those who deserve a special acknowledgement. I am rather
sure to forget somebody, so I will apologise in advance to those I
forgot to mention.

My first and biggest thanks goes to my supervisor Prof.\ Jon Barrett,
a brilliant and sharp scientist, always keen to discuss topics in
the foundations of quantum theory and, more generally, of physics.
He has guided me in my DPhil research, and he has granted me the freedom
to pursue my independent lines of research, always available to discuss
my results and to provide feedback or interesting insights. Without
his support and guidance, all this research would not have been possible.

Another big thanks goes to my collaborator and former supervisor Prof.\ Giulio
Chiribella, with whom I did most of the research presented here. It
has been a great pleasure to work and have exciting scientific discussions
with him. I have learnt really a lot from him, and he has always provided
very good advice and scientific support, as well as excellent guidance
in the research project of this thesis.

I really want to thank my examiners, Prof.\ Sam Staton and Dr.\ Markus
Müller, for the careful reading of my thesis and the insightful and
stimulating comments they made, which led to a fruitful scientific
discussion during my viva.

Then Prof.\ Bob Coecke deserves a special mention, for co-leading
the Quantum Group at the Department of Computer Science, and for creating
such a nice, relaxed, and friendly atmosphere. This is the ideal environment
to foster new ideas and collaborations. I have had the opportunity
to collaborate with him various times, on a scientific article, on
an MSc thesis supervision, and teaching the classes for his course.
They all were very nice experiences.

Of all the other people in the Quantum Group, Dr.\ Matty Hoban deserves
a special thanks. He has helped me a lot in several ways, sharing
his knowledge and expertise, or simply providing guidance and encouragement
in the less bright times one inevitably encounters in one's DPhil.
I am indebted to him for his patience and support.

In fact all the Quantum Group deserve a thank you, starting from the
faculty members, including Prof.\ Samson Abramsky and Dr.\ Jamie
Vicary, to postdocs, Dr.\ Nathan Walk, Dr.\ Ognyan Oreshkov, Dr.\ Niel
de Beaudrap, Dr.\ Matt Pusey, Dr.\ Stefano Gogioso, down to my fellow
DPhil students Robin Lorenz, Subhayan Roy Moulik, Vojt\v{e}ch Havlí\v{c}ek,
Marietta Stasinou, Fatimah Ahmadi, David Reutter, Fabrizio Genovese.
I want also to mention some former fellow DPhil students, who have
successfully graduated: Dr.\ Ciarán Lee, Dr.\ John Selby, and Dr.\ John-Mark
Allen. It has been a pleasure to work with them, and I have learnt
a lot from them all.

Being at the University of Oxford I have had the opportunity to interact
with brilliant researchers outside the my group. I would like to mention
Dr.\ David Jennings, and Prof.\ Harvey Brown, for the interesting
and inspiring discussions I have had with them.

I would like to thank all my collaborators, from Oxford and other
institutions: Prof.\ Giulio Chiribella, Prof.\ Bob Coecke, Prof.\ Howard
Barnum, Prof.\ Jonathan Oppenheim, Prof.\ Pawe\l{} Horodecki, Dr.\ Ciarán
Lee, Dr.\ John Selby, Dr.\ Stefano Gogioso, Dr.\ Lídia del Rio,
Dr.\ Philippe Faist, Dr.\ Roberto Salazar, Dr.\ Jarek Jorbicz,
and Carlo Sparaciari. A particular thanks goes to Prof.\ Jonathan
Oppenheim and Prof.\ Pawe\l{} Horodecki, who invited me to visit
their institutions. Prof.\ Pawe\l{} Horodecki deserves a special
mention, for being always so nice and friendly, available and willing
to help, despite being very busy as one of the world-leading experts
in quantum information theory. It has been a sheer pleasure to work
with him and learn from his brilliant ideas.

I wish to thank also Perimeter Institute, especially Prof.\ Lucien
Hardy and Prof.\ Rob Spekkens, for hosting me during my visits there;
and the Paris Centre for Quantum Computing, especially Dr.\ Alexei
Grinbaum, for inviting me to a QuPa meeting in May 2018. I had a great
and fruitful time in these visits, and I benefited a lot from inspiring
conversations.

I want to conclude the academic part of my acknowledgements by mentioning
Prof.\ Paolo Ciatti, Prof.\ Pieralberto Marchetti, Prof.\ Kurt
Lechner, Prof.\ Antonio Saggion, Prof.\ Paolo Villoresi at the University
of Padua, and Dr.\ Francesco Sorge for their support, and for passing
their passion for theoretical physics, quantum mechanics, thermodynamics,
and the foundations of physics on to me.

I feel very grateful for the help and support from many (old and new)
friends during my time at Oxford. First of all, I would like to thank
Aris Filos-Ratsikas, who has been an amazing friend in these three
years; I had a great time with him, even in less bright moments, and
he has always given me very good advice. I wish to thank some other
friends for their time together, it would be too long to thank each
of them separately, therefore I will list them here: Stefano Steffenoni,
Enzo Arnosti, Alex Kavvos, Philip Lazos, Ninad Rajagopal, Stefano
Gogioso, Marietta Stasinou, Vanessa Monteventi, Rosalba García Millán,
Sheng Peng, George Vere, Gabriele Damiani, and Stefania Monterisi.
Thanks for making these DPhil years better and lighter.

Clearly, all this would not have been possible without the love and
support of my Mum and Dad, to whom this thesis is dedicated. They
have always supported me and my passions throughout my life, and have
always encouraged me to work hard and do my best to achieve my goals.
They have been my guidance and inspiration.

To conclude, I want to thank Miriam Saccon for her love, support,
and patience during my DPhil and the writing of this thesis; without
her this would have been a much harder time. Despite not being her
area of expertise, she did a very good proofreading of earlier drafts
of this thesis, and I am very grateful to her for this.

Finally, I acknowledge the support of the EPSRC doctoral training
grant 1652538 ``Information-theoretic foundations of quantum thermodynamics'',
and of the Oxford-Google DeepMind graduate scholarship.

\bibliographystyle{hustthesis}
\bibliography{bibliographyPhD}

\end{document}